\documentclass[11pt]{article}
\usepackage{amsmath,amscd,amsthm,amssymb,amsfonts,mathtools}

\usepackage{ytableau}
\usepackage{nicematrix}
\usepackage{mathrsfs}
\usepackage{graphics}
\usepackage[new]{old-arrows}
\usepackage{hyperref}
\usepackage{epsfig}
\usepackage[utf8]{inputenc}
\usepackage[T1]{fontenc}
\usepackage{xfrac} 
\usepackage{stmaryrd}
\usepackage{bbold}
\usepackage{tikz}
\usepackage{mathdots}
\usepackage{yhmath}
\usepackage{cancel}
\usepackage{color}
\usepackage{siunitx}
\usepackage{array}
\usepackage{authblk}
\usepackage{multirow}
\usepackage{gensymb}
\usepackage{tabularx}
\usepackage{extarrows}
\usepackage{booktabs}
\usetikzlibrary{decorations.markings}
\usetikzlibrary{fadings}
\usetikzlibrary{patterns}
\usetikzlibrary{shadows.blur}
\usetikzlibrary{shapes}
\usepackage{graphicx}
\usepackage{xcolor}
\usepackage{float}
\usepackage[top=1.5in,bottom=1in,left=1.25in,right=1.25in]{geometry}
\usepackage{tikz-cd}
\usepackage{braket}

\usepackage{concmath}

\DeclareMathOperator{\Hom}{Hom}

\DeclareMathOperator{\Pic}{Pic}
\DeclareMathOperator{\Gr}{Gr}
\DeclareMathOperator{\bGr}{\overline{Gr}}

\DeclareMathOperator{\gl}{\mathfrak{gl}}

\DeclareMathOperator{\SL}{SL}
\DeclareMathOperator{\GL}{GL}
\DeclareMathOperator{\PGL}{PGL}

\DeclareMathOperator{\Aut}{Aut}

\DeclareMathOperator{\End}{End}

\DeclareMathOperator{\Spec}{Spec}

\DeclareMathOperator{\Hilb}{Hilb}
\DeclareMathOperator{\Quot}{Quot}

\DeclareMathOperator{\ch}{ch}

\newcommand*{\cM}{\mathcal{M}}

\newcommand*{\bA}{\mathbb{A}}
\newcommand*{\bC}{\mathbb{C}}
\newcommand*{\bN}{\mathbb{N}}
\newcommand*{\bZ}{\mathbb{Z}}
\newcommand*{\bR}{\mathbb{R}}
\newcommand*{\bP}{\mathbb{P}}

\newtheorem{innercustomthm}{Theorem}
\newenvironment{customthm}[1]
  {\renewcommand\theinnercustomthm{#1}\innercustomthm}
  {\endinnercustomthm}

\theoremstyle{plain}
  \newtheorem{theorem}{Theorem}[section]
  \newtheorem{proposition}[theorem]{Proposition}
  \newtheorem{lemma}[theorem]{Lemma}
  \newtheorem{corollary}[theorem]{Corollary}

\theoremstyle{definition}
  \newtheorem{definition}{Definition}[section]
  
  \newtheorem{notation}{Notation}

\theoremstyle{remark}
  \newtheorem{example}{Example}[section]
  \newtheorem{remark}{Remark}[section]

\numberwithin{equation}{section}

\title{On the Hilbert Space of the Chern-Simons Matrix Model, Deformed Double Current Algebra Action, and the Conformal Limit}

\author[1]{Sen Hu\thanks{shu@ustc.edu.cn}}
\author[2]{Si Li\thanks{sili@mail.tsinghua.edu.cn}}
\author[1]{Dongheng Ye\thanks{ydrj163@mail.ustc.edu.cn}}
\author[3]{Yehao Zhou\thanks{yehao.zhou@ipmu.jp}}

\affil[1]{School of Mathematical Sciences, University of Science and Technology of China, Hefei, China}
\affil[2]{Department of Mathematical Sciences and Yau Mathematical Sciences Center, Tsinghua University, Beijing, China}
\affil[3]{Kavli Institute for the Physics and Mathematics of the Universe (WPI), University of Tokyo, Kashiwa, Chiba, Japan}

\date{}

\begin{document}

\NiceMatrixOptions{code-for-first-row = \scriptstyle,code-for-first-col = \scriptstyle }

\maketitle

\begin{abstract}
A Chern-Simons matrix model was proposed by Dorey, Tong, and Turner to describe non-Abelian fractional quantum Hall effect. In this paper we study the Hilbert space of the Chern-Simons matrix model from a geometric quantization point of view. We show that the Hilbert space of the Chern-Simons matrix model can be identified with the space of sections of a line bundle on the quiver variety associated to a framed Jordan quiver. We compute the character of the Hilbert space using localization technique. Using a natural isomorphism between vortex moduli space and a Beilinson-Drinfeld Schubert variety, we prove that the ground states wave functions are flat sections of a bundle of conformal blocks associated to a WZW model. In particular they solve a Knizhnik-Zamolodchikov equation. We show that there exists a natural action of the deformed double current algebra (DDCA) on the Hilbert space, moreover the action is irreducible.

We define and study the conformal limit of the Chern-Simons matrix model. We show that the conformal limit of the Hilbert space is an irreducible integrable module of $\widehat{\mathfrak{gl}}(n)$ with level identified with the matrix model level. Moreover, we prove that $\widehat{\mathfrak{gl}}(n)$ generators can be obtained from scaling limits of matrix model operators, which settles a conjecture of Dorey-Tong-Turner. The key to the proof is the construction of a Yangian $Y(\mathfrak{gl}_n)$ action on the conformal limit of the Hilbert space, which we expect to be equivalent to the $Y(\mathfrak{gl}_n)$ action on the integrable $\widehat{\mathfrak{gl}}(n)$ modules constructed by Uglov. We also characterize eigenvectors and eigenvalues of the matrix model Hilbert space with respect to a maximal commutative subalgebra of Yangian.

\end{abstract}

{\tableofcontents}

\section{Introduction}
Susskind suggested that the hydrodynamic properties of the quantum Hall fluid are captured by a non-commutative Chern-Simons theory \cite{Susskind}. Polychronakos proposed a matrix model as a regularization of the aforementioned non-commutative Chern-Simons theory \cite{Poly,polychronakos2006physics}. Quantization of the matrix models were constructed and the Laughlin wave function are reproduced in \cite{Poly, Hellerman-Raamsdonk, Karabali-Sakita}. 

A generalization of the above matrix model describing a class of non-Abelian quantum Hall states was introduced by Dorey, Tong, and Turner in \cite{Dorey-Tong-Turner}. It is shown in \cite{Dorey-Tong-Turner} that this model describes the microscopic dynamics of $N$ vortices in the $\frac{\mathrm{U}(1)_{(k+n)n}\times \mathrm{SU}(n)_k}{\mathbb Z_n}$ Chern-Simons theory, which generalizes the previous work on the $n=1$ case \cite{morariu2001finite}. From now on the terminology ``Chern-Simons matrix model'' refers to such model, which will be reviewed in Section \ref{sec:model}. Canonical quantization of the Chern-Simons matrix model results in a Hilbert space which will be denoted by $\mathcal H_N(n,k)$.

Recently, Bourgine and Matsuo \cite{bourgine2024calogero} pointed out the relationship between Chern-Simons matrix model and a higher spin generalization of Calogero model. The phase space of Chern-Simons matrix model is a Lagrangian inside the Nakajima quiver variety of the Jordan quiver, so the geometric quantization procedure tells us that the quantized ring of function on the Nakajima quiver variety acts on the Hilbert space of the Chern-Simons matrix model. This is the ``Calogero representation'' studied by Gaiotto, Rapčák, and the last author in \cite{Gaiotto-Rapcek-Zhou}. We note that relation between Nakajima quiver varieties of the Jordan quiver and (spin) Calogero models \cite{gibbons1984generalisation} was extensively studied \cite{bielawski2011symplectic,tacchella2015family}. There are also relations between Nakajima quiver varieties of the cyclic quiver quivers and generalizations of spin Calogero systems \cite{chalykh2017kp}. Relations between quantized multiplicative quiver varieties \cite{jordan2014quantized} and (generalizations of) spin Ruijsenaars–Schneider systems were explored in \cite{chalykh2017multiplicative,chalykh2020hamiltonian,fairon2019spin,fairon2021integrable}. It hints that a super Chern-Simons matrix model is related to quantization of a quiver super-variety \cite{tong2015quantum,okazaki2018matrix}.

The $N\to \infty$ limit of the Chern-Simons matrix model is expected to capture the physics of $\mathrm{SU}(n)_k$ WZW model \cite{dorey2016matrix}. It was proposed in \cite{dorey2016matrix} that there should be an isomorphism
\begin{align}\label{eq:limit of hilbert space_intro}
    `` \lim_{M\to \infty}" \mathcal H_{nM+r}(n,k)\cong L_{k\varpi_r}(\widehat{\mathfrak{sl}}(n)_k)\otimes \text{Fock space of }\widehat{\gl}(1),
\end{align}
where $r\in \{0,1,\cdots,n-1\}$, and $\varpi_r$ is the $r$-th fundamental weight of $\mathfrak{sl}_n$, and $L_{k\varpi_r}(\widehat{\mathfrak{sl}}(n)_k)$ is the irreducible integrable module of $\widehat{\mathfrak{sl}}(n)_k$ with highest weight $k\varpi_r$. Moreover \eqref{eq:limit of hilbert space_intro} should  be compatible with the natural graded $\gl_n$ module structures on two sides. The limit is dressed with quotation mark because it was not defined in \cite{dorey2016matrix}, rather \eqref{eq:limit of hilbert space_intro} was proposed as a result of comparing characters of two sides.

Moreover, it was argued in \cite{dorey2016matrix} that $\widehat{\mathfrak{sl}}(n)$ generators can be obtained as certain scaling limit of matrix model operators. Namely if we set
\begin{align}
    \mathcal J^a_{b,m}=\left(\frac{n}{(n+k)N}\right)^{\frac{|m|}{2}}\cdot \begin{cases}
        \varphi^{\dag a}(Z)^m\varphi_b-\frac{\delta^a_b}{n}\varphi^{\dag c}(Z)^m\varphi_c &,\text{ if }m> 0,\\
        \varphi^{\dag a}(Z^\dag)^{-m}\varphi_b-\frac{\delta^a_b}{n}\varphi^{\dag c}(Z^\dag)^{-m}\varphi_c &,\text{ if }m\le 0,\\
    \end{cases}
\end{align}
where $(Z,Z^\dag,\varphi,\varphi^\dag)$ are quantized matrix model operators (see Section \ref{sec:Quantization and Hilbert Space}), then $\mathcal J^a_{b,m}$ is conjectured to satisfy the $\widehat{\mathfrak{sl}}(n)_k$ commutation relation asymptotically:
\begin{align}\label{eq:asymptotic hat sl(n)_intro}
    [{\mathcal J}^a_{b,m},{\mathcal J}^c_{d,l}]=\delta^c_b {\mathcal J}^a_{d,m+l}-\delta^a_d {\mathcal J}^c_{b,m+l}+km\delta_{m+l,0}\left(\delta^a_d\delta^b_c-\frac{1}{n}\delta^a_b\delta^c_d\right)+O(1/N).
\end{align}
Here we say a sequence of operators $\{{\mathcal A}_N\in \End(\mathcal H_N(n,k))\}$ is of order $O(1/N^h)$ if for every $E\in \mathbb N$ there exists a constant $C_E>0$ such that
\begin{align}
    \rVert {\mathcal A}_N(v)\rVert_N \le C_E N^{-h} \rVert v\rVert_N\text{ holds for all }v\in \mathcal H_N(n,k)_{\le E}.
\end{align}
Here $\mathcal H_N(n,k)_{\le E}$ denotes the subspace of states which has energy at most $E$ above the ground states (see Definition \ref{def:energy grading}), and $\rVert \cdot\rVert_N$ is the natural norm on $\mathcal H_N(n,k)$ (see Section \ref{subsec:Hermitian inner product}). Evidence of \eqref{eq:asymptotic hat sl(n)_intro} was provided in \cite{Dorey-Tong-Turner}, but a full proof is still lacking, due to the difficulty of proving Identity 1 and Identity 2 in \cite{Dorey-Tong-Turner} in the quantum theory (semi-classical analysis was given there). In \cite{hu2023quantum}, it was conjectured that $\widehat{\gl}(1)$ generators can also be obtained by the scaling limit of matrix model operators. Namely, if we set 
\begin{align}
    \mathcal B_m:=\left(\frac{n}{(n+k)N}\right)^{\frac{|m|}{2}}\cdot \begin{cases}
        \frac{1}{k}\varphi^{\dag a}(Z)^m\varphi_a &,\text{ if }m< 0,\\
        \frac{1}{k}\varphi^{\dag a}(Z^\dag)^{-m}\varphi_a &,\text{ if }m\le 0,\\
    \end{cases}
\end{align}
then it is conjectured in \cite{hu2023quantum} that
\begin{align}\label{eq:asymptotic hat gl(1)_intro}
    [\mathcal B_m,\mathcal J^a_{b,l}]=O(1/N),\;\text{ and }\; [\mathcal B_m,\mathcal B_l]=\frac{n}{n+k}m\delta_{m+l,0}+O(1/N).
\end{align}
Note that $\mathcal B_0=kN$ by the moment map equation \eqref{eq:constraint}, in particular it is central. In the case of $n=1$, a proof of \eqref{eq:asymptotic hat gl(1)_intro} was given in \cite{hu2023quantum}. However, the proof in \cite{hu2023quantum} used a particular choice of basis of $\mathcal H_N(1,k)$ that does not obviously generalize to arbitrary $\mathcal H_N(n,k)$. 

In this paper we will settle the conjectures \eqref{eq:limit of hilbert space_intro}, \eqref{eq:asymptotic hat sl(n)_intro}, and \eqref{eq:asymptotic hat gl(1)_intro}, following two observations: one made by Bourgine and Matsuo in \cite{bourgine2024calogero} about the level-rank duality relating $\mathcal H_{N}(n,k)$ and $\mathcal H_{kN}(kn,1)$, another one made by Gaiotto, Rapčák, and the last author in \cite{Gaiotto-Rapcek-Zhou} about the deformed double current algebra action on $\mathcal H_N(n,k)$. Although proving conjectures \eqref{eq:limit of hilbert space_intro}, \eqref{eq:asymptotic hat sl(n)_intro}, and \eqref{eq:asymptotic hat gl(1)_intro} is the motivation of this study, it is not the only goal in this paper. Along the way paving to the proof, we will also diverge to the discussions of other aspects of the Chern-Simons matrix model and its conformal limit, for example the Knizhnik-Zamolodchikov equation that ground states solve, and the Gelfand-Tsetlin bases of the Hilbert space, and relation to higher spin generalization of Calogero-Sutherland model.

\subsection{Techniques and outline}
We will review the definition of a $\mathrm{U}(N)$ Chern-Simons matrix model of level $k$ with $n$ non-relativistic matter contents in Section \ref{sec:model}. The phase space $\cM(N,n)$ is studied in \ref{sec:phase space}, and we will see that $\cM(N,n)$ has three equivalent characterizations: as a quiver variety of the Jordan quiver:
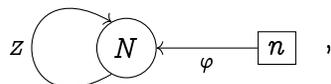
\begin{equation*}
\begin{tikzpicture}[x={(2cm,0cm)}, y={(0cm,2cm)}, baseline=0cm]
  \node[draw,circle,fill=white] (Gauge) at (0,0) {$N$};
  \node[draw,rectangle,fill=white] (Framing) at (1,0) {$n$};
  \node (Z) at (-.73,0) {\scriptsize $Z$};
 \draw[<-] (Gauge.0) -- (Framing.180) node[midway,below] {\scriptsize $\varphi$};

  \draw[->,looseness=6] (Gauge.240) to[out=210,in=150] (Gauge.120);

\end{tikzpicture}\quad ,
\end{equation*}
as a Quot scheme $\mathrm{Quot}^{N}({\mathcal O}^{\oplus n}_{{\bA}^{1}})$, and as a Beilinson-Drinfeld Schubert variety $\bGr^{\omega_1,\cdots,\omega_1}_{\GL_n,\bA^{(N)}}$.

\bigskip In Section \ref{sec:Quantization and Hilbert Space}, we study the quantization of the Chern-Simons matrix model, which produces a Hilbert space $\mathcal H_N(n,k)$. $\mathcal H_{N}(n,k)$ is the subspace of the polynomial ring $\bC[Z^\dag,\varphi^\dag]$ in variables $\{Z^{\dag i}_j,\varphi^{\dag a}_j\:|\: 1\le i,j\le N, 1\le a\le n\}$ which consists of $f\in \bC[Z^\dag,\varphi^\dag]$ that satisfies the moment map constraints:
\begin{align*}
    (\mu^i_j+k\delta^i_j)f=0, \quad \mu^i_j=Z^{\dag i}_l Z^l_j-Z^{\dag l}_j Z^i_l- \varphi^{\dag a}_j \varphi^i_a.
\end{align*}
Here $\mu^i_j$ generates infinitesimal $\GL_N$ action on $\bC[Z^\dag,\varphi^\dag]$, such that $Z^\dag$ transforms as adjoint representation and $\varphi^\dag$ transforms as dual vector representation, and $\mathcal H_{N}(n,k)$ is the space of semi-invariants $\bC[Z^\dag,\varphi^\dag]^{\GL_N,-k}$. In Proposition \ref{prop: H(N,n) = global section} we show that 
\begin{align*}
    \mathcal H_{N}(n,k)\cong \Gamma(\cM(N,n),\mathcal L_{\det}^{\otimes k}),
\end{align*}
where $\mathcal L_{\det}$ is the determinant line bundle on $\cM(N,n)$. $\mathcal L_{\det}$ is a $\GL_n\times\bC^{\times}_q$-equivariant line bundle on $\cM(N,n)$, so the equivariant Euler characteristic of $\mathcal L_{\det}^{\otimes k}$ is defined, and we denote it by $\chi_{q,\mathbf a}(\cM(N,n),\mathcal L_{\det}^{\otimes k})$, where $\mathbf a=(\mathbf a_1,\cdots,\mathbf a_n)$ (resp. $q$) are the equivariant variables of $\GL_n$ (resp. $\bC^{\times}_q$). By a standard cohomology-vanishing argument (Lemma \ref{lem: pushforward of line bundle}), $\chi_{q,\mathbf a}(\cM(N,n),\mathcal L_{\det}^{\otimes k})$ is equal to the character of zeroth cohomology of $\mathcal L_{\det}^{\otimes k}$, i.e. $\mathcal H_{N}(n,k)$. The character of $\mathcal H_{N}(n,k)$ was computed in \cite{dorey2016matrix} using contour integrals. In Section \ref{sec:Quantization and Hilbert Space} we compute $\mathcal H_{N}(n,k)$ using the identification of a convolution product operator on affine Grassmannian and the Jing operator \cite{jing1991vertex} that is used to define transformed Hall-Littlewood polynomials.

\begin{customthm}{A}[First appeared in \cite{dorey2016matrix}. Theorem \ref{thm: Hilbert series}, Proposition \ref{prop: ground states}]
The character of the Hilbert space $\mathcal H_N(n,k)$ is
\begin{align}
    \mathrm{ch}_{q,\mathbf{a}}(\mathcal H_N(n,k))=H_{(k^N)}(\mathbf{a};q)\prod_{i=1}^N\frac{1}{1-q^i}.
\end{align}
Here $H_{(k^N)}(\mathbf{a};q)$ is the transformed Hall-Littlewood polynomial associated to the partition $(k^N)$. Moreover, the character for the space of ground states $\mathcal H_N(n,k)_0$ is
\begin{align}
\ch_{q,\mathbf{a}}(\mathcal H_N(n,k)_0)=\mathfrak{A}^{kL} s_{k\varpi_r}(\mathbf a)q^{\frac{k}{2}L(L-1)n+krL},
\end{align}
where $L=\lfloor\frac{N}{n}\rfloor$, $r=N-nL$, $\mathfrak{A}=\prod_{i=1}^n \mathbf a_i$, $\varpi_r$ is the $r$-th fundamental weight of $\GL_n$, and $s_{k\varpi_r}(\mathbf a)$ is the Schur polynomial associated to the weight $k\varpi_r$.
\end{customthm}

We point out that restriction-to-torus-fixed-points map $\Gamma(\cM(N,n),\mathcal L_{\det})\to \Gamma(\cM(N,n)^T,\mathcal L_{\det})$ is an isomorphism (Proposition \ref{prop: restrict to fixed pts}). This is part of the induction procedure in the proof of \cite[Theorem 0.2.2]{zhu2009affine}, see \cite[2.1.3]{zhu2009affine}. Using this isomorphism, we prove that $\mathcal H_N(n,1)$ is isomorphic to a fermion Fock space $\mathcal F_N(n)$ (Proposition \ref{prop: fock space}). $\mathcal F_N(n)$ is explicitly presented as a wedge space $\bigwedge^N[\psi^a_m\:|\: 1\le a\le n,\: m\in \bZ_{\ge 0}]$.

\bigskip In Section \ref{sec:operators} we study  operators that arise from the phase space quantization. The quantization amounts to double the quiver:
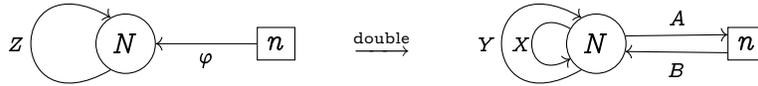
\begin{equation*}
\begin{tikzpicture}[x={(2cm,0cm)}, y={(0cm,2cm)}, baseline=0cm]
  \node[draw,circle,fill=white] (Gauge) at (0,0) {$N$};
  \node[draw,rectangle,fill=white] (Framing) at (1,0) {$n$};
  \node (Z) at (-.73,0) {\scriptsize $Z$};
 \draw[<-] (Gauge.0) -- (Framing.180) node[midway,below] {\scriptsize $\varphi$};

  \draw[->,looseness=6] (Gauge.240) to[out=210,in=150] (Gauge.120);

\end{tikzpicture}
\qquad 
\underrightarrow{\text{\tiny double}}
\qquad
\begin{tikzpicture}[x={(2cm,0cm)}, y={(0cm,2cm)}, baseline=0cm]
  \node[draw,circle,fill=white] (Gauge) at (0,0) {$N$};
  \node[draw,rectangle,fill=white] (Framing) at (1,0) {$n$};
  \node (Zdag) at (-.5,0) {\scriptsize $X$};
  \node (Z) at (-.73,0) {\scriptsize $Y$};
 \draw[->] (Gauge.15) -- (Framing.155) node[midway,above] {\scriptsize $A$};
 \draw[<-] (Gauge.345) -- (Framing.205) node[midway,below] {\scriptsize $B$};

  \draw[<-,looseness=5] (Gauge.210) to[out=210,in=150] (Gauge.150);
  \draw[->,looseness=6] (Gauge.240) to[out=210,in=150] (Gauge.120);

\end{tikzpicture}
\end{equation*}
Here we make the change of variables: $X=Z^{\dag }, Y=Z, A=\varphi^{\dag },B= \varphi$. The right-hand-side is the well-known quiver for the instanton moduli space on $\bC^2$, whose quantization is defined by quantum Hamiltonian reduction of the Weyl algebra generated by $X,Y,A,B$. In \cite{Gaiotto-Rapcek-Zhou} it was shown that quantization of the right-hand-side is a quotient of deformed double current algebra (DDCA)\footnote{DDCA is also called the $1$-shifted affine Yangian, see \cite{rapcak2023cohomological,ishtiaque2024r,ashwinkumar2024r}. This algebra also shows up in the context of twisted holography, see \cite{costello2016m,Costello:2017fbo,Gaiotto:2019wcc,oh2021feynman,Oh:2021wes}.} \cite{guay2005cherednik,guay2007affine,guay2017deformed,etingof2023new,kalinov2021deformed,Gaiotto-Rapcek-Zhou}. 

The DDCA $\mathsf A^{(n)}$ is a $\bC[\epsilon_1,\epsilon_2]$-algebra generated by $\{\mathsf T_{p, q}(x), \mathsf t_{p,q}\:|\: x \in \gl_{n}, (p, q) \in {\bN}^{2}\}$ with relations \eqref{eqn: A0}-\eqref{eqn: A4}, see Definition \ref{def:DDCA}. According to \cite{Gaiotto-Rapcek-Zhou,hu2023quantum}, $\mathsf A^{(n)}$ acts on $\mathcal H_N(n,k)$ by the assignment of generators:
\begin{align*}
    \epsilon_1\mapsto 1,\quad\epsilon_2\mapsto k,\quad \mathsf T_{p,q}(E^a_b)\mapsto A^a\mathrm{Sym}(Y^p X^q)B_b, \quad \mathsf t_{p,q}\mapsto \mathrm{Tr}\:\mathrm{Sym}(Y^p X^q).
\end{align*}
We note that $\mathsf A^{(n)}$ contains two copies of $U(\gl_n[z])$. One is generated by $\{\mathsf T_{0, m}(x)\:|\: x \in \gl_{n}, m \in {\bN}\}$ which is identified with the infinitesimal $\GL_n[z]$ action that comes from the symmetry of $\cM(N,n)$ (see Lemma \ref{lem: gl_n[z] action}). The other one is generated by $\{\mathsf T_{m, 0}(x)\:|\: x \in \gl_{n}, m \in {\bN}\}$. There is also a Yangian subalgebra $Y(\gl_{n}) \subset \mathsf A^{(n)}$ such that $Y(\gl_{n})\ni T^{a}_{b}(u) \mapsto \delta^{a}_{b} + A^{a} \frac{1}{u - XY} B_{b}$ when it acts on ${\cal H}_{N}(n, k)$. Here $T^{a}_{b}(u)$ obeys RTT relation: $(u-v) [T^{a}_{b}(u), T^{c}_{d}(v)] = T^{c}_{b}(u) T^{a}_{d}(v) - T_{b}^{c}(v) T^{a}_{d}(u)$. We denote $T^{a}_{b}(u) = \delta^{a}_{b} + \Sigma_{n \ge 0} T^{a}_{b;n} u^{-n-1}$.

\begin{customthm}{B}[Corollary \ref{cor: cyclic}, Corollary \ref{cor:simplicity}, Corollary \ref{cor:semisimple Yangian module}]
$\mathcal H_N(n,k)$ is a simple $\mathsf A^{(n)}$ module. Moreover,
\begin{itemize}
    \item $\mathcal H_N(n,k)$ is a cyclic $U(\gl_n[z])$ module which is generated by an arbitrary nonzero ground state element $v\in \mathcal H_N(n,k)_0$, where $\gl_n[z]$ generators are given by $E^a_b\otimes z^m\mapsto \mathsf T_{0,m}(E^a_b)$,
    \item $\mathcal H_N(n,k)$ is a semisimple $Y(\gl_n)$ module, where RTT generators of $Y(\gl_n)$ are given by $T^a_b(u)\mapsto \delta^a_b+A^a\frac{1}{u-XY}B_b$.
\end{itemize}
\end{customthm}
The following level-rank relation was proposed by Bourgine and Matsuo in \cite{bourgine2024calogero}, and we give a proof using the geometric methods in this paper.

\begin{customthm}{C}[Proposition \ref{prop: level-rank map}]
Take $k,n,N\in \mathbb Z_{>0}$, then there is a graded $\gl_n[z]$-equivariant surjective map
\begin{align}\label{level-rank map_intro}
    \mathcal H_{kN}(kn,1)_{\mathfrak{sl}_k[z]}\twoheadrightarrow\mathcal H_{N}(n,k).
\end{align}
\end{customthm}
Such map is constructed geometrically: there is a closed embedding $\Quot^N(\mathcal O^{\oplus n}_{\bA^1})\hookrightarrow \Quot^{kN}(\mathcal O^{\oplus kn}_{\bA^1})$ sending a subsheaf $\mathcal F\subset \mathcal O^{\oplus n}_{\bA^1}$ to the subsheaf $\mathcal F^{\oplus k}\subset \mathcal O^{\oplus kn}_{\bA^1}$. This realizes $\cM(N,n)$ as the $\SL_k$ fixed point locus in $\cM(kN,kn)$.
\begin{remark}
We do not know if \eqref{level-rank map_intro} is isomorphism or not in general, nevertheless we can show the following:
\begin{itemize}
    \item \eqref{level-rank map_intro} is an isomorphism if $n=1$ (Corollary \ref{cor:level-rank map: n=1}) or $N=1$ (Corollary \ref{cor:level-rank map: N=1}).
    \item In general, \eqref{level-rank map_intro} becomes an isomorphism after inverting the function $\mathrm{Disc}:=\prod_{i<j}(x_i-x_j)^2$, i.e. 
\begin{align}
    \mathcal H_{kN}(kn,1)_{\mathfrak{sl}_k[z]}[\mathrm{Disc}^{-1}]\cong \mathcal H_{N}(n,k)[\mathrm{Disc}^{-1}].
\end{align}
See Theorem \ref{thm: level-rank}.
\end{itemize}
\end{remark}

\bigskip In Section \ref{sec:wave functions} we study the wave function presentation of $\mathcal H_N(n,k)$. This means an embedding
\begin{align*}
    \mathfrak{W}:\mathcal H_N(n,k)\hookrightarrow \Gamma((\mathbb P^{n-1})^{N}\times \bA^N_{\mathrm{disj}},\mathcal O(k)^{\boxtimes N})\cong \bC[\bA^N_{\mathrm{disj}}]\otimes (S^k\bC^n)^{\otimes N}.
\end{align*}
The map is defined by first restricting to the open subset $h^{-1}(\bA^{(N)}_{\mathrm{disj}})\subset \cM(N,n)$ where $\bA^{(N)}_{\mathrm{disj}}$ parametrizes $N$ disjoint points in $\bA^1$ and $h$ is the Hilbert-Chow morphism $\cM(N,n)\to \bA^{(N)}$, then followed by de-symmetrization i.e. pullback to $\bA^N_{\mathrm{disj}}$ along the covering map $\bA^N_{\mathrm{disj}}\to \bA^{(N)}_{\mathrm{disj}}$. Note that fiber of $h$ at every point in $\bA^{(N)}_{\mathrm{disj}}$ is isomorphic to $(\mathbb P^{n-1})^{N}$, and the restriction of $\mathcal L_{\det}^{\otimes k}$ to the fiber $(\bP^{n-1})^N$ is isomorphic to $\mathcal O(k)^{\boxtimes N}$. $\mathfrak{W}$ is explicitly given by sending $f(X,A)\ket{\emptyset}\in \mathcal H_N(n,k)$ to 
\begin{align}
    f(\mathrm{diag}(x_1,\cdots,x_N),(y^a_i))\in \bC[\bA^N_{\mathrm{disj}}]\otimes (S^k\bC^n)^{\otimes N},
\end{align}
where $\{x_i\}_{1\le i\le N}$ are the coordinates on $\bA^N$ and the $i$-th copy of $S^k\bC^n$ is represented by homogeneous polynomials in $\{y^1_i,\cdots,y^n_i\}$ of degree $k$.

The DDCA operators $\{\mathsf T_{p,q}(x),\mathsf t_{p,q}\}$ act on $\bC[\bA^N_{\mathrm{disj}}]\otimes (S^k\bC^n)^{\otimes N}$ as differential operators dressed with $\gl_n^{\otimes N}$. For example the following was computed in \cite{Gaiotto-Rapcek-Zhou}:
\begin{align*}
    \mathsf t_{2,0}\mapsto \sum_{i=1}^N\Delta^{-1}\partial_{i}^2\Delta-2\sum_{i<j}^N\frac{ \Omega_{ij}+k}{(x_i-x_j)^2}, \text{ where }\Delta=\prod_{i<j}(x_i-x_j),\;\Omega_{ij}=E^a_{b,i}E^b_{a,j},
\end{align*}
$E^a_{b,i}$ are the $\gl_n$ generators that act on the $i$-th $S^k\bC^n$. $\mathsf t_{2,0}$ is a higher-spin analog of Calogero-Moser Hamiltonian. In Lemma \ref{lem:Calogero-Sutherland} we compute the following higher-spin analog of Calogero-Sutherland Hamiltonian:
\begin{align}\label{eq:higher spin CS_intro}
    \mathrm{Tr}((XY)^2)\mapsto \sum_{i=1}^N\Delta^{-1}(x_i\partial_{i})^2\Delta-2\sum_{i<j}^N\frac{x_ix_j(\Omega_{ij}+k)}{(x_i-x_j)^2}-(N-1)\sum_{i=1}^Nx_i\partial_i-\frac{N(N-1)(2N-1)}{6}.
\end{align}

A particularly interesting feature of the wave function presentation is that, the ground states wave functions are expected to solve a Knizhnik-Zamolodchikov equation \cite{Dorey-Tong-Turner}. The $N=nL$ case was proved in \cite{Dorey-Tong-Turner}, and the general case was proved in \cite{bourgine2024calogero}. Both of the proofs are computational and involve careful analysis on explicit formulae of the ground states. In Section \ref{sec:wave functions} we give a new proof based on the geometric construction of conformal block \cite{frenkel2004vertex}. The proof essentially boils down to the simple fact that the KZ connection $\nabla_i$ is $\bC^{\times}_q$ equivariant of weight $-1$, so it lowers the degree of the ground states, which then must be zero. Our method does not require explicit formulae of the ground states.

\begin{customthm}{D}[Corollary \ref{cor:KZ eqn}]
For an arbitrary ground state wave function $\phi\in \mathfrak{W}(\mathcal H_N(n,k)_0)$, $\phi$ solves the Knizhnik-Zamolodchikov equation$:$
\begin{align*}
    (k+n)\partial_i\phi-\sum_{j\neq i}\frac{\Omega_{ij}+k}{x_i-x_j}\phi=0,\; \forall i\in\{1,\cdots,N\}.
\end{align*}
\end{customthm}

\bigskip Section \ref{sec:Conformal Limit, Part I} is devoted to the construction of the left-hand-side of \eqref{eq:limit of hilbert space_intro}. Our approach is based on the following simple observation: there is a natural closed embedding $\iota:\mathrm{Quot}^{N}({\mathcal O}^{\oplus n}_{{\bA}^{1}})\hookrightarrow \mathrm{Quot}^{N+n}({\mathcal O}^{\oplus n}_{{\bA}^{1}})$ sending a subsheaf $\mathcal F$ to $\mathcal F\otimes\mathcal O_{\bA^1}(-[0])$, where $[0]$ is the divisor of the point $0\in \bA^1$ (see Section \ref{subsec:transition map_geometric}). We show in Lemma \ref{lem: transition for det line bundle} that $\iota^*\mathcal L_{\det}$ is $\GL_n[z]\rtimes \bC^{\times}_q$-equivariantly isomorphic to $\mathcal L_{\det}\otimes \chi$, where $\chi$ is the character of $\GL_n[z]\rtimes \bC^{\times}_q$ that maps $(g[z],t)$ to $t^N\cdot\det g[0]$.

We show in Lemma \ref{lem: surjectivity for transition map} that the induced map $\mathcal H_{N+n}(n,k)\to \mathcal H_{N}(n,k)$ is surjective. This allows us to define a projective limit $\underset{\substack{\longleftarrow\\ L}}{\lim}\: \mathcal H_{nL+r}(n,k)$. Precise construction involves replacing $\mathcal H_{N}(n,k)$ by $\widetilde{\mathcal H}_{N}(n,k)$ (the same vector space endowed with a shifted grading) and taking projective limit degree-wise, see Definition \ref{def:conformal limit}. The resulting limit vector space will be denoted by $\widetilde{\mathcal H}^{(r)}_{\infty}(n,k)$. Although the transition map $\mathcal H_{N+n}(n,k)\to \mathcal H_N(n,k)$ is not $\gl_n[z]$-equivariant, we will cure the non-equivariance by twisting $\gl_n[z]$ action by a central character, see Definition \ref{def:p_N and sigma_N}. Then $\gl_n[z]$ acts on $\widetilde{\mathcal H}^{(r)}_{\infty}(n,k)$ by taking inverse limit of compatible $\gl_n[z]$ actions. This $\gl_n[z]$ action is actually cyclic: $\widetilde{\mathcal H}^{(r)}_{\infty}(n,k)$ is generated from any nonzero ground state by $\gl_n[z]$ action (see Proposition \ref{prop:cyclic at limit}). 

We also give an algebraic construction of transition map $p_N:\mathcal H_{N+n}(n,k)\to \mathcal H_{N}(n,k)$ which is equivalent to the aforementioned geometric one. See Section \ref{subsec:transition map_algebraic}.

It is worth noting that the Hermitian inner product on $\mathcal H_{N+n}(n,k)$ induces a natural section of $p_N$, and we denote it by $\sigma_N$ (see Definition \ref{def:p_N and sigma_N}). Using the inductive system generated by $\sigma_N: \mathcal H_{N}(n,k)\to \mathcal H_{N+n}(n,k)$, we can define a inductive limit $\underset{\substack{\longrightarrow\\L}}{\lim}\: \widetilde{\mathcal H}_{nL+r}(n,k)$, and we show in Proposition \ref{prop:ind limit=proj limit} that this inductive limit is actually isomorphic to $\widetilde{\mathcal H}^{(r)}_{\infty}(n,k)$.

When $k=1$, we can identify $\widetilde{\mathcal H}_N(n,1)$ with a shifted fermion Fock space $\widetilde{\mathcal F}_N(n)=\bigwedge^N[\psi^a_m\:|\: 1\le a\le n,\: m\in \bZ_{\ge -L}]$, where $L=\lfloor\frac{N}{n}\rfloor$. In the limit $L\to \infty$, $\widetilde{\mathcal H}^{(r)}_{\infty}(n,1)$ is isomorphic to the charge $r$ semi-infinite wedge space $\widetilde{\mathcal F}^{(r)}_\infty(n)$, which is spanned by elements of the form $\psi^{a_1}_{m_1}\wedge \psi^{a_2}_{m_2}\wedge \cdots$ with $m_j=\lfloor\frac{r-j}{n}\rfloor $ when $j\gg 0$.

\bigskip Section \ref{sec:Conformal Limit, Part II} is devoted to the proof of conjectures \eqref{eq:limit of hilbert space_intro}, \eqref{eq:asymptotic hat sl(n)_intro}, and \eqref{eq:asymptotic hat gl(1)_intro}. To prove \eqref{eq:limit of hilbert space_intro}, we need to construct the full $\widehat{\gl}(n)$ action on $\widetilde{\mathcal H}^{(r)}_{\infty}(n,k)$. In the case of $k=1$, it is known that $\widetilde{\mathcal F}^{(r)}_\infty(n)$ is isomorphic to irreducible integrable module $L_{\varpi_r}(\widehat{\gl}(n)_1)$, with $\widehat{\gl}(n)$ generators acting by $J^a_{b,m}\mapsto\sum_{\ell\ge 0} \psi^a_{\ell}\psi^*_{b,\ell+m}-\sum_{\ell<0} \psi^*_{b,\ell+m}\psi^a_{\ell}$ \cite[Lecture 9]{raina2013bombay}. For the general $k$, we construct $\widehat{\gl}(n)$ action on $\widetilde{\mathcal H}^{(r)}_{\infty}(n,k)$ using the level-rank duality, see Corollary \ref{cor:affine gl(n) action on conformal limit}. 

To prove \eqref{eq:asymptotic hat sl(n)_intro} and \eqref{eq:asymptotic hat gl(1)_intro}, we first develop a framework of conformal limit of operators in Section \ref{subsec:conformal limit of operators}. Let $\{\prescript{L}{}{\mathcal O}\in \End(\widetilde{\mathcal H}_{nL+r}(n,k))\}_{L\in \bN}$ be a collection of linear operators. We say that $$\lim_{L\to\infty}\prescript{L}{}{\mathcal O}=\prescript{\infty}{}{\mathcal O}\in \End(\widetilde{\mathcal H}^{(r)}_{\infty}(n,k))$$ if $\sigma^\infty_{nL+r}\circ \prescript{L}{}{\mathcal O}\circ p^\infty_{nL+r}$ point-wise converge to $\prescript{\infty}{}{\mathcal O}$ with respect to a fixed norm $\rVert\cdot\rVert$ on $\widetilde{\mathcal H}^{(r)}_{\infty}(n,k)$. Here $p^\infty_{nL+r}:\widetilde{\mathcal H}^{(r)}_{\infty}(n,k)\rightleftarrows \widetilde{\mathcal H}_{nL+r}(n,k): \sigma^\infty_{nL+r}$ are natural projection and section respectively (see Definition \ref{def:conformal limit}). For our purpose, we will only work with sequences of operators with uniformly bounded degree, i.e. there exists $C\in \bZ$ such that $\forall d\in \bN$ and $\forall L\in \bN$, $\prescript{L}{}{\mathcal O}(\widetilde{\mathcal H}_{nL+r}(n,k)_{d})\subset \bigoplus_{i\le d+C}\widetilde{\mathcal H}_{nL+r}(n,k)_{i}$. See Definition \ref{def:conformal limit of operators} for details. Moreover, we say that $\prescript{L}{}{\mathcal O}$ converges to $\prescript{\infty}{}{\mathcal O}$ with error term of order $O(L^{-h})$, notation $\prescript{L}{}{\mathcal O}\xrightarrow{O(L^{-h})}\prescript{\infty}{}{\mathcal O}$, if for all $v\in \widetilde{\mathcal H}_{nL+r}(n,k)$ there exists constant $C_v>0$ such that
\begin{align*}
    \lVert \prescript{\infty}{}{\mathcal O}(v)-\sigma^\infty_{nL+r}\circ \prescript{L}{}{\mathcal O}\circ p^\infty_{nL+r}(v)\rVert \le C_v L^{-h}\text{ holds for all }L.
\end{align*}
See Definition \ref{def:speed of convergence}. We note that the conformal limit of operators is additive and multiplicative (Lemma \ref{lem:linearity and multiplication}), and the conformal limit of operators that are compatible with projection maps $p_{nL+r}$ is exactly the projective limit (Proposition \ref{prop:inv lim is conf lim}).

\begin{customthm}{E}[Corollary \ref{cor:affine gl(n) action on conformal limit}, Theorem \ref{thm:conformal limit of T[m,0]}, Theorem \ref{thm:L[1]}]
The $\gl_n[z]$-action on $\widetilde{\mathcal H}^{(r)}_{\infty}(n,k)$ extends to an $\widehat{\gl}(n)_k$-action, such that $\widetilde{\mathcal H}^{(r)}_{\infty}(n,k)$ is isomorphic to irreducible integrable module $L_{k\varpi_r}(\widehat{\mathfrak{sl}}(n)_k)\otimes \mathrm{Fock}_{kr}(\widehat{\gl}(1)_{kn})$. Moreover, 
\begin{itemize}
    \item $\widehat{\gl}(n)_k$ annihilation operators $J^a_{b,m}$ $(m> 0)$ are obtained from scaling limit of $\mathsf T_{m,0}(E^a_b):$
\begin{align}\label{eq:limit of T[m,0]_intro}
    \frac{1}{(k+n)^{m}L^{m}} \prescript{L}{}{\mathsf T}_{m,0}(E^a_b)\xrightarrow{O(L^{-1})} J^a_{b,m}-\frac{\delta^a_b}{k+n}J^c_{c,m}. 
\end{align}
\item Let $\mathcal T(z)=\sum_{m\in \bZ}\mathcal L_m z^{-m-2}$ be the Sugawara's stress-operator of affine vertex algebra associated to $\widehat{\mathfrak{sl}}(n)_k\oplus \widehat{\gl}(1)_{kn}$, then $\mathcal L_{1}$ is obtained from conformal limit by
\begin{align}\label{eq:limit of t[2,1]_intro}
    \frac{1}{(k+n)L} \prescript{L}{}{\mathsf t}_{2,1}-\prescript{L}{}{\mathsf t}_{1,0}\xrightarrow{O(L^{-1})}-\mathcal L_{1}-\frac{k(n+2r)}{n(n+k)}\cdot J^a_{a,1}.
\end{align}
\end{itemize}
\end{customthm}
The proof of \eqref{eq:limit of T[m,0]_intro} and \eqref{eq:limit of t[2,1]_intro} is the most technical part of this paper. A key ingredient in the proof is the following modified Yangian RTT generators
\begin{align}\label{eq:modified Yangian_intro}
    \widetilde{T}^a_b(u)\mapsto \left[\delta^a_b+A^a\frac{1}{u+(k+n)L-XY }B_b\right]\prod_{j=1}^L \frac{u+(n+k)j}{u+k+(n+k)j},
\end{align}
and we will prove in Theorem \ref{thm:Yangian comptible with p_N and sigma_N} that $p_{nL+r}\circ \prescript{L+1}{}{\widetilde{T}}^a_b(u)=\prescript{L}{}{\widetilde{T}}^a_b(u)\circ p_{nL+r}$. Although both \eqref{eq:modified Yangian_intro} and the statement in Theorem \ref{thm:Yangian comptible with p_N and sigma_N} are explicit, we could not find an explicit proof of Theorem \ref{thm:Yangian comptible with p_N and sigma_N}. Instead, our approach is representation-theoretic, see Section \ref{subsec:Yangian action on H_N(n,k)} for details. 

We will prove in Section \ref{subsec:Comparison with DTT} that \eqref{eq:limit of T[m,0]_intro} implies \eqref{eq:asymptotic hat sl(n)_intro} and \eqref{eq:asymptotic hat gl(1)_intro}. A key step in the proof is to relate the norm on $\widetilde{\mathcal H}_N(n,k)$ to the norm on $\widetilde{\mathcal H}^{(r)}_{\infty}(n,k)$, see Lemma \ref{lem:bound of norm}.

\bigskip Finally, in Section \ref{sec:Application} we present some applications of our studies in previous sections. Namely, we characterize the Yangian simple submodules of $\widetilde{\mathcal H}^{(r)}_{\infty}(n,k)$; we show that the action of the Gelfand-Tsetlin subalgebra of $Y(\gl_n)$ on $\widetilde{\mathcal H}^{(r)}_{\infty}(n,k)$ has simple spectrum which leads to the Yangian Gelfand-Tsetlin basis of $\widetilde{\mathcal H}^{(r)}_{\infty}(n,k)$; we also compute the eigenvalues of Gelfand-Tsetlin subalgebra of $Y(\gl_n)$, in particular the quantum determinant. 

\begin{customthm}{F}[Theorem \ref{thm:spectral decomposition}, Theorem \ref{thm:Yangian irrep decomposition}]\label{thm:F}
$\prescript{\infty}{}{\widetilde{T}}^a_b(u):=\underset{\substack{\longleftarrow\\ L}}{\lim}\prescript{L}{}{\widetilde{T}}^a_b(u)$ defines an action of RTT generators of Yangian algebra $Y(\gl_n)$ on $\widetilde{\mathcal H}^{(r)}_{\infty}(n,k)$ such that the latter decomposes into simple $Y(\gl_n)$-modules:
\begin{align}
    \widetilde{\mathcal H}^{(r)}_{\infty}(n,k)=\bigoplus_{\substack{\lambda=(\lambda_1\ge \lambda_2\ge \cdots)\in \bZ^\infty\\  \lambda_{j}=-k\lfloor\frac{j-1-r}{n}\rfloor \text{ for }j\gg 0.}} \widetilde{\mathcal H}(\lambda).
\end{align}
Moreover,
\begin{itemize}
\item $\widetilde{\mathcal H}(\lambda)$ is homogeneous of degree $\sum_{j=1}^\infty\left(\lambda_j+k\lfloor\frac{j-1-r}{n}\rfloor\right)$ with respect to the shifted energy grading (see Definition \ref{def:p_N and sigma_N}).
\item The Drinfeld polynomials of $\widetilde{\mathcal H}(\lambda)$ are
\begin{align}
    P_m(u)=\prod_{\substack{\text{$(i,j)$ is a top box of}\\ \text{a height $m$ column}\\ \text{in }\mathrm{SYD}(\lambda/\lambda^\shortdownarrow)}}(u-r+i-j),\quad m=1,\cdots, n-1,
\end{align}
where $\lambda^\shortdownarrow$ is the downward shift of $\lambda$ by $k$-units (Definition \ref{def:decrease and increase}), and $\mathrm{SYD}(\lambda/\lambda^\shortdownarrow)$ is the skew Young diagram associated to $\lambda/\lambda^\shortdownarrow$ (Definition \ref{def:SYD}).
\item The eigenvector decomposition of $\widetilde{\mathcal H}(\lambda)$ with respect to quantum minors $\{\widetilde A_m(u)\}_{1\le m\le n}$ is 
\begin{align}
    \widetilde{\mathcal H}(\lambda)=\bigoplus_{\substack{\Lambda\in \mathrm{GT}^{(r)}_{\infty}(n)\\ \Lambda_1=\lambda}}V_{\Lambda},
\end{align}
where $\mathrm{GT}^{(r)}_{\infty}(n)$ is the set of semi-infinite GT patterns of height $n$ and type $r$ (Definition \ref{def:semi-inf GT pattern}).
\item The eigenvalue of $\widetilde{A}_m(u)$ on $V_{\Lambda}$ is
\begin{align}
    \prod_{i=1}^r\left(\frac{u-r-\Lambda_{m+1,i}+i}{u-r-\Lambda_{1,i}+i}\right)\times \prod_{i=r+1}^\infty\left(\frac{u-r-\Lambda_{m+1,i}+i}{u-r-\Lambda^{\mathrm{vac}}_{m+1,i}+i}\times \frac{u-r-\Lambda^{\mathrm{vac}}_{1,i}+i}{u-r-\Lambda_{1,i}+i}\right),
\end{align}
where $\Lambda^{\mathrm{vac}}_{i,j}=-k\lfloor\frac{j+i-2-r}{n}\rfloor$ for all $i,j$.
\end{itemize}
\end{customthm}

In the case $k=1$, the Chern-Simons matrix model becomes the spin Calogero model, and Theorem \ref{thm:F} recovers the result of Takemura and Uglov in \cite{takemura1997orthogonal}. Notably, the Yangian Gelfand-Tsetlin bases in the spin Calogero model can be represented by $\gl_n$ Jack polynomials \cite{uglov1998yangian}, see also \cite{bernard1993yang,bernevig2008model}. It will be interesting to explore the relation between the Yangian Gelfand-Tsetlin bases at a general level $k$ and symmetric polynomials.

Another application is the complete solution to the higher spin generalization of Calogero-Sutherland model defined in \eqref{eq:higher spin CS_intro}.

\begin{customthm}{G}[{Theorem \ref{thm:spectrum of H_CS}}]
The eigenspace decomposition of $\mathcal H_N(n,k)$ with respect to the action of $H_{\mathrm{CS}}=\mathrm{Tr}((XY)^2)$ is given by
\begin{align}
    \mathcal H_N(n,k)=\bigoplus_{\lambda} \widetilde{\mathcal H}_N(\lambda),
\end{align}
where the sum is taken for all $\lambda=(\lambda_1\ge\cdots\ge\lambda_N)\in \bN^N$ such that $\lambda$ is admissible (Definition \ref{def:admissible and cuttable}), and $\widetilde{\mathcal H}_N(\lambda)$ is defined in Definition \ref{def:H(la)}. The eigenvalue of $H_{\mathrm{CS}}$ on $\widetilde{\mathcal H}_N(\lambda)$ is
\begin{align}
    \sum_{i=1}^N (\lambda_i-k)(\lambda_i-k+N+1-2i).
\end{align}
\end{customthm}

\bigskip In Appendix \ref{sec:hall-littlewood} we briefly review the Hall-Littlewood polynomials which show up in the character formula for $\mathcal H_N(n,k)$. In Appendix \ref{sec:affine Gr} we give a crash introduction to affine Grassmannians, and also a geometric interpretation of the Jing operator which is used to define transformed Hall-Littlewood polynomials. In Appendix \ref{sec:Identities} we prove some identities in the quantized Nakajima quiver variety for the Jordan quiver.

\section{Definition of the Model}\label{sec:model}

The action of the Dorey-Tong-Turner matrix model is given in \cite{Poly,Dorey-Tong-Turner,dorey2016matrix}:
\begin{equation}\label{Action}
    S=i\int dt \left [ \mathrm{Tr}\left \{Z^{\dag}(\dot{Z}-i[\alpha,Z])+i (k+n) \alpha+i \omega (Z^{\dag}Z)\right \}+\sum_{a=1}^{n}\varphi^{\dag a} \left( \dot{\varphi}_a-i\alpha\varphi_a\right) \right ]
\end{equation}
Here $Z, \alpha$ are $N \times N$ complex matrices and $\{\varphi_a\}_{a=1}^n$ are $n$-tuples of 
$N$-dimensional vectors, $k$ is called the level of the model which is the reminiscent of noncommutativity \cite{Poly,tong2004quantum}.

The action \eqref{Action} is invariant with respect to the following gauge transformations:
\begin{align}\label{Gauge invariant}
Z \mapsto U Z U^{-1},\quad \partial_{t} - i \alpha \mapsto U (\partial_{t} - i \alpha) U^{-1},\quad \varphi_a\mapsto U\varphi_a,
\end{align}
where $ U \in \mathrm{U}(N)$ is a unitary $N\times N$ matrix. The gauge group $\mathrm{U}(N)$ is better considered as area preserving diffeomorphism symmetry over a non-commutative space \cite{Susskind}. We collect $(Z,\varphi)$ into the vector space $V(N,n):=\mathrm{End}(\mathbb C^N)\times \mathrm{Hom}(\mathbb C^n,\mathbb C^N)$, then \eqref{Gauge invariant} is the natural action of $\mathrm{U}(N)$ on $V(N,n)$. As a complex vector space, $V(N,n)$ has a standard K\"ahler structure $\omega=\frac{i}{2}\left(d Z^l_j\wedge d Z^{\dag j}_l+d\varphi_a^j\wedge d\varphi^{\dag a}_j\right)$, which is invariant with respect to the $\mathrm{U}(N)$ action. Moreover, the $\mathrm{U}(N)$ action on $V(N,n)$ is Hamiltonian with moment map
\begin{equation}\label{eq: moment map}
\begin{aligned}
\mu_{\mathbb R}: &V(N,n)\longrightarrow \mathrm{Lie}(\mathrm{U}(N))^{*},\\
 \mu_{\mathbb R}&(Z, \varphi)  = -i\left([Z, Z^\dag] + \varphi \varphi^\dag\right).
\end{aligned}
\end{equation}
We can rewrite the action \eqref{Action} using the moment map:
\begin{align}
 S=i\int dt \left [ \mathrm{Tr}\left (Z^{\dag}\dot{Z}+i \omega Z^{\dag}Z\right )+\sum_{a=1}^{n}{\varphi^{\dag}}^{a} \dot{\varphi}_a \right ]+i\int dt\: \mathrm{Tr}\left[ \alpha(\mu_{\bR}(Z,\varphi)+i(k+n)\cdot\mathbf{1}) \right]
\end{align}
The matrix model can be viewed as a Hamiltonian system with constraints which is given by variation with respect to the auxiliary field $\alpha$:
\begin{equation}\label{Classical constraints}
    \frac{\delta S}{\delta \alpha}=(k+n)\cdot{\bf 1}-i\mu_{\mathbb R}(Z,\varphi)=0.
\end{equation}
The phase space is then $\mu_{\mathbb R}^{-1}(-i(k+n)\cdot{\bf 1})/ \mathrm{U}(N)$.

Let us take $(Z,\varphi)\in \mu_{\mathbb R}^{-1}(-i(k+n)\cdot{\bf 1})$, and then we have $\mathrm{Tr}(i \mu_{\mathbb R}(Z, \varphi))=(k+n) N$. On the other hand $\mathrm{Tr}( i\mu_{\mathbb R}(Z, \varphi))=\mathrm{Tr}(\varphi \varphi^\dag)\ge 0$. Therefore in order for $\mu_{\mathbb R}^{-1}(-i(k+n)\cdot{\bf 1})$ to be nonempty, we must have $k+n\ge 0$.

In the critical case $k=-n$, the above argument shows that $\varphi=0$ and $[Z,Z^\dag]=0$. Then the $U(N)$ action on $\mu_{\mathbb R}^{-1}(-i(k+n)\cdot{\bf 1})$ is not free, and the quotient space $\mu_{\mathbb R}^{-1}(-i(k+n)\cdot{\bf 1})/U(N)$ is badly-behaved.

Therefore, we shall assume that $k>-n$ in order to get a sensible phase space. We will see shortly that after taking quantization into account, it makes sense to assume that $k\in \bZ_{>0}$ in order to get nontrivial Hilbert space.

\section{Geometry of Phase Space}\label{sec:phase space}
Define $$V(N,n)=\mathrm{End}(\mathbb C^N)\times \mathrm{Hom}(\mathbb C^n,\mathbb C^N).$$
As we have seen in the previous section, the phase space of the Chern-Simons matrix model \eqref{Action} is
\begin{align}
    \cM(N,n)=\{(Z,\varphi)\in V(N,n)\:|\: [Z,Z^{\dag}]+\varphi\varphi^{\dag}=(k+n) \cdot \mathbf{1}\}/\mathrm{U}(N).
\end{align}
Since we have assume that $k>-n$, then by the Kempf-Ness theorem \cite{kempf1979length}, more precisely by \cite[Proposition 3.1, Corollary 6.2]{king1994moduli}, $\cM(N,n)$ admits a natural complex structure which is induced from the following GIT quotient:
\begin{align}\label{eq: M(N,n)}
    \cM(N,n)\cong \{(Z,\varphi)\in V(N,n)\:|\: \mathbb C[Z]\cdot\mathrm{Im}(\varphi)=\mathbb C^N\}\sslash \mathrm{GL}_N(\mathbb C).
\end{align}
This is the quiver variety of the Jordan quiver:
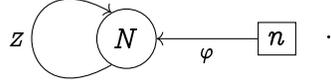
\begin{equation*}
\begin{tikzpicture}[x={(2cm,0cm)}, y={(0cm,2cm)}, baseline=0cm]
  \node[draw,circle,fill=white] (Gauge) at (0,0) {$N$};
  \node[draw,rectangle,fill=white] (Framing) at (1,0) {$n$};
  \node (Z) at (-.73,0) {\scriptsize $Z$};
 \draw[<-] (Gauge.0) -- (Framing.180) node[midway,below] {\scriptsize $\varphi$};

  \draw[->,looseness=6] (Gauge.240) to[out=210,in=150] (Gauge.120);

\end{tikzpicture}\quad .
\end{equation*}
By the geometric invariant theory \cite{mumford1994geometric}, we have isomorphism
\begin{align}\label{eq: M(N,n) as GIT quotient}
    \cM(N,n)\cong \mathrm{Proj}\:\bigoplus_{m\le 0}\mathbb C[V(N,n)]^{\GL_N,m},
\end{align}
where $\mathbb C[V(N,n)]^{\GL_N,m}$ denotes the subspace of $\mathbb C[V(N,n)]$ on which $\GL_N$ acts by the multiplication of the character $g\mapsto \det(g)^m$. $\bigoplus_{m\ge 0}\mathbb C[V(N,n)]^{\GL_N,m}$ is graded since $$\mathbb C[V(N,n)]^{\GL_N,m}\cdot \mathbb C[V(N,n)]^{\GL_N,s}\subset\mathbb C[V(N,n)]^{\GL_N,m+s},$$ and this allows us to define its projective spectra.

\begin{remark}
By construction, there is a natural projective morphism
\begin{align}
    \cM(N,n)\to \Spec\:\bC[V(N,n)^{\GL_n}]=V(N,n)\sslash \GL_N.
\end{align}
We note that $\bC[V(N,n)^{\GL_n}]$ consists of $\GL_N$-invariant polynomials in $Z$ only (since the central $\mathbb C^{\times}\subset \GL_N$ acts on $\varphi$ with weight $-1$), thus $V(N,n)\sslash \GL_N\cong \gl_N\sslash\GL_N$. $\gl_N\sslash\GL_N$ is isomorphic to $\bA^N/S_N$ by the Kostant theorem \cite{kostant1963}.
\end{remark}

\begin{lemma}
$\cM(N,n)$ is a smooth connected algebraic variety of complex dimension $nN$.
\end{lemma}

\begin{proof}
By the geometric invariant theory \cite{mumford1994geometric}, the stable subset $R(N,n):=\{(Z,\varphi)\in V(N,n)\:|\: \mathbb C[Z]\cdot\mathrm{Im}(\varphi)=\mathbb C^N\}$ is open in $V(N,n)$, so $R(N,n)$ is smooth and connected. Moreover $\GL_N$ acts on $R(N,n)$ freely \cite{mumford1994geometric}, thus the quotient is smooth and connected. The dimension is computed by $\dim \cM(N,n)=\dim R(N,n)-\dim \GL_N=nN$.
\end{proof}

\begin{lemma}\label{lem: M(N,n)=Quot}
$\cM(N,n)$ is isomorphic to $\mathrm{Quot}^{N}({\mathcal O}^{\oplus n}_{{\mathbb A}^{1}})$, the Quot scheme of affine line $\mathbb A^1$ parametrizing length $N$ finite quotients of $\mathcal O_{\mathbb A^1}^{\oplus n}$. Moreover, the flavour symmetry is the $\GL_{n}$-action on $\mathcal O_{\mathbb A^1}^{\oplus n}$ and the $\bC^{\times}$-scaling of $X$ with weight $-1$ is mapped to the $\bC_q^{\times}$-rotation of the $\mathbb A^1$-plane.
\end{lemma}

\begin{proof}
Consider the prequotient $R(N,n):=\{(Z,\varphi)\in \mathrm{End}(\mathbb C^N)\times \mathrm{Hom}(\mathbb C^n,\mathbb C^N)\:|\: \mathbb C[Z]\cdot\mathrm{Im}(\varphi)=\mathbb C^N\}$, it admits a natural map $q: R(N,n) \to \mathrm{Quot}^{N}({\mathcal O}^{\oplus n}_{{\mathbb A}^{1}})$, defined as follows. Consider a point $(Z,\varphi)\in R(N,n)$, the action of the matrix $Z$ on $\mathbb C^N$ makes it into a $\mathbb C[z]$-module such that $z$ acts as the matrix $Z$, and $\varphi: \mathbb C^n\to \mathbb C^N$ is equivalent to a $\mathbb C[z]$-module map $\widetilde\varphi:{\mathcal O}^{\oplus n}_{{\mathbb A}^{1}}\to \mathbb C^N$ where $\mathbb A^1=\Spec\mathbb C[z]$, and the stability condition is equivalent to the surjectivity of $\widetilde\varphi$. This defines a morphism $q:R(N,n)\to \mathrm{Quot}^{N}({\mathcal O}^{\oplus n}_{{\mathbb A}^{1}})$. Every point in $\mathrm{Quot}^{N}({\mathcal O}^{\oplus n}_{{\mathbb A}^{1}})$ is contained in the image of $q$, and $\GL_N$ acts on $R(N,n)$ by changing the basis of $\mathbb C^N$, thus $q$ is a principal $\GL_N$-bundle, whence $\mathrm{Quot}^{N}({\mathcal O}^{\oplus n}_{{\mathbb A}^{1}})\cong R(N,n)/\GL_N=\cM(N,n)$.
\end{proof}

\begin{remark}\label{rmk: Hilbert-Chow}
We note that the Hilbert-Chow morphism $h: \mathrm{Quot}^{N}({\mathcal O}^{\oplus n}_{{\mathbb A}^{1}})\to \mathbb A^{(N)}=\mathrm{Sym}^N(\mathbb A^1)$ is identified with the natural projection $\cM(N,n)\to \gl_N\sslash \GL_N\cong \mathfrak{t}/S_N\cong \mathbb A^{(N)}$, which maps a couple $(Z,\varphi)$ to the eigenvalues of $Z$.

We also note that the Hilbert-Chow morphism $h: \mathrm{Quot}^{N}({\mathcal O}^{\oplus n}_{{\mathbb A}^{1}})\to \mathbb A^{(N)}$ is proper, because we have natural isomorphism $$\mathrm{Quot}^{N}({\mathcal O}^{\oplus n}_{{\mathbb A}^{1}})\cong \mathrm{Quot}^{N}({\mathcal O}^{\oplus n}_{{\mathbb P}^{1}})\times_{\mathrm{Sym}^N(\mathbb P^1)} \mathrm{Sym}^N(\mathbb A^1),$$
where $\mathrm{Quot}^{N}({\mathcal O}^{\oplus n}_{{\mathbb P}^{1}})\to {\mathrm{Sym}^N(\mathbb P^1)}$ is the Hilbert-Chow morphism for $\mathrm{Quot}^{N}({\mathcal O}^{\oplus n}_{{\mathbb P}^{1}})$ and the morphism from $\mathrm{Sym}^N(\mathbb A^1)$ to $ \mathrm{Sym}^N(\mathbb P^1)$ is the open immersion which is induced by the natural inclusion $\mathbb A^1\hookrightarrow\mathbb P^1$. Since both $\mathrm{Quot}^{N}({\mathcal O}^{\oplus n}_{{\mathbb P}^{1}})$ and $ {\mathrm{Sym}^N(\mathbb P^1)}$ are proper, it follows that $\mathrm{Quot}^{N}({\mathcal O}^{\oplus n}_{{\mathbb P}^{1}})\to {\mathrm{Sym}^N(\mathbb P^1)}$ is proper, whence $h$ is proper.
\end{remark}

\begin{proposition}\label{prop: isom between Quot and Gr}
$\mathrm{Quot}^{N}({\mathcal O}^{\oplus n}_{{\mathbb A}^{1}})$ is isomorphic to $\overline{\Gr}^{\omega_1,\cdots,\omega_1}_{\GL_{n},\bA^{(N)}}$ (c.f. Appendix \ref{sec:affine Gr} for definition). Moreover the isomorphism is $\GL_{n}\times \bC_q^{\times}$-equivariant and commutes with projections to $\bA^{(N)}$,
\begin{center}
\begin{tikzcd}[row sep=1.2em,column sep=0.6em]
\mathrm{Quot}^{N}({\mathcal O}^{\oplus n}_{{\mathbb A}^{1}}) \arrow[dr,"h"']  \arrow[rr, "p"] & & \overline{\Gr}^{\omega_1,\cdots,\omega_1}_{\GL_{n},\bA^{(N)}} \arrow[dl,"\pi"]\\
& \bA^{(N)}  &
\end{tikzcd}
\end{center}
Here $h$ is the Hilbert-Chow map, and $\pi$ is the structure map of symmetrized Beilinson-Drinfeld Grassmannian.
\end{proposition}

\begin{proof}
By definition, ${\Gr}_{\GL_{n},\bA^{(N)}}$ is the moduli space of couples $(\mathcal F,D)$, here $D$ is a divisor of degree $N$ on $\mathbb A^1$ and $\mathcal F$ is a subsheaf of $j_*\mathcal O_{\mathbb A^1\setminus D}^{\oplus n}$ which is required to be locally free of rank $n$, where $j:\mathbb A^1\setminus D\hookrightarrow \mathbb A^1$ is the inclusion. We define the morphism $p: \mathrm{Quot}^{N}({\mathcal O}^{\oplus n}_{{\mathbb A}^{1}})\to {\Gr}_{\GL_{n},\bA^{(N)}}$ between functors of points. Namely for a $\mathbb C$-scheme $S$, $p$ maps a quotient $(f:{\mathcal O}^{\oplus n}_{{\mathbb A}^{1}\times S}\twoheadrightarrow \mathcal E)\in \mathrm{Quot}^{N}({\mathcal O}^{\oplus n}_{{\mathbb A}^{1}})(S)$ to a couple $(\ker(f),h(f))\in {\Gr}_{\GL_{n},\bA^{(N)}}(S)$, where $h(f)\in \mathbb A^{(N)}(S)$ is the image of $f$ under the Hilbert-Chow morphism $\mathrm{Quot}^{N}({\mathcal O}^{\oplus n}_{{\mathbb A}^{1}})\to \mathbb A^{(N)}$, i.e. $h(f)$ is the divisor on $\mathbb A^1\times S$ associated to $\mathcal E$. By construction, we have $\pi\circ p=h$, where $\pi:{\Gr}_{\GL_{n},\bA^{(N)}}\to \bA^{(N)}$ is the natural projection that maps a couple $(\mathcal F, D)$ to the divisor $D$. $p$ is a monomorphism (c.f. \cite[12.18]{gortz2010algebraic}), because a quotient $f:{\mathcal O}^{\oplus n}_{{\mathbb A}^{1}\times S}\twoheadrightarrow \mathcal E$ is uniquely determined by its kernel. Since $h$ is proper (c.f. Remark \ref{rmk: Hilbert-Chow}) and $\pi$ is ind-proper \cite{zhu2016introduction}, it follows that $p$ is proper, whence $p$ is a closed immersion by \cite[Corollary 12.92]{gortz2010algebraic}. 

It remains to show that the image of $p$ is $\overline{\Gr}^{\omega_1,\cdots,\omega_1}_{\GL_{n},\bA^{(N)}}$. In fact, let $\vec{x}:=(x_1,\cdots,x_N)\in \bA^{(N)}$ be such that $x_i\neq x_j$ whenever $i\neq j$, then $h^{-1}(\vec{x})$ is isomorphic to the moduli space of $N$-tuples $(\mathcal Q_1,\cdots,\mathcal Q_N)$ where $\mathcal Q_i$ is one dimensional quotient module of $\bC[z]$ which is supported at $x_i$. Since $\bGr^{\omega_1}_{\GL_n}$ is the moduli space of one-dimensional quotient module of $\mathbb C[\![z]\!]$, it follows that $h^{-1}(\vec{x})\cong \left(\bGr^{\omega_1}_{\GL_n}\right)^{\times N}$. Denote by $\Delta\subset \bA^{(N)}$ the divisor consisting of $(y_1,\cdots,y_N)$ such that $y_i=y_j$ for some $i\neq j$, then $h^{-1}(\bA^{(N)}\setminus \Delta)\cong \left(\bGr^{\omega_1}_{\GL_n}\right)^{\times N}\times (\bA^{(N)}\setminus \Delta)$ by the above argument. Since $\cM(N,n)$ is connected, it follows that the image of $p$ is the closure of $p(h^{-1}(\bA^{(N)}\setminus \Delta))$ in ${\Gr}_{\GL_{n},\bA^{(N)}}$. By definition, $\overline{\Gr}^{\omega_1,\cdots,\omega_1}_{\GL_{n},\bA^{(N)}}$ is the closure of $\left(\bGr^{\omega_1}_{\GL_n}\right)^{\times N}\times (\bA^{(N)}\setminus \Delta)$ in ${\Gr}_{\GL_{n},\bA^{(N)}}$, whence the image of $p$ is $\overline{\Gr}^{\omega_1,\cdots,\omega_1}_{\GL_{n},\bA^{(N)}}$.
\end{proof}

\begin{corollary}\label{cor:flatness}
The Hilbert-Chow morphism $h:\mathrm{Quot}^{N}({\mathcal O}^{\oplus n}_{{\mathbb A}^{1}})\to \bA^{(N)}$ is flat.
\end{corollary}

\begin{proof}
Let $\vec{x}=(x_1,\cdots,x_N)\in \bA^{(N)}$, then the fiber $h^{-1}(\vec{x})$ can be described as follows. Let the divisor corresponding to $\vec{x}$ be $\sum_{j=1}^m s_j\cdot [y_j]$ such that $y_j\neq y_{j'}$ whenever $j\neq j'$, then $h^{-1}(\vec{x})$ is isomorphic to $\overline{\Gr}^{s_1\omega_1}_{\GL_{n}}\times \cdots\times \overline{\Gr}^{s_m\omega_1}_{\GL_{n}}$. Therefore the fiber $h^{-1}(\vec{x})$ is irreducible of dimension $(n-1)N$. On the other hand $\mathrm{Quot}^N(\mathcal O_{\mathbb A^1}^{\oplus n})$ is smooth of dimension $nN$, thus $h:\mathrm{Quot}^N(\mathbb A^1,\mathcal O_{\mathbb A^1}^{\oplus n})\longrightarrow \mathbb A^{(N)}$ is flat by the miracle flatness theorem \cite[\href{https://stacks.math.columbia.edu/tag/00R4}{Tag 00R4}]{stacks-project}.
\end{proof}

Summarizing what we have discussed above, there are three equivalent description of the moduli space $\cM(N,n)$:
\begin{enumerate}
    \item as a GIT quotient: $\cM(N,n)\cong \left(\mathrm{End}(\mathbb C^N)\times \mathrm{Hom}(\mathbb C^n,\mathbb C^N)\right)\sslash_{k}\GL_N$,
    \item as a Quot scheme $\cM(N,n)\cong \Quot^N(\mathcal O^{\oplus n}_{\bA^1})$,
    \item as a closed subvariety $\overline{\Gr}^{\omega_1,\cdots,\omega_1}_{\GL_{n},\bA^{(N)}}$ of Beilinson-Drinfeld Grassmannian ${\Gr}_{\GL_{n},\bA^{(N)}}$.
\end{enumerate}

\subsection{Line bundles on \texorpdfstring{$\mathcal M(N,n)$}{M(N,n)}}\label{subsec:Pic(M(N,n))}

It is is known that there is a distinguished line bundle on ${\Gr}_{\GL_{n},\bA^{(N)}}$ which is denoted  by $\mathcal O(1)$ \cite{zhu2016introduction}. Roughly speaking, the fiber of $\mathcal O(1)$ at a point $(\mathcal F,D)\in {\Gr}_{\GL_{n},\bA^{(N)}}$ is given by ``$\det(\bC[z]^{\oplus n})/\det(\mathcal F)$'', we put quotation mark because neither $\det(\bC[z]^{\oplus n})$ nor $\det(\mathcal F)$ are well-defined since they are infinite dimensional. Nevertheless when restricted to subvariety $\overline{\Gr}^{\omega_1,\cdots,\omega_1}_{\GL_{n},\bA^{(N)}}$, $\mathcal F$ is a subsheaf of $\mathcal O^{\oplus n}_{\bA^1}$ with finite dimensional quotient, and the fiber $\mathcal O(1)$ at $(\mathcal F,D)\in \overline{\Gr}^{\omega_1,\cdots,\omega_1}_{\GL_{n},\bA^{(N)}}$ is given by $\det(\mathcal O^{\oplus n}_{\bA^1}/\mathcal F)$. 

Using the Quot scheme description, $p^*\mathcal O(1)$ is the determinant of $\mathcal E_{\mathrm{univ}}$, which is the universal quotient sheaf $\mathcal O_{\bA^1\times \cM(N,n)}^{\oplus n}\twoheadrightarrow \mathcal E_{\mathrm{univ}}$ on $\cM(N,n)$. We define
\begin{align}
    \mathcal L_{\det}=\det(\mathcal E_{\mathrm{univ}}).
\end{align}
In the GIT description $\mathcal E_{\mathrm{univ}}$ is the rank $N$ vector bundle $R(N,n)\times_{\GL_N}\bC^N\to \cM(N,n)$. We also note that 

\begin{lemma}
The Picard group of $\cM(N,n)$ is generated by $\mathcal L_{\det}$.
\end{lemma}

\begin{proof}
Since $R(N,n)\to \cM(N,n)$ is a principal $\GL_N$-bundle, we have
\begin{align*}
    \Pic(\cM(N,n))\cong \Pic^{\GL_N}(R(N,n))
\end{align*}
where $\Pic^{\GL_N}(R(N,n))$ is the $\GL_N$-equivariant Picard group. Since $R(N,n)$ is a $\GL_N$-equivariant open subvariety of $V(N,n)$, we have surjective map $\Pic^{\GL_N}(V(N,n))\twoheadrightarrow \Pic^{\GL_N}(R(N,n))$. Since $V(N,n)$ is an affine space, every $\GL_N$-equivariant line bundle on it is isomorphic to $\mathcal O_{V(N,n)}\otimes\chi$ where $\chi$ is a character of $\GL_N$. We note that $\chi(g)=\det(g)^m$ for some $m$, therefore $\Pic^{\GL_N}(V(N,n))=\mathbb Z\cdot [\chi_1]$, where $\chi_1(g)=\det(g)$. It follows that $\Pic^{\GL_N}(R(N,n))$ is generated by $\mathcal O_{R(N,n)}\otimes\chi_1$. The image of $\mathcal O_{R(N,n)}\otimes\chi_1$ in $\Pic(\cM(N,n))$ is exactly $\mathcal L_{\det}$, thus $\Pic(\cM(N,n))$ is generated by $\mathcal L_{\det}$.
\end{proof}

\begin{lemma}\label{lem: pushforward of line bundle}
Let $m\in \bZ_{\ge 0}$, then $R^ih_*(\mathcal L_{\det}^{\otimes m})=0$ for all $i>0$, and $h_*(\mathcal L_{\det}^{\otimes m})$ is locally free on $\bA^{(N)}$.
\end{lemma}

\begin{proof}
Since $h$ is a proper flat morphism, $Rh_*(\mathcal L_{\det}^{\otimes m})$ is a perfect complex with positive tor-amplitudes in $D^b_{\mathrm{coh}}(\bA^{(N)})$, so we only need to show that $R^ih_*(\mathcal L_{\det}^{\otimes m})=0$ for all $i>0$. Consider the $\bC_q^{\times}$ action on $\bA^1$ by scaling with weight $-1$, then $\bC_q^{\times}$ naturally acts on $\Quot^N(\mathcal O^{\oplus n}_{\bA^1})$ and $\bA^{(N)}$, and $h$ is $\bC_q^{\times}$-equivariant. It follows that $R^ih_*(\mathcal L_{\det}^{\otimes m})$ is a $\bC_q^{\times}$-equivariant coherent sheaf on $\bA^{(N)}$. The $\bC_q^{\times}$-action on $\bA^{(N)}$ is repelling, thus to show that $R^ih_*(\mathcal L_{\det}^{\otimes m})=0$ for all $i>0$, it is enough to show that the central fiber has the cohomology vanishing property:
\begin{align*}
    H^i(h^{-1}(0),\mathcal L_{\det}^{\otimes m}|_{h^{-1}(0)})=0,\; \forall i>0.
\end{align*}
We note that $h^{-1}(0)$ is isomorphic to $\overline{\Gr}^{N\omega_1}_{\GL_n}$, and the restriction of $\mathcal L_{\det}$ to $h^{-1}(0)$ is $\mathcal{O}(1)$. It is known that $H^i(\overline{\Gr}^{N\omega_1}_{\GL_n},\mathcal O(m))$ vanishes for all $i>0$ whenever $m\ge 0$, c.f. \cite[Ch.XVIII]{mathieu1988formules}. This finishes the proof.
\end{proof}

\begin{example}
In the case when $n=1$, $\cM(N,1)\cong \Hilb^N(\bA^1)=\bA^{(N)}$ by Lemma \ref{lem: M(N,n)=Quot}. In this case $h$ is isomorphism, and $\mathcal L_{\det}$ is a trivial line bundle. Consider the $\bC_q^{\times}$ action on $\bA^1$ by scaling with weight $-1$, which induces a $\bC_q^{\times}$ action on $\bA^{(N)}$ and $\mathcal L_{\det}$. Let $\mathrm{ch}_q(V)$ be the character of a $\bC_q^{\times}$-module $V$, i.e. $\mathrm{ch}_q(V)=\sum_{n\in \mathbb Z}\dim(V_n)q^n$ where $V_n$ is the weight $n$ eigenspace of $V$. Then we have $$\mathrm{ch}_q(\Gamma(\bA^{(N)},\mathcal L_{\det}^{\otimes m}))=\mathrm{ch}_q(\bC[\bA^{(N)}])\cdot \mathrm{ch}_q(\mathcal L_{\det}^{\otimes m}|_{0}).$$ We note that $\mathcal L_{\det}|_{0}\cong \det(\mathbb C[z]/(z^N))$, so $\mathrm{ch}_q(\mathcal L_{\det}^{\otimes m}|_{0})=q^{\frac{N(N-1)}{2}m}$, thus
\begin{align}
    \mathrm{ch}_q(\Gamma(\bA^{(N)},\mathcal L_{\det}^{\otimes m}))=q^{\frac{N(N-1)}{2}m}\prod_{i=1}^N\frac{1}{1-q^i}.
\end{align}
\end{example}

\section{Quantization and Hilbert Space}\label{sec:Quantization and Hilbert Space}

Let us canonically quantize the matrix model as follows:
\begin{align}\label{eq: canonical quantization}
    [Z^i_j,Z^{\dag k}_l]=\delta^i_l\delta^k_j,\quad [\varphi^j_b,\varphi^{\dag a}_i]=\delta^j_i\delta^a_b,
\end{align}
and the (unconstrained) Hilbert space $\mathcal H_N(n)$ is generated from a distinguished vector $|0\rangle$ by the action of $Z,Z^\dag,\varphi,\varphi^\dag$ with relations:
\begin{align}
    Z|0\rangle=0,\quad \varphi|0\rangle=0.
\end{align}
\begin{notation}
For a cleaner presentation, let us introduce the following notation for the variables:
\begin{align}
    X^i_j:=Z^{\dag i}_j,\quad Y^i_j:=Z^i_j,\quad A^a_i:=\varphi^{\dag a}_i,\quad B^i_a:= \varphi^i_a.
\end{align}
\end{notation}
Then $\mathcal H_N(n)$ is the space of polynomial functions $\mathbb C[V(N,n)]$ on the affine space $V(N,n)=\mathrm{End}(\mathbb C^N)\times \mathrm{Hom}(\mathbb C^n,\mathbb C^N)$. The isomorphism is such that $f(X,A)|0\rangle$ is mapped to the function $f(X,A)\in \mathbb C[V(N,n)]$, and $Y, B$ acts on $\mathbb C[V(N,n)]$ by differential forms:
\begin{align}
    Y^i_j=\frac{\partial}{\partial X^j_i},\quad B^{i}_a=\frac{\partial}{\partial A^a_i}.
\end{align}
Then the algebra generated by $(X,Y,A,B)$ with relations \eqref{eq: canonical quantization} is isomorphic to the algebra of differential operators $D(V(N,n))$ on the affine variety $V(N,n)$. The dual of the moment map \eqref{eq: moment map} is complexified to a Lie algebra homomorphism $\mu^*_{\bC}:\gl_N\to D(V(N,n))$:
\begin{align}\label{complexified moment map}
    \mu^*_{\bC}(E^i_j)=X^i_l Y^l_j-X^l_j Y^i_l- A^a_j B^i_a,
\end{align}
where $E^i_j$ is the elementary matrix for the $i$-th row and $j$-th column. We note that the action of $\mu^*_{\bC}(E^i_j)$ on $\mathbb C[V(N,n)]$ is exactly the the action of $E^i_j$ on $\mathbb C[V(N,n)]$, i.e.
\begin{align*}
    \mu^*_{\bC}(x)(f)=x\cdot f,\quad\forall x\in \gl_N,\forall f\in \mathbb C[V(N,n)],
\end{align*}
where the latter action is obtained from the differential of the action of $\GL_N$ on $V(N,n)$.

\bigskip The physical Hilbert space is obtain from $\mathcal H_N(n)$ by imposing the constraint\footnote{Note that we have already normally ordered $(BA)^i_j$ to $A^a_jB^i_a$ in \eqref{complexified moment map}, which results in the shift of moment map value $k+n\mapsto k$.}
\begin{align}\label{eq:constraint}
    \mu^*_{\bC}(E^i_j)+k\delta^i_j=0.
\end{align}
Denote the physical Hilbert space corresponding to the Chern-Simons level $k$ by $\mathcal H_N(n,k)$, then we have
\begin{align}
    \mathcal H_N(n,k)=\mathbb C[V(N,n)]^{\GL_N,-k}.
\end{align}

\begin{lemma}[Quantization condition]\label{lem: quantization condition}
$\mathcal H_N(n,k)$ is nontrivial if and only if $k\in \mathbb Z_{\ge 0}$.
\end{lemma}

\begin{proof}
Let $\mathbb C^{\times}\subset\GL_N$ be the center of $\GL_N$, then $\mathbb C^{\times}$ acts on $X$ trivially and it acts on $\varphi$ with weight $-1$. It follows that $\mathcal H_N(n,k)=0$ if $k<0$. Let us show that $\mathcal H_N(n,k)$ is nontrivial if $k\in \mathbb Z_{\ge 0}$. If $k=0$, then $\mathcal H_N(n,0)=\mathbb C[\End(\bC^N)]^{\GL_N}=\bC[\bA^{(N)}]$ is nontrivial. If $k>0$, then this follows from the Theorem \ref{thm: Hilbert series} below.
\end{proof}

In the rest of this paper, we impose the assumption that\footnote{As we have seen in the proof of Lemma \ref{lem: quantization condition}, $\mathcal H_N(n,0)=\bC[\bA^{(N)}]$, which is nonzero, but not interesting for our purpose.}
\begin{equation*}
\boxed{k\in \mathbb Z_{>0}.}
\end{equation*}

\begin{example}\label{ex: H_N(1,k)}
In the case when $n=1$, the following element $\ket{\Omega_k}$ 
\begin{align}\label{eq: ground state n=1}
    \ket{\Omega_k}=\left[\epsilon^{i_{1}i_{2}\cdots i_{N}}(A)_{i_{1}}(A X)_{i_{2}}(A X^2)_{i_{3}}\cdots (A X^{N-1})_{i_{N}}\right]^{k}\ket{\emptyset},
\end{align}
belongs to the physical Hilbert space $\mathcal H_N(1,k)$ \cite{Dorey-Tong-Turner}. Here $\epsilon^{i_{1}i_{2}\cdots i_{N}}$ is the Levi-Civita tensor. Since $\mathbb C[V(N,n)]$ is an integral domain, it follows that $\mathcal H_N(1,k)$ contains a subspace $\mathbb C[\bA^{(N)}]\cdot \ket{\Omega_k}$ which is isomorphic to $\mathbb C[\bA^{(N)}]$. Later we will show that $\mathcal H_N(1,k)=\mathbb C[\bA^{(N)}]\cdot \ket{\Omega_k}$.
\end{example}

\subsection{Hermitian inner product on \texorpdfstring{$\mathcal H_N(n,k)$}{Hilbert space}}\label{subsec:Hermitian inner product}

$\mathcal H_N(n)$ is the Hilbert space of $(N^2+nN)$ dimensional harmonic oscillator, it is equipped with a canonical nondegenerate Hermitian inner product $\langle\cdot|\cdot\rangle$, which is uniquely determined by the following two properties:
\begin{enumerate}
    \item let $\ket{v}=f(X,A)\ket{\emptyset}$ where $f(X,A)\in \bC[V(N,n)]$, then $\langle \emptyset \ket{v}=f(0,0)$,
    \item let $\ket{v}=f(X,A)\ket{\emptyset},\ket{w}=g(X,A)\ket{\emptyset}$ where $f(X,A),g(X,A)\in \bC[V(N,n)]$, then $\langle w\ket{v}=\langle \emptyset|(g(X,A)^\dag\ket{v})$.
\end{enumerate}
Here the Hermitian conjugate is defined by
\begin{align*}
    (X^i_j)^\dag:=Y^j_i=\frac{\partial}{\partial X^i_j},\quad (A^a_i)^\dag:=B^{i}_a=\frac{\partial}{\partial A^a_i}.
\end{align*}
Then the moment map \eqref{complexified moment map} satisfies $\mu^*_{\bC}(E^i_j)^\dag=\mu^*_{\bC}(E^j_i)$. It follows that the Hermitian inner product $\langle\cdot|\cdot\rangle$ is $\mathrm{U}(N)$-invariant, i.e. $\langle g\cdot w\ket{ g\cdot v}=\langle  w\ket{ v}$ for all $\ket{v},\ket{w}\in \mathcal H_N(n)$ and all $g\in \mathrm{U}(N)$. $\mathcal H_N(n)$ decomposes as direct sum of $\mathrm{U}(N)$ modules
\begin{align*}
    \mathcal H_N(n)=\bigoplus_{\chi} V_\chi\otimes M_\chi
\end{align*}
where $V_\chi$ is the irreducible $\mathrm{U}(N)$ module of highest weight $\chi$, and $M_\chi$ is the multiplicity space of $V_\chi$ in $\mathcal H_N(n)$. These direct summands are orthogonal to each other by $\mathrm{U}(N)$-invariance of $\langle\cdot|\cdot\rangle$, i.e.
\begin{align*}
    \langle V_\chi\otimes M_\chi\ket{V_{\chi'}\otimes M_{\chi'}}=0,\;\text{whenever }\chi\neq \chi'.
\end{align*}
This implies that for all summands $V_\chi\otimes M_\chi$ the restriction $\langle\cdot|\cdot\rangle|_{V_\chi\otimes M_\chi}$ is nondegenerate, because $\langle\cdot|\cdot\rangle$ is nondegenerate. Let us take $V_\chi$ to be the one dimensional representation that $g\in \mathrm{U}(N)$ acts by $\det(g)^{-k}$, then $V_\chi\otimes M_\chi=\mathcal H_N(n,k)$, then the restriction $\langle\cdot|\cdot\rangle|_{\mathcal H_N(n,k)}$ is nondegenerate.

\subsection{Hilbert space as global sections of line bundle on \texorpdfstring{$\mathcal M(N,n)$}{M(N,n)}}

Using the GIT description \eqref{eq: M(N,n) as GIT quotient}, there is a canonical map:
\begin{align}\label{eq: global to local}
    \mathbb C[V(N,n)]^{\GL_N,-k}\to \Gamma(\cM(N,n),\mathcal L_{\det}^{\otimes k}).
\end{align}
In fact, 
\begin{align*}
    \Gamma(\cM(N,n),\mathcal L_{\det}^{\otimes k})=\Gamma(R(N,n),q^*\mathcal L_{\det}^{\otimes k})^{\GL_N},
\end{align*}
where $q:R(N,n)\to \cM(N,n)$ is the quotient map. Notice that $q^*\mathcal L_{\det}$ is the trivial bundle with $\GL_N$-equivariant structure given by the character $g\mapsto \det(g)$, thus we have
\begin{align*}
    \Gamma(R(N,n),q^*\mathcal L_{\det}^{\otimes k})^{\GL_N}\cong \Gamma(R(N,n),\mathcal O_{R(N,n)})^{\GL_N,-k}.
\end{align*}
Then \eqref{eq: global to local} is given by the restriction of functions on $V(N,n)$ to open subset $R(N,n)$.

\begin{proposition}\label{prop: H(N,n) = global section}
The map \eqref{eq: global to local} is an isomorphism, in particular we have isomorphism
\begin{align}\label{eq: H(N,n) = global section}
    \mathcal H_N(n,k)\cong \Gamma(\cM(N,n),\mathcal L_{\det}^{\otimes k}).
\end{align}
\end{proposition}

\begin{proof}
Since $R(N,n)$ is open and dense in $V(N,n)$, the map \eqref{eq: global to local} is injective by construction. It remains to show its surjectivity. We discuss the two situations $n=1$ and $n>1$ separately.

\bigskip $\bullet\; n=1$. It is shown in Example \ref{ex: H_N(1,k)} that $\mathcal H_N(1,k)\supset\mathbb C[\bA^{(N)}]\cdot \ket{\Omega_k}$ where $\ket{\Omega_k}$ is given by \eqref{eq: ground state n=1}, so we get an embedding
\begin{align}\label{embedding}
    \mathbb C[\bA^{(N)}]\cdot \ket{\Omega_k}\hookrightarrow  \Gamma(\cM(N,1),\mathcal L_{\det}^{\otimes k}).
\end{align}
Consider the $\bC_q^{\times}$ action on $\bA^1$ by scaling with weight $-1$, then $\ket{\Omega_k}$ is a $\bC_q^{\times}$ eigenvector of weight $\frac{N(N-1)}{2}k$, thus
\begin{align*}
    \mathrm{ch}_q(\mathbb C[\bA^{(N)}]\cdot \ket{\Omega_k})=q^{\frac{N(N-1)}{2}k}\prod_{i=1}^N\frac{1}{1-q^i}=\mathrm{ch}_q(\Gamma(\bA^{(N)},\mathcal L_{\det}^{\otimes k})),
\end{align*}
which implies that the embedding \eqref{embedding} must be an isomorphism.

\bigskip $\bullet\; n>1$. Consider the $\SL_N$-quotient $V(N,n)\sslash\SL_N=\Spec\:\bC[V(N,n)]^{\SL_N}$ and denote it by $\widetilde{\cM}(N,n)$, then $\bC^{\times}\cong \GL_N/\SL_N$ naturally acts on $\widetilde{\cM}(N,n)$, and we have
\begin{align}\label{eq: C[V(N,n)] in terms of tilde M}
    \mathbb C[V(N,n)]^{\GL_N,-k}=\mathbb C[\widetilde{\cM}(N,n)]^{\bC^{\times},-k}.
\end{align}
Moreover, the $\mathbb C^{\times}$-weights in $\mathbb C[\widetilde{\cM}(N,n)]$ are non-positive, thus there is an natural embedding $\iota:\bA^{(N)}\cong \Spec\:\mathbb C[\widetilde{\cM}(N,n)]^{\bC^{\times}}\hookrightarrow \widetilde{\cM}(N,n)$ such that $\mathrm{Im}(\iota)=\widetilde{\cM}(N,n)^{\bC^{\times}}$. The stable locus of $\bC^{\times}$-action on $\widetilde{\cM}(N,n)$ is the complement of fixed point set $\widetilde{\cM}(N,n)^{\bC^{\times}}$ in $\widetilde{\cM}(N,n)$, thus we have
\begin{align}\label{eq: Gamma(L^k) in terms of tilde M}
    \Gamma(\cM(N,n),\mathcal L_{\det}^{\otimes k})\cong \Gamma\left(\widetilde{\cM}(N,n)\setminus \widetilde{\cM}(N,n)^{\bC^{\times}},\mathcal O_{\widetilde{\cM}(N,n)}\right)^{\bC^{\times},-k}.
\end{align}
In view of \eqref{eq: C[V(N,n)] in terms of tilde M} and \eqref{eq: Gamma(L^k) in terms of tilde M}, the map \eqref{eq: global to local} is given by the restriction of functions on $\widetilde{\cM}(N,n)$ to the open subset $\widetilde{\cM}(N,n)\setminus \widetilde{\cM}(N,n)^{\bC^{\times}}$. We notice that $\widetilde{\cM}(N,n)$ is a normal variety (because $V(N,n)$ is smooth), and
\begin{align*}
    \dim \widetilde{\cM}(N,n)=\dim\left( \widetilde{\cM}(N,n)\setminus \widetilde{\cM}(N,n)^{\bC^{\times}}\right)=\dim {\cM}(N,n)+1=nN+1
\end{align*}
and
\begin{align*}
    \dim \widetilde{\cM}(N,n)^{\bC^{\times}}=\dim\bA^{(N)}=N,
\end{align*}
thus $\widetilde{\cM}(N,n)^{\bC^{\times}}$ has codimension at least $2$ in $\widetilde{\cM}(N,n)$. Then it follows from algebraic Hartogs theorem that the restriction map
\begin{align*}
    \Gamma\left(\widetilde{\cM}(N,n),\mathcal O_{\widetilde{\cM}(N,n)}\right)\to \Gamma\left(\widetilde{\cM}(N,n)\setminus \widetilde{\cM}(N,n)^{\bC^{\times}},\mathcal O_{\widetilde{\cM}(N,n)}\right)
\end{align*}
is an isomorphism. This finishes the proof of the proposition.
\end{proof}

\begin{definition}\label{def:energy grading}
Let us consider the $\bC_q^{\times}$ action of $V(N,n)$ by scaling the matrices $\End(\bC^{N})$ with weight $-1$ and fixes $\Hom(\bC^n,\bC^N)$, then this $\bC_q^{\times}$ action descends to $\cM(N,n)$. We note that this $\bC_q^{\times}$ action on $\cM(N,n)$ agrees with the $\bC_q^{\times}$ action on $\Quot^N(\mathcal O_{\bA^1}^{\oplus n})$ that is induced from scaling $\bA^1$ with weight $-1$. $\mathcal L_{\det}$ is a $\bC_q^{\times}$-equivariant line bundle, so $\mathcal H_N(n,k)$ has a natural $\bC_q^{\times}$-module structure. We shall call the grading on $\mathcal H_N(n,k)$ induced by this $\bC_q^{\times}$-action the \textit{energy grading}. Explicitly, the energy grading is given by setting
\begin{align}\label{eq:energy grading}
    \deg X=1,\quad \deg A=0.
\end{align}
\end{definition}

We will see in the Theorem \ref{thm: Hilbert series} that $\mathcal H_N(n,k)$ is positive under the energy grading and every weight space of $\mathcal H_N(n,k)$ is finite dimensional. 

Before we proceed to the statement, let us introduce the flavor symmetry, which is the $\GL_n$-action on $V(N,n)$ via dual vector representation on $\Hom(\bC^n,\bC^N)$ and trivial on $\End(\bC^{N})$. The $\GL_n$-action commutes with the aforementioned $\bC_q^{\times}$-action, and it descends to a $\GL_n$-action on $\cM(N,n)$, which makes $\mathcal L_{\det}$ is a $\GL_n\times \bC_q^{\times}$-equivariant line bundle, so $\mathcal H_N(n,k)$ has a natural $\GL_n\times \bC_q^{\times}$-module structure.

The canonical map \eqref{eq: global to local} is $\GL_n\times \bC_q^{\times}$-equivariant, thus \eqref{eq: H(N,n) = global section} is a $\GL_n\times \bC_q^{\times}$-equivariant isomorphism.

\subsection{Character (partition function)}

Let $P$ denote the weight lattice of $\GL_n$. $P$ is the group of characters for the maximal torus $T\cong \bC^{\times n}\subset\GL_n$, so $P\cong\bZ^n$. Then for a $\GL_n\times \bC_q^{\times}$-module $V$, it decomposes into a direct sum
\begin{align*}
    V=\bigoplus _{\substack{n\in \mathbb Z\\ \lambda\in P}} V_{n,\lambda},
\end{align*}
where $V_{n,\lambda}$ is the subspace of $V$ such that $\bC_q^{\times}$-weight is $n$ and $T$-weight is $\lambda$. We define
\begin{align}\label{def:character}
    \mathrm{ch}_{q,\mathbf{a}}(V)=\sum_{\substack{n\in \mathbb Z\\ \lambda\in P}} \dim(V_{n,\lambda})q^n \mathbf{a}^\lambda
\end{align}
where $\mathbf{a}^\lambda$ is the short-hand notation for $\mathbf{a}_1^{\lambda_1}\cdots \mathbf{a}_n^{\lambda_n}$. Here we assume that $\dim(V_{n,\lambda})<\infty$ for all pairs $(n,\lambda)$, and also assume that $\bC_q^{\times}$-weights of $V$ are bounded from below, so that we can regard $\mathrm{ch}_{q,\mathbf{a}}(V)$ as an element in $\mathbb Z[\mathbf{a}_1^{\pm},\cdots,\mathbf{a}_n^{\pm}]^{S_n}[\![q]\!]$. The following result first appears in \cite{dorey2016matrix}, and we provide a geometric proof in this paper.

\begin{theorem}\label{thm: Hilbert series}
As an element in $\mathbb Z[\mathbf{a}_1^{\pm},\cdots,\mathbf{a}_n^{\pm}]^{S_n}[\![q]\!]$,
\begin{align}\label{eq: character of Hilbert space}
    \mathrm{ch}_{q,\mathbf{a}}(\mathcal H_N(n,k))=H_{(k^N)}(\mathbf{a};q)\prod_{i=1}^N\frac{1}{1-q^i}.
\end{align}
Here $H_{(k^N)}(\mathbf{a};q)$ is the transformed Hall-Littlewood polynomial\footnote{See Appendix \ref{sec:hall-littlewood} for a review of (transformed) Hall-Littlewood polynomial.} associated to the partition $(k^N)$.
\end{theorem}

\begin{proof}
According to Proposition \ref{prop: H(N,n) = global section}, we have $\mathrm{ch}_{q,\mathbf{a}}(\mathcal H_N(n,k))=\mathrm{ch}_{q,\mathbf{a}}(\Gamma(\cM(N,n),\mathcal L_{\det}^{\otimes k}))$. By the Lemma \ref{lem: pushforward of line bundle}, $H^i(\cM(N,n),\mathcal L_{\det}^{\otimes k})$ vanishes for $i>0$, so we have
\begin{align*}
    \mathrm{ch}_{q,\mathbf{a}}(\mathcal H_N(n,k))=\chi_{q,\mathbf{a}}(\cM(N,n),\mathcal L_{\det}^{\otimes k}),
\end{align*}
where the right-hand-side is the $\GL_n\times \bC_q^{\times}$-equivariant Euler characteristic. Apply the $\bC_q^{\times}$-equivariant localization formula to the sheaf $h_*\mathcal L_{\det}^{\otimes k}$ on $\bA^{(N)}$, we obtain
\begin{align*}
    \chi_{q,\mathbf{a}}(\cM(N,n),\mathcal L_{\det}^{\otimes k})=\chi_{q,\mathbf{a}}(\bA^{(N)},Rh_*\mathcal L_{\det}^{\otimes k})=\frac{\chi_{q,\mathbf{a}}(h^{-1}(0),\mathcal L_{\det}^{\otimes k}|_{h^{-1}(0)})}{\mathrm{ch}_{q,\mathbf{a}}(\bigwedge^{\bullet} T_0^*\bA^{(N)})}.
\end{align*}
Here $T_0^*\bA^{(N)}$ is the cotangent space of $\bA^{(N)}$ at $0$, and $\mathrm{ch}_{q,\mathbf{a}}(T_0^*\bA^{(N)})=q+q^2+\cdots+q^N$. For the numerator, we notice that $h^{-1}(0)\cong \overline{\Gr}^{N \omega_1}_{\GL_N}$ and $\mathcal L_{\det}|_{h^{-1}(0)}\cong\mathcal O(1)$, and it is proven in \cite[Corollary B.3]{moosavian2021towards} that
\begin{align}
    \chi_{q,\mathbf{a}}(\overline{\Gr}^{N \omega_1}_{\GL_n},\mathcal O(k))=H_{(k^N)}(\mathbf{a};q),
\end{align}
whence \eqref{eq: character of Hilbert space} follows. We summarize the idea of the proof of \cite[Corollary B.3]{moosavian2021towards} in Appendix \ref{subsec:Geometric description of Jing operator}.
\end{proof}

\subsection{Ground states}

\begin{definition}
The subspace of $\mathcal H_N(n,k)$ which has the lowest $\bC_q^{\times}$-weights is called the space of ground states, denoted by $\mathcal H_N(n,k)_0$.
\end{definition}

It is shown in \cite{Dorey-Tong-Turner} that $\mathcal H_N(n,k)_0$ is spanned by elements of the form 
\begin{align}\label{eq: ground states general n}
    \ket{a_1,\cdots,a_N}=\left[\epsilon^{i_{1}i_{2}\cdots i_{N}}\prod_{j=1}^{N}(A^{a_{j}} X^{\lfloor\frac{j-1}{n}\rfloor})_{i_{j}}\right]^{k}\ket{\emptyset}.
\end{align}
In this subsection we give a geometric description of $\mathcal H_N(n,k)_0$.

\begin{proposition}\label{prop: ground states}
Let $L=\lfloor\frac{N}{n}\rfloor$ and $r=N-nL$, then 
\begin{align}\label{leading term}
    \ch_{q,\mathbf{a}}(\mathcal H_N(n,k)_0)=\mathfrak{A}^{kL} s_{k\varpi_r}(\mathbf a)q^{\frac{k}{2}L(L-1)n+krL}.
\end{align}
Here $\varpi_r$ is the $r$-th fundamental weight of $\GL_n$ and $s_{k\varpi_r}(\mathbf a)$ is the Schur polynomial associated to the weight $k\varpi_r$, and $\mathfrak{A}:=\prod_{i=1}^n \mathbf a_i$.
\end{proposition}

\begin{proof}
The idea is to use the $\bC_q^{\times}$-localization on $\cM(N,n)$ to analyse the leading order term in the $q$ expansion of $\chi_{q,\mathbf{a}}(\cM(N,n),\mathcal L_{\det}^{\otimes k})$. To this extend, we make use of the affine Grassmannian description $\cM(N,n)\cong \overline{\Gr}^{\omega_1,\cdots,\omega_1}_{\GL_{n},\bA^{(N)}}$ and the Quot scheme description $\cM(N,n)\cong \mathrm{Quot}^{N}({\mathcal O}^{\oplus n}_{{\mathbb A}^{1}})$ in the Section \ref{sec:phase space}.

First of all, the $\bC_q^{\times}$-fixed points of $\overline{\Gr}^{\omega_1,\cdots,\omega_1}_{\GL_{n},\bA^{(N)}}$ agrees with the $\bC_q^{\times}$-fixed points of the central fiber $\pi^{-1}(0)=\overline{\Gr}^{N \omega_1}_{\GL_n}$, and we have
\begin{align*}
    \pi^{-1}(0)^{\bC_q^{\times}}=\coprod_{\substack{\mu\dashv N\\ l(\mu)\le n}}F_\mu,\quad F_\mu=\GL_n\cdot \{z^{\mu}\}.
\end{align*}
Here $\mu_1\ge \mu_2\ge \cdots\ge \mu_n$ is a partition of $N$, and $\{z^{\mu}\}$ is the point on $\overline{\Gr}^{N \omega_1}_{\GL_n}$ corresponding to the subsheaf $$\mathcal F_\mu=\bigoplus_{i=1}^n z^{\mu_i}\cdot \mathcal O_{\bA^1}\subset \mathcal O_{\bA^1}^{\oplus n},$$ and $\GL_n\cdot \{z^{\mu}\}$ is the $\GL_n$-orbit through $\{z^{\mu}\}$. The fiber of $\mathcal O(1)$ at $z^{\mu}$ has $\bC_q^{\times}$-weight
\begin{align}\label{eq: E(mu)}
    E(\mu)=\sum_{i=1}^n\frac{\mu_i(\mu_i-1)}{2}.
\end{align}
Using the $\bC_q^{\times}$-equivariant localization, we have
\begin{align}\label{Contribution of fixed locus}
\ch_{q,\mathbf{a}}(\mathcal H_N(n,k))=\sum_{\substack{\mu\dashv N\\ l(\mu)\le n}}q^{E(\mu)}\sum_{w\in S_n/S_n^{\mu}}w\left(\frac{\mathbf{a}^{k\mu}}{\bigwedge^{\bullet} T_{z^{\mu}}^*(F_{\mu})}\prod_{i}\frac{1}{1-q^{a_i}A_i}\prod_{j}\frac{1}{1-q^{-b_j}B_j}\right).
\end{align}
Here $S^{\mu}_n$ is the subgroup of the permutation group $S_n$ such that $\mu$ is fixed under $S^{\mu}_n$, $T_{z^{\mu}}^*(F_{\mu})$ is the cotangent space of $F_{\mu}$ at $\{z^{\mu}\}$, $a_i,b_j$ are positive integers and $A_i,B_j\in \mathbb Z[\mathbf{a}_1^{\pm},\cdots,\mathbf{a}_n^{\pm}]$. Note that $1/(1-q^{a_i}A_i)$ (resp. $1/(1-q^{-b_j}B_j)$) terms correspond to repelling (resp. contracting) direction of $\bC_q^{\times}$-action. It is possible that there is no contracting direction, in this case the $\prod_{j}\frac{1}{1-q^{-b_j}B_j}$ term in \eqref{Contribution of fixed locus} is $1$. In the $q$ expansion of \eqref{Contribution of fixed locus}, each summand has leading order term
\begin{align}
    q^{E(\mu)+R(\mu)}\sum_{w\in S_n/S_n^{\mu}}w\left(\frac{\mathbf{a}^{k\mu}}{\bigwedge^{\bullet} T_{z^{\mu}}^*(F_{\mu})}\prod_{j}(-B_j)^{-1}\right),\quad \text{where }R(\mu)=\sum_j b_j.
\end{align}
We claim that
\begin{itemize}
    \item[(1)] The minimum value of $E(\mu)$ is obtained exactly when $\mu$ equals to $\lambda=\varpi_r+L\cdot\varpi_n$,
    \item[(2)] There is no contracting direction at $\{z^{\lambda}\}$, i.e. $R(\lambda)=0$.
\end{itemize}
To prove the first claim, we rewrite \eqref{eq: E(mu)} using the identity $\sum_{i=1}^N\mu_i=N$:
\begin{align*}
    E(\mu)=\frac{(\mu,\mu)}{2}-N,
\end{align*}
where $(\cdot,\cdot)$ is the killing form on the coweight space of $\GL_n$. We note that $\varpi_r$ is minuscule coweight, so every dominant coweight $\mu$ with $|\mu|=N$ can be written as $\mu=\lambda+\beta$, where $\beta=\sum_{i=1}^{n-1}m_i\alpha_i$ is a linear combination of simple coroots $\{\alpha_i\}_{i=1}^{n-1}$ such that $m_i\ge 0$ for all $i$. Since $\lambda$ is dominant, we have $(\lambda,\beta)\ge 0$, therefore
\begin{align*}
    E(\mu)=\frac{(\lambda,\lambda)}{2}-N+\frac{2(\lambda,\beta)+(\beta,\beta)}{2}\ge \frac{(\lambda,\lambda)}{2}-N.
\end{align*}
$E(\mu)$ obtains its minimal value exactly when $\beta=0$, i.e. $\mu=\lambda$. 

To prove the second claim, we notice that the tangent space of $\mathrm{Quot}^N(\mathbb A^1,\mathcal O_{\mathbb A^1}^{\oplus n})$ at $\{z^{\mu}\}$ is
\begin{align*}
    \Hom_{\mathcal O_{\mathbb A^1}}(\mathcal F_{\mu},\mathcal O_{\bA^1}^{\oplus n}/\mathcal F_{\mu}).
\end{align*}
It is elementary to see that when $\mu=\lambda=\varpi_r+L\cdot\varpi_n$, $ \Hom_{\mathcal O_{\mathbb A^1}}(\mathcal F_{\lambda},\mathcal E_{\lambda})$ has no positive $\bC_q^{\times}$-weight space, so there is no contracting direction at $\{z^{\lambda}\}$. Thus we see that the leading term of $ \ch_{q,\mathbf{a}}(\mathcal H_N(n,k))$ in the $q$ expansion is
\begin{align*}
    q^{E(\lambda)}\sum_{w\in S_n/S_n^{\lambda}}w\left(\frac{\mathbf{a}^{k\lambda}}{\bigwedge^{\bullet} T_{z^{\lambda}}^*(F_{\lambda})}\right)=q^{\frac{k}{2}L(L-1)n+krL}\chi_{q,\mathbf{a}}(F_\lambda,\mathcal L_{k\lambda}).
\end{align*}
Here $\mathcal L_{k\lambda}$ is the $\GL_n$ equivariant line bundle $\GL_n\times_{P_\lambda}\bC_{k\lambda}$ on $F_\lambda\cong \GL_n/{P_\lambda}$, where $P_\lambda$ is the stabilizer of $\{z^\lambda\}$ in $\GL_n$ and $P_\lambda$ acts on $\bC_{k\lambda}$ by the character $k\lambda$. Using the Borel-Weil-Bott theorem, we have $\chi_{q,\mathbf{a}}(F_\lambda,\mathcal L_{k\lambda})=s_{k\lambda}(\mathbf{a})$. It is elementary to see that $s_{k\lambda}(\mathbf{a})=\mathfrak{A}^{kL}s_{k\varpi_r}(\mathbf{a})$, whence \eqref{leading term} follows.
\end{proof}

\begin{remark}\label{rmk: ground states}
The proof of Proposition \ref{prop: ground states} shows that $\mathcal H_N(n,k)_0$ maps isomorphically onto $\Gamma(F_\lambda,\mathcal L_{\det}^{\otimes k}|_{F_\lambda})$ along the restriction map $\Gamma(\cM(N,n),\mathcal L_{\det}^{\otimes k})\to\Gamma(F_\lambda,\mathcal L_{\det}^{\otimes k}|_{F_\lambda})$, where $\lambda=\varpi_r+L\varpi_n$ and $F_\lambda=\GL_n\cdot \{z^\lambda\}$. Moreover, \eqref{leading term} implies that $\mathcal H_N(n,k)_0$ is an irreducible $\GL_{n}$-module of highest weight $k\varpi_r+kL\varpi_n$, and the ground states energy is
\begin{align}
    E_0=\frac{k}{2}L(L-1)n+krL.
\end{align}
The restriction map $\Gamma(F_\lambda,\mathcal L_{\det}^{\otimes k}|_{F_\lambda})\to \mathcal L_{\det}^{\otimes k}|_{\{z^\lambda\}}$ projects the ground states to highest weight space.
\end{remark}

\subsection{Restriction to torus fixed points}\label{subsec:Restriction to torus fixed points}

Let $T=\bC^{\times n}\subset \GL_n$ be the maximal torus, then the $T$-fixed points of $\cM(N,n)$ can be obtained using \cite[Propsoition 2.3.1]{maulik2012quantum}:
\begin{equation}\label{T-fixed pts}
\begin{split}
    \cM(N,n)^T&=\coprod_{N_1+\cdots+N_n=N}\cM(N_1,1)\times\cdots\times \cM(N_n,1)\\
    &\cong \coprod_{N_1+\cdots+N_n=N}\bA^{(N_1)}\times\cdots\times \bA^{(N_n)}.
\end{split}
\end{equation}
We note that each connected component $\bA^{(N_1)}\times\cdots\times \bA^{(N_n)}$ projects to the base $\bA^{(N)}$ via the symmetrization map. In particular, $h|_{\cM(N,n)^T}: \cM(N,n)^T\to \bA^{(N)}$ is a flat morphism.

\bigskip We have a natural restriction map
\begin{align}
    \mathfrak{R}:\Gamma(\cM(N,n),\mathcal L_{\det}^{\otimes k})\longrightarrow \Gamma(\cM(N,n)^T,\mathcal L_{\det}^{\otimes k}|_{\cM(N,n)^T}),
\end{align}
which is $T[z]\rtimes\bC_q^{\times}$-equivariant.

\begin{proposition}\label{prop: restrict to fixed pts}
The restriction map $\mathfrak{R}$ is surjective. Moreover, if $k=1$ then $\mathfrak{R}$ is an isomorphism.
\end{proposition}

\begin{proof}
In the view of the isomorphism $\cM(N,n)\cong \overline{\Gr}^{\omega_1,\cdots,\omega_1}_{\GL_{n},\bA^{(N)}}$, the proposition can be regarded as part of the induction procedure in the proof of \cite[Theorem 0.2.2]{zhu2009affine}, see \cite[2.1.3]{zhu2009affine}. We provide an exposition here for transparency. Since $h|_{\cM(N,n)^T}: \cM(N,n)^T\to \bA^{(N)}$ is a flat morphism, the pushforward $\left(h|_{\cM(N,n)^T}\right)_*(\mathcal L^{\otimes k}_{\det}|_{\cM(N,n)^T})$ is a locally free sheaf on $\bA^{(N)}$, and its rank is
\begin{align*}
    \sum_{N_1+\cdots+N_n=N}\frac{N!}{N_1!\cdots N_n!}=n^N.
\end{align*}
By Lemma \ref{lem: pushforward of line bundle}, $h_* (\mathcal L^{\otimes k}_{\det})$ is a locally free sheaf on $\bA^{(N)}$, and its rank is ${\binom{n+k-1}{k}}^N$. Since $\mathfrak R$ is the induced map by applying $\Gamma(\bA^{(N)},-)$ to $\mathfrak{r}:h_* (\mathcal L^{\otimes k}_{\det})\to \left(h|_{\cM(N,n)^T}\right)_*(\mathcal L^{\otimes k}_{\det}|_{\cM(N,n)^T})$, we only need to show that $\mathfrak{r}$ is surjective. By construction, $\mathfrak{r}$ is $\bC_q^{\times}$-equivariant, thus we only need to show that $$\mathfrak{r}|_{\{0\}}:\Gamma(\bGr^{N\omega_1}_{\GL_n},\mathcal O(k))\to \Gamma\left(\bGr^{N\omega_1,T}_{\GL_n},\mathcal O(k)|_{\bGr^{N\omega_1,T}_{\GL_n}}\right)$$ is surjective, which is a special case of \cite[2.1.1]{zhu2009affine}. This finishes the proof.
\end{proof}

\begin{corollary}\label{cor:k=1 fixed pt}
If $k=1$, then we have isomorphism between $T[z]\rtimes\bC_q^{\times}$-modules:
\begin{align}
    \mathcal H_N(n,1)\cong \bigoplus_{N_1+\cdots+N_n=N} \mathcal H_{N_1}(1,1)\otimes\cdots\otimes\mathcal H_{N_n}(1,1).
\end{align}
In particular, we have
\begin{align}\label{eq: character of Hilbert space k=1}
    \mathrm{ch}_{q,\mathbf{a}}(\mathcal H_N(n,1))=\sum_{N_1+\cdots+N_n=N}\prod_{l=1}^n\left(\mathbf a_l^{N_l}q^{\frac{N_l(N_l-1)}{2}}\prod_{i=1}^{N_l}\frac{1}{1-q^i}\right).
\end{align}
\end{corollary}

\begin{corollary}\label{cor: restrict to A-fixed pts}
Let $A\subset T$ be a subtorus, then the natural restriction to $A$-fixed locus map:
\begin{align}\label{restrict to A-fixed pts}
    \mathfrak{R}_A:\Gamma(\cM(N,n),\mathcal L_{\det})\longrightarrow \Gamma(\cM(N,n)^A,\mathcal L_{\det}|_{\cM(N,n)^A}),
\end{align}
is a $T[z]\rtimes\bC_q^{\times}$-equivariant isomorphism.
\end{corollary}

\begin{proof}
Since $A$ commutes with $T[z]\rtimes\bC_q^{\times}$, $\Gamma(\cM(N,n)^A,\mathcal L_{\det}|_{\cM(N,n)^A})$ possesses a $T[z]\rtimes\bC_q^{\times}$-module structure such that $\mathfrak{R}_A$ is $T[z]\rtimes\bC_q^{\times}$-equivariant. To show that $\mathfrak{R}_A$ is an isomorphism, we notice that \eqref{restrict to A-fixed pts} fits into a commutative diagram of restriction maps:
\begin{equation*}
\begin{tikzcd}
\Gamma(\cM(N,n),\mathcal L_{\det})\ar[dr,"\mathfrak{R}_T" '] \ar[r,"\mathfrak{R}_A"] & \Gamma\left(\cM(N,n)^{A},\mathcal L_{\det}|_{\cM(N,n)^{A}}\right) \ar[d,"\mathfrak{R}_{T/A}"]\\
& \Gamma(\cM(N,n)^{T},\mathcal L_{\det}|_{\cM(N,n)^{T}}) 
\end{tikzcd}
\end{equation*}
$\mathfrak{R}_{T}$ and $\mathfrak{R}_{T/A}$ are isomorphisms according to Proposition \ref{prop: restrict to fixed pts}, thus $\mathfrak{R}_{A}$ is also an isomorphism. 
\end{proof}

\subsection{Fermion Fock space}\label{sec:fermion fock}

Consider $n$ pairs of free fermion oscillators $\{\psi^a_m,\psi^*_{a,m}\:|\: 1\le a\le n,m\in \mathbb Z_{\ge 0}\}$, which satisfy anti-commutation relations
\begin{align*}
    \{\psi^a_m,\psi^*_{b,m'}\}=\delta^a_b\delta_{mm'},\quad \{\psi^a_m,{\psi}^b_{m'}\}=0, \quad \{\psi^*_{a,m},\psi^*_{b,m'}\}=0.
\end{align*}
The algebra generated by these fermions with the above relations is denoted by $\mathrm{Cl}(n)$. The fermion Fock space $\mathcal F(n)$ is a $\mathrm{Cl}(n)$-module generated by $\ket{\emptyset}$ with relations $\psi^*_{a,m}\ket{\emptyset}=0$, for all $a,m$. 

\bigskip There is an algebra homomorphism from $U(\gl_n[z])$ to a completion $ \mathrm{Cl}(n)^{\wedge}$ which is determined by
\begin{align}
    E^a_b\otimes z^m\mapsto \sum_{\ell=0}^{\infty} \psi^a_{\ell+m}\psi^*_{b,\ell}.
\end{align}
This algebra map induces a $\GL_n[z]$-module structure on the Fock space $\mathcal F(n)$, then the latter has eigenspace decomposition with respect to the action of $\mathbf{1}\otimes z^0$:
\begin{align}\label{weight decomposition of F(n)}
    \mathcal F(n)=\bigoplus_{N\ge 0}\mathcal F_N(n).
\end{align}
We note that $\mathcal F_N(n)$ is the subspace of $\mathcal F(n)$ spanned by elements of the form ${\psi}^{a_1}_{m_1}\cdots{\psi}^{a_N}_{m_N}\ket{\emptyset}$.

\bigskip We can decompose $\mathcal F(n)$ into a tensor product of its factors $\mathcal F(n)=\mathcal F(1)^{\otimes n}$, where the $a$-th tensor component $\mathcal F(1)$ is the Fock space for the $a$-th pair $\{\psi^a_m,\psi^*_{a,m}\:|\: m\in \mathbb Z_{\ge 0}\}$. In view of this tensor decomposition, each direct summand on the right-hand-side of \eqref{weight decomposition of F(n)} decomposes as:
\begin{align}\label{flavor decomposition of F_N(n)}
    \mathcal F_N(n)= \bigoplus_{N_1+\cdots+N_n=N} \mathcal F_{N_1}(1)\otimes\cdots\otimes\mathcal F_{N_n}(1).
\end{align}
We give a $\mathbb Z$-grading (i.e. $\bC_q^{\times}$-equivariant structure) on $\mathrm{Cl}(n)$ by setting 
\begin{align*}
    \deg\psi^a_p=p,\quad\deg\psi^*_{b,q}=-q.
\end{align*}
Then $\mathrm{Cl}(n)$ is a graded algebra and $\mathcal F(n)$ is a graded $\mathrm{Cl}(n)$-module. We note that $\deg(E^a_b\otimes z^m)=m$ in this grading, which is compatible with the natural grading on $\gl_n[z]$. 

We can compute the $q$-character of $\mathcal F_N(1)$ as follows:
\begin{align*}
    \mathrm{ch}_q(\mathcal F_N(1))=\sum_{m_1>\cdots>m_N\ge 0}q^{\sum_{i=1}^N m_i}=q^{\frac{N(N-1)}{2}}\prod_{i=1}^N\frac{1}{1-q^i}.
\end{align*}
Using the decomposition \eqref{flavor decomposition of F_N(n)} we have
\begin{align}\label{eq: character of F_N(n)}
\begin{split}
\mathrm{ch}_{q,\mathbf{a}}(\mathcal F_N(n))&=\sum_{N_1+\cdots+N_n=N}\prod_{l=1}^n\left(\mathbf a_l^{N_l}\mathrm{ch}_q(\mathcal F_{N_l}(1))\right)\\
\text{\small by \eqref{eq: character of Hilbert space k=1}}\quad &=\mathrm{ch}_{q,\mathbf{a}}(\mathcal H_N(n,1)).
\end{split}
\end{align}

\begin{proposition}\label{prop: fock space}
The linear map $\mathfrak{f}:\mathcal F_N(n)\to \mathcal H_N(n,1)$ which maps the bases by
\begin{align}\label{eq: fock space}
    {\psi}^{a_1}_{m_1}\cdots{\psi}^{a_N}_{m_N}\ket{\emptyset}\mapsto \epsilon^{i_{1}i_{2}\cdots i_{N}}(A^{a_1}X^{m_1})_{i_{1}}\cdots (A^{a_N}X^{m_N})_{i_{N}}\ket{\emptyset}
\end{align}
is a $\GL_n[z]\rtimes\bC_q^{\times}$-equivariant isomorphism.
\end{proposition}

\begin{proof}
It is straightforward to compute that $\mathfrak{f}$ respects the $\mathbb Z$-grading and intertwines the $\gl_n[z]$-actions on $\mathcal F_N(n)$ and $\mathcal H_N(n,1)$ respectively, so $\mathfrak{f}$ is $\GL_n[z]\rtimes\bC_q^{\times}$-equivariant. Let $L=\lfloor\frac{N}{n}\rfloor$ and $r=N-nL$, then 
\begin{align*}
    \mathfrak{f}( {\psi}^{1}_{0}\cdots\psi^{r}_0\psi^{1}_{1}\cdots\psi^{n}_1\cdots{\psi}^{1}_{L}\cdots{\psi}^{n}_{L}\ket{\emptyset})
\end{align*}
is a nonzero element in the ground states $\mathcal H_N(n,1)_0$ \cite{Dorey-Tong-Turner}. By the $\gl_n[z]$-equivariant property of $\mathfrak{f}$ and Corollary \ref{cor: cyclic}, $\mathfrak{f}$ is surjective. Since $\mathcal F_N(n)$ has the same character as $\mathcal H_N(n,1)$ by \eqref{eq: character of F_N(n)}, $\mathfrak{f}$ is an isomorphism.
\end{proof}

\begin{remark}
The isomorphism $\mathfrak{f}$ is not an isometry with respect to the natural Hermitian structures on $\mathcal F_N(n)$ and on $\mathcal H_N(n,1)$. To see this, let us compute the commutator between $\mathbf 1\otimes z=\sum_{\ell=0}^{\infty} \psi^a_{\ell+1}\psi^*_{a,\ell}$ and its Hermitian conjugate:
\begin{align*}
    [(\mathbf 1\otimes z)^\dag,\mathbf 1\otimes z]=\left[\sum_{\ell=0}^{\infty} \psi^a_{\ell}\psi^*_{a,\ell+1},\; \sum_{m=0}^{\infty} \psi^b_{m+1}\psi^*_{b,m}\right]=\psi^a_{0}\psi^*_{a,0}.
\end{align*}
Meanwhile $\mathbf 1\otimes z$ acts on $\mathcal H_N(n,1)$ as $\mathsf t_{0,1}$, and the commutator reads:
\begin{align*}
    [\mathsf t_{0,1}^\dag,\mathsf t_{0,1}]=[\mathsf t_{1,0},\mathsf t_{0,1}]=\mathsf t_{0,0}=N.
\end{align*}
The above example shows that $\mathfrak{f}$ does not preserve the Hermitian conjugate of operator, thus $\mathfrak{f}$ can not be an isometry.
\end{remark}

Using the multiplication map $\mathcal H_N(n,1)^{\otimes k}\to \mathcal H_N(n,k)$ which is surjective by Corollary \ref{cor: tensor product surjective}, one obtain fermion Fock space presentation for general $k$:
\begin{align}
    \mathcal F_N(n)^{\otimes k}\twoheadrightarrow \mathcal H_N(n,k).
\end{align}
This map has nontrivial kernel if $k\ge 2$, which is closely related to level-rank duality, see Section \ref{subsec:level-rank} and also \cite[4.1]{bourgine2024calogero}.

\section{Operators}\label{sec:operators}

\subsection{\texorpdfstring{$\gl_n[z]$}{Positive loop group} action on \texorpdfstring{$\mathcal H_{N}(n,k)$}{Hilbert space}}

We have already seen that the flavour symmetry $\GL_n$ acts on $\cM(N,n)$ naturally. In the Quot scheme description $\cM(N,n)\cong \mathrm{Quot}^{N}({\mathcal O}^{\oplus n}_{{\mathbb A}^{1}}) $, the $\GL_n$ action is by the automorphism on ${\mathcal O}^{\oplus n}_{{\mathbb A}^{1}}$. It turns out that ${\mathcal O}^{\oplus n}_{{\mathbb A}^{1}}$ admits a bigger automorphism group, denoted by $\GL_n[z]$. $\GL_n[z]$ is the functor that associate a $\bC$-scheme $S$ to $\GL_n(\bA^1\times S)$. $\GL_n[z]$ is represented by a group ind-scheme by \cite[Lemma 4.1.4]{zhu2016introduction}. 

$\GL_n(\bA^1\times S)$ acts on ${\mathcal O}^{\oplus n}_{{\mathbb A}^{1}\times S}$ naturally, and the action is obviously compatible with base change $S'\to S$, thus there is a natural action of $\GL_n[z]$ on $\mathrm{Quot}^{N}({\mathcal O}^{\oplus n}_{{\mathbb A}^{1}})$. $\GL_n[z]$ acts on $\mathcal E_{\mathrm{univ}}$ by pullback, so the action extends to the determinant line bundle $\mathcal L_{\det}$. 

The Lie algebra of $\GL_n[z]$ is the positive part of loop algebra of $\gl_n$, i.e.
\begin{align*}
    \mathrm{Lie}(\GL_n[z])= \{g\in \GL_n[z](\Spec\:\bC[\epsilon]/(\epsilon^2))\:|\: g\equiv \mathrm{id}\pmod{\epsilon}\}\cong \gl_n[z].
\end{align*}
The $\GL_n[z]$ action on $\mathcal L_{\det}$ induces a $\gl_n[z]$-module structure on $\Gamma(\cM(N,n),\mathcal L_{\det}^{\otimes k})\cong \mathcal H_{N}(n,k)$. Moreover, for every $\vec{x}\in \bA^{(N)}$, $\GL_n[z]$ acts on the fiber $h^{-1}(\vec{x})$, and the specialization map 
\begin{align*}
    \mathcal H_N(n,k)\twoheadrightarrow \Gamma(h^{-1}(\vec{x}),\mathcal L_{\det}^{\otimes k}|_{h^{-1}(\vec{x})})
\end{align*}
is a $\gl_n[z]$-module map.

We note that the $\GL_n[z]$-action on $\cM(N,n)$ is induced from the following $\GL_n[z]$-action on $V(N,n)$. In fact, $V(N,n)$ represents the following moduli data
\begin{align*}
   V(N,n)(S)=\{(\mathcal E,\varphi,\psi)\:|\:\mathcal E\in \mathrm{Coh}(\bA^1\times S),\:\varphi\in \mathrm{Hom}_{\mathcal O_{\bA^1\times S}}(\mathcal O_{\bA^1\times S}^{\oplus n},\mathcal E),\: \psi:\mathcal E\cong \mathcal O_{ S}^{\oplus N}\}/\mathrm{equivalence}.
\end{align*}
Then $\GL_n(\bA^1\times S)$ acts on $V(N,n)(S)$ via its natural action on $\mathcal O_{\bA^1\times S}^{\oplus n}$, which is obviously compatible with base change $S'\to S$. Thus there is a natural action of $\GL_n[z]$ on $V(N,n)$, which induces a $\gl_n[z]$-module structure on $\bC[V(N,n)]\cong \mathcal H_N(n)$.

\begin{lemma}\label{lem: gl_n[z] action}
The $\gl_n[z]$-module structure on $\mathcal H_N(n)$ agrees with the one induced by the Lie algebra map $\rho: \gl_n[z]\to D(V(N,n))$ such that
\begin{align}\label{eq: gl_n[z] action}
    \rho(E^a_b\otimes z^m)=A^a_i(X^m)^i_j\frac{\partial}{\partial A^b_j}.
\end{align}
\end{lemma}

\begin{proof}
It is enough to work out the action of $E^a_b\otimes z^m\in \gl_n[z]$ on $V(N,n)$. Let $(X,A)\in V(N,n)$, then $1+\epsilon E^a_b z^m\in \mathrm{Lie}(\GL_n[z])$ acts on $(X,A)$ by precomposing $A:\mathcal O^{\oplus n}_{\bA^1}\to \mathcal E$ with the automorphism $1+\epsilon E^a_bz^m$. Since $A$ is $\mathcal O^{\oplus n}_{\bA^1}$-module map, the action can be explicitly written as
\begin{align*}
    (X,A)\mapsto (X,A+\epsilon X^m\circ A\circ E^a_b).
\end{align*}
Equivalently, the above formula of the action determines a tangent field $A^a_i(X^m)^i_j\frac{\partial}{\partial A^b_j}$. This finishes the proof.
\end{proof}

Since the action of $\gl_n[z]$ on $\mathcal H_N(n,k)$ is the induced from the action of $\gl_n[z]$ on $\mathcal H_N(n)$ by restriction to the subspace $\mathcal H_N(n,k)=\mathcal H_N(n)^{\GL_N,-k}$, we conclude that the $\gl_n[z]$-module structure of $\mathcal H_N(n,k)$ is the one induced from the Lie algebra map \eqref{eq: gl_n[z] action}. As a corollary, we have the following
\begin{proposition}\label{prop: image of gl_1[z]}
The action of the central $\gl_1[z]\subset \gl_n[z]$ on $\mathcal H_N(n,k)$ is given by 
\begin{align*}
    \mathbf{1}\otimes z^m\mapsto \text{ multiplication by }k\cdot\mathrm{Tr}(X^m)\in \bC[\bA^{(N)}].
\end{align*}
\end{proposition}

\begin{proof}
Applying the moment map relation $\mu^*_{\bC}(E^i_j)+k\delta^i_j=0$ to \eqref{eq: gl_n[z] action}, we get
\begin{align*}
    \rho(\mathbf{1}\otimes z^m)&=k\cdot\mathrm{Tr}(X^m)-(X^m)^i_jX^l_i\frac{\partial}{\partial X^l_j}+(X^m)^i_jX^j_l\frac{\partial}{\partial X^i_l}\\
    &=k\cdot\mathrm{Tr}(X^m)-(X^{m+1})^l_j\frac{\partial}{\partial X^l_j}+(X^{m+1})^i_l\frac{\partial}{\partial X^i_l}\\
    &=k\cdot\mathrm{Tr}(X^m).
\end{align*}
\end{proof}

In particular, we see that the action of $\gl_n[z]$ on the central fiber $\Gamma(h^{-1}(0),\mathcal L_{\det}^{\otimes k}|_{h^{-1}(0)})$ factors through $\mathfrak{sl}_n[z]\oplus \gl_1$. 

\bigskip We have seen that $\Gamma(h^{-1}(0),\mathcal L_{\det}^{\otimes k}|_{h^{-1}(0)})\cong \Gamma(\bGr^{N\omega_1}_{\GL_n},\mathcal O(k))$. According to \cite[4.1]{mirkovic2007quiver}, $\bGr^{N\omega_1}_{\GL_n}$ is a closed subvariety of a connected component $\Gr^{(N)}_{\GL_n}$ of the affine Grassmannian $\Gr_{\GL_n}$ which is the moduli ind-scheme of lattices $\Lambda\subset \bC(\!(z)\!)^{\oplus n}$ such that 
\begin{align*}
    \dim (\bC[\![z]\!]^{\oplus n}/z^{M}\Lambda)=N+Mn,\quad\forall M\gg 0.
\end{align*}
The $\GL_n[z]$-action on $\bGr^{N\omega_1}_{\GL_n}$ is the restriction of the natural $\GL_n[z]$-action on $\Gr^{(N)}_{\GL_n}$, and we have an induced $\gl_n[z]$-module map  \footnote{$\Gamma(\bGr^{N\omega_1}_{\GL_n},\mathcal O(k))^*$ is called an affine Demazure module of $\GL_n[z]$ \cite{fourier2006tensor,zhu2009affine}.}
\begin{align}\label{eq: truncation of integrable module}
    \Gamma(\Gr^{(N)}_{\GL_n},\mathcal O(k))\to \Gamma(\bGr^{N\omega_1}_{\GL_n},\mathcal O(k)).
\end{align}
The follow well-known result is a Borel-Weil-Bott type theorem for the Kac-Moody Lie algebra $\widehat{\mathfrak{gl}}(n)$.

\begin{proposition}[Theorem 2.5.5 in \cite{zhu2016introduction}]\label{prop: BWB}
The action of $\mathfrak{gl}_n[z]$ on $\Gamma(\Gr^{(N)}_{\GL_n},\mathcal O(k))$ extends to an action of $\widehat{\mathfrak{sl}}(n)_k\oplus \widehat{\gl}(1)_{kn}$. Moreover, $\Gamma(\Gr^{(N)}_{\GL_n},\mathcal O(k))$ is dual to $L_{k\varpi_{n-r}}(\widehat{\mathfrak{sl}}(n)_k)\otimes\mathrm{Fock}_{-kN}(\widehat{\gl}(1)_{kn})$, where $L_{k\varpi_{n-r}}(\widehat{\mathfrak{sl}}(n)_k)$ is the level $k$ integrable representation with highest weight $k\varpi_{n-r}$ of $\widehat{\mathfrak{sl}}(n)$ and $\mathrm{Fock}_{-kN}(\widehat{\gl}(1)_{kn})$ is the Fock module of $\widehat{\gl}(1)_{kn}$ of weight $-kN$. Here $r=N-\lfloor\frac{N}{n}\rfloor n$.
\end{proposition}

It is known that the lowest energy subspace $L_{k\varpi_{n-r}}(\widehat{\mathfrak{sl}}(n)_k)^*_0\subset L_{k\varpi_{n-r}}(\widehat{\mathfrak{sl}}(n)_k)^*$ is isomorphic to $L_{k\varpi_r}({\mathfrak{sl}}_n)$ (irreducible $\mathfrak{sl}_n$ module of highest weight $k\varpi_r$), and $L_{k\varpi_{n-r}}(\widehat{\mathfrak{sl}}(n)_k)^*$ is generated from $L_{k\varpi_{n-r}}(\widehat{\mathfrak{sl}}(n)_k)^*_0$ by the actions of $\mathfrak{sl}_n[z]$. Since $L_{k\varpi_r}({\mathfrak{sl}}_n)$ is an irreducible $\mathfrak{sl}_n$-module, it follows that for every nonzero $v\in L_{k\varpi_{n-r}}(\widehat{\mathfrak{sl}}(n)_k)^*_0$, the action map $U(\mathfrak{sl}_n[z])\to U(\mathfrak{sl}_n[z])\cdot v$ is surjective onto $L_{k\varpi_{n-r}}(\widehat{\mathfrak{sl}}(n)_k)^*$.

According to \cite[Ch.XVIII]{mathieu1988formules}, the map \eqref{eq: truncation of integrable module} is surjective. Moreover, \eqref{eq: truncation of integrable module} maps the lowest energy subspace $L_{k\varpi_{n-r}}(\widehat{\mathfrak{sl}}(n)_k)^*_0$ isomorphically onto the lowest energy subspace $\Gamma(\bGr^{N\omega_1}_{\GL_n},\mathcal O(k))_0\subset \Gamma(\bGr^{N\omega_1}_{\GL_n},\mathcal O(k))$. Thus for every nonzero $v\in \Gamma(\bGr^{N\omega_1}_{\GL_n},\mathcal O(k))_0$, the action map $U(\mathfrak{sl}_n[z])\to U(\mathfrak{sl}_n[z])\cdot v$ is surjective onto $\Gamma(\bGr^{N\omega_1}_{\GL_n},\mathcal O(k))$. Using graded Nakayama lemma, one can extend the above surjectivity results to $\mathcal H_N(n,k)$, namely we have the following.

\begin{corollary}\label{cor: cyclic}
For every nonzero $v\in \mathcal H_N(n,k)_0$, the action map $U(\gl_n[z])\to U(\gl_n[z])\cdot v$ is surjective onto $\mathcal H_N(n,k)$. In particular, $\mathcal H_N(n,k)$ is a cyclic $\gl_n[z]$-module.
\end{corollary}

\begin{proof}
Denote by $M:=U(\gl_n[z])\cdot v$ the subspace generated from $v$ by the action of $\gl_n[z]$. As we have explained above, $M/\mathfrak{m}M\to \mathcal H_N(n,k)/\mathfrak{m}\mathcal H_N(n,k)$ is surjective where $\mathfrak m$ is the maximal ideal of $\bC[\bA^{(N)}]$ corresponding to the origin $0\in \bA^{(N)}$. The action $\gl_n[z]$ on $\mathcal H_N(n,k)$ is $\bC_q^{\times}$-equivariant, thus $M$ is a graded subspace. It follows from graded Nakayama lemma that $M\to \mathcal H_N(n,k)$ is surjective. We present the detail below for the convenience of reader. 

Suppose that $M\neq \mathcal H_N(n,k)$, then let $w\in \mathcal H_N(n,k)\setminus M$ be a homogeneous element such that its energy grading is smallest among $\mathcal H_N(n,k)\setminus M$. The existence of such element is a consequence of the fact that the energy grading of $\mathcal H_N(n,k)$ is bounded from below.

We claim that the image of $w$ under the specialization map $\mathcal H_N(n,k)\to \Gamma(h^{-1}(0),\mathcal L_{\det}^{\otimes k}|_{h^{-1}(0)})$ can not be zero. Otherwise $w\in \mathfrak m\cdot \mathcal H_N(n,k)$. Write $w=\sum_{i}\alpha_i \cdot h_i$ such that $\alpha_i$'s are homogeneous elements in $\mathfrak{m}$ and $h_i$'s are homogeneous elements in $\mathcal H_N(n,k)$. Since $\alpha_i$ has positive energy, then $h_i$ must belongs to $M$ by the choice of $w$. By the Proposition \ref{prop: image of gl_1[z]}, for every $\alpha_i$ there exists $\beta_i\in U(\gl_1[z])$ such that $\rho(\beta_i)=\alpha_i$, thus $w\in M$, a contradiction. This proves our claim.

Next, let $\bar{w}$ be the image of $w$ in $\Gamma(h^{-1}(0),\mathcal L_{\det}^{\otimes k}|_{h^{-1}(0)})$, and let $\bar{v}$ be the image of $v$ in the lowest energy subspace $\Gamma(h^{-1}(0),\mathcal L_{\det}^{\otimes k}|_{h^{-1}(0)})_0$. According to our previous discussion, there exists $g\in U(\mathfrak{sl}_n[z])$ such that $\bar{w}=g\cdot \bar{v}$. We note that
\begin{align*}
    \deg (w)=\deg (\bar w)=\deg (g\cdot \bar{v})=\deg(g\cdot v),
\end{align*}
thus $w-g\cdot v$ is a homogeneous element in $\mathcal H(n,k)$ which has the same degree with $w$. We also note that $w-g\cdot v\notin M$ because $w\notin M$ and $g\cdot v\in M$. Then according to our claim, the image of $w-g\cdot v$ in $\Gamma(h^{-1}(0),\mathcal L_{\det}^{\otimes k}|_{h^{-1}(0)})$ is nonzero, this contradicts with $\bar{w}=g\cdot \bar{v}$. Therefore $M= \mathcal H_N(n,k)$.
\end{proof}

\begin{corollary}\label{cor: tensor product surjective}
The multiplication map $\Gamma(\cM(N,n),\mathcal L_{\det})^{\otimes k}\to \Gamma(\cM(N,n),\mathcal L_{\det}^{\otimes k})$ is surjective.
\end{corollary}

\begin{proof}
According to Corollary \ref{cor: cyclic}, it suffices to show that the map between ground states
$$\Gamma(\cM(N,n),\mathcal L_{\det})_0^{\otimes k}\to \Gamma(\cM(N,n),\mathcal L_{\det}^{\otimes k})_0$$
is nonzero. Using Remark \ref{rmk: ground states}, it is equivalent to show that $\Gamma(F_\lambda,\mathcal L_{\det}|_{F_\lambda})^{\otimes k}\to \Gamma(F_\lambda,\mathcal L_{\det}^{\otimes k}|_{F_\lambda})$ is nonzero. Consider the projections to the highest weight spaces
\begin{align*}
    \Gamma(F_\lambda,\mathcal L_{\det}|_{F_\lambda})\to \mathcal L_{\det}|_{\{z^\lambda\}},\qquad
    \Gamma(F_\lambda,\mathcal L_{\det}^{\otimes k}|_{F_\lambda})\to \mathcal L_{\det}^{\otimes k}|_{\{z^\lambda\}}
\end{align*}
which are both nonzero. The tensor product map for the fiber at $\{z^\lambda\}$ is an isomorphism, i.e. $(\mathcal L_{\det}|_{\{z^\lambda\}})^{\otimes k}\cong \mathcal L_{\det}^{\otimes k}|_{\{z^\lambda\}}$. This implies that $\Gamma(F_\lambda,\mathcal L_{\det}|_{F_\lambda})^{\otimes k}\to \Gamma(F_\lambda,\mathcal L_{\det}^{\otimes k}|_{F_\lambda})$ is nonzero, whence the corollary follows.
\end{proof}

\begin{remark}\label{rmk:explicit spanning set}
Combine Corollary \ref{cor: tensor product surjective} and Proposition \ref{prop: fock space}, we see that $\mathcal H_N(n,k)$ is spanned by elements of form
\begin{align}\label{eq:explicit spanning set}
    \prod_{\alpha=1}^k \left[\epsilon^{i_{1}i_{2}\cdots i_{N}}(A^{a^{(\alpha)}_1}X^{m^{(\alpha)}_1})_{i_{1}}\cdots (A^{a^{(\alpha)}_N}X^{m^{(\alpha)}_N})_{i_{N}}\right]\ket{\emptyset},\quad \text{where }\;  1\le a^{(\alpha)}_j\le n,\; m^{(\alpha)}_j\in \mathbb N.
\end{align}
\end{remark}

\subsection{DDCA action on \texorpdfstring{$\mathcal H_N(n,k)$}{Hilbert space}}\label{subsec:DDCA module}

The deformed double current algebra (DDCA) was defined by Guay for type A in \cite{guay2007affine} and later defined for other types by Guay and Yang in \cite{guay2017deformed}. In this paper we focus on the type A, and we use a slightly different version of DDCA, which was introduced in \cite{Gaiotto-Rapcek-Zhou}. The relation between our definition and Guay's original definition is pointed out in \cite[2.5]{Gaiotto-Rapcek-Zhou}.

\begin{definition}[{see \cite[Definition 2.0.5]{Gaiotto-Rapcek-Zhou}}]\label{def:DDCA}
The deformed double current algebra $\mathsf A^{(n)}$ is defined to be the $\mathbb C[\epsilon_1,\epsilon_2]$-algebra generated by $\{\mathsf T_{p, q}(x), \mathsf t_{p,q}\:|\: x \in \gl_{n}, (p, q) \in {\bN}^{2}\}$ with the relations \eqref{eqn: A0}-\eqref{eqn: A4} as follows.
\begin{equation}\label{eqn: A0}
    \mathsf T_{r,s}(1)=\epsilon_2 \mathsf t_{r,s},\; \mathsf T_{r,s}(ax+by)=a\mathsf T_{r,s}(x)+b\mathsf T_{r,s}(y),\;\forall (a,b)\in \mathbb C^2,\tag{A0}
\end{equation}
\begin{equation}\label{eqn: A1}
    [\mathsf T_{0,0}(x),\mathsf T_{0,m}(y)]=\mathsf T_{0,m}([x,y]),\; [\mathsf T_{0,0}(x),\mathsf t_{0,m}]=0,\tag{A1}
\end{equation}
\begin{equation}\label{eqn: A2}
\text{ for }p+q\le 2,
\begin{cases}
&[\mathsf t_{p,q},\mathsf T_{r,s}(x)]=(sp-rq)\mathsf T_{p+r-1,q+s-1}(x),\\
&[\mathsf t_{p,q},\mathsf t_{r,s}]=(sp-rq)\mathsf t_{p+r-1,q+s-1},
\end{cases}
\tag{A2}
\end{equation}
To write down \eqref{eqn: A3}-\eqref{eqn: A4}, we introduce notation $\epsilon_3=-n\epsilon_1-\epsilon_2$, and $$\mathsf T_{u,r,t,s}(x\otimes y):=\mathsf T_{u,r}(x)\mathsf T_{t,s}(y)$$ for $x,y\in \mathfrak{gl}_n$, and $\Omega:=E^a_b\otimes E^b_a\in \mathfrak{gl}_n^{\otimes 2}$, then
\begin{equation}\label{eqn: A3}
\begin{cases}
    &\begin{aligned}
        [\mathsf T_{1,0}(x),\mathsf T_{0,m}(y)]=&\mathsf T_{1,m}([x,y])-\frac{\epsilon_3 m}{2}\mathsf T_{0,m-1}(\{x,y\})-n\epsilon_1\mathrm{tr}(y) \mathsf T_{0,m-1}(x)\\
&+\epsilon_1\sum_{j=0}^{m-1}\frac{j+1}{m+1}\mathsf T_{0,j,0,m-1-j}(([x,y]\otimes 1)\cdot \Omega)\\
&+\epsilon_1 \sum_{j=0}^{m-1}\mathsf T_{0,j,0,m-1-j}((x\otimes y-xy\otimes 1)\cdot \Omega)
    \end{aligned}\\
    &[\mathsf T_{1,0}(x),\mathsf t_{0,m}]=m \mathsf T_{0,m-1}(x),
\end{cases}\tag{A3}
\end{equation}
\begin{equation}\label{eqn: A4}
    \begin{split}
        &[\mathsf t_{3,0},\mathsf t_{0,m}]=3m \mathsf t_{2,m-1}+\frac{m(m-1)(m-2)}{4}(\epsilon_1^2-\epsilon_2\epsilon_3) \mathsf t_{0,m-3}\\
&-\frac{3\epsilon_1}{2}\sum_{j=0}^{m-3}(j+1)(m-2-j)(\mathsf T_{0,j,0,m-3-j}(\Omega)+\epsilon_1\epsilon_2\mathsf t_{0,j}\mathsf t_{0,m-3-j}), \;(m\ge 3)
    \end{split}\tag{A4}
\end{equation}
\end{definition}

\begin{remark}
Let us take $x=E^a_b$, $y=E^c_d$ in \eqref{eqn: A3}, then the 2nd and the 3rd line of \eqref{eqn: A3} can be written explicitly as follows:
\begin{multline}
\sum_{j=0}^{m-1}\frac{j+1}{m+1}\mathsf T_{0,j,0,m-1-j}(([x,y]\otimes 1)\cdot \Omega)+\sum_{j=0}^{m-1}\mathsf T_{0,j,0,m-1-j}((x\otimes y-xy\otimes 1)\cdot \Omega)\\
=\sum_{j=0}^{m-1}\mathsf T_{0,j}(E^a_d)\mathsf T_{0,m-1-j}(E^c_b)-\sum_{j=0}^{m-1}\frac{j+1}{m+1}\left(\delta^a_d\mathsf T_{0,j}(E^c_f)\mathsf T_{0,m-1-j}(E^f_b)+\delta^c_b\mathsf T_{0,m-1-j}(E^a_f)\mathsf T_{0,j}(E^f_d)\right).
\end{multline}
\end{remark}

We note that the relations \eqref{eqn: A0}-\eqref{eqn: A4} become linear after setting $\epsilon_1=0$. In fact, $\mathsf A^{(n)}/(\epsilon_1)$ is isomorphic to the universal enveloping algebra of the Lie algebra $D_{\epsilon_2}(\mathbb C)\otimes \mathfrak{gl}_K^{\sim}$. Here $D_{\epsilon_2}(\mathbb C)$ is the algebra of $\epsilon_2$-differential operators on $\bC$, i.e. $D_{\epsilon_2}(\mathbb C)\cong \bC[\epsilon_2]\langle x,y\rangle/([y,x]=\epsilon_2)$. $D_{\epsilon_2}(\mathbb C)\otimes \mathfrak{gl}_K^{\sim}$ is defined to be the $\bC[\epsilon_2]$-submodule of $D_{\epsilon_2}(\mathbb C)\otimes \mathfrak{gl}_K[\epsilon^\pm]$ generated by $D_{\epsilon_2}(\mathbb C)\otimes \mathfrak{sl}_K$ and $\frac{1}{\epsilon_2}\cdot D_{\epsilon_2}(\mathbb C)\otimes 1$. $D_{\epsilon_2}(\mathbb C)\otimes \mathfrak{gl}_K^{\sim}$ is actually a Lie subalgebra of $D_{\epsilon_2}(\mathbb C)\otimes \mathfrak{gl}_K[\epsilon^\pm]$ \cite[Lemma 2.0.2]{Gaiotto-Rapcek-Zhou}. We can regard $D_{\epsilon_2}(\mathbb C)\otimes \mathfrak{gl}_K^{\sim}$ as a version of ``double current algebra'', meaning that it is obtained from $\gl_K$ by joining two variables $x$ and $y$. Since $\mathsf A^{(n)}$ is a deformation of $U(D_{\epsilon_2}(\mathbb C)\otimes \mathfrak{gl}_K^{\sim})$, this explains the terminology ``deformed double current algebra''.

\bigskip According to \cite{hu2023quantum,Gaiotto-Rapcek-Zhou}, the deformed double current algebra $\mathsf A^{(n)}$ acts on $\mathcal H_N(n,k)$ via the following assignment:
\begin{align}\label{DDCA action}
    \epsilon_1\mapsto 1,\quad\epsilon_2\mapsto k,\quad \mathsf T_{p,q}(E^a_b)\mapsto A^a\mathrm{Sym}(Y^p X^q)B_b, \quad \mathsf t_{p,q}\mapsto \mathrm{Tr}\:\mathrm{Sym}(Y^p X^q).
\end{align}
Here $\mathrm{Sym}(\cdots)$ means averaging over permutations, and we have omitted the dummy indices that are contracted, for example $$\mathsf T_{1,1}(E^a_b)\mapsto \frac{1}{2}A^a_iY^i_j X^j_l B^{l}_b+\frac{1}{2}A^a_iX^i_j Y^j_lB^{l}_b.$$
We note that $\{\mathsf T_{0,p}(E^a_b)\:|\: 1\le  a,b\le n,p\in \mathbb N\}$ generates a subalgebra $U(\gl_n[z])\subset \mathsf A^{(n)}$, and the action of $U(\gl_n[z])$ agrees with the one defined by \eqref{eq: gl_n[z] action}.

\bigskip There is a natural grading on $\mathsf A^{(n)}$ given by \cite[(2.6)]{Gaiotto-Rapcek-Zhou}
\begin{align}\label{grading on A}
    \deg (\mathsf T_{p,q}(E^a_b))=\deg (\mathsf t_{p,q})= q-p,\quad \deg(\epsilon_1)=\deg (\epsilon_2)=0.
\end{align}
Under this grading, the action of $\mathsf A^{(n)}$ on $\mathcal H_N(n,k)$ is $\bC_q^{\times}$-equivariant. 

\bigskip $\mathsf A^{(n)}$ has a skew-linear anti-involution $\sigma:\mathsf A^{(n)}\cong \mathsf A^{(n)}$ which is determined by 
\begin{align}
    \sigma(\epsilon_1)=\epsilon_1,\quad\sigma(\epsilon_2)=\epsilon_2,\quad \sigma(\mathsf T_{p,q}(E^a_b))=\mathsf T_{q,p}(E^b_a),\quad \sigma(\mathsf t_{p,q})=\mathsf t_{q,p}.
\end{align}
Skew-linearity means $\sigma(c\cdot \mathfrak a)=\bar{c}\cdot\sigma(\mathfrak a)$ for all $c\in \bC$ and all $\mathfrak{a}\in \mathsf A^{(n)}$. Anti-involution means that $\sigma(\mathfrak{a}\cdot \mathfrak{b})=\sigma(\mathfrak{b})\cdot\sigma(\mathfrak{a})$ for all $\mathfrak{a},\mathfrak{b}\in \mathsf A^{(n)}$ and $\sigma^2=\mathrm{id}$.

\begin{proposition}\label{prop: skew-invariance}
The Hermitian inner product $\langle\cdot | \cdot\rangle$ on $\mathcal H_N(n,k)$ is skew-invariant under the action of $\mathsf A^{(n)}$, i.e.
\begin{align}\label{eq: skew-invariance}
    \langle \sigma(\mathfrak{a})\cdot w | {v}\rangle=\langle w | \mathfrak{a}\cdot v\rangle
\end{align}
for all $\ket{w},\ket{v}\in \mathcal H_N(n,k)$ and all $\mathfrak{a}\in \mathsf A^{(n)}$.
\end{proposition}

\begin{proof}
Let $\rho_N:\mathsf A^{(n)}\to \mathrm{End}(\mathcal H_N(n,k))$ be the algebra map determined by \eqref{DDCA action}, then we see that $$\rho_N(\sigma(\mathfrak a))=\rho_N(\mathfrak a)^\dag$$ for all $\mathfrak{a}\in \mathsf A^{(n)}$, whence the proposition follows.
\end{proof}

\begin{corollary}\label{cor:simplicity}
$\mathcal H_N(n,k)$ is a simple module of $\mathsf A^{(n)}$.
\end{corollary}

\begin{proof}
It suffices to show that every nonzero element $\ket{w}\in \mathcal H_N(n,k)$ generates $\mathcal H_N(n,k)$ via the $\mathsf A^{(n)}$-action. First of all we notice that $\mathsf t_{1,1}$ acts on $\mathcal H_N(n,k)$ as the energy grading operator (up to a constant shift). Let us we write $\ket{w}=\sum_{i\in I} \ket{w_i}$ where $I$ is a finite index set and each $\ket{w_i}$ is homogeneous such that $\deg \ket{w_i}\neq \deg \ket{w_j}$ whenever $i\neq j$, then $\{\ket{w_i}\}_{i\in I}$ belongs to $\mathbb C[\mathsf t_{1,1}]\cdot \ket{w}$. Therefore we can assume that $\ket{w}$ is homogeneous without loss of generality.

Since the Hermitian inner product $\langle\cdot | \cdot\rangle$ on $\mathcal H_N(n,k)$ is nondegenerate, there exists $\ket{w'}\in \mathcal H_N(n,k)$ such that $\langle w'| w\rangle \neq 0$. Let us choose a nonzero $\ket{v_0}\in \mathcal H_N(n,k)_0$, then by Corollary \ref{cor: cyclic}, there exists $\mathfrak{a}\in U(\gl_n[z])\subset \mathsf A^{(n)}$ such that $\ket{w'}=\mathfrak a \ket{v_0}$. Without loss of generality, we can also take $\mathfrak{a}$ to be homogeneous with respect to the grading \eqref{grading on A}. By the skew-invariance \eqref{eq: skew-invariance}, we have
\begin{align*}
    0\neq \langle w'| w\rangle = \langle v_0|\sigma(\mathfrak{a})| w\rangle.
\end{align*}
Denote $\ket{v}:=\sigma(\mathfrak{a})| w\rangle$, then $\ket{v}$ is homogeneous and $\langle v_0| v\rangle \neq 0$, which implies that $\ket{v}$ belongs to the lowest energy subspace $\mathcal H_N(n,k)_0$, because $\mathcal H_N(n,k)_0$ is orthogonal to $\mathcal H_N(n,k)_{>0}$. According to Corollary \ref{cor: cyclic}, $\mathcal H_N(n,k)=U(\gl_n[z])\ket{v}$, thus $\mathcal H_N(n,k)=\mathsf A^{(n)}\ket{w}$. This finishes the proof.
\end{proof}

It is shown in \cite[2.8]{Gaiotto-Rapcek-Zhou} that there is a Yangian subalgebra $Y(\gl_n)\subset \mathsf A^{(n)}$ such that
\begin{align}\label{eq:naive Yangian}
    Y(\gl_n)\ni T^a_b(u)\mapsto \delta^a_b+A^a\frac{1}{u-XY }B_b
\end{align}
when it acts on $\mathcal H_N(n,k)$. Here $T^a_b(u)=\delta^a_b+\sum_{n\ge 0}T^a_{b;n}u^{-n-1}, 1\le a,b\le n$ is the RTT generator of $Y(\gl_n)$, i.e. it satisfies the RTT relation:
\begin{align*}
    (u-v)[T^a_b(u),T^c_d(v)]=T^c_b(u)T^a_d(v) -T^c_b(v)T^a_d(u).
\end{align*}
We note that the Yangian action on $\mathcal H_N(n,k)$ preserves the energy grading, i.e.
$[Y(\gl_n),\mathsf t_{1,1}]=0$, thus every energy eigenspace $\mathcal H_N(n,k)^{\bC_q^{\times},d}$ of $\bC_q^{\times}$-weight $d$ is a Yangian submodule. 

Moreover, $\mathcal H_N(n,k)$ is a unitary module of $Y(\gl_n)$, i.e. the Hermitian inner product $\langle\cdot | \cdot\rangle$ on $\mathcal H_N(n,k)$ is skew-invariant under the action of $Y(\gl_n)$, i.e.
\begin{align*}
    \langle T^a_b(u)\cdot w | {v}\rangle=\langle w | T^b_a(u)\cdot v\rangle
\end{align*}
for all $\ket{w},\ket{v}\in \mathcal H_N(n,k)$. If $S\subset \mathcal H_N(n,k)^{\bC_q^{\times},d}$ is a $Y(\gl_n)$ submodule, then its orthogonal complement $S^\perp$ in $\mathcal H_N(n,k)^{\bC_q^{\times},d}$ is also a $Y(\gl_n)$ submodule by the unitarity. Therefore we have the following.

\begin{corollary}\label{cor:semisimple Yangian module}
$\mathcal H_N(n,k)$ is a semisimple module of the Yangian $Y(\gl_n)$ with finite dimensional simple constituent.
\end{corollary}

See Remark \ref{rmk:simple Y(gl_n) submodules} for a characterization of simple $Y(\gl_n)$-submodules of $\mathcal H_N(n,k)$. See also Theorem \ref{thm:Yangian irrep decomposition}.

\subsection{Level-rank relation}\label{subsec:level-rank}
Consider the tensor decomposition $\bC^{kN}=\bC^{N}\otimes\bC^k,\bC^{kn}=\bC^{n}\otimes\bC^k$ with $\GL_k$ naturally acting on the second tensor components, then we have isomorphism
\begin{align}\label{fixed point induce level-rank}
    \cM(kN,kn)^{\GL_k}\cong \cM(N,n).
\end{align}
Let us denote $\mathcal L_{\det,k}$ to be the determinant line bundle on $\cM(kN,kn)$, then the restriction of $\mathcal L_{\det,k}$ to the $\GL_k$-fixed locus $\cM(N,n)$ is isomorphic to $\mathcal L_{\det,1}^{\otimes k}$. The restriction map gives a homomorphism:
\begin{align}\label{level-rank map}
    \tau: \mathcal H_{kN}(kn,m)=\Gamma(\cM(kN,kn),\mathcal L_{\det,k}^{\otimes m})\longrightarrow \Gamma(\cM(N,n),\mathcal L_{\det,1}^{\otimes km})=\mathcal H_{N}(n,km),
\end{align}
which we shall call it the \textit{level-rank map}.
\begin{proposition}\label{prop: level-rank map}
The level-rank map $\tau$ in \eqref{level-rank map} is $\mathbb Z$-graded, surjective and $\gl_n[z]\oplus\mathfrak{sl}_k[z]$-equivariant, where $\gl_n[z]\oplus\mathfrak{sl}_k[z]$ acts on $\mathcal H_{kN}(kn,m)$ via the inclusion $\gl_n[z]\oplus\mathfrak{sl}_k[z]\subset \gl_{kn}[z]$ and  $\gl_n[z]\oplus\mathfrak{sl}_k[z]$ acts on $\mathcal H_{N}(n,km)$ via the projection to the first component.
\end{proposition}

\begin{proof}
Recall that the $\mathbb Z$-graded $\gl_{kn}[z]$-module structure on $\mathcal H_{kN}(kn,m)$ is induced from the $\GL_{kn}[z]\rtimes\bC_q^{\times}$ action on the line bundle $\mathcal L_{\det,k}^{\otimes m}\to \cM(kN,kn)$. We notice that the $\GL_k$ fixed point locus $\cM(kN,kn)^{\GL_k}$ is $\bC_q^{\times}$-closed, and it is also fixed by $\GL_k[z]$ such that $\SL_k[z]$ acts on the restriction of determinant line bundle $\mathcal L_{\det,1}|_{\cM(kN,kn)^{\GL_k}}$ trivially, thus \eqref{level-rank map} is $\mathbb Z$-graded and $\gl_n[z]\oplus\mathfrak{sl}_k[z]$-equivariant. 

By the $\gl_n[z]$-equivariance of \eqref{level-rank map} and the cyclicity of $\gl_n[z]$-action on $\mathcal H_{N}(n,km)$ (Corollary \ref{cor: cyclic}), to show that \eqref{level-rank map} is surjective, it suffices to show that the induced map between ground states $\mathcal H_{kN}(kn,m)_0\to\mathcal H_{N}(n,km)_0$ is nonzero. Using Remark \ref{rmk: ground states}, we only need to show that the restriction map
\begin{align*}
    \Gamma(\GL_{nk}\cdot\{z^{\lambda_k}\},\mathcal L_{\det,k}^{\otimes m}|_{\GL_{nk}\cdot\{z^{\lambda_k}\}})\longrightarrow \Gamma(\GL_{n}\cdot\{z^{\lambda_k}\},\mathcal L_{\det,k}^{\otimes m}|_{\GL_{n}\cdot\{z^{\lambda_k}\}}).
\end{align*}
is nonzero, where $\lambda_k=\varpi_{kr}+L\cdot\varpi_{kn}$ is a coweight for $\GL_{kn}$. Notice that both sides in the above map have nontrivial image after further restriction to $\mathcal L_{\det,k}^{\otimes m}|_{\{z^{\lambda_k}\}}$, thus the above map is nontrivial, which proves the surjectivity of \eqref{level-rank map}.
\end{proof}

Proposition \ref{prop: level-rank map} implies that the level-rank map $\tau$ in \eqref{level-rank map} factors through a $\gl_n[z]$-equivariant surjective map
\begin{align}\label{level-rank map_main}
    \bar\tau:\mathcal H_{kN}(kn,m)_{\mathfrak{sl}_k[z]}\twoheadrightarrow \mathcal H_{N}(n,km),
\end{align}
where $\mathcal H_{kN}(kn,m)_{\mathfrak{sl}_k[z]}$ is the space of coinvariants, i.e. $\mathcal H_{kN}(kn,m)/\mathfrak{sl}_k[z]\cdot \mathcal H_{kN}(kn,m)$. We shall call $\bar\tau$ the \textit{reduced level-rank map}.

\subsubsection{The case of \texorpdfstring{$m=1$}{m=1}}

Let us focus on the case when $m=1$ which is closely related to the level-rank duality, as pointed out in \cite{bourgine2024calogero}.

Consider the $\mathfrak{h}[z]$-coinvariant $ \mathcal H_{kN}(kn,1)_{\mathfrak{h}[z]}$ where $\mathfrak{h}\subset\mathfrak{sl}_k$ is the Cartan subalgebra. In the view of isomorphism
\begin{align}
    \cM(kN,kn)^{T}\cong \coprod_{N_1+\cdots+N_k=kN}\cM(N_1,n)\times\cdots\times \cM(N_k,n),
\end{align}
where $T\subset\GL_k$ is the maximal torus, Corollary \ref{cor: restrict to A-fixed pts} implies that the restriction to $T$-fixed point locus induces an isomorphism between $\mathfrak{t}[z]$-modules:
\begin{align}\label{T_k fixed pt}
    \mathcal H_{kN}(kn,1)\cong \bigoplus_{N_1+\cdots+N_k=kN}\mathcal H_{N_1}(n,1)\otimes\cdots\otimes \mathcal H_{N_k}(n,1),
\end{align}
where $\mathfrak{t}\cong\bC^k$ is the Lie algebra of $T$. According to Proposition \ref{prop: image of gl_1[z]} $\mathfrak{t}[z]$ acts on $\mathcal H_{N_1}(n,1)\otimes\cdots\otimes \mathcal H_{N_k}(n,1)$ by
\begin{align*}
    (t_1,\cdots,t_k )\otimes z^r \mapsto (t_1\mathrm{Tr}(X^{r}_1),\cdots,t_k\mathrm{Tr}(X^{r}_k))\in \bC[\bA^{(N_1)}\times\cdots\times \bA^{(N_k)}].
\end{align*}
In particular $\mathfrak{t}$ acts on $\mathcal H_{N_1}(n,1)\otimes\cdots\otimes \mathcal H_{N_k}(n,1)$ with weight $(N_1,\cdots,N_k)$, therefore the $\mathfrak{h}$-coinvariant $\mathcal H_{kN}(kn,1)_{\mathfrak{h}}$ is $\mathcal H_{N}(n,1)\otimes\cdots\otimes \mathcal H_{N}(n,1)$. Further taking $z\mathfrak{h}[z]$-covariant amounts to taking quotient of $\mathcal H_{N}(n,1)^{\otimes k}$ by the action of 
\begin{align*}
    \{\mathrm{Tr}(X^{r}_i)-\mathrm{Tr}(X^{r}_{i+1})\:|\:1\le i< k,r\in \mathbb Z_{>0}\},
\end{align*}
which is nothing but taking tensor product with respect to the base ring $\bC[\bA^{(N)}]$. Therefore we obtain the following:

\begin{proposition}
The isomorphism \eqref{T_k fixed pt} induces a $\mathbb Z$-graded isomorphism between $\gl_n[z]$-modules:
\begin{align}\label{h[z] coinvariant}
    \mathcal H_{kN}(kn,1)_{\mathfrak{h}[z]}\cong \underbrace{\mathcal H_{N}(n,1)\otimes_{\bC[A^{(N)}]}\cdots\otimes_{\bC[A^{(N)}]}\mathcal H_{N}(n,1)}_{k\text{ copies}}.
\end{align}
\end{proposition}

By the isomorphism \eqref{h[z] coinvariant} and Lemma \ref{lem: pushforward of line bundle}, $\mathcal H_{kN}(kn,1)_{\mathfrak{h}[z]}$ is a locally free module of finite rank on $\bA^{(N)}$. We note that there is a natural surjective map $\mathcal H_{kN}(kn,1)_{\mathfrak{h}[z]}\twoheadrightarrow \mathcal H_{kN}(kn,1)_{\mathfrak{sl}_k[z]}$ whose composition with $m=1$ case in \eqref{level-rank map_main} gives a surjective map 
\begin{align}\label{h[z] coinv map to H_N(n,k)}
    \mathcal H_{kN}(kn,1)_{\mathfrak{h}[z]}\cong \mathcal H_{N}(n,1)\otimes_{\bC[A^{(N)}]}\cdots\otimes_{\bC[A^{(N)}]}\mathcal H_{N}(n,1)\twoheadrightarrow \mathcal H_{N}(n,k).
\end{align}
Since the above map is induced from the diagonal embedding $\cM(N,n)\hookrightarrow \cM(N,n)^{\times k}\subset \cM(kN,kn)^T$, we see that it is nothing but the tensor multiplication map $\mathcal H_N(n,1)^{\otimes k}\to \mathcal H_N(n,k)$.

Using the fermion presentation $\mathfrak{f}:\mathcal F_{kN}(kn)\cong \mathcal H_{kN}(kn,1)$ in Section \ref{sec:fermion fock}, we can explicitly write down $\tau\circ\mathfrak{f}:\mathcal F_{kN}(kn)\to \mathcal H_{N}(n,k)$ as follows. Under the fixed point decomposition \ref{T_k fixed pt}, the relevant component $\mathcal H_N(n,1)\otimes\cdots\otimes \mathcal H_N(n,1)$ is identified with the $\mathfrak{h}$-invariant subspace $\mathcal H_{kN}(kn,1)^{\mathfrak{h}}$, where $\mathfrak{h}\subset \mathfrak{sl}_k$ is the Cartan subalgebra. Correspondingly $\mathcal F_{kN}(kn)^{\mathfrak{h}}\cong \mathcal F_{N}(n)\otimes \cdots\otimes \mathcal F_{N}(n)$, where the isomorphism is given by 
\begin{align*}
\mathcal F_{N}(n)\otimes \cdots\otimes \mathcal F_{N}(n)\ni \bigotimes_{\alpha=1}^k\left({\psi}^{a^{(\alpha)}_{1},\alpha}_{m^{(\alpha)}_{1}}\cdots{\psi}^{a^{(\alpha)}_{N},\alpha}_{m^{(\alpha)}_{N}}\ket{\emptyset}\right)\mapsto \prod_{\alpha=1}^k{\psi}^{a^{(\alpha)}_{1},\alpha}_{m^{(\alpha)}_{1}}\cdots{\psi}^{a^{(\alpha)}_{N},\alpha}_{m^{(\alpha)}_{N}}\ket{\emptyset}\in \mathcal F_{kN}(kn)^{\mathfrak{h}}.
\end{align*}
Here $({\psi}^{a,\alpha}_{m})_{1\le a\le n,1\le \alpha\le k,m\in \mathbb Z_{\ge 0}}$ are the modes for $n\times k$ fermion fields. Since the map $\mathcal H_N(n,1)\otimes\cdots\otimes \mathcal H_N(n,1)\to \mathcal H_N(n,k)$ is the multiplication map, we obtain the following.

\begin{proposition}\label{prop:level-rank via fermion fock}
The composition of the maps $\mathcal F_{kN}(kn)\overset{\mathfrak{f}}{\longrightarrow}\mathcal H_{kN}(kn,1)\overset{\tau}{\longrightarrow} \mathcal H_{N}(n,k)$ is given by 
\begin{multline}\label{eq:level-rank via fermion fock}
\tau\circ\mathfrak{f}\left(\prod_{\alpha=1}^k{\psi}^{a^{(\alpha)}_{1},\alpha}_{m^{(\alpha)}_{1}}\cdots{\psi}^{a^{(\alpha)}_{N_\alpha},\alpha}_{m^{(\alpha)}_{N_\alpha}}\ket{\emptyset}\right)\\
=\tau\left(\epsilon^{i^{(1)}_1\cdots i^{(k)}_{N_k}}\prod_{\alpha=1}^k\left[(A^{a^{(\alpha)}_1,\alpha}X^{m^{(\alpha)}_1})_{i^{(\alpha)}_{1}}\cdots (A^{a^{(\alpha)}_{N_\alpha},\alpha}X^{m^{(\alpha)}_{N_\alpha}})_{i^{(\alpha)}_{N_\alpha}}\right]\ket{\emptyset}\right)\\
=\begin{cases}
\prod_{\alpha=1}^k \left[\epsilon^{i_{1}i_{2}\cdots i_{N}}(A^{a^{(\alpha)}_1}X^{m^{(\alpha)}_1})_{i_{1}}\cdots (A^{a^{(\alpha)}_N}X^{m^{(\alpha)}_N})_{i_{N}}\right]\ket{\emptyset}  &, \text{ if }\forall\alpha, N_\alpha=N,\\
0 &, \text{ otherwise}.
\end{cases}
\end{multline}
\end{proposition}

\begin{corollary}\label{cor: bmf}
$\forall k\in \bZ_{\ge 0}$, $\mathcal L_{\det}^{\otimes k}$ is generated by global sections, i.e. the natural map between sheaves $$\Gamma(\cM(N,n),\mathcal L_{\det}^{\otimes k})\otimes \mathcal O_{\cM(N,n)}\to \mathcal L_{\det}^{\otimes k}$$ is surjective.
\end{corollary}

\begin{proof}
As we have seen in the above, the $k$-th tensor power map $$\Gamma(\cM(N,n),\mathcal L_{\det})^{\otimes k}\to \Gamma(\cM(N,n),\mathcal L_{\det}^{\otimes k})$$ is surjective, so it is enough to show that $\mathcal L_{\det}$ is generated by global sections. Let $T=\bC^{\times n}\subset \GL_n$ be the maximal torus, then $\Gamma(\cM(N,n),\mathcal L_{\det})\otimes \mathcal O_{\cM(N,n)}\to \mathcal L_{\det}$ is $T$-equivariant, so the support of its cokernel is $T$-invariant, therefore we need to show that fibers of $\mathcal L_{\det}$ at every $T$-fixed point are generated by global sections. By \eqref{T-fixed pts} $\cM(N,n)^T$ is an affine variety, so the natural map $$\Gamma(\cM(N,n)^T,\mathcal L_{\det}|_{\cM(N,n)^T})\otimes \mathcal O_{\cM(N,n)^T}\to \mathcal L_{\det}|_{\cM(N,n)^T}$$ is surjective. By Proposition \ref{prop: restrict to fixed pts}, the restriction map $\Gamma(\cM(N,n),\mathcal L_{\det})\to\Gamma(\cM(N,n)^T,\mathcal L_{\det}|_{\cM(N,n)^T})$ is an isomorphism. This finishes the proof.
\end{proof}

\begin{corollary}\label{cor:level-rank map: n=1}
The reduced level-rank map \eqref{level-rank map_main} in the case $m=n=1:$
\begin{align}\label{level-rank map: n=1}
    \mathcal H_{kN}(k,1)_{\mathfrak{sl}_k[z]}\longrightarrow \mathcal H_{N}(1,k)
\end{align}
is an isomorphism.
\end{corollary}

\begin{proof}
Consider the composition of surjective maps $\mathcal H_{kN}(k,1)_{\mathfrak{h}[z]}\twoheadrightarrow\mathcal H_{kN}(k,1)_{\mathfrak{sl}_k[z]}\twoheadrightarrow \mathcal H_{N}(1,k)$. It is enough to show that $\mathcal H_{kN}(k,1)_{\mathfrak{h}[z]}\twoheadrightarrow\mathcal H_{N}(1,k)$ is an isomorphism. In fact, $\mathcal H_{N}(1,k)$ (in particular $\mathcal H_{N}(1,1)$) is a rank one free module over $\bC[\bA^{(N)}]$, then it follows that \eqref{h[z] coinv map to H_N(n,k)} is an isomorphism.
\end{proof}

\begin{corollary}\label{cor:level-rank map: N=1}
The reduced level-rank map \eqref{level-rank map_main} in the case $m=N=1:$
\begin{align}\label{level-rank map: N=1}
    \mathcal H_{k}(kn,1)_{\mathfrak{sl}_k[z]}\longrightarrow \mathcal H_{1}(n,k)
\end{align}
is an isomorphism.
\end{corollary}

\begin{proof}
We notice that the natural map $\mathcal H_{kN}(kn,1)_{\mathfrak{h}[z]}\twoheadrightarrow\mathcal H_{kN}(kn,1)_{\mathfrak{sl}_k[z]}$ factors through $$\mathcal H_{kN}(kn,1)_{\mathfrak{h}[z]\rtimes S_k}\cong \left(\mathcal H_{N}(n,1)\otimes_{\bC[A^{(N)}]}\cdots\otimes_{\bC[A^{(N)}]}\mathcal H_{N}(n,1)\right)^{S_k}. $$
Take $N=1$ and the right-hand-side of the above isomorphism becomes $\bC[\bA^1]\otimes S^k\bC^n$, where $\mathfrak{sl}_n[z]$ acts on $S^k\bC^n$ via evaluation map $\mathfrak{sl}_n[z]\overset{z=0}{\longrightarrow} \mathfrak{sl}_n$. On the other hand, $\mathcal H_{1}(n,k)$ is also isomorphic to $\bC[\bA^1]\otimes S^k\bC^n$ as a $\bC[\bA^1]\otimes U(\mathfrak{sl}_n[z])$-module. Thus the surjective map $\mathcal H_{k}(kn,1)_{\mathfrak{h}[z]\rtimes S_k}\twoheadrightarrow \mathcal H_{1}(n,k)$ is an isomorphism, which implies that $\mathcal H_{k}(kn,1)_{\mathfrak{sl}_k[z]}\twoheadrightarrow \mathcal H_{1}(n,k)$ is an isomorphism.
\end{proof}

In general, we do not know if the map $\mathcal H_{kN}(kn,1)_{\mathfrak{sl}_k[z]}\twoheadrightarrow \mathcal H_{N}(n,k)$ is isomorphism or not. Nevertheless, we can show that it becomes an isomorphism after localization. Namely, let 
\begin{align}
    \mathrm{Disc}=\prod_{i<j}(x_i-x_j)^2\in \bC[\bA^{(N)}]
\end{align}
be the discriminant, then the nonvanishing locus of $\mathrm{Disc}$ is $\bA^{(N)}_{\mathrm{disj}}$, i.e. $\bC[\bA^{(N)}][\mathrm{Disc}^{-1}]=\bC[\bA^{(N)}_{\mathrm{disj}}]$.

\begin{theorem}\label{thm: level-rank}
If $m=1$, then the reduced level-rank map \eqref{level-rank map_main} becomes an isomorphism after localizing $\mathrm{Disc}$, i.e. we have isomorphism
\begin{align}\label{level-rank duality}
    \mathcal H_{kN}(kn,1)_{\mathfrak{sl}_k[z]}[\mathrm{Disc}^{-1}]\cong \mathcal H_{N}(n,k)[\mathrm{Disc}^{-1}],
\end{align}
for all $k,n,N\in \mathbb Z_{>0}$.
\end{theorem}

\begin{proof}
Since $\mathcal H_{N}(n,k)$ is a free module over $\bC[\bA^{(N)}]$, it suffices to show that the fiber of map $\mathcal H_{kN}(kn,1)_{\mathfrak{sl}_k[z]}\to \mathcal H_{N}(n,k)$ at $\vec x=(x_1,\cdots,x_N)\in \bA^{(N)}$ is an isomorphism whenever $\vec x\in \bA^{(N)}_{\mathrm{disj}}$.

Denote the fibers of $\mathcal H_{kN}(kn,1)$ and $\mathcal H_{kN}(kn,1)_{\mathfrak{sl}_k[z]}$ at ${\vec x}$ by $\mathcal H_{kN}(kn,1)_{\vec x}$ and $\mathcal H_{kN}(kn,1)_{\mathfrak{sl}_k[z],\vec{x}}$ respectively. Applying $\mathfrak{sl}_k[z]$-coinvariant to the natural map $\mathcal H_{kN}(kn,1)_{\vec{x}}\twoheadrightarrow \mathcal H_{kN}(kn,1)_{\mathfrak{sl}_k[z],\vec{x}}$ and we get:
\begin{align*}
    \left[\mathcal H_{kN}(kn,1)_{\vec x}\right]_{\mathfrak{sl}_k[z]}\twoheadrightarrow \mathcal H_{kN}(kn,1)_{\mathfrak{sl}_k[z],\vec{x}}.
\end{align*}
It is enough to show that the surjective map $\left[\mathcal H_{kN}(kn,1)_{\vec x}\right]_{\mathfrak{sl}_k[z]}\twoheadrightarrow \mathcal H_{N}(n,k)_{\vec{x}}$ is an isomorphism.

Using the isomorphism \eqref{eq: H(N,n) = global section} and Proposition \ref{prop: isom between Quot and Gr}, we have 
\begin{align*}
    \mathcal H_{kN}(kn,1)_{\vec x}=\bigotimes_{i=1}^N\Gamma(\bGr^{k\omega_1}_{\GL_{kn},x_i},\mathcal O(1)),\quad \mathcal H_{N}(n,k)_{\vec x}=\bigotimes_{i=1}^N\Gamma(\bGr^{\omega_1}_{\GL_{n},x_i},\mathcal O(k)).
\end{align*}
Here $\bGr^{k\omega_1}_{\GL_{kn},x_i}$ and $\bGr^{\omega_1}_{\GL_{n},x_i}$ the affine Schubert varieties of the affine Grassmannians supported at $x_i$. The action of $\mathfrak{sl}[z]$ on $\Gamma(\bGr^{k\omega_1}_{\GL_{kn},x_i},\mathcal O(1))$ factors through a Lie algebra quotient $\mathfrak{sl}[z]\twoheadrightarrow \mathfrak{sl}[z]/(z-x_i)^r$ for some $r$. We note that the Lie algebra map 
\begin{align*}
    \mathfrak{sl}[z]\longrightarrow \bigoplus_{i=1}^N\mathfrak{sl}[z]/(z-x_i)^r
\end{align*}
is surjective since $x_i$ are distinct. It follows that 
\begin{align}\label{coinv of fiber}
    \left[\mathcal H_{kN}(kn,1)_{\vec x}\right]_{\mathfrak{sl}_k[z]}=\bigotimes_{i=1}^N\Gamma(\bGr^{k\omega_1}_{\GL_{kn},x_i},\mathcal O(1))_{\mathfrak{sl}_k[z]}.
\end{align}
We notice that $\Gamma(\bGr^{k\omega_1}_{\GL_{kn},x_i},\mathcal O(1))_{\mathfrak{sl}_k[z]}$ can be identified with $\left[\mathcal H_{k}(kn,1)_{x_i}\right]_{\mathfrak{sl}_k[z]}$. By Corollary \ref{cor:level-rank map: N=1}, the coinvariant space $\left[\mathcal H_{k}(kn,1)_{x_i}\right]_{\mathfrak{sl}_k[z]}$ is isomorphic to $\mathcal H_{1}(n,k)_{x_i}$, and the latter is isomorphic to $\Gamma(\bGr^{\omega_1}_{\GL_{n},x_i},\mathcal O(k))$. This implies that the map $\left[\mathcal H_{kN}(kn,1)_{\vec x}\right]_{\mathfrak{sl}_k[z]}\to \mathcal H_{N}(n,k)_{\vec{x}}$ is an isomorphism. This finishes the proof.
\end{proof}

\section{Wave Functions}\label{sec:wave functions}

Restricted to the open subset $\bA^{(N)}_{\mathrm{disj}}\subset \bA^{(N)}$ consisting of $N$ disjoint points in $\bA^1$, the fibers of Hilbert-Chow morphism $h:\cM(N,n)\to \bA^{(N)}$ are isomorphic to $(\bP^{n-1})^N$. Moreover, the restriction of $\mathcal L_{\det}^{\otimes k}$ to the fiber $(\bP^{n-1})^N$ is isomorphic to $\mathcal O(k)^{\boxtimes N}$. So for every $\vec{x}\in \bA^{(N)}_{\mathrm{disj}}$, the fiber $h_*(\mathcal L_{\det}^{\otimes k})|_{\vec{x}}$ is isomorphic to $\Gamma(\mathbb P^{n-1},\mathcal O(k))^{\otimes N}\cong (S^k\bC^n)^{\otimes N}$. 

Denote by $\mathfrak{W}$ the composition of the the restriction map $$\Gamma(\cM(N,n),\mathcal L_{\det}^{\otimes k})\to \Gamma\left(h^{-1}(\bA^{(N)}_{\mathrm{disj}}),\mathcal L_{\det}^{\otimes k}|_{h^{-1}(\bA^{(N)}_{\mathrm{disj}})}\right)$$ and the base change map $$\Gamma\left(h^{-1}(\bA^{(N)}_{\mathrm{disj}}),\mathcal L_{\det}^{\otimes k}|_{h^{-1}(\bA^{(N)}_{\mathrm{disj}})}\right)\to \Gamma\left(h^{-1}(\bA^{(N)}_{\mathrm{disj}}),\mathcal L_{\det}^{\otimes k}|_{h^{-1}(\bA^{(N)}_{\mathrm{disj}})}\right)\otimes_{\bC[\bA^{(N)}_{\mathrm{disj}}]}\bC[\bA^{N}_{\mathrm{disj}}]$$
Since both maps are injective, $\mathfrak{W}$ is a embedding
\begin{align}\label{wave function map}
    \mathfrak{W}:\mathcal H_N(n,k)\hookrightarrow \bC[\bA^N_{\mathrm{disj}}]\otimes (S^k\bC^n)^{\otimes N}.
\end{align}
We call $\mathfrak{W}$ the wave function presentation map, and for $\ket{v}\in \mathcal H_N(n,k)$ we call $\mathfrak{W}(\ket{v})$ the wave function of $\ket{v}$.

Let us write down $\mathfrak{W}$ explicitly. We diagonalize $X=H\mathrm{diag}(x_1,\cdots,x_N)H^{-1}$, then define $y^a_i=(A H)^a_i$. The coordinate on $\bA^N_{\mathrm{disj}}$ is then $\{x_1,\cdots,x_N\}$, and the $i$-th copy of $S^k\bC^n$ is represented by homogeneous polynomials in $\{y^1_i,\cdots,y^n_i\}$ of degree $k$. Every element in $\mathcal H_N(n,k)$ can be presented as $f(X,A)\ket{\emptyset}$ where $f(X,A)$ is a $\GL_N$ semi-invariant polynomial, the wave function $\mathfrak{W}(f(X,A)\ket{\emptyset})$ is the function
\begin{align}\label{eqn:wave function}
    f(\mathrm{diag}(x_1,\cdots,x_N),y)\in \bC[\bA^N_{\mathrm{disj}}]\otimes (S^k\bC^n)^{\otimes N}.
\end{align}

\begin{example}
If $k=1$, then we can write down the wave function of the image of the fermion Fock space presentation in Proposition \ref{prop: fock space}:
\begin{equation}
\begin{split}
\mathfrak{W}\circ\mathfrak{f}({\psi}^{a_1}_{m_1}\cdots{\psi}^{a_N}_{m_N}\ket{\emptyset})&=\mathfrak{W}(\epsilon^{i_{1}i_{2}\cdots i_{N}}(A^{a_1}X^{m_1})_{i_{1}}\cdots (A^{a_N}X^{m_N})_{i_{N}}\ket{\emptyset})\\
&=\det(x_i^{m_j}y^{a_j}_i)_{ij}.
\end{split}
\end{equation}
\end{example}

More generally, using the multiplication map $\mathcal H_N(n,1)^{\otimes k}\to \mathcal H_N(n,k)$ which is surjective by Corollary \ref{cor: tensor product surjective}, one can always present the wave function by a product of functions of form $\det(x_i^{m_j}y^{a_j}_i)_{ij}$, see \cite[3.2]{bourgine2024calogero}.

\subsection{Calogero representation of DDCA}\label{subsec:Calogero rep}

According to \cite[Lemma A.12]{Gaiotto-Rapcek-Zhou}, the Calogero representation of the deformed double current algebra $\mathsf A^{(n)}$ is an algebra map from $\mathsf A^{(n)}$ to the twisted differential operator algebra $D^{k}((\mathbb P^{n-1})^{N}\times \bA^N_{\mathrm{disj}})$. The latter naturally acts on $\Gamma((\mathbb P^{n-1})^{N}\times \bA^N_{\mathrm{disj}},\mathcal O(k)^{\boxtimes N})\cong \bC[\bA^N_{\mathrm{disj}}]\otimes (S^k\bC^n)^{\otimes N}$. Moreover, the construction of the map $\mathsf A^{(n)}\to D^{k}((\mathbb P^{n-1})^{N}\times \bA^N_{\mathrm{disj}})$ implies that 
\begin{align*}
    \mathfrak W: \Gamma(\cM(N,n),\mathcal L_{\det}^{\otimes k})\to \Gamma((\mathbb P^{n-1})^{N}\times \bA^N_{\mathrm{disj}},\mathcal O(k)^{\boxtimes N})
\end{align*}
intertwines with the $\mathsf A^{(n)}$-actions. Therefore we can present the $\mathsf A^{(n)}$-action on $\mathcal H_N(n,k)$ using the Calogero representation \cite[(A.14)]{Gaiotto-Rapcek-Zhou}:
\begin{align}\label{Calogero rep}
    \mathsf t_{2,0}\mapsto \sum_{i=1}^N\Delta^{-1}\partial_{i}^2\Delta-2\sum_{i<j}^N\frac{ \Omega_{ij}+k}{(x_i-x_j)^2},\qquad \mathsf T_{0,n}(E^a_{b})\mapsto \sum_{i=1}^N E^a_{b,i} x_i^n.
\end{align}
Here $E^a_{b,i}$ is the $a$-th row $b$-th column elementary matrix $E^a_{b}\in \gl_n$, which acts on $i$-th $S^k\bC^n$ as $k$-th symmetric power of vector representation $\bC^n$. $\Omega_{ij}=E^a_{b,i}E^b_{a,j}$ is the quadratic Casimir of $ij$ sites. $\Delta$ is the Vandermonde factor
\begin{align}\label{eqn:Vandermonde}
    \Delta=\prod_{i>j}^N(x_i-x_j).
\end{align}

\begin{remark}
Using the Calogero representation \eqref{Calogero rep}, the Calogero-like Hamiltonian in \cite[(3.8)]{bourgine2024calogero} with $B=1$ can be written as
\begin{align}\label{Calogero Hamiltonian}
    \tilde{\mathcal H}=-\mathsf t_{2,0}+2\mathsf t_{1,1}.
\end{align}
Since $\{\mathsf t_{2,0},\mathsf t_{1,1},\mathsf t_{0,2}\}$ is an $\mathfrak{sl}_2$-triple, the Hamiltonian $\tilde{\mathcal H}$ is triangular on the $\mathsf t_{1,1}$-eigenvectors. This provides an alternative proof of \cite[(3.20),(3.21)]{bourgine2024calogero}.
\end{remark}

\begin{remark}
As we have discussed in Section \ref{subsec:DDCA module}, there is a Yangian subalgebra $Y(\gl_n)\subset \mathsf A^{(n)}$ such that $[Y(\gl_n),\mathsf t_{1,1}]=0$, so $\mathcal H_N(n,k)$ admits a Yangian action which preserves the energy grading. Moreover, the Calogero Hamiltonian \eqref{Calogero Hamiltonian} can be written as $\tilde{\mathcal H}=e^{-\frac{1}{2}\mathrm{ad}_{\mathsf t_{2,0}}}(2\mathsf t_{1,1})$, so $\tilde{\mathcal H}$ commutes with Yangian algebra $e^{-\frac{1}{2}\mathrm{ad}_{\mathsf t_{2,0}}}(Y(\gl_n))$. This confirms a conjecture in \cite{bourgine2024calogero} that the Calogero Hamiltonian \eqref{Calogero Hamiltonian} admits a Yangian symmetry.
\end{remark}

\begin{lemma}\label{lem:Calogero-Sutherland}
The Calogero representation of the operator $H_{\mathrm{CS}}:=\mathrm{Tr}((XY)^2)$ is a higher-spin analog of Calogero-Sutherland Hamiltonian$:$
\begin{align}\label{eq:H_CS}
    H_{\mathrm{CS}}=\sum_{i=1}^N\Delta^{-1}(x_i\partial_{i})^2\Delta-2\sum_{i<j}^N\frac{x_ix_j(\Omega_{ij}+k)}{(x_i-x_j)^2}-(N-1)\sum_{i=1}^Nx_i\partial_i-\frac{N(N-1)(2N-1)}{6}.
\end{align}
\end{lemma}

\begin{proof}
The proof is similar to that of \cite[Lemma A.1.2]{Gaiotto-Rapcek-Zhou}. We diagonalize $X=H\mathrm{diag}(x_1,\cdots,x_N)H^{-1}$. Define $\overline{XY}^i_j:=:(H^{-1}XYH)^i_j:$ where the normal ordering means $H$ and $X$ are put at the right of $Y$ (the ordering of $H$ and $X$ does not matter since they commute), that is $\overline{XY}^i_j=Y^l_s(H^{-1})^i_tX^t_lH^s_j$. Using \cite[(A.17)]{Gaiotto-Rapcek-Zhou}\footnote{The variables $X$ and $Y$ in this paper are denoted by $Y$ and $X$ in \cite[Appendix A]{Gaiotto-Rapcek-Zhou} respectively.} we compute that 
\begin{align}\label{eq:XY}
    \overline{XY}^i_j=\begin{cases}
        \frac{x_i}{x_j-x_i}u^a_jv^i_a &, \text{ if }i\neq j,\\
        x_i\partial_i &, \text{ if }i= j,
    \end{cases}\qquad
    \text{where }u^a_i=(AH)^a_i,\; v^j_a=(H^{-1}B)^j_a.
\end{align}
On the other hand, we have
\begin{align}\label{eq:(XY)^2}
\overline{XY}^i_j\overline{XY}^j_i=\mathrm{Tr}((XY)^2)-(H^{-1})^i_l[(XY)^l_m,H^m_p]\overline{XY}^p_i+(H^{-1})^i_l[(XY)^l_q,H^t_i](XY)^q_t.
\end{align}
Using \cite[(A.21))]{Gaiotto-Rapcek-Zhou}, we have
\begin{align}\label{eq:[XY,H]}
    [(XY)^a_b,H^c_d]=\sum_{e\neq d}^N\frac{x_d}{x_d-x_e}H^a_dH^c_e(H^{-1})^e_b.
\end{align}
Plug \eqref{eq:[XY,H]} into \eqref{eq:(XY)^2}, and we get
\begin{align}
    \overline{XY}^i_j\overline{XY}^j_i=\mathrm{Tr}((XY)^2)-\sum_{p=1}^N\sum_{q\neq p}^N\frac{x_p+x_q}{x_p-x_q}\overline{XY}^p_p.
\end{align}
Plug \eqref{eq:XY} to the above equation, and we get
\begin{align}
\mathrm{Tr}((XY)^2)=\sum_{i=1}^N(x_i\partial_{i})^2-2\sum_{i<j}^N\frac{x_ix_j}{(x_i-x_j)^2}u^a_iv^j_au^b_jv^i_b-(N-1)\sum_{i=1}^Nx_i\partial_i+2\sum_{i\neq j}\frac{x_i^2}{x_i-x_j}\partial_i.
\end{align}
According to \cite[Lemma A.1.1]{Gaiotto-Rapcek-Zhou}, the $i$-th $\gl_n$ generators are defined by $E^a_{b,i}=u^a_iv^i_b$ ($i$ is not summed). It is straightforward to compute that $u^a_iv^j_au^b_jv^i_b=\Omega_{ij}+k$, thus the right-hand-side of above equation is equal to \eqref{eq:H_CS}.
\end{proof}

\subsection{Ground states}

Let $\mathcal E_N^{n,k}:=s^*h_*(\mathcal L_{\det}^{\otimes k})$ be the locally free sheaf on $\bA^{N}$, where $s: \bA^{N}\to \bA^{(N)}$ is the symmetrization map. According to Proposition \ref{prop: isom between Quot and Gr}, $\mathcal E_N^{n,k}\cong \tilde\pi_*\mathcal O(k)$. Here $\tilde\pi:\overline{\Gr}^{\omega_1,\cdots,\omega_1}_{\GL_{n},\bA^{N}}\to \bA^N$ is the structure map of Beilinson-Drinfeld Grassmannian, and $\mathcal O(k)$ is the $k$-th tensor power of determinant line bundle $\mathcal O(1)$ on Beilinson-Drinfeld Grassmannian. 

Let $\bA^N_{\mathrm{disj}}$ be the open subset in $\bA^N$ where $N$ points are distinct, then $\tilde \pi: \overline{\Gr}^{\omega_1,\cdots,\omega_1}_{\GL_{n},\bA^{N}_{\mathrm{disj}}} \to \bA^N_{\mathrm{disj}}$ is a trivial $(\bP^{n-1})^N$-fibration. Therefore we have a canonical isomorphism 
\begin{align}\label{restriction of E to open}
    \mathcal E_N^{n,k}|_{\bA^N_{\mathrm{disj}}}\cong \mathcal O_{\bA^N_{\mathrm{disj}}}\otimes (S^k\bC^n)^{\otimes N}.
\end{align}
\begin{remark}\label{rmk: wave function via E}
The isomorphism \eqref{restriction of E to open} composed with the embedding $\mathcal H_N(n,k)=\Gamma(\bA^N,\mathcal E_N^{n,k})^{S_N}\hookrightarrow \Gamma(\bA^N_{\mathrm{disj}},\mathcal E_N^{n,k}|_{\bA^N_{\mathrm{disj}}})$ gives rise to the wave function map \eqref{wave function map}. 
\end{remark}

Since Hilbert-Chow map $h:\cM(N,n)\to \bA^{(N)}$ is flat (Corollary \ref{cor:flatness}), by the flat base change theorem we have
\begin{align}\label{global section of E}
    \Gamma(\bA^N,\mathcal E_N^{n,k})\cong \mathcal H_N(n,k)\otimes_{\bC[\bA^{(N)}]} \bC[\bA^{N}].
\end{align}
The above isomorphism is $\GL_n[z]\rtimes\bC_q^{\times}$-equivariant, where $\GL_n[z]$ action on the left-hand-side is induced by its action on the Beilinson-Drinfeld Grassmannian $\overline{\Gr}^{\omega_1,\cdots,\omega_1}_{\GL_{n},\bA^{N}}$, and $\bC_q^{\times}$ action on the left-hand-side is induced by the scaling of $\bA^1$ with weight $-1$. Note that the permutation group $S_N$ naturally acts on the second tensor component of $\mathcal H_N(n,k)\otimes_{\bC[\bA^{(N)}]} \bC[\bA^{N}]$, and $\mathcal H_N(n,k)$ is the symmetric part, i.e. $\mathcal H_N(n,k)=\Gamma(\bA^N,\mathcal E_N^{n,k})^{S_N}$.

Although $\Gamma(\bA^N,\mathcal E_N^{n,k})$ differs from $\mathcal H_N(n,k)$, the difference disappears when restricted to lowest $\bC_q^{\times}$-weight components. Namely, let $\Gamma(\bA^N,\mathcal E_N^{n,k})_0$ be the $\bC_q^{\times}$-eigenspace of $\Gamma(\bA^N,\mathcal E_N^{n,k})$ with smallest $\bC_q^{\times}$-weight, then we have the following.
\begin{lemma}\label{lem: ground states unsymmetrized}
\eqref{global section of E} induces an isomorphism $\Gamma(\bA^N,\mathcal E_N^{n,k})_0\cong \mathcal H_N(n,k)_0$.
\end{lemma}

\begin{proof}
This follows from the fact that $\bC[\bA^{N}]$ is non-negatively graded and its degree zero component is $\bC$.
\end{proof}

Consider the full Beilinson-Drinfeld Grassmannian $\Gr_{\GL_n,\bA^N}$. The connected components of $\Gr_{\GL_n,\bA^N}$ are labelled by $N$-tuple of integers $(m_1,\cdots,m_N)\in \bZ^N$ such that $\Gr^{(m_1,\cdots,m_N)}_{\GL_n,\bA^N}$ is characterized by its fiber over a point $\vec x\in \bA^N_{\mathrm{disj}}$ being isomorphic to $\prod_{i=1}^N\Gr^{(m_i)}_{\GL_n}$. Here $\Gr^{(m_i)}_{\GL_n}$ is the $m_i$-th connected component of $\Gr_{\GL_n}$ which is the moduli ind-scheme of lattices $\Lambda\subset \bC(\!(z)\!)^{\oplus n}$ such that 
\begin{align*}
    \dim (\bC[\![z]\!]^{\oplus n}/z^{M}\Lambda)=m_i+Mn,\quad\forall M\gg 0.
\end{align*}
$\overline{\Gr}^{\omega_1,\cdots,\omega_1}_{\GL_{n},\bA^{N}} $ is contained in the connected component $\Gr^{\mathbf 1}_{\GL_n,\bA^N}$ where $\mathbf 1$ is the vector in $\bZ^N$ whose components are $1$.

Let $\tilde\pi^{\mathbf 1}: \Gr^{\mathbf 1}_{\GL_n,\bA^N}\to\bA^N$ be the structure map of the Beilinson-Drinfeld Grassmannian, and we define the following sheaf on $\bA^N$:
\begin{align}
    \mathcal V_N^{n,k}:=\tilde\pi^{\mathbf 1}_*\mathcal O(k).
\end{align}
Since there is a natural $\GL_n[z]\rtimes \bC_q^{\times}$-action on $\Gr^{\mathbf 1}_{\GL_n,\bA^N}$ such that $\mathcal O(k)$ is equivariant, and such action restricts to the $\GL_n[z]\rtimes \bC_q^{\times}$-action on $\overline{\Gr}^{\omega_1,\cdots,\omega_1}_{\GL_{n},\bA^{N}}$, we have the restriction map 
\begin{align}\label{eq: truncation of integrable module_BD}
    \mathcal V_N^{n,k}\to \tilde\pi_*\mathcal O(k)=\mathcal E_N^{n,k},
\end{align}
which is $\GL_n[z]\rtimes \bC_q^{\times}$-equivariant by construction. Applying the result in \cite[Ch.XVIII]{mathieu1988formules} to each fiber of $\tilde\pi^{\mathbf 1}$, we see that the map \eqref{eq: truncation of integrable module_BD} is surjective. Moreover \eqref{eq: truncation of integrable module_BD} induces isomorphism between ground states, namely we have the following.

\begin{proposition}\label{prop: identify ground states}
Let $\Gamma(\bA^N,\mathcal V_N^{n,k})_0$ be the $\bC_q^{\times}$-eigenspace of $\Gamma(\bA^N,\mathcal V_N^{n,k})$ with the smallest $\bC_q^{\times}$-weight, then \eqref{eq: truncation of integrable module_BD} induces isomorphism
\begin{align}
    \Gamma(\bA^N,\mathcal V_N^{n,k})_0\cong \Gamma(\bA^N,\mathcal E_N^{n,k})_0.
\end{align}
\end{proposition}

\begin{lemma}\label{lem: pro-finite rank}
There exists a set of finite rank locally free sheaves $\mathcal V_m$ on $\bA^N$ ($m\in \mathbb N$) with surjective maps $\mathcal V_{m+1}\twoheadrightarrow \mathcal V_{m}$, such that 
\begin{align*}
    \mathcal V_N^{n,k}\cong \underset{\substack{\longleftarrow\\ m}}{\lim}\:\mathcal V_{m}.
\end{align*}
Moreover, for any $\vec x\in \bA^N$ let $\mathcal V_{m,\vec x}$ be the fiber of $\mathcal V_{m}$ at $\vec x$,
\begin{align*}
    \Gamma(\tilde\pi^{-1}(\vec x),\mathcal O(k))\cong \underset{\substack{\longleftarrow\\ m}}{\lim}\:\mathcal V_{m,\vec x}.
\end{align*}
\end{lemma}

\begin{proof}
It is known that $\Gr^{\mathbf 1}_{\GL_n,\bA^N}\cong \Gr^{\mathbf 1}_{\PGL_n,\bA^N}\times_{\bA^N} \Gr^{\mathbf 1}_{\bC^{\times},\bA^N}$ as ind-schemes. $\Gr^{\mathbf 1}_{\bC^{\times},\bA^N}$ is inductive limit of finite flat schemes over $\bA^N$. In fact, 
\begin{align*}
    \Gr^{\mathbf 1}_{\bC^{\times},\bA^N}\cong \underset{\substack{\longrightarrow\\ m}}{\lim}\Spec_{\bA^N}\left(\bigoplus_{i=0}^m \mathrm{Sym}^i(\mathcal F_m)\right),\quad \mathcal F_m=\Hom_{\mathcal O_{\bA^N}}(\mathcal O_{\bA^N\times\bA^1}(m\cdot D)/\mathcal O_{\bA^N\times\bA^1},\mathcal O_{\bA^N}),
\end{align*}
where $D$ is the divisor on $\bA^N\times\bA^1=\Spec\bC[x_1,\cdots,x_N]\times\Spec \bC[z]$ given by equation $\prod_{i=1}^N(z-x_i)$. The inductive limit is induced by the projection $\mathcal F_{m+1}\twoheadrightarrow \mathcal F_m$. Since $\mathcal O_{\bA^N\times\bA^1}(m\cdot D)/\mathcal O_{\bA^N\times\bA^1}$ is locally free of rank $m-1$ on $\bA^N$, we see that $\mathcal F_m$ and $\bigoplus_{i=0}^m \mathrm{Sym}^i(\mathcal F_m)$ are locally free of finite rank on $\bA^N$. By the flat base change theorem, to prove the lemma for $\mathcal V_N^{n,k}=\tilde\pi_*\mathcal O(k)$, it suffices to prove the analogous statement for $\tilde\pi'_*\mathcal O(k)$, i.e. there exists a set of finite rank locally free sheaves $\mathcal V'_m$ on $\bA^N$ ($m\in \mathbb N$) with surjective maps $\mathcal V'_{m+1}\twoheadrightarrow \mathcal V'_{m}$, such that 
\begin{align*}
    \tilde\pi'_*\mathcal O(k)\cong \underset{\substack{\longleftarrow\\ m}}{\lim}\:\mathcal V'_{m},\quad \Gamma(\tilde\pi'^{-1}(\vec x),\mathcal O(k))\cong\underset{\substack{\longleftarrow\\ m}}{\lim}\:\mathcal V'_{m,\vec x},\forall\vec x\in \bA^N.
\end{align*}

It is known that $\Gr^{\mathbf 1}_{\PGL_n,\bA^N}$ is direct limit of Beilinson-Drinfeld Schubert varieties:
\begin{align*}
    \Gr^{\mathbf 1}_{\PGL_n,\bA^N}=\underset{\substack{\longrightarrow\\ \lambda}}{\lim} \bGr^{\lambda,\cdots,\lambda}_{\PGL_n,\bA^N},
\end{align*}
where $\lambda$ runs through all dominant coweights of $\PGL_n$ such that $\lambda-\omega_1$ is in the root lattice of $\PGL_n$. $\bGr^{\lambda,\cdots,\lambda}_{\PGL_n,\bA^N}$ is contained in $\bGr^{\lambda',\cdots,\lambda'}_{\PGL_n,\bA^N}$ if and only if $\lambda\le \lambda'$ in the Bruhat order. Let $\tilde\pi_\lambda$ be the restriction of $\tilde\pi'$ to $\bGr^{\lambda,\cdots,\lambda}_{\PGL_n,\bA^N}$, then $\tilde\pi_\lambda$ is proper and flat \cite[1.2.4]{zhu2009affine}, thus the derived pushforward $R\tilde\pi_{\lambda *}\mathcal O(k)$ is a perfect complex with positive tor amplitudes in $D^b_{\mathrm{coh}}(\bA^N)$. We note that $\tilde\pi_\lambda^{-1}(0)\cong \bGr^{N\lambda}_{\PGL_n}$, and according to \cite[Ch.XVIII]{mathieu1988formules} we have $H^i(\bGr^{N\lambda}_{\PGL_n},\mathcal O(k))=0$ for all $i>0$. It follows that $R^i\tilde\pi_{\lambda *}\mathcal O(k)$ vanishes in an open neighborhood $U$ of $\{0\}\in \bA^N$ for all $i>0$. By the $\bC_q^{\times}$-equivariance of $\tilde\pi_\lambda$, $U$ is $\bC_q^{\times}$-closed, thus $U=\bA^N$. This implies that $R^i\tilde\pi_{\lambda *}\mathcal O(k)=0$ for all $i>0$ and $\tilde\pi_{\lambda *}\mathcal O(k)$ is locally free of finite rank and its fiber at $\vec x\in \bA^N$ is canonically isomorphic to $\Gamma(\tilde\pi_{\lambda}^{-1}(\vec x),\mathcal O(k))$. Let us define $\mathcal V'_m=\tilde\pi_{(mn+1)\omega_1 *}\mathcal O(k)$.

Finally, it remains to show that the transition map $\tilde\pi_{\lambda' *}\mathcal O(k)\to \tilde\pi_{\lambda *}\mathcal O(k)$ is surjective whenever $\lambda'\ge \lambda$ in Bruhat order. By the $\bC_q^{\times}$-equivariance of the map $\tilde\pi_{\lambda' *}\mathcal O(k)\to \tilde\pi_{\lambda *}\mathcal O(k)$, it suffices to show that the map between fibers at zero is surjective, i.e. $\Gamma(\bGr^{N\lambda'}_{\PGL_n},\mathcal O(k))\to \Gamma(\bGr^{N\lambda}_{\PGL_n},\mathcal O(k))$ is surjective. The surjectivity for the latter is proven in \cite[Ch.XVIII]{mathieu1988formules}. This finishes the proof.
\end{proof}

\begin{proof}[Proof of Proposition \ref{prop: identify ground states}]
Since $\Gamma(\bA^N,\mathcal V_N^{n,k})\to \Gamma(\bA^N,\mathcal E_N^{n,k})$ is surjective, it induces surjective map between ground states, i.e. $\Gamma(\bA^N,\mathcal V_N^{n,k})_0\twoheadrightarrow \Gamma(\bA^N,\mathcal E_N^{n,k})_0$. Thus it suffices to show that $\dim\Gamma(\bA^N,\mathcal V_N^{n,k})_0=\dim \Gamma(\bA^N,\mathcal E_N^{n,k})_0$.

By the lemma \ref{lem: pro-finite rank} and by the $\bC_q^{\times}$-localization on $\bA^N$, the $q$-character of $\Gamma(\bA^N,\mathcal V_N^{n,k})$ is 
\begin{align*}
    \mathrm{ch}_q(\Gamma(\bA^N,\mathcal V_N^{n,k}))=\frac{1}{(1-q)^N} \mathrm{ch}_q(\Gamma(\tilde\pi^{-1}(0),\mathcal O(k))).
\end{align*}
We note that $\tilde\pi^{-1}(0)\cong \Gr^{(N)}_{\GL_n}$. According to \cite[Theorem 2.5.5]{zhu2016introduction}, $\Gamma(\Gr^{(N)}_{\GL_n},\mathcal O(k))$ is dual to $L_{k\varpi_{n-r}}(\widehat{\mathfrak{sl}}(n)_k)\otimes\mathrm{Fock}_{kN}$, where $L_{k\varpi_{n-r}}(\widehat{\mathfrak{sl}}(n)_k)$ is the level $k$ integrable representation with highest weight $k\varpi_{n-r}$ of $\widehat{\mathfrak{sl}}(n)$ and $\mathrm{Fock}_{kN}$ is the Fock module of $\widehat{\gl}(1)_{kn}$ of weight $kN$. Then we have
\begin{align*}
    \dim \Gamma(\bA^N,\mathcal V_N^{n,k})_0=\dim\Gamma(\Gr^{(N)}_{\GL_n},\mathcal O(k))_0=\dim L_{k\varpi_{n-r}}(\widehat{\mathfrak{sl}}(n)_k)^*_0=\dim L_{k\varpi_r}({\mathfrak{sl}}_n)=\dim \mathcal H_N(n,k)_0.
\end{align*}
According to Lemma \ref{lem: ground states unsymmetrized}, we have $\dim \Gamma(\bA^N,\mathcal V_N^{n,k})_0=\dim \Gamma(\bA^N,\mathcal E_N^{n,k})_0$. This proves the proposition.
\end{proof}

\subsubsection{Conformal blocks}

Fix $\ell$ distinct points $y_1,\cdots,y_\ell$ in $\bA^1$, and attach each point $y_i$ a positive integer $d_i$ such that $\sum_{i=1}^\ell d_i=N$. Then we can regard $\mathbf y=\sum_{i=1}^\ell d_iy_i$ as an element in $\bA^N$. The fiber of $\tilde\pi$ at $\mathbf y$ is isomorphic to $\prod_{i=1}^\ell\Gr^{(d_i)}_{\GL_n}$, and the restriction of $\mathcal O(k)$ to $\tilde\pi^{-1}(\mathbf y)$ is the tensor product of $\mathcal O(k)$ on each component $\Gr^{(d_i)}_{\GL_n}$. Applying Proposition \ref{prop: BWB} to $\tilde\pi^{-1}(\mathbf y)$ and we obtain the following.

\begin{proposition}\label{prop: BD BWB}
There is a vector space isomorphism $\Gamma(\tilde\pi^{-1}(\mathbf y),\mathcal O(k))^*\cong \bigotimes_{i=1}^\ell M_i$, where $M_i\cong L_{k\varpi_{n-r_i}}(\widehat{\mathfrak{sl}}(n)_{k})\otimes\mathrm{Fock}_{-kd_i}(\widehat{\gl}(1)_{kn})$, $r_i=d_i-\lfloor\frac{d_i}{n}\rfloor n$. The action of $\gl_n[z]$ on $\bigotimes_{i=1}^\ell M_i$ is via the Lie algebra map $\gl_n[z]\to \bigoplus_{i=1}^\ell \widehat{\gl}(n)_k$, where the $i$-th $\widehat{\gl}(n)_k$ is the central extension
\begin{align*}
    0\longrightarrow \bC\longrightarrow \widehat{\gl}(n)_k\longrightarrow \gl_n(\!(z-y_i)\!)\longrightarrow 0.
\end{align*}
\end{proposition}

Consider the vertex algebra $V_k(\gl_n)$ associated to affine Lie algebra $\widehat{\gl}(n)_k$. We interpret the vector space $M_i$ as a $V_k(\gl_n)$-module inserted at $y_i\in \bA^1\subset \bP^1$. In \cite[13.1.7]{frenkel2004vertex}, the space of \textit{modified conformal blocks} is defined as the Lie algebra invariant
\begin{align}
    C^0_{V_k(\gl_n)}(\bP^1,(y_i),(M_i))_{i=1}^\ell:=\left(\bigotimes_{i=1}^\ell M_i^*\right)^{\mathfrak{g}^0_{\vec y}},
\end{align}
where $\mathfrak{g}^0_{\vec y}$ is the Lie subalgebra of $\gl_n[z,\frac{1}{z-y_i}]_{1\le i\le \ell}$ consisting of those elements $f(z)$ such that $f(\infty)=0$. By Proposition \ref{prop: BD BWB}, the space of modified conformal blocks $C^0_{V_k(\gl_n)}(\bP^1,(y_i),(M_i))_{i=1}^\ell$ is naturally isomorphic to $\Gamma(\tilde\pi^{-1}(\mathbf y),\mathcal O(k))^{\mathfrak{g}^0_{\vec y}}$. 

\begin{theorem}\label{thm: grounds states in CB}
For any element $\phi\in \Gamma(\bA^N,\mathcal V_N^{n,k})_0$, its restriction to the fiber at $\mathbf y$, denoted by $\phi_{\mathbf y}\in \Gamma(\tilde\pi^{-1}(\mathbf y),\mathcal O(k))$, belongs to the subspace of modified conformal blocks. I.e. $\phi_{\mathbf y}$ is $\mathfrak{g}^0_{\vec y}$ invariant.
\end{theorem}

\begin{proof}
Define the Lie algebra $\mathfrak{g}^0$ to be the Lie subalgebra of $\gl_n[z,\frac{1}{z-x_i}]_{1\le i\le N}$ consisting of those elements $f(z)$ such that $f(\infty)=0$. We note that $\bC_q^{\times}$ acts on $\gl_n[z,\frac{1}{z-x_i}]_{1\le i\le N}$ by assigning weight $1$ to $z$ and $-1$ to $\frac{1}{z-x_i}$, so $\mathfrak{g}^0$ is negatively weighted with respect to the $\bC_q^{\times}$-action. $\mathfrak{g}^0$ naturally acts on $\Gr^{\mathbf 1}_{\GL_n,\bA^N}$, and such action lifts to an action of $\mathfrak{g}^0$ on $\mathcal O(k)$, which endows $\Gamma(\bA^N,\mathcal V_N^{n,k})$ a $\mathfrak{g}^0$-module structure. We note that $\Gamma(\tilde\pi^{-1}(\mathbf y),\mathcal O(k))$ is quotient $\mathfrak{g}^0$-module of $\Gamma(\bA^N,\mathcal V_N^{n,k})$, and $\mathfrak{g}^0$ acts on $\Gamma(\tilde\pi^{-1}(\mathbf y),\mathcal O(k))$ via the evaluation map $\mathfrak{g}^0\twoheadrightarrow \mathfrak{g}^0_{\vec y}$
\begin{align*}
     x_i\mapsto y_j , \text{ for } \sum_{u=1}^{j-1}d_u<i\le \sum_{v=1}^{j}d_v.
\end{align*}
Therefore, the theorem will automatically follow from the following statement:
\begin{align*}
    \Gamma(\bA^N,\mathcal V_N^{n,k})_0\subset \Gamma(\bA^N,\mathcal V_N^{n,k})^{\mathfrak{g}^0}.
\end{align*}
This is true because the $\mathfrak{g}^0$-action on $\Gamma(\bA^N,\mathcal V_N^{n,k})$ is $\bC_q^{\times}$-equivariant, and $\mathfrak{g}^0$ is negatively graded, so $\mathfrak{g}^0$ takes $\Gamma(\bA^N,\mathcal V_N^{n,k})_0$ to lower $\bC_q^{\times}$-weight components which must be zero.
\end{proof}

\subsubsection{Knizhnik-Zamolodchikov equation}

It is known that the Beilinsen-Drinfeld Grassmannians $\{\Gr_{\GL_n,\bA^N}\}_{N\in \bZ_{>0}}$ forms a factorization space over the Ran space of $\bA^1$ \cite[20.3.5]{frenkel2004vertex}, and $\mathcal O(1)$ is a factorizable line bundle \cite[20.4.1]{frenkel2004vertex}. Then it follows from the theory of factorization space that $\mathcal V_N^{n,k}=\tilde\pi_*^{\mathbf 1}\mathcal O(k)$ naturally inherits a $D$-module structure from the factorzation structure \cite[20.4.1]{frenkel2004vertex}, i.e. it is equipped with a flat connection $\nabla: \mathcal V_N^{n,k}\to \mathcal V_N^{n,k}\otimes\Omega^1_{\bA^N}$.

\begin{proposition}\label{prop: ground states are flat sections}
Any ground state of $\Gamma(\bA^N,\mathcal V_N^{n,k})$ is a flat section of $\nabla$, i.e. 
\begin{align*}
    \nabla(\phi)=0,\; \forall \phi\in \Gamma(\bA^N,\mathcal V_N^{n,k})_0.
\end{align*}
\end{proposition}

\begin{proof}
$\nabla$ is $\bC_q^{\times}$-equivariant by construction. Let $\nabla_i\in \End(\mathcal V_N^{n,k})$ be the action of $i$-th tangent vector $\partial_i$ on $\mathcal V_N^{n,k}$, then $\nabla_i$ decrease the energy grading by one, hence it annihilates the ground states.
\end{proof}

When restricted to the open subset $\bA^N_{\mathrm{disj}}\subset \bA^N$, we can give a trivialization of $\mathcal V_N^{n,k}$ using the Proposition \ref{prop: BD BWB}:
\begin{align}
    \mathcal V_N^{n,k}|_{\bA^N_{\mathrm{disj}}}\cong \mathcal O_{\bA^N_{\mathrm{disj}}}\otimes \bigotimes_{i=1}^N M_i^*,\quad M_i=L_{k\varpi_{n-1}}(\widehat{\mathfrak{sl}}(n)_{k})\otimes\mathrm{Fock}_{-k}(\widehat{\gl}(1)_{kn}).
\end{align}
According to \cite[13.3.3]{frenkel2004vertex}, for any point $\vec x=(x_1,\cdots,x_N)\in \bA^N_{\mathrm{disj}}$ the space of modified conformal blocks $C^0_{V_k(\gl_n)}(\bP^1,(x_i),(M_i))_{i=1}^N$ is isomorphic to $\left(S^k\bC^n\right)^{\otimes N}$. Moreover, we can vary the point $\vec x$ and the space of modified conformal blocks forms a sub-bundle 
\begin{align*}
    C^0_{k}\left(\mathbf x,\bigotimes_{i=1}^N M_i\right)\subset \mathcal V_N^{n,k}|_{\bA^N_{\mathrm{disj}}},
\end{align*}
whose fiber at a point $\vec x$ is $C^0_{V_k(\gl_n)}(\bP^1,(x_i),(M_i))_{i=1}^N$. According to the construction in \cite[13.3.3]{frenkel2004vertex} the isomorphism $C^0_{V_k(\gl_n)}(\bP^1,(x_i),(M_i))_{i=1}^N\cong \left(S^k\bC^n\right)^{\otimes N}$ is induced by restricting a linear function $\tau\in \bigotimes_{i=1}^N M_i^*$ to the ground states of $\bigotimes_{i=1}^N M_i$, and the latter is dual to $\left(S^k\bC^n\right)^{\otimes N}$. We can translate the above construction into the geometric language: the dual of ground states of $\bigotimes_{i=1}^N M_i$ is naturally identified with $\Gamma((\bP^{n-1})^N,\mathcal O(k)^{\boxtimes N})$ where $(\bP^{n-1})^N$ is the fiber of $\bGr^{\omega_1,\cdots,\omega_1}_{\GL_n,\bA^N}$ at $\vec x$, and the restriction-to-ground-states map is exactly the natural quotient map $\Gamma(\Gr^{\mathbf 1}_{\GL_n,\vec x},\mathcal O(k))\twoheadrightarrow \Gamma(\bGr^{\omega_1,\cdots,\omega_1}_{\GL_n,\vec x},\mathcal O(k))$. Varying $\vec x\in\bA^N_{\mathrm{disj}}$ and we obtain the following.
\begin{proposition}\label{prop: CB=E}
The projection $\mathcal V_N^{n,k}\twoheadrightarrow\mathcal E_N^{n,k}$ induces an isomorphism between vector bundles
\begin{align}\label{CB=E}
    C^0_{k}\left(\mathbf x,\bigotimes_{i=1}^N M_i\right)\cong \mathcal E_N^{n,k}|_{\bA^N_{\mathrm{disj}}}
\end{align}
\end{proposition}

A key property of the sub-bundle $C^0_{k}\left(\mathbf x,\bigotimes_{i=1}^N M_i\right)\subset \mathcal V_N^{n,k}|_{\bA^N_{\mathrm{disj}}}$ is that it is a sub D-module \cite[13.3.7.(2)]{frenkel2004vertex}, i.e. the connection $\nabla$ preserves $C^0_{k}\left(\mathbf x,\bigotimes_{i=1}^N M_i\right)$. The restriction of $\nabla$ to $C^0_{k}\left(\mathbf x,\bigotimes_{i=1}^N M_i\right)$ is known as the Knizhnik-Zamolodchikov connection \cite{knizhnik1984current}. Using the isomorphism \eqref{CB=E} and the trivialization \eqref{restriction of E to open}, we obtain trivialization $C^0_{k}\left(\mathbf x,\bigotimes_{i=1}^N M_i\right)\cong \mathcal O_{\bA^N_{\mathrm{disj}}}\otimes\left(S^k\bC^n\right)^{\otimes N}$. The coordinate form of connection $\nabla$ on $C^0_{k}\left(\mathbf x,\bigotimes_{i=1}^N M_i\right)$ in this trivialization is the following:
\begin{align}\label{eqn:KZ}
    \nabla_i=\partial_i-\frac{1}{k+n}\sum_{j\neq i}\frac{\Omega^{\mathfrak{sl}_n}_{ij}}{x_i-x_j}-\frac{1}{kn}\sum_{j\neq i}\frac{k^2}{x_i-x_j},
\end{align}
where the first summation term corresponds to the module $L_{k\varpi_{n-1}}(\widehat{\mathfrak{sl}}(n)_{k})$ which is worked out in \cite[13.3.8]{frenkel2004vertex} and the second summation term corresponds to the module $\mathrm{Fock}_{-k}(\widehat{\gl}(1)_{kn})$ which is worked out in \cite[13.2.6]{frenkel2004vertex}. Here $\Omega^{\mathfrak{sl}_n}_{ij}$ is the quadratic Casimir for $\mathfrak{sl}_n$. $\Omega^{\mathfrak{sl}_n}_{ij}$ equals to $\Omega_{ij}-\frac{k^2}{n}\mathrm{Id}$ when acting on $S^k\bC^n\otimes S^k\bC^n$, so we can rewrite \eqref{eqn:KZ} as
\begin{align}\label{eqn:KZ 2}
    \nabla_i=\partial_i-\frac{1}{k+n}\sum_{j\neq i}\frac{\Omega_{ij}+k}{x_i-x_j}.
\end{align}
By Theorem \ref{thm: grounds states in CB} the restriction of any ground state $\phi\in \Gamma(\bA^N,\mathcal V_N^{n,k})_0$ to the open locus $\bA^N_{\mathrm{disj}}$ is a section of the sub-bundle $C^0_{k}\left(\mathbf x,\bigotimes_{i=1}^N M_i\right)$, and by the Proposition \ref{prop: ground states are flat sections} the coordinate form of $\phi$ satisfies the KZ equation $\nabla(\phi)=0$ where $\nabla$ is the KZ connection in \eqref{eqn:KZ 2}. Combine Remark \ref{rmk: wave function via E}, Lemma \ref{lem: ground states unsymmetrized}, and Proposition \ref{prop: identify ground states}, we see that the wave functions of ground states in $\mathcal H_N(n,k)$ satisfy the KZ equation. We summarize the result as follows.
\begin{corollary}\label{cor:KZ eqn}
Ground states wave functions solve the KZ equation, i.e.
\begin{align*}
    \nabla(\phi)=0,\; \forall \phi\in \mathfrak{W}(\mathcal H_N(n,k)_0),
\end{align*}
where $\nabla$ is the KZ connection \eqref{eqn:KZ 2}.
\end{corollary}

The result that the Chern-Simons matrix model ground states wave functions satisfy KZ equation is not new, see \cite[5.1]{Dorey-Tong-Turner} for a derivation in the case when $N$ is divisible by $n$ and \cite[Appendix A.2]{bourgine2024calogero} for a proof in general. The previous proofs are computational and involve careful analysis on explicit formulae of the ground states. The new ingredient here is that we provide a geometric proof which essentially boils down to the simple fact that ground states have the lowest $\bC^{\times}_q$ weights among all physical states. Our method does not require explicit formulae of the ground states.

\section{Conformal Limit, Part I: Transition Maps}\label{sec:Conformal Limit, Part I}

In this section, we define and study the conformal limit ($N\to \infty$) of the Hilbert spaces of the matrix model.

To begin with, we introduce a natural transition map $\mathcal H_{N+n}(n,k)\to \mathcal H_{N}(n,k)$. We will define this map in two ways: geometrically and algebraically, which turn out to be equivalent construction. The outcome of the constructions is a $\GL_n[z]\rtimes\bC^{\times}_q$-equivariant map $\mathcal H_{N+n}(n,k)\to \mathcal H_{N}(n,k)\otimes \chi_N^{\otimes k}$ where $\chi_N$ is certain character of $\GL_n[z]\rtimes\bC^{\times}_q$ which plays the role of retaining equivariance. And we shall define the conformal limit Hilbert space as $\bC^{\times}_q$-finite subspace of the inverse limit of $\mathcal H_{N}(n,k)$ equipped with twisted $\GL_n[z]\rtimes\bC^{\times}_q$-actions.

\subsection{Geometric construction of the transition map}\label{subsec:transition map_geometric}

Consider the closed embedding 
\begin{align}\label{embedding of Quot}
    \iota_N:\Quot^N(\mathcal O_{\bA^1}^{\oplus n})\hookrightarrow \Quot^{N+n}(\mathcal O_{\bA^1}^{\oplus n})
\end{align}
which maps a subsheaf $\mathcal E\subset \mathcal O_{\bA^1}^{\oplus n}$ of codimension $N$ to the subsheaf $\mathcal E\otimes\mathcal O_{\bA^1}(-[0])\subset \mathcal O_{\bA^1}^{\oplus n}$ which has codimension $N+n$. Here $\mathcal O_{\bA^1}(-[0])$ is the ideal sheaf corresponding to the origin $0\in \bA^1$.

\begin{lemma}\label{lem: transition for det line bundle}
The pullback of the determinant line bundle $\iota_N^*\mathcal L_{\det}$ is $\GL_n[z]\rtimes \bC^{\times}_q$-equivariantly isomorphic to $\mathcal L_{\det}\otimes \chi_N$, where $\chi_N$ is the equivariant line bundle which is the pullback of the character
\begin{align}\label{twisting character}
    \GL_n[z]\rtimes \bC^{\times}_q\to \bC^{\times},\quad (g[z],t)\mapsto t^N\cdot\det g[0],
\end{align}
along the projection from $\Quot^N(\mathcal O_{\bA^1}^{\oplus n})$ to a point.
\end{lemma}

\begin{proof}
Consider the line bundle $\mathcal L_{\det}^{-1}\otimes \iota_N^*\mathcal L_{\det}$, we claim that it is (non-equivariantly) isomorphic to the structure sheaf of $\Quot^N(\mathcal O_{\bA^1}^{\oplus n})$. Let $D\subset \bA^{(N)}$ be the prime divisor given by the equation $\prod_{i=1}^Nx_i=0$ \footnote{The equation $\prod_{i=1}^Nx_i=0$ on $\bA^{N}$ cuts out $N$ hyperplanes, and the symmetric group acts on $N$ irreducible components transitively, therefore its image in $\bA^{(N)}$ is irreducible, i.e. $D$ is a prime divisor.}. The preimage of $\bA^{(N)}\setminus D$ along the Hilbert-Chow map $h:\Quot^N(\mathcal O_{\bA^1}^{\oplus n})\to \bA^{(N)}$ is the moduli of codimensional $N$ subsheaves $\mathcal E\subset \mathcal O_{\bA^1}^{\oplus n}$ such that $\{0\}$ is not in the support of $\mathcal O_{\bA^1}^{\oplus n}/\mathcal E$. Therefore the fiber of $\iota_N^*\mathcal L_{\det}$ at a point $(\mathcal E\subset \mathcal O_{\bA^1}^{\oplus n})$ in $h^{-1}(\bA^{(N)}\setminus D)$ is $\det(\mathcal O_{\bA^1}^{\oplus n}/\mathcal E)\otimes \det (\mathcal O_{\bA^1}^{\oplus n}|_{\{0\}})$. It follows that $\left(\mathcal L_{\det}^{-1}\otimes \iota_N^*\mathcal L_{\det}\right)|_{h^{-1}(\bA^{(N)}\setminus D)}$ is isomorphic to the structure sheaf tensoring the line $\det (\mathcal O_{\bA^1}^{\oplus n}|_{\{0\}})$ which is non-equivariantly isomorphic to the structure sheaf, thus $\mathcal L_{\det}^{-1}\otimes \iota_N^*\mathcal L_{\det}$ is isomorphic to a divisor supported on $h^{-1}(D)$. Since $h$ is flat by Corollary \ref{cor:flatness}, $h^{-1}(D)$ is a divisor. Note that the restriction of $h$ to $h^{-1}(\bA^{(N)}\setminus \Delta)$ is a locally trivial fibration with fibers isomorphic to $(\bP^{n-1})^{\times N}$, where $\Delta$ is the diagonal divisor, therefore $h^{-1}(D)$ is prime because the intersection $D\cap (\bA^{(N)}\setminus \Delta)$ is nonempty. As a result $\mathcal L_{\det}^{-1}\otimes \iota_N^*\mathcal L_{\det}\cong \mathcal O_{\Quot^N(\mathcal O_{\bA^1}^{\oplus n})}(m\cdot h^{-1}(D))$ for some $m$. Since $\mathcal O_{\Quot^N(\mathcal O_{\bA^1}^{\oplus n})}(h^{-1}(D))\cong h^*\mathcal{O}_{\bA^{(N)}}(D)$ and $\mathcal{O}_{\bA^{(N)}}(D)$ is isomorphic to $\mathcal{O}_{\bA^{(N)}}$, this proves our claim.

Since $\Gamma(\Quot^N(\mathcal O_{\bA^1}^{\oplus n}),\mathcal O_{\Quot^N(\mathcal O_{\bA^1}^{\oplus n})})^{\times}=\Gamma(\bA^{(N)},\mathcal O_{\bA^{(N)}})^{\times}=\bC^{\times}$, the equivariant structure of the trivial line bundle $\mathcal L_{\det}^{-1}\otimes \iota_N^*\mathcal L_{\det}$ must be pullback of a character $\chi:\GL_n[z]\rtimes \bC^{\times}_q\to \bC^{\times}$ along the projection from $\Quot^N(\mathcal O_{\bA^1}^{\oplus n})$ to a point. It remains to find $\chi$. To this end, we only need to find the action of $\chi$ on $\GL_n[z]$ and on $\bC^{\times}_q$ separately. For $\GL_n[z]$, we notice that $\left(\mathcal L_{\det}^{-1}\otimes \iota_N^*\mathcal L_{\det}\right)|_{h^{-1}(\bA^{(N)}\setminus D)}$ is $\GL_n[z]$-equivariantly isomorphic to $\mathcal O_{h^{-1}(\bA^{(N)}\setminus D)}\otimes\det (\mathcal O_{\bA^1}^{\oplus n}|_{\{0\}})$, thus $\chi:\GL_n[z]\to \bC^{\times}$ is the determinant representation $\chi(g[z])=\det g[0]$. For $\bC^{\times}_q$, we look at the fiber of $\mathcal L_{\det}^{-1}\otimes \iota_N^*\mathcal L_{\det}$ at a $\bC^{\times}_q$-fixed point, say $\oplus_{i=1}^n z^{\lambda_i}\mathcal O_{\bA^1}\subset \mathcal O_{\bA^1}^{\oplus_n}$ ($\sum_{i=1}^n\lambda_i=N$), then the fiber of $\mathcal L_{\det}^{-1}\otimes \iota_N^*\mathcal L_{\det}$ at this point is $\otimes_{i=1}^n z^{\lambda_i}\mathcal O_{\{0\}}$ which has $\bC^{\times}_q$-weight $N$. Thus $\chi:\bC^{\times}_q\to \bC^{\times}$ is given by the $N$-th power map. This finishes the proof.
\end{proof}

\begin{lemma}\label{lem: surjectivity for transition map}
For any positive integer $k$, the induced map on global sections $$\iota_N^*:\Gamma(\Quot^{N+n}(\mathcal O_{\bA^1}^{\oplus n}),\mathcal L_{\det}^{\otimes k})\to \Gamma(\Quot^{N}(\mathcal O_{\bA^1}^{\oplus n}),\mathcal L_{\det}^{\otimes k})\otimes \chi_N^{\otimes k}$$ is surjective.
\end{lemma}

\begin{proof}
The proof is similar to that of the surjectivity statement in Proposition \ref{prop: level-rank map}. By Lemma \ref{lem: transition for det line bundle} $\iota_N^*$ is a $\gl_n[z]$-module map, and by Corollary \ref{cor: cyclic} both $\Gamma(\Quot^{N+n}(\mathcal O_{\bA^1}^{\oplus n}),\mathcal L_{\det}^{\otimes k})$ and $\Gamma(\Quot^{N}(\mathcal O_{\bA^1}^{\oplus n}),\mathcal L_{\det}^{\otimes k})\otimes \chi^{\otimes k}$ are generated from a ground state by the action of $\gl_n[z]$, thus we only need to show that $\iota_N^*(v)\neq 0$ for some ground state $v$. By Remark \ref{rmk: ground states}, the restriction $\Gamma(\Quot^{N+n}(\mathcal O_{\bA^1}^{\oplus n}),\mathcal L_{\det}^{\otimes k})\to \mathcal L_{\det}^{\otimes k}|_{\{z^{\lambda}\}}$ maps the $\gl_n$ highest weight vector of $\mathcal H_{N+n}(n,k)_0$ isomorphically onto $\mathcal L_{\det}^{\otimes k}|_{\{z^{\lambda}\}}$. Here $\lambda_i=\lfloor\frac{N+n-i}{n}\rfloor, 1\le i\le n$, and $\{z^{\lambda}\}$ corresponds to the subsheaf $\oplus_{i=1}^n z^{\lambda_i}\mathcal O_{\bA^1}\subset \mathcal O_{\bA^1}^{\oplus n}$. Thus the image of the $\gl_n$ highest weight vector of $\mathcal H_{N+n}(n,k)_0$ is nonzero in $\Gamma(\Quot^{N}(\mathcal O_{\bA^1}^{\oplus n}),\mathcal L_{\det}^{\otimes k})\otimes \chi_N^{\otimes k}$. This finishes the proof.

\end{proof}

The transition map $\iota_N^*: \mathcal H_{N+n}(n,k)\to \mathcal H_{N}(n,k)\otimes \chi_N^{\otimes k}$ is compatible with the map \eqref{fixed point induce level-rank}. Namely, we have the following commutative diagram of embeddings of moduli spaces:
\begin{equation}
\begin{tikzcd}
\Quot^N(\mathcal O_{\bA^1}^{\oplus n})\ar[r,hook,"\iota_N"] \ar[d,hook,"j_N"] &\Quot^{N+n}(\mathcal O_{\bA^1}^{\oplus n}) \ar[d,hook,"j_{N+n}"]\\
\Quot^{kN}(\mathcal O_{\bA^1}^{\oplus kn})\ar[r,hook,"\iota_{kN}"] &\Quot^{kN+kn}(\mathcal O_{\bA^1}^{\oplus kn}) 
\end{tikzcd},
\end{equation}
where the vertical arrows maps a subsheaf $\mathcal E\subset \mathcal O_{\bA^1}^{\oplus n}$ to $\mathcal E\otimes \mathcal O_{\bA^1}^{\oplus k} \subset \mathcal O_{\bA^1}^{\oplus n}\otimes \mathcal O_{\bA^1}^{\oplus k}=\mathcal O_{\bA^1}^{\oplus kn}$. Note that 
\begin{align*}
    j_N^*\chi_{kN}= \chi_{N}^{\otimes k}.
\end{align*}
Therefore we obtain a commutative diagram of maps between global sections of line bundles:
\begin{equation}
\begin{tikzcd}
\mathcal{H}_{kN+kn}(kn,1)\ar[r,two heads,"\iota_{kN}^*"] \ar[d,two heads,"\tau_{N+n}"] & \mathcal{H}_{kN}(kn,1)\otimes \chi_{kN} \ar[d,two heads,"\tau_{N}"]\\
\mathcal{H}_{N+n}(n,k)\ar[r,two heads,"\iota_{N}^*"]  & \mathcal{H}_{N}(n,k)\otimes \chi_{N}^{\otimes k}
\end{tikzcd},
\end{equation}
and all the maps in the above diagram are $\GL_n[z]\rtimes\bC^{\times}_q$-equivariant.

\subsection{Algebraic construction of the transition map}\label{subsec:transition map_algebraic}

Before we introduce the algebraic construction of transition map which looks unnatural from the first glance, it is instructive to explain why a innocent-looking construction does not work.

Let us choose a splitting $\bC^{N+n}=\bC^N\oplus\bC^n$ which induces an embedding $V(N,n)\hookrightarrow V(N+n,n)$ where $V(N,n)$ is defined in the beginning of Section \ref{sec:phase space}. The restriction-to-subspace maps  polynomial function ring $\bC[V(N+n,n)]$ surjectively onto $\bC[V(N,n)]$. Let us take the semi-invariant subspace $\bC[V(N+n,n)]^{\GL_{N+n},-k}$, which is by definition our Hilbert space $\mathcal H_{N+n}(n,k)$, and we claim that its image in $\bC[V(N,n)]$ is zero. In fact, the restriction map $\bC[V(N+n,n)]\to\bC[V(N,n)]$ is $\GL_N\times \GL_n$ equivariant, where $\GL_N\times \GL_n$ action on $\bC[V(N+n,n)]$ is induced from diagonal embedding $\GL_N\times \GL_n\hookrightarrow\GL_{N+n}$ and $\GL_N\times \GL_n$ action on $\bC[V(N,n)]$ is induced from projection to the first component. Since $\GL_n$ acts on $\bC[V(N+n,n)]^{\GL_{N+n},-k}$ via the character $g\mapsto \det(g)^{-k}$, the image of $\bC[V(N+n,n)]^{\GL_{N+n},-k}$ in $\bC[V(N,n)]$ must be zero. So the naive restriction map $\bC[V(N+n,n)]\to\bC[V(N,n)]$ does not work.

To fix the issue, we consider the following vector space
\begin{align}
    \widetilde{V}(N,n):=\End(\bC^N)\oplus\Hom(\bC^n,\bC^{N+n}),
\end{align}
which embeds into $V(N+n,n)=\End(\bC^{N+n})\oplus\Hom(\bC^n,\bC^{N+n})$ by splitting $\bC^{N+n}=\bC^N\oplus\bC^n$. The notation for linear coordinates on $\widetilde{V}(N,n)$ is that
\begin{align}
    X\in \End(\bC^N),\quad A \in \Hom(\bC^{N},\bC^n),\quad \eta\in \Hom(\bC^{n},\bC^n),
\end{align}
with components $(X^i_j)_{1\le i,j,\le N}$, $(A^a_i)_{1\le a\le n, 1\le i\le N}$, $(\eta^a_{\ell})_{1\le a\le n, N+1\le \ell\le N+n}$.

Let us define a $\gl_n[z]$-action on $\widetilde{V}(N,n)$ by
\begin{align}\label{gl_n[z] action on aux space}
    E^a_b\otimes z^m \mapsto A^a_i(X^m)^i_j\frac{\partial}{\partial A^b_j}+\delta_{m=0}\eta^a_{\ell}\frac{\partial}{\partial \eta^b_{\ell}},
\end{align}
which exponentiates to a $\GL_n[z]$-action. We define the $\bC^{\times}_q$-grading on $\widetilde{V}(N,n)$ by $\deg(X)=1$, $\deg(A)=\deg(\eta)=0$
\begin{lemma}\label{lem:restriction map}
The restriction map $\bC[V(N+n,n)]\to \bC[\widetilde{V}(N,n)]$ is $\GL_n[z]\rtimes \bC^{\times}_q$-equivariant, where the $\GL_n[z]$-action on $\bC[V(N+n,n)]$ is the standard one \eqref{eq: gl_n[z] action} and the $\GL_n[z]$-action on $\bC[\widetilde{V}(N,n)]$ is given by \eqref{gl_n[z] action on aux space}.
\end{lemma}

\begin{proof}
Straightforward computation.
\end{proof}

\begin{lemma}\label{lem:alg construction}
The image of $\mathcal H_{N+n}(n,k)=\bC[V(N+n,n)]^{\GL_{N+n},-k}$ under the restriction map $\bC[V(N+n,n)]\to \bC[\widetilde{V}(N,n)]$ is
\begin{align}
    \det(X)^k\cdot \det(\eta)^k\cdot\bC[X,A]^{\GL_N,-k},
\end{align}
where $\GL_N$ acts on $X$ by adjoint action and acts on $A$ by dual fundamental representation.
\end{lemma}

\begin{proof}
Let us denote $G_1$ (resp. $G_2$) to be the $\GL_N$ (resp. $\GL_n$) factor inside the diagonal embedding $\GL_N\times \GL_n\subset \GL_{N+n}$. Then for $i=1,2$, $G_i$ acts on the image of $\bC[V(N+n,n)]^{\GL_{N+n},-k}$ under the restriction map $\bC[V(N+n,n)]\to \bC[\widetilde{V}(N,n)]$ by the character $g\mapsto \det(g)^{-k}$. The subspace of $\bC[\widetilde{V}(N,n)]$ on which $G_2$ acts by the character $g\mapsto \det(g)^{-k}$ is $\det(\eta)^k\cdot \bC[X,A]$. The subspace of $\det(\eta)^k\cdot \bC[X,A]$ on which $G_1$ acts by the character $g\mapsto \det(g)^{-k}$ is $\det(\eta)^k\cdot \bC[X,A]^{\GL_N,-k}$. Thus the image lies in the subspace $\det(\eta)^k\cdot\bC[X,A]^{\GL_N,-k}$. Since $\bC[V(N+n,n)]\to \bC[\widetilde{V}(N,n)]$ is $\gl_n[z]$-equivariant and both $\bC[V(N+n,n)]^{\GL_{N+n},-k}$ and $\det(X)^k\cdot \det(\eta)^k\cdot\bC[X,A]^{\GL_N,-k}$ are generated from a ground state by $\gl_n[z]$-action (Corollary \ref{cor: cyclic}), it is enough to show that there exists a ground state in $\bC[V(N+n,n)]^{\GL_{N+n},-k}$ whose image in $\bC[\widetilde{V}(N,n)]$ is of the form $\det(X)^k\cdot \det(\eta)^k$ times a ground state in $\bC[X,A]^{\GL_N,-k}$. Using the explicit presentation \eqref{eq: ground states general n}, we can check that
\begin{align}
    \ket{1,\cdots,n,a_1,\cdots,a_N}\mapsto \det(X)^k\cdot \det(\eta)^k\cdot \ket{a_1,\cdots,a_N}.
\end{align}
This finishes the proof.
\end{proof}

We note that the $\GL_n[z]\rtimes \bC^{\times}_q$-action on $\det(X)^k\cdot \det(\eta)^k$ is exactly given by the character $\chi_N$, which agrees with the geometric construction. We claim that the map $\mathcal H_{N+n}(n,k)\twoheadrightarrow \det(X)^k\cdot \det(\eta)^k\cdot\mathcal H_{N}(n,k)$ constructed in Lemma \ref{lem:alg construction} agrees with $\iota_N^*:\mathcal H_{N+n}(n,k)\twoheadrightarrow \chi_N^{\otimes k}\otimes\mathcal H_{N}(n,k)$ in the previous subsection, upon a choice of isomorphism $\det(X)\cdot \det(\eta)\cong \chi_N$. In fact, both $\mathcal H_{N+n}(n,k)$ and $\chi_N^{\otimes k}\otimes\mathcal H_{N}(n,k)$ are generated by ground states, so any $\GL_n[z]\rtimes \bC^{\times}_q$-equivariant map between them is uniquely determined by the $\gl_n$-equivariant map between the ground state $\mathcal H_{N+n}(n,k)_0$ and $\chi_N^{\otimes k}\otimes\mathcal H_{N}(n,k)_0$. Since the ground states are irreducible $\gl_n$-modules by Remark \ref{rmk: ground states}, any two such maps are differed by a scalar multiple according to Schur's lemma. The ambiguity of choosing a scalar multiple is fixed by choosing an isomorphism $\det(X)\cdot \det(\eta)\cong \chi_N$.

\begin{definition}\label{def:p_N and sigma_N}
We define the twisted Hilbert space $\widetilde{\mathcal H}_N(n,k)$ to be the Hermitian vector space ${\mathcal H}_N(n,k)$ together with \textit{shifted energy grading} given by the eigenvalues of 
\begin{align}\label{eq:twisted energy grading}
    X^i_j\frac{\partial}{\partial X^i_j}-\frac{k}{2}L(L-1)n-krL,\;\text{where }L=\lfloor\frac{N}{n}\rfloor,\; r=N-Ln.
\end{align}
We equip $\widetilde{\mathcal H}_N(n,k)$ with the following graded $\gl_n[z]$-action:
\begin{align}\label{eq:twisted gl_n[z] action}
    E^a_b\otimes z^m \mapsto A^a_i(X^m)^i_j \frac{\partial}{\partial A^b_j}-\delta_{m=0}\delta^a_b kL.
\end{align}
We define the transition map $p_N:\widetilde{\mathcal H}_{N+n}(n,k)\twoheadrightarrow \widetilde{\mathcal H}_N(n,k)$ to be $\det(X)^{-k}\cdot \det(\eta)^{-k}$ multiplying with the restriction map $\bC[V(N+n,n)]^{\GL_{N+n},-k}\to \det(X)^k\cdot \det(\eta)^k\cdot\bC[X,A]^{\GL_N,-k}$ in the Lemma \ref{lem:alg construction}. We define a section $\sigma_N$ of the projection $p_N$ to be the composition
\begin{equation}\label{eq:sigma_N}
\begin{tikzcd}
\widetilde{\mathcal H}_N(n,k)\ar[rrr,"\left(p_N|_{\ker(p_N)^\perp}\right)^{-1}"] & & &\ker(p_N)^\perp\ar[r,hook] &\widetilde{\mathcal H}_{N+n}(n,k).
\end{tikzcd}
\end{equation}
Here we have used the fact that $p_N|_{\ker(p_N)^\perp}$ is isomorphism.
\end{definition}

\begin{remark}
The ground states $\widetilde{\mathcal H}_N(n,k)_0$ has degree $0$ with respect to the energy grading \eqref{eq:twisted energy grading}. Moreover, $\widetilde{\mathcal H}_N(n,k)_0$ is an irreducible $\gl_n$- module of highest weight $k\varpi_r$ with respect to the action via taking $m=0$ in \eqref{eq:twisted gl_n[z] action}.
\end{remark}

We shall write the energy grading decomposition as
\begin{align}
    \widetilde{\mathcal H}_N(n,k)=\bigoplus_{d\ge 0} \widetilde{\mathcal H}_N(n,k)_d,
\end{align}
where $\widetilde{\mathcal H}_N(n,k)_d$ is the eigenspace of the operator \eqref{eq:twisted energy grading} with eigenvalue $d$.

\begin{proposition}\label{prop: p_N is gl_n[z] equiv}
The projection $p_N$ and the section $\sigma_N$ preserve the energy grading on $\widetilde{\mathcal H}_{N+n}(n,k)$ and $\widetilde{\mathcal H}_N(n,k)$. Moreover, $p_N$ is $\gl_n[z]$-equivariant.
\end{proposition}

\begin{proof}
$\widetilde{\mathcal H}_N(n,k)$ is isomorphic to $\det(A)^{-kL}\cdot \mathcal H_N(n,k)$ as $\gl_n[z]$-module, therefore $p_N$ is $\gl_n[z]$-equivariant by Lemma \ref{lem:restriction map}. Since $p_N$ maps $\widetilde{\mathcal H}_{N+n}(n,k)_0$ to $\widetilde{\mathcal H}_N(n,k)_0$, and both $\widetilde{\mathcal H}_{N+n}(n,k)$ and $\widetilde{\mathcal H}_N(n,k)$ are generated from ground states by $\gl_n[z]$-actions which respects the energy grading, thus $\gl_n[z]$-equivariance of $p_N$ implies that $p_N$ respects the energy grading. The Hermitian inner product on $\widetilde{\mathcal H}_N(n,k)$ respects the grading, i.e.
\begin{align*}
    \langle \widetilde{\mathcal H}_N(n,k)_d\ket{\widetilde{\mathcal H}_N(n,k)_{d'}}=0,\;\text{whenever }d\neq d'.
\end{align*}
Therefore $\ker(p_N)^\perp$ is a graded subspace, i.e. $\ker(p_N)^\perp=\bigoplus_{d\ge 0} \ker(p_N)^\perp\cap \widetilde{\mathcal H}_N(n,k)_d$, whence $\sigma_N$ respects the energy grading.
\end{proof}

As we have discussed in the previous subsection, the transition map $p_N$ is compatible with the map $\tau_N$ defined in \eqref{level-rank map}, i.e. we have a commutative diagram:
\begin{equation}\label{eq:level-rank and transition}
\begin{tikzcd}
\widetilde{\mathcal H}_{kN+kn}(kn,1)\ar[r,two heads,"p_{kN}"] \ar[d,two heads,"\tau_{N+n}"] & \widetilde{\mathcal H}_{kN}(kn,1) \ar[d,two heads,"\tau_{N}"]\\
\widetilde{\mathcal H}_{N+n}(n,k)\ar[r,two heads,"p_N"]  & \widetilde{\mathcal H}_{N}(n,k)
\end{tikzcd}.
\end{equation}
This can also be directly checked using the explicit formula of $p_N$ given below.

\subsubsection{Explicit formula of \texorpdfstring{$p_N$}{pN}}

By Remark \ref{rmk:explicit spanning set}, every element in $\widetilde{\mathcal H}_{N+n}(n,k)$ is a linear combination of following elements
\begin{align}
\ket{(a^{(\alpha)}_j),(m^{(\alpha)}_j)}=\prod_{\alpha=1}^k \left[\epsilon^{i_{1}i_{2}\cdots i_{N+n}}(A^{a^{(\alpha)}_1}X^{m^{(\alpha)}_1})_{i_{1}}\cdots (A^{a^{(\alpha)}_{N+n}}X^{m^{(\alpha)}_{N+n}})_{i_{N+n}}\right]\ket{\emptyset}.
\end{align}
Without loss of generality, let us assume that $\forall \alpha$, $m^{(\alpha)}_1\ge \cdots\ge m^{(\alpha)}_{N+n}$. Then we have
\begin{align}\label{eq:explicit p_N}
p_N\left(\ket{(a^{(\alpha)}_j),(m^{(\alpha)}_j)}\right) =\delta\times \prod_{\alpha=1}^k \left[\epsilon^{i_{1}i_{2}\cdots i_{N}}(A^{a^{(\alpha)}_1}X^{m^{(\alpha)}_1-1})_{i_{1}}\cdots (A^{a^{(\alpha)}_{N}}X^{m^{(\alpha)}_{N}-1})_{i_{N}}\right]\ket{\emptyset},
\end{align}
where the factor $\delta$ is given by
\begin{align}
\delta=\prod_{\alpha=1}^k\left[\epsilon^{a^{(\alpha)}_{N+1}\cdots a^{(\alpha)}_{N+n}}_{12\cdots n}\prod_{\ell=N+1}^{N+n}\delta_{m^{(\alpha)}_\ell=0}\right].
\end{align}
We note that if $\delta\neq 0$ then it automatically follows that $m^{(\alpha)}_j\ge 1$ for all $\alpha$ and all $j\le N$, so that \eqref{eq:explicit p_N} makes sense.

\subsubsection{Fermion Fock space construction in the case of \texorpdfstring{$k=1$}{k=1}}

When $k=1$, we can similarly define a shifted version of the fermion Fock space $\mathcal F_N(n)$ in Section \ref{sec:fermion fock}. Namely, let $L=\lfloor\frac{N}{n}\rfloor$, $r=N-Ln$, and define $\widetilde{\mathcal F}_N(n)$ to be the wedge space
\begin{align}
    \widetilde{\mathcal F}_N(n)=\bigwedge^N[\psi^a_{m}\:|\: a\in \{1,\cdots,n\}, m\in \mathbb Z_{\ge -L}].
\end{align}
The energy grading on $\widetilde{\mathcal F}_N(n)$ is such that
\begin{align}
    \deg(\psi^{a_1}_{m_1}\wedge\cdots \wedge\psi^{a_N}_{m_N})=\frac{nL(L+1)}{2}+\sum_{i=1}^N m_i.
\end{align}
The space of ground states $\widetilde{\mathcal F}_N(n)_0$ has degree $0$, and we write the energy grading decomposition
\begin{align}
    \widetilde{\mathcal F}_N(n)=\bigoplus_{d\ge 0} \widetilde{\mathcal F}_N(n)_d.
\end{align}
We equip $\widetilde{\mathcal F}_N(n)$ with the following graded $\gl_n[z]$-action:
\begin{align}\label{eq:twisted gl_n[z] action_fermion}
    E^a_b\otimes z^m\mapsto \sum_{\ell\ge \max\{-L,-m\}} \psi^a_{\ell+m}\frac{\partial}{\partial \psi^b_\ell}-\sum_{-L\le\ell<-m} \frac{\partial}{\partial \psi^b_\ell}\psi^a_{\ell+m}.
\end{align}
Then the linear map $\widetilde{\mathfrak{f}}_N:\widetilde{\mathcal F}_N(n)\to \widetilde{\mathcal H}_N(n,1)$ which is given by
\begin{align}
    \widetilde{\mathfrak{f}}_N(\psi^{a_1}_{m_1}\wedge\cdots \wedge\psi^{a_N}_{m_N})=\epsilon^{i_1i_2\cdots i_N}(A^{a_1}X^{m_1+L})_{i_{1}}\cdots (A^{a_{N}}X^{m_{N}+L})_{i_{N}}\ket{\emptyset}
\end{align}
is a graded $\gl_n[z]$-module isomorphism.

Define $p^{\mathfrak{f}}_N:=\widetilde{\mathfrak{f}}^{-1}_{N}\circ p_N\circ \widetilde{\mathfrak{f}}_{N+n}:\widetilde{\mathcal F}_{N+n}(n)\to \widetilde{\mathcal F}_N(n)$ to be the transition map for the shifted fermion Fock spaces. Explicitly, let $\psi^{a_1}_{m_1}\wedge\cdots \wedge\psi^{a_{N+n}}_{m_{N+n}}\in \widetilde{\mathcal F}_{N+n}(n)$ be an element, and we assume that $m_1\ge \cdots\ge m_{N+n}$ without loss of generality, then we have
\begin{align}\label{eq:p_N for fock}
    p^{\mathfrak{f}}_N(\psi^{a_1}_{m_1}\wedge\cdots \wedge\psi^{a_{N+n}}_{m_{N+n}})=\delta'\times \psi^{a_1}_{m_1}\cdots \psi^{a_{N}}_{m_{N}},
\end{align}
where the factor $\delta'$ is given by
\begin{align}
    \delta'=\epsilon^{a_{N+1}\cdots a_{N+n}}_{12\cdots n}\prod_{\ell=N+1}^{N+n}\delta_{m_\ell=-L-1}.
\end{align}

\subsection{Definition of the conformal limit of \texorpdfstring{$\widetilde{\mathcal H}_N(n,k)$}{H(n,k)}}

\begin{definition}\label{def:conformal limit}
Fix an integer $r\in\{0,\cdots,n-1\}$, we define the charge $r$ conformal limit Hilbert space $\widetilde{\mathcal H}^{(r)}_{\infty}(n,k)$ to be the graded vector space which is degree-wise completion of the inverse system $\left\{\widetilde{\mathcal H}_{nL+r}(n,k), p_{nL+r}\right\}_{L\in \mathbb N}$, i.e.
\begin{align}\label{eq:conformal limit}
    \widetilde{\mathcal H}^{(r)}_{\infty}(n,k)=\bigoplus_{d\ge 0}\widetilde{\mathcal H}^{(r)}_{\infty}(n,k)_d,\;\text{where }\;\widetilde{\mathcal H}^{(r)}_{\infty}(n,k)_d:=\underset{\substack{\longleftarrow\\L}}{\lim}\: \widetilde{\mathcal H}_{nL+r}(n,k)_d.
\end{align}
Here we use $p_N$ in Definition \ref{def:p_N and sigma_N} to define the inverse limit. Moreover, we denote the natural projection from the limit to a finite stage by $p^{\infty}_N:\widetilde{\mathcal H}^{(r)}_{\infty}(n,k)\twoheadrightarrow \widetilde{\mathcal H}_{N}(n,k)$.
\end{definition}

\begin{remark}\label{rmk:p_N are module maps}
Since the transition maps are graded $\gl_n[z]$-module maps, it follows that $\widetilde{\mathcal H}^{(r)}_{\infty}(n,k)$ inherits a graded $\gl_n[z]$-module structure such that $p^{\infty}_{N}$ are graded $\gl_n[z]$-module morphisms. 
\end{remark}

\begin{lemma}\label{lem:finiteness}
$\widetilde{\mathcal H}^{(r)}_{\infty}(n,k)_d$ is finite dimensional for all $n,k,d,r$.
\end{lemma}

\begin{proof}
It suffices to show that for fixed $n,k,d$, there exists $C(n,k,d)\in \bN$ such that $\dim \widetilde{\mathcal H}_N(n,k)_d\le C(n,k,d)$ for all $N$. Since the multiplication map $\widetilde{\mathcal H}_N(n,1)^{\otimes k}\to \widetilde{\mathcal H}_N(n,k)$ is surjective, it is enough to prove the lemma for $k=1$. Using Corollary \ref{cor:k=1 fixed pt}, we have graded vector space isomorphism
\begin{align}
    \widetilde{\mathcal H}_N(n,1)\cong \bigoplus_{N_1+\cdots+N_n=N} \widetilde{\mathcal H}_{N_1}(1,1)\otimes\cdots\otimes\widetilde{\mathcal H}_{N_n}(1,1),
\end{align}
therefore it suffices to prove the lemma for $n=k=1$. As explained in Example \ref{ex: H_N(1,k)}, $\widetilde{\mathcal H}_{N}(1,1)$ is isomorphic to $\bC[\bA^{(N)}]$ as graded vector space. Thus $\dim \widetilde{\mathcal H}_{N}(1,1)_d\le p(d)$, where $p(d)$ is the number of partitions of $d$. So we take $C(1,1,d)=p(d)$, and take
\begin{align}
    C(n,k,d)=\sum_{d=\sum_{i=1}^n\sum_{j=1}^k d_{i,j}}\prod_{i=1}^n\prod_{j=1}^k p(d_{i,j}).
\end{align}
This finishes the proof.
\end{proof}

\begin{proposition}\label{prop:p_N stabilizes}
The inverse system in \eqref{eq:conformal limit} stabilizes for $L\gg 0$, i.e. $p^{\infty}_{nL+r}|_{\widetilde{\mathcal H}^{(r)}_{\infty}(n,k)_d}:\widetilde{\mathcal H}^{(r)}_{\infty}(n,k)_d\twoheadrightarrow \widetilde{\mathcal H}_{nL+r}(n,k)_d$ are isomorphisms for $L\gg 0$ and fixed $n,k,d$.
\end{proposition}

\begin{proof}
Since $\dim\widetilde{\mathcal H}_{nL+r}(n,k)_d\le \dim \widetilde{\mathcal H}^{(r)}_{\infty}(n,k)_d<\infty$ for all $L$, there exists $M\in \bN$ such that
\begin{align*}
\dim\widetilde{\mathcal H}_{nL+r}(n,k)_d=\dim\widetilde{\mathcal H}_{nM+r}(n,k)_d \text{ for all }L\ge M,
\end{align*}
whence $p_{nL+r}|_{\widetilde{\mathcal H}_{(n+1)L+r}(n,k)_d}$ are isomorphisms for all $L\ge M$. This implies that $p^{\infty}_{nL+r}|_{\widetilde{\mathcal H}^{(r)}_{\infty}(n,k)_d}$ are isomorphisms for $L\ge M$.
\end{proof}

Another corollary is that $\widetilde{\mathcal H}^{(r)}_{\infty}(n,k)$ can be reconstructed using the direct system $\left\{\widetilde{\mathcal H}_{N}(n,k), \sigma_{N}\right\}$. Namely, let us define $\sigma^\infty_N: \widetilde{\mathcal H}_{N}(n,k)\to \widetilde{\mathcal H}^{(r)}_{\infty}(n,k)$ to be the $M\to \infty$ limit of the system of maps 
\begin{align*}
    \sigma_{N+nM}\circ\cdots\circ\sigma_{N+n}\circ\sigma_N: \widetilde{\mathcal H}_{N}(n,k)\to \widetilde{\mathcal H}_{N+nM+n}(n,k).
\end{align*}
By construction, $\sigma^\infty_N$ is a graded homomorphism, and it is a section of $p^\infty_N$, i.e. $p^\infty_N\circ\sigma^\infty_N=\mathrm{Id}$. Moreover, $\sigma^\infty_{N+n}\circ \sigma_{N}=\sigma^\infty_N$, and this induces a graded homomorphism
\begin{align}\label{eq:sigma^infty_r}
    \sigma^{(r)}_\infty:  \underset{\substack{\longrightarrow\\L}}{\lim}\: \widetilde{\mathcal H}_{nL+r}(n,k)\longrightarrow \widetilde{\mathcal H}^{(r)}_{\infty}(n,k).
\end{align}

\begin{proposition}\label{prop:ind limit=proj limit}
$\sigma^{(r)}_\infty$ in \eqref{eq:sigma^infty_r} is an isomorphism.
\end{proposition}

\begin{proof}
Since $p^\infty_N|_{\widetilde{\mathcal H}^{(r)}_{\infty}(n,k)_d}$ is isomorphism for $N\gg 0$, the same is true for $\sigma^\infty_N|_{\widetilde{\mathcal H}_{N}(n,k)_d}$. Taking direct limit, we see that $\sigma^{(r)}_\infty$ is an isomorphism on degree-$d$ component, therefore $\sigma^{(r)}_\infty$ is an isomorphism.
\end{proof}

The tensor multiplication map $\widetilde{\mathcal H}^{(r)}_{\infty}(n,1)^{\otimes k}\to \widetilde{\mathcal H}^{(r)}_{\infty}(n,k)$ is well-defined as the $L\to \infty$ limit of the system of tensor multiplication maps $\widetilde{\mathcal H}_{nL+r}(n,1)^{\otimes k}\to \widetilde{\mathcal H}_{nL+r}(n,k)$, since the latter are compatible with transition maps. Similarly, the level-rank map $\tau^{(r)}_{\infty}:\widetilde{\mathcal H}^{(kr)}_{\infty}(kn,1)\to \widetilde{\mathcal H}^{(r)}_{\infty}(n,k)$ is well-defined as the $L\to \infty$ limit of the system of level-rank maps $\tau_{nL+r}:\widetilde{\mathcal H}_{k(nL+r)}(kn,1)\to \widetilde{\mathcal H}_{nL+r}(n,k)$, which are compatible with transition maps by \eqref{eq:level-rank and transition}.

\begin{proposition}\label{prop:tensor and level-rank at conformal limit}
The tensor multiplication map $\widetilde{\mathcal H}^{(r)}_{\infty}(n,1)^{\otimes k}\to \widetilde{\mathcal H}^{(r)}_{\infty}(n,k)$ and the level-rank map $\tau^{(r)}_{\infty}:\widetilde{\mathcal H}^{(kr)}_{\infty}(kn,1)\to \widetilde{\mathcal H}^{(r)}_{\infty}(n,k)$ are surjective.
\end{proposition}

\begin{proof}
To show the surjectivity of tensor multiplication map, it is enough to show that it is degree-wise surjective, i.e. for every fixed $d$,
\begin{align*}
    \bigoplus_{d=d_1+\cdots+d_k}\bigotimes_{i=1}^k \widetilde{\mathcal H}^{(r)}_{\infty}(n,1)_{d_i}\longrightarrow \widetilde{\mathcal H}^{(r)}_{\infty}(n,k)_d
\end{align*}
is surjective. We choose $N$ such that $\mathfrak{p}^\infty_N: \widetilde{\mathcal H}^{(r)}_{\infty}(n,1)_{d'}\to \widetilde{\mathcal H}_{N}(n,1)_{d'}$ and $\mathfrak{p}^\infty_N: \widetilde{\mathcal H}^{(r)}_{\infty}(n,k)_{d}\to \widetilde{\mathcal H}_{N}(n,k)_{d}$ are isomorphisms for all $d'\le d$. Then the result follows from the surjectivity of $\widetilde{\mathcal H}_{N}(n,1)^{\otimes k}\to \widetilde{\mathcal H}_{N}(n,k)$. The proof of surjectivity of the level-rank map is similar and we omit the detail.
\end{proof}

\begin{proposition}\label{prop:cyclic at limit}
$\widetilde{\mathcal H}^{(r)}_{\infty}(n,k)$ is a cyclic $\gl_n[z]$-module. Namely, $\widetilde{\mathcal H}^{(r)}_{\infty}(n,k)$ is generated from a ground state $\ket{v}\in \widetilde{\mathcal H}^{(r)}_{\infty}(n,k)_0$ by $\gl_n[z]$-action.
\end{proposition}

\begin{proof}
We need to show that $U(\gl_n[z])\otimes \ket{v}\to \widetilde{\mathcal H}^{(r)}_{\infty}(n,k)$ is surjective. Since this map is graded, it is enough to show that $U(\gl_n[z])_d\otimes \ket{v}\to \widetilde{\mathcal H}^{(r)}_{\infty}(n,k)_d$ is surjective for all $d$. Let us fix $d$ and choose $N$ such that $p^\infty_N:\widetilde{\mathcal H}^{(r)}_{\infty}(n,k)_d\to \widetilde{\mathcal H}_{N}(n,k)_d$ is isomorphism, then the result follows from the surjectivity of $U(\gl_n[z])\otimes p^\infty_N(\ket{v})\to \widetilde{\mathcal H}_{N}(n,k)$ (Corollary \ref{cor: cyclic}).
\end{proof}

Recall the character of a graded $\gl_n$ module defined in \eqref{def:character}, and according to our construction of $\widetilde{\mathcal H}^{(r)}_{\infty}(n,k)$ and Theorem \ref{thm: Hilbert series}, we have
\begin{align*}
    \ch_{q,\mathbf{a}}(\widetilde{\mathcal H}^{(r)}_{\infty}(n,k))=\underset{L\to \infty}{\lim} \: \mathfrak{A}^{-kL} q^{-\frac{k}{2}L(L-1)n-krL}\chi_{q,\mathbf{a}}(\overline{\Gr}^{(nL+r) \omega_1}_{\GL_n},\mathcal O(k))\prod_{i=1}^{nL+r}\frac{1}{1-q^i}.
\end{align*}
Here $\mathfrak{A}=\prod_{i=1}^n\mathbf a_i$. Notice that
\begin{align*}
    \mathfrak{A}^{-kL} q^{-\frac{k}{2}L(L-1)n-krL}\chi_{q,\mathbf{a}}(\overline{\Gr}^{(nL+r) \omega_1}_{\GL_n},\mathcal O(k))=\chi_{q,\mathbf{a}}(\overline{\Gr}^{(nL+r)\omega_1-L\omega_n}_{\GL_n},\mathcal O(k)),
\end{align*}
thus we have
\begin{align*}
    \ch_{q,\mathbf{a}}(\widetilde{\mathcal H}^{(r)}_{\infty}(n,k))=\chi_{q,\mathbf{a}}(\overline{\Gr}^{(r)}_{\GL_n,\mathrm{red}},\mathcal O(k))\prod_{i=1}^{\infty}\frac{1}{1-q^i}.
\end{align*}
Here $\overline{\Gr}^{(r)}_{\GL_n,\mathrm{red}}$ is the $r$-th connected component of affine Grassmannian $\overline{\Gr}^{(r)}_{\GL_n}$ endowed with reduced scheme structure. According to \cite[Theorem 2.5.5]{zhu2016introduction}, $\chi_{q,\mathbf{a}}(\overline{\Gr}^{(r)}_{\GL_n,\mathrm{red}},\mathcal O(k))$ is dual to $L_{k\varpi_{n-r}}(\widehat{\mathfrak{sl}}(n)_k)$, where $L_{k\varpi_{n-r}}(\widehat{\mathfrak{sl}}(n)_k)$ is the level $k$ integrable representation with highest weight $k\varpi_{n-r}$ of $\widehat{\mathfrak{sl}}(n)$. Therefore we get
\begin{align}\label{eq:char of H_infty}
    \ch_{q,\mathbf{a}}(\widetilde{\mathcal H}^{(r)}_{\infty}(n,k))=\ch_{q,\mathbf{a}}(L_{k\varpi_{r}}(\widehat{\mathfrak{sl}}(n)_k))\prod_{i=1}^{\infty}\frac{1}{1-q^i}.
\end{align}
If we interpret $\prod_{i=1}^{\infty}\frac{1}{1-q^i}$ as the character of a Fock space, then $\widetilde{\mathcal H}^{(r)}_{\infty}(n,k)$ has the same character as $L_{k\varpi_{r}}(\widehat{\mathfrak{sl}}(n)_k)\otimes\mathrm{Fock}$. In the next section we will show that $\widetilde{\mathcal H}^{(r)}_{\infty}(n,k)$ is an integrable $\widehat{\gl}(n)$-module of level $k$ and is isomorphic to $L_{k\varpi_{r}}(\widehat{\gl}(n)_k)$. $\widetilde{\mathcal H}^{(r)}_{\infty}(n,k)$ is already a cyclic $\gl_n[z]$ module, and we will see in the next section that the annihilation operators in $\widehat{\gl}(n)$ are conformal limits of certain rescaling of operators $\mathsf T_{m,0}(E^a_b)$ in \eqref{DDCA action}.

\subsubsection{Semi-infinite wedge construction in the case of \texorpdfstring{$k=1$}{k=1}}

In the case of $k=1$, our matrix model reduces to the spin Calogero-Sutherland model, and Uglov constructed its conformal limit via the semi-infinite wedge \cite{uglov1996semi,uglov1998yangian}. We briefly recall his construction here.

\begin{definition}\label{def:conformal limit_k=1}
Fix an integer $r\in\{0,\cdots,n-1\}$, we define the charge $r$ conformal limit fermion Fock space $\widetilde{\mathcal F}^{(r)}_{\infty}(n)$ to be the graded vector space which is degree-wise completion of the inverse system $\{\widetilde{\mathcal F}_{nL+r}(n), p^{\mathfrak{f}}_{nL+r}\}_{L\in \mathbb N}$, i.e.
\begin{align}\label{eq:conformal limit_k=1}
    \widetilde{\mathcal F}^{(r)}_{\infty}(n)=\bigoplus_{d\ge 0}\widetilde{\mathcal F}^{(r)}_{\infty}(n)_d,\;\text{where }\;\widetilde{\mathcal F}^{(r)}_{\infty}(n)_d:=\underset{\substack{\longleftarrow\\L}}{\lim}\: \widetilde{\mathcal F}_{nL+r}(n)_d.
\end{align}
\end{definition}

By construction, $\widetilde{\mathcal F}^{(r)}_{\infty}(n)$ is isomorphic to $\widetilde{\mathcal H}^{(r)}_{\infty}(n,1)$ via the isomorphism $\mathfrak{f}_{\infty}$ which is the limit of isomorphisms $\{\mathfrak{f}_{nL+r}\}_{L\in \bN}$.

$\widetilde{\mathcal F}^{(r)}_{\infty}(n)$ has a basis 
\begin{equation}
\begin{split}
&\psi^{a_1}_{m_1}\wedge \psi^{a_2}_{m_2}\wedge \cdots,\quad (m_1\ge m_2\ge \cdots, \text{ and } a_i<a_j \text{ when }m_i=m_j),\\
& m_j=\lfloor\frac{r-j}{n}\rfloor\text{ when }j\gg 0.
\end{split}
\end{equation}
The energy grading on $\widetilde{\mathcal F}^{(r)}_{\infty}(n)$ is such that
\begin{align}
    \deg(\psi^{a_1}_{m_1}\wedge \psi^{a_2}_{m_2}\wedge \cdots)=\sum_{j=1}^{\infty}m_j-\lfloor\frac{r-j}{n}\rfloor.
\end{align}
The space of ground states $\widetilde{\mathcal F}^{(r)}_{\infty}(n)_0$ is spanned by
\begin{equation}
\begin{split}
&\psi^{a_1}_{0}\wedge \cdots \wedge\psi^{a_r}_{0}\wedge \psi^{1}_{-1} \wedge \cdots \wedge\psi^{n}_{-1}\wedge \psi^{1}_{-2} \wedge \cdots \wedge\psi^{n}_{-2}\wedge\cdots.
\end{split}
\end{equation}
The graded $\gl_n[z]$-action on $\widetilde{\mathcal F}_N(n)$ is given by:
\begin{align}\label{eq:gl_n[z] action_semi-infinite}
     E^a_b\otimes z^m\mapsto \sum_{\ell\ge -m} \psi^a_{\ell+m}\frac{\partial}{\partial \psi^b_\ell}-\sum_{\ell<-m} \frac{\partial}{\partial \psi^b_\ell}\psi^a_{\ell+m}.
\end{align}
If we allow $m$ in the above equation to be negative, then we get a graded $\widehat{\gl}(n)$-action on $\widetilde{\mathcal F}_N(n)$. For the convenience of later discussions, we swap the positive and negative modes to align with VOA conventions.
\begin{proposition}[{\cite[Lecture 9]{raina2013bombay}}]
The assignment 
\begin{align}
    J^a_{b,m}\mapsto \sum_{\ell\ge 0} \psi^a_{\ell}\frac{\partial}{\partial \psi^b_{\ell+m}}-\sum_{\ell<0} \frac{\partial}{\partial \psi^b_{\ell+m}}\psi^a_{\ell}
\end{align}
defines a graded $\widehat{\gl}(n)_1$-action on $\widetilde{\mathcal F}^{(r)}_{\infty}(n):$
\begin{align}
    [J^a_{b,p}, J^c_{d,q}]=\delta^c_bJ^a_{d,p+q}-\delta^a_dJ^c_{b,p+q}+kp\delta_{p+q,0}\delta^a_d\delta^c_b.
\end{align}
Moreover, $\widetilde{\mathcal F}^{(r)}_{\infty}(n)$ is isomorphic to $L_{\varpi_r}(\widehat{\gl}(n)_1)$, the $r$-th fundamental representation of $\widehat{\gl}(n)_1$, with respect to the above action. In particular, $\widetilde{\mathcal F}^{(r)}_{\infty}(n)$ is irreducible and integrable.
\end{proposition}

Note that $J^a_{b,-m}= E^a_b\otimes z^m$ for all $m\ge 0$. The above proposition is the starting point of our discussion on the conformal limit of operators acting on $\widetilde{\mathcal H}^{(r)}_{\infty}(n,k)$, which we will carried out in detail in the next section.

\section{Conformal Limit, Part II: Operators}\label{sec:Conformal Limit, Part II}

\subsection{Level-rank duality and \texorpdfstring{$\widehat{\gl}(n)$}{affine gl(n)} module structure}

As we explained in the last section, the $\gl_n[z]$-action on $\widetilde{\mathcal H}^{(r)}_{\infty}(n,1)\cong \widetilde{\mathcal F}^{(r)}_{\infty}(n)$ extends to an integrable $\widehat{\gl}(n)_1$-action, such that $\widetilde{\mathcal H}^{(r)}_{\infty}(n,1)$ is isomorphic to the $r$-th fundamental representation $L_{\varpi_r}(\widehat{\gl}(n)_1)$. 

Now consider the semi-infinite wedge space $\widetilde{\mathcal F}^{(kr)}_{\infty}(kn)$, which is isomorphic to $L_{\varpi_{kr}}(\widehat{\gl}(kn)_1)$. Since $\widehat{\gl}(kn)_1$ contains subalgebra $\widehat{\mathfrak{sl}}(k)_n\oplus \widehat{\mathfrak{sl}}(n)_k\oplus \widehat{\gl}(1)_{kn}$ in a diagonal manner, $L_{\varpi_{kr}}(\widehat{\gl}(kn)_1)$ decomposes accordingly into direct sum of tensor products of irreducible modules of $\widehat{\mathfrak{sl}}(k)_n$, $ \widehat{\mathfrak{sl}}(n)_k$, and $\widehat{\gl}(1)_{kn}$. The following level-rank duality result is due to I. Frenkel, see also \cite[Theorem AFF]{hasegawa1989spin}.

\begin{theorem}[{\cite[Theorem 1.6]{frenkel2006representations}}]
As an $\widehat{\mathfrak{sl}}(k)_n\oplus \widehat{\mathfrak{sl}}(n)_k\oplus \widehat{\gl}(1)_{kn}$-module,
\begin{align}\label{eq:level-rank duality}
    L_{\varpi_{kr}}(\widehat{\gl}(kn)_1)\cong \bigoplus_{\lambda}L_{\lambda}(\widehat{\mathfrak{sl}}(k)_n)\otimes L_{\lambda^{\mathrm{t}}}(\widehat{\mathfrak{sl}}(n)_k)\otimes \mathrm{Fock}_{kr}(\widehat{\gl}(1)_{kn}),
\end{align}
where the summation is taken for all $\lambda=(\lambda_1\ge\cdots\ge\lambda_k)\in \bZ^k$ such that $\sum_{i=1}^k\lambda_i=kr$ and $\lambda_1-\lambda_k\le n$. Here $\lambda^{\mathrm{t}}$ is the transpose of $\lambda$ in the sense of a Maya diagram \cite{hasegawa1989spin}.
\end{theorem}

By the proposition \ref{prop:tensor and level-rank at conformal limit}, there exists a graded surjective $\gl_n[z]$-module map $\tau^{(r)}_\infty: \widetilde{\mathcal H}^{(kr)}_{\infty}(kn,1)\twoheadrightarrow \widetilde{\mathcal H}^{(r)}_{\infty}(n,k)$. Since the $\mathfrak{sl}_k[z]$ subalgebra of $\gl_{kn}[z]$ acts on $\widetilde{\mathcal H}^{(r)}_{\infty}(n,k)$ trivially, $\tau^{(r)}_\infty$ factors through the $\mathfrak{sl}_k[z]$-coinvariant, i.e. 
\begin{align}\label{eq:reduced level-rank at conformal limit}
    \bar{\tau}^{(r)}_\infty: \widetilde{\mathcal H}^{(kr)}_{\infty}(kn,1)_{\mathfrak{sl}_k[z]}\twoheadrightarrow \widetilde{\mathcal H}^{(r)}_{\infty}(n,k).
\end{align}

\begin{proposition}\label{prop:level-rank map isom at infty}
The map $\bar{\tau}^{(r)}_\infty$ in \eqref{eq:reduced level-rank at conformal limit} is an isomorphism.
\end{proposition}

\begin{proof}
Since $\bar{\tau}^{(r)}_\infty$ is surjective, it is enough to show that two sides of $\bar{\tau}^{(r)}_\infty$ have the same character. The character of $\widetilde{\mathcal H}^{(r)}_{\infty}(n,k)$ is computed in \eqref{eq:char of H_infty}, which equals to the character of $L_{k\varpi_r}(\widehat{\mathfrak{sl}}(n)_k)\otimes \mathrm{Fock}$. In view of the direct sum decomposition of \eqref{eq:level-rank duality}, the only summand which is nonzero after taking $\mathfrak{sl}_k[z]$-coinvariant is the one corresponding to $\lambda=(n,\cdots,n)$. In fact, we have
\begin{align*}
L_{\lambda}(\widehat{\mathfrak{sl}}(k)_n)_{\mathfrak{sl}_k[z]}=
\begin{cases}
    \bC, & \text{if }\forall (i,j), \lambda_i=\lambda_j, \\
    0, & \text{otherwise}.
\end{cases}
\end{align*}
It follows that 
\begin{align*}
    \ch_{q,\mathbf a}(\widetilde{\mathcal H}^{(kr)}_{\infty}(kn,1)_{\mathfrak{sl}_k[z]})=\ch_{q,\mathbf a}(L_{k\varpi_r}(\widehat{\mathfrak{sl}}(n)_k)\otimes \mathrm{Fock})=\ch_{q,\mathbf a}(\widetilde{\mathcal H}^{(r)}_{\infty}(n,k)).
\end{align*}
This finishes the proof.
\end{proof}

\begin{corollary}\label{cor:affine gl(n) action on conformal limit}
The $\gl_n[z]$-action on $\widetilde{\mathcal H}^{(r)}_{\infty}(n,k)$ extends to an $\widehat{\gl}(n)_k$-action, such that $\widetilde{\mathcal H}^{(r)}_{\infty}(n,k)$ is isomorphic to $L_{k\varpi_r}(\widehat{\mathfrak{sl}}(n)_k)\otimes \mathrm{Fock}_{kr}(\widehat{\gl}(1)_{kn})$. In particular, $\widetilde{\mathcal H}^{(r)}_{\infty}(n,k)$ is irreducible and integrable.
\end{corollary}

\begin{definition}
We denote the generators of $\widehat{\gl}(n)_k$ acting on $\widetilde{\mathcal H}^{(r)}_{\infty}(n,k)$ by $\{J^a_{b,m}\:|\: 1\le a,b\le n, m\in \mathbb Z\}$ such that $\forall m\ge 0$, $J^a_{b,-m}=E^a_b\otimes z^m$ where the latter is $\gl_n[z]$ generator. In this convention $J^a_{b,>0}\cdot \widetilde{\mathcal H}^{(r)}_{\infty}(n,k)_0=0$. We denote
\begin{align}\label{eq:trace+tracelss decomposition}
    J^a_{b,m}=\bar J^a_{b,m}+\frac{\delta^a_b}{n}\alpha_m,
\end{align}
where $\bar J^a_{b,m}$ is the traceless part (i.e. $\widehat{\mathfrak{sl}}(n)_k$ generators), and $\alpha_m$ is the trace part (i.e. $\widehat{\gl}(1)_{kn}$ generators).
\end{definition}

We will show that the annihilation operators $J^a_{b,m},(m>0)$  arises from the matrix model operators $\mathsf T_{m,0}(E^a_b)$ via certain scaling limit, see Theorem \ref{thm:conformal limit of T[m,0]}. 

\subsection{\texorpdfstring{$Y(\gl_n)$}{Y(gl(n))} action on \texorpdfstring{$\widetilde{\mathcal H}^{(r)}_{\infty}(n,k)$}{H(n,k)}}\label{subsec:Yangian action on H_N(n,k)}

Consider the assignment 
\begin{equation}\label{eq:Yangian}
\begin{split}
&\widetilde{T}^a_b(u)\mapsto \left[\delta^a_b+A^a\frac{1}{u+(k+n)L-XY }B_b\right]\frac{\left(1+\frac{u}{n+k}\right)_L}{\left(1+\frac{u+k}{n+k}\right)_L},\\
&\text{where }L=\lfloor\frac{N}{n}\rfloor,\; r=N-nL.\\
&\text{Here we use the Pochhammer symbol notation $(x)_a=\frac{\Gamma(x+a)}{\Gamma(x)}$}.
\end{split}
\end{equation}
Note that $\widetilde{T}^a_b(u)$ in \eqref{eq:Yangian} is obtained from ${T}^a_b(u)$ in \eqref{eq:naive Yangian} by spectral parameter shift $u\mapsto u+(k+n)L$ followed by multiplying a function $f(u)$, therefore \eqref{eq:Yangian} gives rise to a Yangian algebra $Y(\gl_n)$-action on $\widetilde{\mathcal H}_N(n,k)$, i.e. $\widetilde{T}^a_b(u)$ satisfies the RTT relation
\begin{align*}
    (u-v)[\widetilde{T}^a_b(u),\widetilde{T}^c_d(v)]=\widetilde{T}^c_b(u)\widetilde{T}^a_d(v) -\widetilde{T}^c_b(v)\widetilde{T}^a_d(u).
\end{align*}
\begin{definition}\label{def:Yangian notation_finite}
Fix $r\in \{1,\cdots,n-1\}$. Write $\widetilde{T}^a_b(u)=\delta^a_b+\sum_{m\ge 0}\widetilde{T}^a_{b;m}u^{-m-1}$. We denote the Yangian generator $\widetilde{T}^a_{b;m}$ which acts on $\widetilde{\mathcal H}_{nL+r}(n,k)$ by $\prescript{L}{}{\widetilde{T}}^a_{b;m}$.
\end{definition}

\begin{theorem}\label{thm:Yangian comptible with p_N and sigma_N}
$p_N$ and $\sigma_N$ are $Y(\gl_n)$-module maps with respect to the $Y(\gl_n)$-actions defined by \eqref{eq:Yangian}, i.e.
\begin{align}
    p_N\circ \prescript{L+1}{}{\widetilde{T}}^a_b(u)=\prescript{L}{}{\widetilde{T}}^a_b(u)\circ p_N,\quad \sigma_N\circ \prescript{L}{}{\widetilde{T}}^a_b(u)=\prescript{L+1}{}{\widetilde{T}}^a_b(u)\circ \sigma_N.
\end{align}
\end{theorem}

\begin{remark}\label{rmk:Yangian compatible 1}
Since the Hermitian inner product $\langle\cdot | \cdot\rangle$ on $\widetilde{\mathcal H}_N(n,k)$ is skew-invariant under the action of $Y(\gl_n)$, i.e.
\begin{align*}
    \langle \widetilde{T}^a_b(u)\cdot w | {v}\rangle=\langle w | \widetilde{T}^b_a(u)\cdot v\rangle
\end{align*}
for all $\ket{w},\ket{v}\in \widetilde{\mathcal H}_N(n,k)$, then $p_N$ being $Y(\gl_n)$-module map implies $\sigma_N$ being $Y(\gl_n)$-module map. In fact, suppose $p_N$ is a $Y(\gl_n)$-module map, then $\ker(p_N)$ is $Y(\gl_n)$-submodule of $\widetilde{\mathcal H}_{N+n}(n,k)$, thus $\ker(p_N)^\perp$ is also a $Y(\gl_n)$-submodule of $\widetilde{\mathcal H}_{N+n}(n,k)$, therefore $\sigma_N$ is a $Y(\gl_n)$-module map by construction. So we only need to prove that $p_N$ is a $Y(\gl_n)$-module map.
\end{remark}

\begin{remark}\label{rmk:Yangian compatible 2}
Let us recall ${T}^a_b(u)$ defined in \eqref{eq:naive Yangian}: ${T}^a_b(u)= \delta^a_b+A^a\frac{1}{u-XY }B_b$, then it is elementary to see that
\begin{align}\label{eq:Yangian compatible 2}
    p_N\circ \prescript{L+1}{}{\widetilde{T}}^a_b(u)=\prescript{L}{}{\widetilde{T}}^a_b(u)\circ p_N\;\Longleftrightarrow\; p_N\circ \prescript{L+1}{}{T}^a_b(u)=\frac{u+k}{u}\cdot\prescript{L}{}{T}^a_b(u-n-k)\circ p_N.
\end{align}
Therefore it suffices to prove the equation for ${T}^a_b(u)$ in \eqref{eq:Yangian compatible 2} in order to prove that $p_N$ is a $Y(\gl_n)$-module map.
\end{remark}

\begin{remark}\label{rmk:Yangian compatible 3}
Recall that for $m\in \{1,\cdots,n\}$, the quantum minor $A_m(u)$ is defined as 
\begin{align}\label{GZ generators}
    A_m(u):=\mathrm{qdet}\:T^a_b(u)_{1\le a,b\le m}=\sum_{\sigma\in S_m}\mathrm{sgn}(\sigma)\prod_{1\le i\le m}^{\rightarrow}T^{\sigma(i)}_i(u-i+1).
\end{align}
It is well-known that $Y(\gl_n)$ is generated by $\{T^a_{b;0}\:|\: 1\le a,b\le n\}$ and coefficients of $\{A_m(u)\:|\: 1\le m\le n\}$. Then, in order to prove the equation for ${T}^a_b(u)$ in \eqref{eq:Yangian compatible 2}, it is enough to show that
\begin{align}\label{eq:Yangian compatible 3.1}
    p_N\circ \prescript{L+1}{}{T}^a_{b;0}=\left(\prescript{L}{}{T}^a_{b;0}+k\delta^a_b\right)\circ p_N,
\end{align}
and that
\begin{align}\label{eq:Yangian compatible 3.2}
    p_N\circ \prescript{L+1}{}{A}_m(u)=\frac{(u-m+1+k)_m}{(u-m+1)_m}\cdot\prescript{L}{}{A}_m(u-n-k)\circ p_N, \quad (1\le m\le n).
\end{align}
Since $T^a_{b;0}=\mathsf T_{0,0}(E^a_b)$, \eqref{eq:Yangian compatible 3.1} follows from Proposition \ref{prop: p_N is gl_n[z] equiv} and the definition of $\gl_n[z]$ action \eqref{eq:twisted gl_n[z] action}.
\end{remark}

Now, Remark \ref{rmk:Yangian compatible 1}, \ref{rmk:Yangian compatible 2}, and \ref{rmk:Yangian compatible 3} reduce the proof of Theorem \ref{thm:Yangian comptible with p_N and sigma_N} to the proof of equation \eqref{eq:Yangian compatible 3.2}, which is the main goal of the rest of this subsection.

\bigskip\noindent\textbf{Observation.} Let us consider the $(n+N)\times(n+N)$ matrix of operators:
\begin{equation}\label{eq:mat E}
\mathsf E :=\quad  \begin{pNiceArray}{c:c}[first-row, first-col, columns-width=auto]
    & n & N \\
n   & AB & AY \\
\hdottedline
N & XB & XY
\end{pNiceArray}
\end{equation}
then $\mathsf E$ satisfies the $\gl_{n+N}$ commutation relations:
\begin{align*}
[\mathsf E^\mu_\nu, \mathsf E^\rho_\sigma]= \delta^\rho_\nu\mathsf E^\mu_\sigma-\delta^\mu_\sigma\mathsf E^\rho_\nu, \quad (1\le \mu,\nu,\rho,\sigma\le n+N).
\end{align*}
$\mathsf E$ acts on $\bC[V(N,n)]=\bC[X,A]$ naturally, so we consider the following operator acting on $\bC[X,A]$:
\begin{align}\label{commutant construction}
\mathfrak T^a_b(u):=\delta^a_b+\sum_{s\ge 1}\frac{(\mathsf E^s)^a_b}{u^s},\quad (1\le a,b\le n),
\end{align}
i.e. $\mathfrak T(u)$ is the restriction of $(1-\mathsf E/u)^{-1}$ to the upper $n\times n$ block. $\mathfrak T(u)$ commutes with $(XY)^i_j$, $(YX)^i_j$, and $(BA)^i_j$, so $\mathfrak T(u)$ commutes with the moment map $\mu^*_{\bC}(\gl_N)$ defined in \eqref{complexified moment map}, therefore $\mathfrak T(u)$'s action on $\bC[X,A]$ leaves $\bC[X,A]^{\GL_N,-k}$ invariant.
\begin{lemma}\label{lem:commutant}
When restricted to $\bC[X,A]^{\GL_N,-k}$, we have the following equation between operators:
\begin{align}
    \mathfrak T^a_b(u)=T^a_b(u-n-k).
\end{align}
\end{lemma}

\begin{proof}
First of all, it is straightforward to compute that
\begin{align}
    \mathfrak T^a_b(u)=\delta^a_b+A^a\frac{1}{u-YX-B A }B_b.
\end{align}
To relate the right-hand-side of the above equation to $T^a_b(u)$, we notice that for a family of operator $\{\mathcal O^i\}_{1\le i\le N}$ transforming in the $\gl_N$ vector representation\footnote{This means that $[\mu^*_{\bC}(E^i_j),\mathcal O^\ell]= \delta^\ell_j\mathcal O^i$.}, we have 
\begin{align*}
    [\mu^*_{\bC}(E^i_j),\mathcal O^j]= N\mathcal O^i.
\end{align*}
Using the definition of $\mu^*_{\bC}$, we have $\mu^*_{\bC}(E^i_j)=(XY)^i_j-(Y X)^i_j-(BA)^i_j+(N+n)\delta^i_j$. Since $\mu^*_{\bC}(E^i_j)=-k\delta^i_j$ when acting on $\bC[X,A]^{\GL_N,-k}$, we then have
\begin{align}\label{eq:moment map vector rep op}
    (XY)^i_j\mathcal O^j+(n+k)\mathcal{O}^i=\left((YX)^i_j+(BA)^i_j\right)\mathcal O^j.
\end{align}
Therefore we have
\begin{align*}
A^a\frac{1}{u-YX-BA }B_b&=\sum_{s\ge 0}u^{-s-1}A^a(Y X+BA)^sB_b\\
\text{\small by \eqref{eq:moment map vector rep op}}\quad &=\sum_{s\ge 0}u^{-s-1}A^a(XY +n+k)^sB_b\\
&=A^a\frac{1}{u-n-k-XY }B_b.
\end{align*}
This finishes the proof.
\end{proof}

\begin{definition}\label{def:mE}
Define $\prescript{}{m}{\mathsf E}$ to be the lower $(n+N-m)\times (n+N-m)$ block of $\mathsf E$:
\begin{equation*}
\prescript{}{m}{\mathsf E}=\qquad
\begin{pNiceArray}{c:c}[first-row, first-col, columns-width=auto]
    & m & n+N-m \\
m   & 0 & 0 \\
\hdottedline
n+N-m & 0 & \star
\end{pNiceArray}
\end{equation*}
In particular $\mathsf E=\prescript{}{0}{\mathsf E}$. We define the Capelli determinant of $\prescript{}{m}{\mathsf E}$ to be
\begin{align}\label{Capelli det}
    C_m(u):=\sum_{\sigma\in S_{n+N-m}}\mathrm{sgn}(\sigma)\prod_{1\le i\le n+N-m}^{\rightarrow}\left((u+i+m-n-N)\delta^{\sigma(i)}_i-\prescript{}{m}{\mathsf E}^{\sigma(i)}_i\right)
\end{align}
\end{definition}

\begin{remark}
If $m\le n$, then $C_m(u)$ commutes with $(XY)^i_j$, $(YX)^i_j$, and $(BA)^i_j$, so $C_m(u)$ commutes with the moment map $\mu^*_{\bC}(\gl_N)$ defined in \eqref{complexified moment map}, therefore $C_m(u)$'s action on $\bC[X,A]$ leaves $\bC[X,A]^{\GL_N,-k}$ invariant. 
\end{remark}

\begin{remark}\label{rmk:Capelli det action on irrep}
Coefficients of $C_m(u)$ commutes with $\prescript{}{m}{\mathsf E}$, i.e. they are central elements in the universal enveloping algebra $U(\gl_{n+N-m})$ generated by $\prescript{}{m}{\mathsf E}$ \cite{molev2003yangians}. Moreover, if $V_\mu$ is an irreducible $\gl_{n+N-m}$ module with highest weight $\mu=(\mu_1\ge \cdots\ge \mu_{n+N-m})\in \bZ^{n+N-m}$, then $C_m(u)$ acts on $V_\mu$ as the scalar \cite{molev2003yangians}:
\begin{align}\label{eq:C_m(u) on irrep}
\prod_{i=1}^{n+N-m}(u-\mu_i+i+m-n-N).
\end{align}
\end{remark}

\begin{proposition}[{\cite[Theorem 2.33]{molev2003yangians}}]\label{prop:molev's identity}
For $1\le m\le n$, we have identity:
\begin{align}\label{eq:molev's identity}
    \mathrm{qdet}\:\left(1+\frac{\prescript{}{0}{\mathsf E}}{u}\right)\cdot \mathrm{qdet}\:\mathfrak T^a_b(-u+n+N-1)_{1\le a,b\le m}=\mathrm{qdet}\:\left(1+\frac{\prescript{}{m}{\mathsf E}}{u}\right).
\end{align}
Here the quantum determinant $\mathrm{qdet}\left(1+\frac{\prescript{}{m}{\mathsf E}}{u}\right)$ is defined by
\begin{align}
    \sum_{\sigma\in S_{n+N-m}}\mathrm{sgn}(\sigma)\prod_{1\le i\le n+N-m}^{\rightarrow}\left(\delta^{\sigma(i)}_i+\frac{\prescript{}{m}{\mathsf E}^{\sigma(i)}_i}{u-i+1}\right).
\end{align}
\end{proposition}

\begin{corollary}\label{cor:A_m in terms of C_m}
When restricted to $\bC[X,A]^{\GL_N,-k}$, we have the following equation between operators:
\begin{align}\label{eq:A_m in terms of C_m}
    A_m(u-n-k)=\frac{C_m(u-m)}{C_0(u)}(u-m+1)_m.
\end{align}
\end{corollary}

\begin{proof}
Plug \eqref{Capelli det} into \eqref{eq:molev's identity}, then use Lemma \ref{lem:commutant}.
\end{proof}

Therefore, to study how $A_m(u)$ transforms under the projection map $p_N$, it suffices to study how $C_m(u)$ transforms under $p_N$. We state the result as follows.

\begin{theorem}\label{thm:Capelli det transforms under p_N}
Fix $r\in \{1,\cdots,n-1\}$ and $m\in \{0,\cdots,n\}$. We denote the Capelli determinant $C_m(u)$ which acts on $\widetilde{\mathcal H}_{nL+r}(n,k)$ by $\prescript{L}{}{C}_m(u)$. Then we have
\begin{align}\label{eq:Capelli det transforms under p_N}
    p_N\circ \prescript{L+1}{}{C}_m(u)=(u+1-n)_{n}\cdot\prescript{L}{}{C}_m(u-n-k)\circ p_N.
\end{align}
\end{theorem}

Assume Theorem \ref{thm:Capelli det transforms under p_N} for now, and we can deduce \eqref{eq:Yangian compatible 3.2} as follows:
\begin{align*}
&p_N\circ \prescript{L+1}{}{A}_m(u)\\
\text{\small by \eqref{eq:A_m in terms of C_m}}\quad &=(u+1+n+k-m)_m\cdot p_N\circ \frac{\prescript{L+1}{}{C}_m(u+n+k-m)}{\prescript{L+1}{}{C}_0(u+n+k)}\\
\text{\small by \eqref{eq:Capelli det transforms under p_N}}\quad &=(u+1+n+k-m)_m\cdot \frac{(u+k+1-m)_n}{(u+k+1)_n}\cdot\frac{\prescript{L}{}{C}_m(u-m)}{\prescript{L}{}{C}_0(u)}\circ p_N\\
&=(u+k+1-m)_m\cdot\frac{\prescript{L}{}{C}_m(u-m)}{\prescript{L}{}{C}_0(u)}\circ p_N\\
\text{\small by \eqref{eq:A_m in terms of C_m}}\quad &=\frac{(u+k+1-m)_m}{(u+1-m)_m}\cdot\prescript{L}{}{A}_m(u-n-k)\circ p_N.
\end{align*}
This proves Theorem \ref{thm:Yangian comptible with p_N and sigma_N}. The remaining of this subsection is devoted to the proof of Theorem \ref{thm:Capelli det transforms under p_N}.

\subsubsection{Decomposition of \texorpdfstring{$\widetilde{\mathcal H}_{N}(n,k)$}{H(n,k)}}

Let us write $X,A$ in the form of a $(N+n)\times N$ matrix:
\begin{equation}\label{phi,X matrix}
\begin{pNiceArray}{c}[first-row, first-col]
    & N  \\
n   & A  \\
\hdottedline
N & X 
\end{pNiceArray}
\end{equation}
The above matrix transforms as vector $\otimes$ dual vector representation with respect to the $\GL_{N+n}\times \GL_N$ action, where the generator of $\gl_{N+n}$ is the matrix $\mathsf E$ defined in \eqref{eq:mat E} and the generator of $\gl_{N}$ is $-A^a_jB^{i}_a-X^l_jY^i_l$, $(1\le i,j\le N)$. Therefore, by the Schur-Weyl duality we have the following decomposition of polynomial ring $\bC[X,A]$ into simple $\GL_{N+n}\times \GL_N$-modules\footnote{Let $V$, $W$ be two vector spaces, then by the Schur-Weyl duality we have $(V\otimes W^*)^{\otimes d}=V^{\otimes d}\otimes W^{*\otimes d}=\bigoplus_{\lambda,\mu}V_\lambda\otimes S_\lambda\otimes W_{\mu}^*\otimes S_\mu$, where $V_\lambda$ (resp. $W_\mu$) is the irreducible $\GL(V)$-module (resp. $\GL(W)$-module) with highest weight given by Young diagram $\lambda$ (resp. $\mu$), and $S_\lambda$, $S_\mu$ are irreducible $S_d$-module given by Young diagrams $\lambda$, $\mu$ respectively. Then $\mathrm{Sym}^d(V\otimes W^*)=\bigoplus_{\lambda,\mu}V_\lambda\otimes W_{\mu}^*\otimes (S_\lambda\otimes S_\mu)^{S_d}$, where $S_d$ acts diagonally, thus $\mathrm{Sym}^d(V\otimes W^*)=\bigoplus_{\lambda}V_\lambda\otimes W_{\lambda}^*$. Applying the above argument to $V=\bC^{N+n}$ and $W=\bC^N$ gives \eqref{Schur-Weyl}.}:
\begin{align}\label{Schur-Weyl}
    \bC[X,A]=\bigoplus_{\lambda=(\lambda_1\ge\cdots\ge\lambda_N)\in \bN^N}V^{N+n}_{\lambda}\otimes (V^{N}_{\lambda})^*,
\end{align}
where $V^{N+n}_{\lambda}$ (resp. $V^{N}_{\lambda}$) is the irreducible $\GL_{N+n}$-module (resp. $\GL_{N}$-module) with highest weight $\lambda$. Consider the ``diagonal'' subgroup $\GL_N^{\mathrm{diag}}\subset\GL_{N+n}\times \GL_N$ which is generated by $X^i_lY^l_j-A^a_jB^{i}_a-X^l_jY^i_l$, $(1\le i,j\le N)$. Then according to the definition of $\widetilde{\mathcal H}_{N}(n,k)$, we have $\widetilde{\mathcal H}_{N}(n,k)=\bC[X,A]^{\GL_N^{\mathrm{diag}},-k}$, therefore we get a decomposition:
\begin{align}\label{decompose H_N(n,k)}
    \widetilde{\mathcal H}_{N}(n,k)=\bigoplus_{\lambda=(\lambda_1\ge\cdots\ge\lambda_N)\in \bN^N} \left(V^{N+n}_{\lambda}\otimes (V^{N}_{\lambda})^*\right)^{\GL_N^{\mathrm{diag}},-k}.
\end{align}
\begin{definition}\label{def:H(la)}
We define the subspace $\widetilde{\mathcal H}_{N}(\lambda):=\left(V^{N+n}_{\lambda}\otimes (V^{N}_{\lambda})^*\right)^{\GL_N^{\mathrm{diag}},-k}$ in \eqref{decompose H_N(n,k)}.
\end{definition}

\begin{definition}\label{def:decrease and increase}
For a tuple of integers $\mu=(\mu_1,\cdots,\mu_N)\in \bZ^N$, we define $\mu^{\shortdownarrow}\in \bZ^N$ such that $\mu^{\shortdownarrow}_i=\mu_i-k$, ($1\le i\le N$). Similarly we define $\mu^{\shortuparrow}\in \bZ^N$ such that $\mu^{\shortuparrow}_i=\mu_i+k$, ($1\le i\le N$).
\end{definition}

\begin{lemma}\label{lem:H_N(lambda) nontrivial}
$\widetilde{\mathcal H}_{N}(\lambda)\neq 0$ if and only if $\lambda^{\shortdownarrow}\in \bN^N$ and $\forall i\in \{1,\cdots,N-n\}$, $\lambda_i\ge\lambda_{i+n}+k$.
\end{lemma}

\begin{proof}
$V^{N+n}_{\lambda}$ decomposes further into irreducible $\GL_n\times \GL_N$ modules:
\begin{align}
    V^{N+n}_{\lambda}=\bigoplus_{\substack{\nu=(\nu_1\ge\cdots\ge \nu_n)\in \bN^n\\ \mu=(\mu_1\ge\cdots\ge \mu_N)\in \bN^N}} M^\lambda_{\nu,\mu}\otimes V^n_{\nu}\otimes V^N_{\mu}.
\end{align}
Here $M^\lambda_{\nu,\mu}$ is the multiplicity space which has dimension given by the Littlewood-Richardson coefficient $c^\lambda_{\nu\mu}$. Therefore we have
\begin{equation}
\begin{split}
\widetilde{\mathcal H}_{N}(\lambda)&=\bigoplus_{\substack{\nu=(\nu_1\ge\cdots\ge \nu_n)\in \bN^n\\ \mu=(\mu_1\ge\cdots\ge \mu_N)\in \bN^N}}M^\lambda_{\nu,\mu}\otimes V^n_{\nu}\otimes\Hom_{\GL_N}(V^N_{\lambda^\shortdownarrow},V^N_\mu)\\
&\cong\begin{cases}
    \bigoplus_{\nu=(\nu_1\ge\cdots\ge \nu_n)\in \bN^n}M^\lambda_{\nu,\lambda^\shortdownarrow}\otimes V^n_{\nu} &, \text{ if }\lambda^\shortdownarrow\in \bN^N,\\
    0 & ,\text{ otherwise.}
\end{cases}
\end{split}
\end{equation}
Assume that $\lambda^\shortdownarrow\in \bN^N$, then we have 
\begin{align}
    \dim \widetilde{\mathcal H}_{N}(\lambda)=\sum_{\nu=(\nu_1\ge\cdots\ge \nu_n)\in \bN^n} c^\lambda_{\nu\lambda^\shortdownarrow}\dim V^n_\nu=s_{\lambda/\lambda^\shortdownarrow}(1,\cdots,1),
\end{align}
where $s_{\lambda/\lambda^\shortdownarrow}(x_1,\cdots,x_n)$ is the $n$-variable skew Schur function of the shape $\lambda/\lambda^\shortdownarrow$. Thus $\dim \widetilde{\mathcal H}_{N}(\lambda)$ equals to the number of semi-standard skew Young tableaux with shape $\lambda/\lambda^\shortdownarrow$ and entries $\in \{1,\cdots,n\}$ \cite{macdonald1998symmetric}. The latter number is nonzero if and only if there is no column in $\lambda/\lambda^\shortdownarrow$ with length greater than $n$, which is equivalent to that $\forall i\in \{1,\cdots,N-n\}$, $\lambda_i\ge\lambda_{i+n}+k$.
\end{proof}

\begin{remark}\label{rmk:simple Y(gl_n) submodules}
We note that $\GL_N^{\mathrm{diag}}$ commutes with operators $\mathfrak{T}^a_b(u)$ defined in \eqref{commutant construction}. Then it follows from Lemma \ref{lem:commutant} that $\widetilde{\mathcal H}_{N}(\lambda)$ is a $Y(\gl_n)$-submodule of $\widetilde{\mathcal H}_{N}(n,k)$, where RTT generators of $Y(\gl_n)$ acts as $T^a_b(u)=\delta^a_b+A^a\frac{1}{u-XY }B_b$. In fact, $\widetilde{\mathcal H}_{N}(\lambda)$ is a special case of a class of irreducible $Y(\gl_n)$-modules $\{V_\omega\:|\: \omega \text{ is a skew Young diagram}\}$ constructed by Nazarov-Tarasov \cite{nazarov1998representations}, see also \cite{nazarov2002representations}. We identify that $\widetilde{\mathcal H}_{N}(\lambda)\cong V_{\lambda/\lambda^\shortdownarrow}(k-N)^*$.
\end{remark}

\begin{definition}\label{def:admissible and cuttable}
Let $\lambda=(\lambda_1\ge \cdots\ge \lambda_N)\in \bN^N$ be a Young diagram of length $N$. We say that $\lambda$ is \textbf{admissible} if and only if it satisfies the conditions in Lemma \ref{lem:H_N(lambda) nontrivial}, that is, $\lambda^{\shortdownarrow}\in \bN^N$ and $\forall i\in \{1,\cdots,N-n\}$, $\lambda_i\ge\lambda_{i+n}+k$. When $N\ge n$, we say that $\lambda$ is \textbf{cuttable} if and only $\lambda$ is admissible and $\lambda^{\shortdownarrow}_{N-n+1}=0$. When $\lambda$ is cuttable, we denote its \textit{cut} $\lambda^\natural:=(\lambda_1\ge\cdots\ge\lambda_{N-n})\in \bN^{N-n}$.
\end{definition}

\begin{remark}
If $\lambda$ is cuttable then $\lambda^{\shortdownarrow}$ can be regarded as an element in $\bN^{N-n}$, and in fact $\lambda^{\shortdownarrow}=\lambda^{\natural\shortdownarrow}$.
\end{remark}

\begin{lemma}\label{lem:1 to 1 corr for cutting}
There is a one to one correspondence:
\begin{align}
    \text{cuttable Young diagrams of length }N+n\:\overset{1:1}{\longleftrightarrow}\:\text{admissible Young diagrams of length }N.
\end{align}
The map from the left-hand-side to the right-hand-side is given by $\lambda\mapsto \lambda^{\natural\shortdownarrow}$.
\end{lemma}

The proof of the above lemma is straightforward and we omit it.

Consider the following block decomposition of the $(A,X)$ matrix of size $(N+2n)\times(N+n)$ defined in \eqref{phi,X matrix}: 
\begin{equation}\label{phi,X matrix decomposition}
\begin{pNiceArray}{c:c}[first-row, first-col]
    & N & n \\
n   & A & \eta \\
\hdottedline
N+n & X_1 & X_2
\end{pNiceArray}
\end{equation}
Here $X=(X_1,X_2)$ is the corresponding block decomposition of the matrix $X$. Then we have
\begin{align*}
    \bC[X_1,X_2,A,\eta]^{\GL_{N+n}^{\mathrm{diag}},-k}\subseteq \bC[X_1,A]^{\GL_{N}^{\mathrm{diag}},-k}\otimes \bC[X_2,\eta]^{\GL_{n}^{\mathrm{diag}},-k},
\end{align*}
where $\GL_{N}^{\mathrm{diag}}\times \GL_{n}^{\mathrm{diag}}\hookrightarrow \GL_{N+n}^{\mathrm{diag}}$ as block diagonal matrices:
\begin{equation*}
\begin{pNiceArray}{c:c}
 \GL_{N} & 0 \\
\hdottedline
 0 & \GL_{n}
\end{pNiceArray}
\subset \GL_{N+n}
\end{equation*}
The polynomial ring $\bC[X_1,X_2,A,\eta]$ is equipped with natural action of $\GL_{N+2n}\times \GL_N\times \GL_n$. Here $\left(\substack{A\\ X_1}\right)$ transforms as vector under $\GL_{N+2n}$, as dual vector under $\GL_N$, and trivially under $\GL_n$, and $\left(\substack{\eta\\ X_2}\right)$ transforms as vector under $\GL_{N+2n}$, as dual vector under $\GL_n$, and trivially under $\GL_N$. Using the Schur-Weyl duality, we can decomposes $\bC[X_1,X_2,A,\eta]$ as
\begin{equation}\label{vertical splitting}
\begin{split}
\bC[X_1,X_2,A,\eta]&=\bigoplus_{\substack{\lambda'=(\lambda'_1\ge\cdots\ge\lambda'_N)\in \bN^N\\ \lambda''=(\lambda''_1\ge\cdots\ge\lambda''_n)\in \bN^n}}V^{N+2n}_{\lambda'}\otimes (V^{N}_{\lambda'})^*\otimes V^{N+2n}_{\lambda''}\otimes (V^{n}_{\lambda''})^*\\
&=\bigoplus_{\substack{\lambda=(\lambda_1\ge\cdots\ge\lambda_{N+n})\in \bN^{N+n}\\ \lambda'=(\lambda'_1\ge\cdots\ge\lambda'_N)\in \bN^N\\ \lambda''=(\lambda''_1\ge\cdots\ge\lambda''_n)\in \bN^n}} M^{\lambda}_{\lambda',\lambda''}\otimes V^{N+2n}_{\lambda}\otimes (V^{N}_{\lambda'})^*\otimes (V^{n}_{\lambda''})^*.
\end{split}
\end{equation}
Here $M^{\lambda}_{\lambda',\lambda''}$ is the multiplicity space.
\begin{definition}
We define the subspace
\begin{align}
\mathcal V_{N+n}(\lambda',\lambda''):=\left(V^{N+2n}_{\lambda'}\otimes (V^{N}_{\lambda'})^*\right)^{\GL_{N}^{\mathrm{diag}},-k}\otimes \left(V^{N+2n}_{\lambda''}\otimes (V^{n}_{\lambda''})^*\right)^{\GL_{n}^{\mathrm{diag}},-k}
\end{align}
\end{definition}
We have inclusion
\begin{align}\label{H(la) into V(la', la'')}
    \widetilde{\mathcal H}_{N+n}(\lambda)\subset \bigoplus_{c^{\lambda}_{\lambda'\lambda''}\neq 0} \mathcal V_{N+n}(\lambda',\lambda''),
\end{align}
where $c^{\lambda}_{\lambda'\lambda''}$ is the Littlewood-Richardson coefficient.

\smallskip Now consider the restriction map $\bC[V(N+n,n)]\to \bC[\widetilde{V}(N,n)]$ defined in Section \ref{subsec:transition map_algebraic}.

\begin{lemma}\label{lem:restrict V(la', la'')}
Under the restriction map $\bC[V(N+n,n)]\to \bC[\widetilde{V}(N,n)]$,
\begin{align}
    \text{image of }\mathcal V_{N+n}(\lambda',\lambda'')\subseteq
    \begin{cases}
        \det(\eta)^k\cdot \widetilde{\mathcal{H}}_N(\lambda') &,\text{ if $\lambda''=(k^n)$ and $\lambda'$ is admissible},\\
        0 &, \text{ otherwise}.
    \end{cases}
\end{align}
\end{lemma}

\begin{proof}
The restriction map sets the last $n$ rows of $X_1$ and the whole $X_2$ to zero, thus the image of $\mathcal V_{N+n}(\lambda',\lambda'')$ is contained in 
\begin{align}\label{image of V(la', la'')}
    \left(V^{N+n}_{\lambda'}\otimes (V^{N}_{\lambda'})^*\right)^{\GL_{N}^{\mathrm{diag}},-k}\otimes V^{n}_{\lambda''}\otimes\left( (V^{n}_{\lambda''})^*\right)^{\GL_{n},-k}.
\end{align}
The $\GL_{N}^{\mathrm{diag}}$ semi-invariant is $\widetilde{\mathcal{H}}_{N}(\lambda')$, and 
\begin{align*}
    \left( (V^{n}_{\lambda''})^*\right)^{\GL_{n},-k}\cong \begin{cases}
        \bC &, \text{ if $\lambda''=(k^n)$},\\
        0 &, \text{ otherwise}.
    \end{cases}
\end{align*}
Thus the vector space \eqref{image of V(la', la'')} is nontrivial if and only if $\lambda'$ is admissible (by Lemma \ref{lem:H_N(lambda) nontrivial}) and $\lambda''=(k^n)$. Moreover, when $\lambda''=(k^n)$, $V^{n}_{\lambda''}\otimes\left( (V^{n}_{\lambda''})^*\right)^{\GL_{n},-k}$ is one-dimensional and it is generated by $\det(\eta)^k$. This finishes the proof.
\end{proof}

\begin{proposition}\label{prop:p_N and decomposition}
Let $\lambda=(\lambda_1\ge \cdots\ge \lambda_{N+n})\in \bN^{N+n}$. Then we have
\begin{align}\label{p_N(H(la))}
p_N(\widetilde{\mathcal H}_{N+n}(\lambda))\subseteq\begin{cases}
    \widetilde{\mathcal H}_{N}(\lambda^{\natural\shortdownarrow}) &,\text{ if $\lambda$ is cuttable},\\
    0 &, \text{ otherwise}.
\end{cases}
\end{align}
\end{proposition}

\begin{proof}
Let us assume that $\lambda$ is admissible, otherwise $\widetilde{\mathcal H}_{N+n}(\lambda)=0$ by Lemma \ref{lem:H_N(lambda) nontrivial} and \eqref{p_N(H(la))} trivially holds. In view of \eqref{H(la) into V(la', la'')} and the definition of $p_N$ (Definition \ref{def:p_N and sigma_N}), Lemma \ref{lem:restrict V(la', la'')} implies that
\begin{align*}
    p_N(\widetilde{\mathcal H}_{N+n}(\lambda))\subseteq \sum_{c^{\lambda}_{\lambda',(k^n)}\neq 0}\widetilde{\mathcal H}_{N}(\lambda'^{\shortdownarrow}),
\end{align*}
where the $\widetilde{\mathcal H}_{N}(\lambda')$ in Lemma \ref{lem:restrict V(la', la'')} becomes $\widetilde{\mathcal H}_{N}(\lambda'^{\shortdownarrow})$ in the above because we need to multiply $\det(X)^{-k}$ to get $p_N$. We claim there exists $\lambda'$ that makes $c^{\lambda}_{\lambda',(k^n)}$ nonzero if and only if $\lambda$ is cuttable and $\lambda'=\lambda^\natural$. If $c^{\lambda}_{\lambda',(k^n)}\neq 0$, then $\lambda'\subseteq \lambda$, in fact $\lambda'\subseteq \lambda^\natural$ because $\lambda'$ has length $N$. On the other hand, $\lambda$ being admissible implies that $\lambda_i\ge k$ for all $N<i\le N+n$. Therefore $|\lambda|-|\lambda'|-|(k^n)|\ge 0$, and the equality holds if and only if $\lambda_i= k$ for all $N<i\le N+n$ and $\lambda'=\lambda^\natural$. So we see that $c^{\lambda}_{\lambda',(k^n)}\neq 0$ implies that $\lambda$ is cuttable and $\lambda'=\lambda^\natural$. Conversely, for a cuttable $\lambda$, $c^{\lambda}_{\lambda^\natural,(k^n)}=1$. This proves our claim, and the proposition follows from the claim.
\end{proof}

\subsubsection{Branching}
For $m\in \{1,\cdots,n\}$, let us consider the block-diagonal subgroup $\GL_m\times \GL_{n+N-m}\subset \GL_{n+N}$:
\begin{equation}\label{branching subgroup}
\begin{pNiceArray}{c:c}
 \GL_m & 0 \\
\hdottedline
 0 & \GL_{n+N-m}
\end{pNiceArray}
\end{equation}
We can replace the $\GL_{n+N}$ by this subgroup in the decomposition \eqref{Schur-Weyl}, and get
\begin{align}\label{Schur-Weyl branching}
    \bC[X,A]=\bigoplus_{\substack{\lambda=(\lambda_1\ge\cdots\ge\lambda_N)\in \bN^N\\ \nu=(\nu_1\ge\cdots\ge\nu_m)\in \bN^m\\ \mu=(\mu_1\ge\cdots\ge\mu_N)\in \bN^N}}M^{\lambda}_{\nu\mu}\otimes V^{m}_{\nu}\otimes V^{n+N-m}_{\mu}\otimes (V^{N}_{\lambda})^*,
\end{align}
where $M^{\lambda}_{\nu\mu}$ is the multiplicity space. Taking 
$\GL_N^{\mathrm{diag}}$ semi-invariant, we get a decomposition:
\begin{align}\label{decompose H_N(n,k) branching}
    \widetilde{\mathcal H}_{N}(n,k)=\bigoplus_{\substack{\lambda=(\lambda_1\ge\cdots\ge\lambda_N)\in \bN^N\\ \nu=(\nu_1\ge\cdots\ge\nu_m)\in \bN^m\\ \mu=(\mu_1\ge\cdots\ge\mu_N)\in \bN^N}}M^{\lambda}_{\nu\mu}\otimes V^{m}_{\nu}\otimes \left(V^{n+N-m}_{\mu}\otimes (V^{N}_{\lambda})^*\right)^{\GL_N^{\mathrm{diag}},-k}.
\end{align}
\begin{definition}\label{def:branching}
We define the subspace $\widetilde{\mathcal H}_{N}\left(\substack{\lambda\\ \mu}\right):=\bigoplus_\nu M^{\lambda}_{\nu\mu}\otimes V^{m}_{\nu}\otimes \left(V^{n+N-m}_{\mu}\otimes (V^{N}_{\lambda})^*\right)^{\GL_N^{\mathrm{diag}},-k}$ in \eqref{decompose H_N(n,k) branching}.
\end{definition}

Compare Definition \ref{def:branching} with Definition \ref{def:H(la)}, we find:
\begin{align*}
    \widetilde{\mathcal H}_{N}(\lambda)=\bigoplus_{\mu}\widetilde{\mathcal H}_{N}\left(\substack{\lambda\\ \mu}\right).
\end{align*}
Let us describe the pair of Young diagrams $(\lambda,\mu)$ such that $\widetilde{\mathcal H}_{N}\left(\substack{\lambda\\ \mu}\right)$ is nonzero.

\begin{definition}\label{def:branching weights}
Let us fix $m,N\in \bZ_{\ge 1}$, and for a pair of integer tuples $\lambda,\mu\in \bZ^N$, we write $\lambda\overset{m}{\searrow}\mu$ if and only if for all $i\in \{1,\cdots,N\}$, $\lambda_i\ge \mu_i\ge \lambda_{i+m}$ (we set $\lambda_i=0$ if $i\notin \{1,\cdots,N\}$).
\end{definition}

\begin{remark}
Let $\lambda=(\lambda_1\ge \cdots\ge \lambda_{N})$ and $\mu=(\mu_1\ge\cdots\ge \mu_{N})$ be Young diagrams of length $N$, then $\lambda\overset{m}{\searrow}\mu$ if and only if $\mu\subseteq\lambda$ and the skew Young diagram $\lambda/\mu$ has no column with length greater than $m$, equivalently the $m$-variable skew Schur function $s_{\lambda/\mu}(x_1,\cdots,x_m)$ is nonzero.
\end{remark}

\begin{remark}\label{rmk:additive}
If $\lambda\overset{p}{\searrow}\mu \overset{q}{\searrow}\nu$, then $\lambda\overset{p+q}{\searrow}\nu$.
\end{remark}

\begin{lemma}\label{lem:H_N(lambda mu) nontrivial}
$\widetilde{\mathcal H}_{N}\left(\substack{\lambda\\ \mu}\right)\neq 0$ if and only if $\lambda$ is admissible and $\lambda\overset{m}{\searrow}\mu\overset{n-m}{\searrow}\lambda^\shortdownarrow$.
\end{lemma}

\begin{proof}
If $\lambda$ is not admissible, then $\widetilde{\mathcal H}_N(\lambda)=0$ by Lemma \ref{lem:H_N(lambda) nontrivial}, whence $\widetilde{\mathcal H}_{N}\left(\substack{\lambda\\ \mu}\right)=0$, so we may assume that $\lambda$ is admissible in the following. We observe that $\widetilde{\mathcal H}_{N}\left(\substack{\lambda\\ \mu}\right)\neq 0$ if and only if $\left(V^{n+N-m}_{\mu}\otimes (V^{N}_{\lambda})^*\right)^{\GL_N^{\mathrm{diag}},-k}\neq 0$ and $\bigoplus_\nu M^{\lambda}_{\nu\mu}\otimes V^{m}_{\nu}\neq 0$. Using the same argument as the proof of Lemma \ref{lem:H_N(lambda) nontrivial}, the former holds if and only if $\lambda^\shortdownarrow\subseteq\mu$ and $s_{\mu/\lambda^\shortdownarrow}(x_1,\cdots,x_{n-m})\neq 0$, which is equivalent to $\mu\overset{n-m}{\searrow}\lambda^\shortdownarrow$. Using the identity 
\begin{align*}
    s_{\lambda/\mu}(x_1,\cdots,x_m)=\sum_{\nu=(\nu_1\ge\cdots\ge\nu_m)\in \bN^m} c^{\lambda}_{\nu\mu}s_\nu(x_1,\cdots,x_m),
\end{align*}
we see that $\bigoplus_\nu M^{\lambda}_{\nu\mu}\otimes V^{m}_{\nu}\neq 0$ if and only if $s_{\lambda/\mu}(x_1,\cdots,x_m)\neq 0$, which is equivalent to $\lambda\overset{m}{\searrow}\mu$. This finishes the proof.
\end{proof}

We collect some combinatorial facts in Lemma \ref{lem:tail of mu} and \ref{lem:1 to 1 corr for cutting and branching}.

\begin{lemma}\label{lem:tail of mu}
Let $\lambda=(\lambda_1\ge \cdots\ge \lambda_{N})$ and $\mu=(\mu_1\ge\cdots\ge \mu_{N})$ be Young diagrams of length $N$, and assume that $\lambda$ is cuttable and $\lambda\overset{m}{\searrow}\mu\overset{n-m}{\searrow}\lambda^\shortdownarrow$, then
\begin{equation}
\begin{split}
    &\mu_i=k, \;\forall N-n<i\le N-m,\text{ and}\\
    &\mu_i=0,\; \forall N-m<i.
\end{split}
\end{equation}
\end{lemma}

\begin{proof}
Since $\lambda$ is cuttable and $\lambda\overset{m}{\searrow}\mu$, then for all $N-n<i\le N-m$, $\lambda_i=k\ge \mu_i\ge k=\lambda_{i+m}$, this forces $\mu_i=k$. Since $\lambda$ is cuttable and $\mu\overset{n-m}{\searrow}\lambda^\shortdownarrow$, then  for all $i> N-m$, $\lambda^{\shortdownarrow}_{i-n+m}=0\ge \mu_i$, this forces $\mu_i=0$. 
\end{proof}

\begin{lemma}\label{lem:1 to 1 corr for cutting and branching}
Let us fix a cuttable $\lambda\in \bN^{N+n}$, then there is a one to one correspondence:
\begin{align}
    \{\mu\in \bN^{N+n}\:|\: \lambda\overset{m}{\searrow}\mu\overset{n-m}{\searrow}\lambda^\shortdownarrow\}\:\overset{1:1}{\longleftrightarrow}\:\{\nu\in \bN^{N}\:|\: \lambda^{\natural\shortdownarrow}\overset{m}{\searrow}\nu\overset{n-m}{\searrow}\lambda^{\natural\shortdownarrow\shortdownarrow}\}.
\end{align}
The map from the left-hand-side to the right-hand-side is given by $\mu\mapsto \mu^{\natural\shortdownarrow}$.
\end{lemma}

\begin{proof}
Let us give the inverse map: starting from a $\nu\in \bN^{N} $ such that $ \lambda^{\natural\shortdownarrow}\overset{m}{\searrow}\nu\overset{n-m}{\searrow}\lambda^{\natural\shortdownarrow\shortdownarrow}$, we append the rectangular Young diagram $(k^{n-m})$ to the right of $\nu^\shortuparrow$. This is the inverse map of $\mu\mapsto \mu^{\natural\shortdownarrow}$ due to Lemma \ref{lem:tail of mu}.
\end{proof}

In view of the decomposition \eqref{vertical splitting}, we can replace the $\GL_{N+2n}$ action by the subgroup $\GL_m\times\GL_{N+2n-m}$ action, where the latter is the given in the block matrix \eqref{branching subgroup} with $N$ replaced by $N+n$. The associated decomposition reads:
\begin{multline}
\bC[X_1,X_2,A,\eta]=\left(\bigoplus_{\lambda',\mu',\nu'}M^{\lambda'}_{\nu'\mu'}\otimes V^m_{\nu'}\otimes V^{N+2n-m}_{\mu'}\otimes (V^{N}_{\lambda'})^*\right)\\
\otimes \left(\bigoplus_{\lambda'',\mu'',\nu''}M^{\lambda''}_{\nu''\mu''}\otimes V^m_{\nu''}\otimes V^{N+2n-m}_{\mu''}\otimes (V^{n}_{\lambda''})^*\right).
\end{multline}

\begin{definition}
We define the subspace
\begin{multline}
\mathcal V_{N+n}\left(\substack{\lambda',\lambda''\\\mu', \mu''}\right):=\left(\bigoplus_{\nu'}M^{\lambda'}_{\nu'\mu'}\otimes V^m_{\nu'}\otimes \left(V^{N+2n-m}_{\mu'}\otimes (V^{N}_{\lambda'})^*\right)^{\GL_{N}^{\mathrm{diag}},-k}\right)\\
\otimes \left(\bigoplus_{\nu''}M^{\lambda''}_{\nu''\mu''}\otimes V^m_{\nu''}\otimes \left(V^{N+2n-m}_{\mu''}\otimes (V^{n}_{\lambda''})^*\right)^{\GL_{n}^{\mathrm{diag}},-k}\right).
\end{multline}
\end{definition}
We have inclusion
\begin{align}\label{H(la mu) into V(la', la'', mu', mu'')}
    \widetilde{\mathcal H}_{N+n}\left(\substack{\lambda\\ \mu}\right)\subset \bigoplus_{\substack{c^{\lambda}_{\lambda'\lambda''}\neq 0\\ c^{\mu}_{\mu'\mu''}\neq 0}} \mathcal V_{N+n}\left(\substack{\lambda',\lambda''\\\mu', \mu''}\right).
\end{align}

\begin{lemma}\label{lem:restrict V(la', la'', mu', mu'')}
Under the restriction map $\bC[V(N+n,n)]\to \bC[\widetilde{V}(N,n)]$, the image of $\mathcal V_{N+n}\left(\substack{\lambda',\lambda''\\\mu', \mu''}\right)$ is contained in
\begin{align}
    \begin{cases}
        \det(\eta)^k\cdot \widetilde{\mathcal{H}}_N\left(\substack{\lambda'\\ \mu'}\right) &,\text{ if $\lambda''=(k^n)$ and $\mu''=(k^{n-m})$ and $\lambda'$ is admissible and }\lambda'\overset{m}{\searrow}\mu'\overset{n-m}{\searrow}\lambda'^\shortdownarrow,\\
        0 &, \text{ otherwise}.
    \end{cases}
\end{align}
\end{lemma}

\begin{proof}
The restriction map sets the last $n$ rows of $X_1$ and the whole $X_2$ to zero, thus the image of $\mathcal V_{N+n}\left(\substack{\lambda',\lambda''\\\mu', \mu''}\right)$ is contained in 
\begin{multline}\label{image of V(la', la'', mu', mu'')}
\left(\bigoplus_{\nu'}M^{\lambda'}_{\nu'\mu'}\otimes V^m_{\nu'}\otimes \left(V^{N+n-m}_{\mu'}\otimes (V^{N}_{\lambda'})^*\right)^{\GL_{N}^{\mathrm{diag}},-k}\right)\\
\otimes \left(\bigoplus_{\nu''}M^{\lambda''}_{\nu''\mu''}\otimes V^m_{\nu''}\otimes V^{n-m}_{\mu''}\otimes \left((V^{n}_{\lambda''})^*\right)^{\GL_{n},-k}\right).
\end{multline}
The tensor component involving $\GL_{N}^{\mathrm{diag}}$ semi-invariant is $\widetilde{\mathcal{H}}_{N}\left(\substack{\lambda'\\ \mu'}\right)$, and we have
\begin{align*}
    \left( (V^{n}_{\lambda''})^*\right)^{\GL_{n},-k}\cong \begin{cases}
        \bC &, \text{ if $\lambda''=(k^n)$},\\
        0 &, \text{ otherwise}.
    \end{cases}
\end{align*}
Thus the vector space \eqref{image of V(la', la'', mu', mu'')} is nontrivial only if $\lambda'$ is admissible and $\lambda'\overset{m}{\searrow}\mu'\overset{n-m}{\searrow}\lambda'^\shortdownarrow$ (by Lemma \ref{lem:H_N(lambda mu) nontrivial}) and $\lambda''=(k^n)$. Let us assume that the above conditions are satisfied, then the vector space $M^{\lambda''}_{\nu''\mu''}\otimes V^m_{\nu''}\otimes V^{n-m}_{\mu''}$ is the $(\nu',\mu')$ branch of $V^n_{(k^n)}$ ($k$-th power of determinant representation) into $\GL_m\times\GL_{n-m}$ representations, which is nontrivial if an only if $\nu'=(k^m)$ and $\mu'=(k^{n-m})$. When $\mu'=(k^{n-m})$, $\bigoplus_{\nu'}M^{\lambda''}_{\nu''\mu''}\otimes V^m_{\nu''}\otimes V^{n-m}_{\mu''}$ is one-dimensional and it is generated by $\det(\eta)^k$. This finishes the proof.
\end{proof}

\begin{proposition}\label{prop:p_N and branching}
Let $\lambda=(\lambda_1\ge \cdots\ge \lambda_{N+n})$ and $\mu=(\mu_1\ge\cdots\ge \mu_{N+n})$ be Young diagrams of length $N+n$, and assume that $\lambda$ is admissible and $\lambda\overset{m}{\searrow}\mu\overset{n-m}{\searrow}\lambda^\shortdownarrow$, then
\begin{align}
p_N(\widetilde{\mathcal H}_{N+n}\left(\substack{\lambda\\ \mu}\right))\subseteq \begin{cases}
    \widetilde{\mathcal H}_{N}\left(\substack{\lambda^{\natural\shortdownarrow}\\ \mu^{\natural\shortdownarrow}}\right) &,\text{ if $\lambda$ is cuttable},\\
    0 &, \text{ otherwise}.
\end{cases}
\end{align}
\end{proposition}

\begin{proof}
Let us assume that $\lambda$ is cuttable, otherwise $p_N(\widetilde{\mathcal H}_{N+n}\left(\substack{\lambda\\ \mu}\right))\subseteq p_N(\widetilde{\mathcal H}_{N+n}\left(\lambda\right))=0$ by Proposition \ref{prop:p_N and decomposition}. In view of \eqref{H(la mu) into V(la', la'', mu', mu'')} and the definition of $p_N$ (Definition \ref{def:p_N and sigma_N}), Lemma \ref{lem:restrict V(la', la'', mu', mu'')} implies that
\begin{align*}
    p_N(\widetilde{\mathcal H}_{N+n}\left(\substack{\lambda\\ \mu}\right))\subseteq \sum_{\substack{c^{\lambda}_{\lambda',(k^n)}\neq 0\\ c^{\mu}_{\mu',(k^{n-m})}\neq 0}} \widetilde{\mathcal H}_{N}\left(\substack{\lambda'^\shortdownarrow\\ \mu'^\shortdownarrow}\right),
\end{align*}
where the $\widetilde{\mathcal H}_{N}\left(\substack{\lambda'\\ \mu'}\right)$ in Lemma \ref{lem:restrict V(la', la'', mu', mu'')} becomes $\widetilde{\mathcal H}_{N}\left(\substack{\lambda'^\shortdownarrow\\ \mu'^\shortdownarrow}\right)$ in the above because we need to multiply $\det(X)^{-k}$ to get $p_N$. According to the proof of Proposition \ref{prop:p_N and decomposition}, the only $\lambda'$ that makes $c^{\lambda}_{\lambda',(k^n)}$ nonzero is $\lambda'=\lambda^\natural$.

If $c^{\mu}_{\mu',(k^{n-m})}\neq 0$, then $\mu'\subseteq \mu$, in fact $\mu'\subseteq \mu^\natural$ because $\mu'$ has length $N$. Therefore we have
\begin{align*}
|\mu|-|\mu'|-|(k^{n-m})|&\ge \sum_{i=N+1}^{N+n-m}\mu_i-k(n-m)\\
&=\sum_{i=N+1}^{N+n-m}(\mu_i-k)\\
\text{\small by Lemma \ref{lem:tail of mu}}\quad &=0.
\end{align*}
The equality holds if and only if $\mu'=\mu^\natural$. So we conclude that the unique $\mu'$ that makes $c^{\mu}_{\mu',(k^{n-m})}$ nonzero is $\mu'=\mu^\natural$. This finishes the proof.
\end{proof}

\subsubsection{Proof of Theorem \ref{thm:Capelli det transforms under p_N}}

According to Remark \ref{rmk:Capelli det action on irrep}, $\prescript{L+1}{}{C}_m(u)$ acts on $\widetilde{\mathcal H}_{N+n}(\substack{\lambda\\ \mu})$ as the scalar
\begin{align}\label{C(u) before p_N, original}
\prod_{i=1}^{2n+N-m}(u-\mu_i+i+m-2n-N).
\end{align}
By Lemma \ref{lem:tail of mu}, $\mu_i=0$ for $i>n+N-m$, so \eqref{C(u) before p_N, original} equals to
\begin{align}\label{C(u) before p_N}
    (u-n+1)_n\cdot\prod_{i=1}^{n+N-m}(u-\mu_i+i+m-2n-N).
\end{align}
According to Remark \ref{rmk:Capelli det action on irrep}, $\prescript{L}{}{C}_m(u)$ acts on $\widetilde{\mathcal H}_{N}\left(\substack{\lambda^{\natural\shortdownarrow}\\ \mu^{\natural\shortdownarrow}}\right)$ as the scalar
\begin{align}\label{C(u) after p_N}
\prod_{i=1}^{n+N-m}(u-\mu^{\natural\shortdownarrow}_i+i+m-n-N).
\end{align}
Compare \eqref{C(u) before p_N} and \eqref{C(u) after p_N}, and we compute that 
\begin{align*}
p_N\circ \prescript{L+1}{}{C}_m(u)\bigg\rvert_{\widetilde{\mathcal H}_{N+n}\left(\substack{\lambda\\ \mu}\right)}&=(u-n+1)_n\cdot\prod_{i=1}^{n+N-m}(u-\mu_i+i+m-2n-N)\cdot p_N\bigg\rvert_{\widetilde{\mathcal H}_{N+n}\left(\substack{\lambda\\ \mu}\right)}\\
&=(u-n+1)_n\cdot\prod_{i=1}^{n+N-m}((u-n-k)-\mu^{\natural\shortdownarrow}_i+i+m-n-N)\cdot p_N\bigg\rvert_{\widetilde{\mathcal H}_{N+n}\left(\substack{\lambda\\ \mu}\right)}\\
\text{\small by Proposition \ref{prop:p_N and branching}}\quad &= (u-n+1)_n\cdot\prescript{L}{}{C}_m(u-n-k)\circ p_N\bigg\rvert_{\widetilde{\mathcal H}_{N+n}\left(\substack{\lambda\\ \mu}\right)}.
\end{align*}
In other words, the equation \eqref{eq:Capelli det transforms under p_N} holds when restricted to the subspace $\widetilde{\mathcal H}_{N}(\substack{\lambda\\ \mu})$ for a cuttable $\lambda$. On the other hand, if $\lambda$ is not cuttable, then according to Proposition \ref{prop:p_N and branching} $p_N$ maps $\widetilde{\mathcal H}_{N}(\substack{\lambda\\ \mu})$ to zero, thus the equation \eqref{eq:Capelli det transforms under p_N} trivially holds on the subspace $\widetilde{\mathcal H}_{N}(\substack{\lambda\\ \mu})$. Since $\widetilde{\mathcal H}_{N}(n,k)=\bigoplus_{\lambda,\mu}\widetilde{\mathcal H}_{N}(\substack{\lambda\\ \mu})$, the equation \eqref{eq:Capelli det transforms under p_N} holds on the whole space $\widetilde{\mathcal H}_{N}(n,k)$. This concludes the proof of Theorem \ref{thm:Capelli det transforms under p_N}.

\begin{definition}\label{def:Yangian notation_infinite}
Fix $r\in \{1,\cdots,n-1\}$. Define the inverse limit of the system $\{\prescript{L}{}{\widetilde{T}}^a_{b;m}\}_{L\in \bN}$ which acts on $\widetilde{\mathcal H}^{(r)}_{\infty}(n,k)$ by $\prescript{\infty}{}{\widetilde{T}}^a_{b;m}$.
\end{definition}

\subsection{Conformal limit of operators with bounded degree}\label{subsec:conformal limit of operators}

\begin{definition}
Let $V=\bigoplus_{d\ge 0}V_d$ be a non-negatively graded vector space. A linear operator $\mathcal O\in \End(V)$ is said to be of bounded degree if there exists $C\in\bZ$ such that $\forall d\in \bN$, $\mathcal O(V_d)\subset V_{\le d+C}$, where 
\begin{align*}
     V_{\le m}:=\bigoplus_{i\le m} V_i.
\end{align*}
\end{definition}

It is easy to see that finite linear combinations and compositions of bounded degree linear operators are again bounded degree linear operators.

\begin{definition}\label{def:conformal limit of operators}
Fix $r\in \{0,\cdots,n-1\}$. Let $\{\prescript{L}{}{\mathcal O}\in \End(\widetilde{\mathcal H}_{nL+r}(n,k))\}_{L\in \bN}$ be a collection of linear operators of uniformly bounded degree, i.e. there exists $C\in \bZ$ such that $\forall d\in \bN$ and $\forall L\in \bN$, 
\begin{align*}
    \prescript{L}{}{\mathcal O}(\widetilde{\mathcal H}_{nL+r}(n,k)_{d})\subset \bigoplus_{i\le d+C}\widetilde{\mathcal H}_{nL+r}(n,k)_{i}.
\end{align*}
We say that $\lim_{L\to\infty}\prescript{L}{}{\mathcal O}$ exists if and only if there exists linear operator $\prescript{\infty}{}{\mathcal O}\in \End(\widetilde{\mathcal H}^{(r)}_{\infty}(n,k))$ with bounded degree such that $\forall v\in \widetilde{\mathcal H}^{(r)}_{\infty}(n,k)$ the following holds:
\begin{align}
    \underset{L\to \infty}{\lim}\lVert \prescript{\infty}{}{\mathcal O}(v)-\sigma^\infty_{nL+r}\circ \prescript{L}{}{\mathcal O}\circ p^\infty_{nL+r}(v)\rVert=0.
\end{align}
Here $\lVert\cdot\rVert$ is a norm on $\widetilde{\mathcal H}^{(r)}_{\infty}(n,k)$. If the above holds, we say that the conformal limit of $\{\prescript{L}{}{\mathcal O}\}_{L\in \bN}$ is $\prescript{\infty}{}{\mathcal O}$, and write $\underset{L\to \infty}{\lim}\prescript{L}{}{\mathcal O}=\prescript{\infty}{}{\mathcal O}$.
\end{definition}

\begin{remark}\label{rmk:independence of norm}
The above definition does not depend on the choice of norm $\lVert\cdot\rVert$. Suppose that $\lVert\cdot\rVert'$ is another norm on $\widetilde{\mathcal H}^{(r)}_{\infty}(n,k)$, then for every $d\in \bN$, there exists $C_d>1$ such that $\forall v\in \widetilde{\mathcal H}^{(r)}_{\infty}(n,k)_{\le d}$, the following inequality holds:
\begin{align*}
    C_d^{-1}\lVert v\rVert < \lVert v\rVert' < C_d\lVert v\rVert.
\end{align*}
This is because $\widetilde{\mathcal H}^{(r)}_{\infty}(n,k)_{\le d}$ is a finite dimensional vector space (Lemma \ref{lem:finiteness}), and all norms on a finite dimensional vector space are equivalent. According to the assumption on the boundedness of degrees of $\{\prescript{L}{}{\mathcal O}\}_{L\in \bN}$ and $\prescript{\infty}{}{\mathcal O}$, we have $\forall v\in \widetilde{\mathcal H}^{(r)}_{\infty}(n,k)$, there exists $d\in \bZ$ such that $\forall L\in \bN$, $\prescript{\infty}{}{\mathcal O}(v)-\sigma^\infty_{nL+r}\circ \prescript{L}{}{\mathcal O}\circ p^\infty_{nL+r}(v)\in \widetilde{\mathcal H}^{(r)}_{\infty}(n,k)_{\le d}$, therefore
\begin{align*}
    \underset{L\to \infty}{\lim}\lVert \prescript{\infty}{}{\mathcal O}(v)-\sigma^\infty_{nL+r}\circ \prescript{L}{}{\mathcal O}\circ p^\infty_{nL+r}(v)\rVert=0\Longleftrightarrow \underset{L\to \infty}{\lim}\lVert \prescript{\infty}{}{\mathcal O}(v)-\sigma^\infty_{nL+r}\circ \prescript{L}{}{\mathcal O}\circ p^\infty_{nL+r}(v)\rVert'=0.
\end{align*}
\end{remark}

There is an equivalent definition of conformal limit, see the next lemma.
\begin{lemma}\label{lem:equivalent def of convergence}
Let $\{\prescript{L}{}{\mathcal O}\in \End(\widetilde{\mathcal H}_{nL+r}(n,k))\}_{L\in \bN}$ be a collection of linear operators of uniformly bounded degree, and let $\prescript{\infty}{}{\mathcal O}\in \End(\widetilde{\mathcal H}^{(r)}_{\infty}(n,k))$ be a linear operator with bounded degree, then
$\underset{L\to \infty}{\lim}\prescript{L}{}{\mathcal O}=\prescript{\infty}{}{\mathcal O}$ if and only if $\forall d\in \bN$ and $\forall\epsilon\in \bR_{>0}$, there exists $M\in \bN$ such that 
\begin{align}\label{eq:equivalent form of conformal limit}
    \left(\prescript{\infty}{}{\mathcal O}-\sigma^\infty_{nL+r}\circ \prescript{L}{}{\mathcal O}\circ p^\infty_{nL+r}\right)(\mathbb B_d)\subseteq \epsilon\cdot\mathbb B_{d+C}\text{ holds for all }L\ge M,
\end{align}
where $\mathbb B_i$ is the unit ball in $\widetilde{\mathcal H}^{(r)}_{\infty}(n,k)_{\le i}$ with respect to the norm $\lVert\cdot\rVert$, i.e. $$\mathbb B_i=\{v\in \widetilde{\mathcal H}^{(r)}_{\infty}(n,k)_{\le i}:\lVert v\rVert\le 1 \}.$$
\end{lemma}

\begin{proof}
The ``if" direction is obvious, let us show the ``only if" direction. Define $\prescript{L}{}{\mathcal A}:=\prescript{\infty}{}{\mathcal O}-\sigma^\infty_{nL+r}\circ \prescript{L}{}{\mathcal O}\in \End(\widetilde{\mathcal H}^{(r)}_{\infty}(n,k))$. By our assumption, $\prescript{L}{}{\mathcal A}$ has uniformly bounded degree, i.e. $\exists C\in \bZ$ such that $\prescript{L}{}{\mathcal A}(\widetilde{\mathcal H}^{(r)}_{\infty}(n,k)_{\le d})\subseteq \widetilde{\mathcal H}^{(r)}_{\infty}(n,k)_{d+C}$ holds for all $d\in \bN$. Moreover, $\underset{L\to \infty}{\lim}\lVert \prescript{L}{}{\mathcal A}(v)\rVert=0$ for all $v\in \widetilde{\mathcal H}^{(r)}_{\infty}(n,k)$. Now fix an arbitrary $d\in \bN$, then $\prescript{L}{}{\mathcal A}|_{\widetilde{\mathcal H}^{(r)}_{\infty}(n,k)_{\le d}}:\widetilde{\mathcal H}^{(r)}_{\infty}(n,k)_{\le d}\to \widetilde{\mathcal H}^{(r)}_{\infty}(n,k)_{d+C}$ is represented by a matrix of finite size, then point-wise convergence $\prescript{L}{}{\mathcal A}|_{\widetilde{\mathcal H}^{(r)}_{\infty}(n,k)_{\le d}}\to 0$ implies that every matrix component converges zero, thus the matrix norm $\rVert\prescript{L}{}{\mathcal A}|_{\widetilde{\mathcal H}^{(r)}_{\infty}(n,k)_{\le d}}\rVert$ defined by 
\begin{align*}
    \rVert\prescript{L}{}{\mathcal A}|_{\widetilde{\mathcal H}^{(r)}_{\infty}(n,k)_{\le d}}\rVert:= \sup_{\substack{v\in \widetilde{\mathcal H}^{(r)}_{\infty}(n,k)_{\le d}\\ \rVert v\rVert=1}}\rVert\prescript{L}{}{\mathcal A}(v)\rVert,
\end{align*}
converges to zero. This proves the \eqref{eq:equivalent form of conformal limit}.
\end{proof}

Here we list some elementary properties of conformal limit.

\begin{lemma}
Let $\{\prescript{L}{}{\mathcal O}\}_{L\in \bN}$ be a collection of linear operators of uniformly bounded degree, then its conformal limit is unique if it exists.
\end{lemma}

\begin{proof}
Suppose that there are two linear operator $\prescript{\infty}{}{\mathcal O}_{1},\prescript{\infty}{}{\mathcal O}_{2}\in \End(\widetilde{\mathcal H}^{(r)}_{\infty}(n,k))$ of bounded degrees, such that $\underset{L\to \infty}{\lim}\prescript{L}{}{\mathcal O}=\prescript{\infty}{}{\mathcal O}_1$ and $\underset{L\to \infty}{\lim}\prescript{L}{}{\mathcal O}=\prescript{\infty}{}{\mathcal O}_2$ simultaneously holds. For arbitrary $v\in \widetilde{\mathcal H}^{(r)}_{\infty}(n,k)$, we have
\begin{align*}
    &\lVert \prescript{\infty}{}{\mathcal O}_1(v)-\prescript{\infty}{}{\mathcal O}_2(v)\rVert\\
    &=\underset{L\to \infty}{\lim}\lVert \left(\prescript{\infty}{}{\mathcal O}_1(v)-\sigma^\infty_{nL+r}\circ \prescript{L}{}{\mathcal O}\circ p^\infty_{nL+r}(v)\right)-\left(\prescript{\infty}{}{\mathcal O}_2(v)-\sigma^\infty_{nL+r}\circ \prescript{L}{}{\mathcal O}\circ p^\infty_{nL+r}(v)\right)\rVert\\
    &\le \underset{L\to \infty}{\lim}\lVert \prescript{\infty}{}{\mathcal O}_1(v)-\sigma^\infty_{nL+r}\circ \prescript{L}{}{\mathcal O}\circ p^\infty_{nL+r}(v)\rVert+\underset{L\to \infty}{\lim}\lVert\prescript{\infty}{}{\mathcal O}_2(v)-\sigma^\infty_{nL+r}\circ \prescript{L}{}{\mathcal O}\circ p^\infty_{nL+r}(v)\rVert\\
    &=0.
\end{align*}
Thus $\prescript{\infty}{}{\mathcal O}_1(v)=\prescript{\infty}{}{\mathcal O}_2(v)$. Since $v$ is arbitrary, we conclude that $\prescript{\infty}{}{\mathcal O}_1=\prescript{\infty}{}{\mathcal O}_2$.
\end{proof}

\begin{lemma}
Let $\{\prescript{L}{}{\mathcal O}\}_{L\in \bN}$ be a collection of linear operators of homogeneous degree $d$, and assume that $\underset{L\to \infty}{\lim}\prescript{L}{}{\mathcal O}=\prescript{\infty}{}{\mathcal O}$, then $\prescript{\infty}{}{\mathcal O}$ is homogeneous of degree $d$.
\end{lemma}

\begin{proof}
For arbitrary $f\in \bN$, let us take arbitrary $v\in \widetilde{\mathcal H}^{(r)}_{\infty}(n,k)_{f}$, then by our assumption $\sigma^\infty_{nL+r}\circ \prescript{L}{}{\mathcal O}\circ p^\infty_{nL+r}(v)\in \widetilde{\mathcal H}^{(r)}_{\infty}(n,k)_{f+d}$ for all $L\in \bN$. Since a finite dimensional subspace in a normed $\bC$-vector space is closed, the limit
\begin{align*}
    \underset{L\to \infty}{\lim}\lVert \prescript{\infty}{}{\mathcal O}(v)-\sigma^\infty_{nL+r}\circ \prescript{L}{}{\mathcal O}\circ p^\infty_{nL+r}(v)\rVert=0
\end{align*}
implies that $ \prescript{\infty}{}{\mathcal O}(v)\in \widetilde{\mathcal H}^{(r)}_{\infty}(n,k)_{f+d}$. This finishes the proof.
\end{proof}

\begin{lemma}\label{lem:scalar limit}
Suppose that $\{\alpha_L\in \bC\}_{L\in \bN}$ is a sequence of complex numbers such that $\underset{L\to \infty}{\lim}\alpha_L=\alpha$, then
\begin{align}
    \underset{L\to \infty}{\lim}{\alpha_L\cdot \mathrm{Id}}_{\widetilde{\mathcal H}_{nL+r}(n,k)}=\alpha\cdot\mathrm{Id}_{\widetilde{\mathcal H}^{(r)}_{\infty}(n,k)}.
\end{align}
\end{lemma}

\begin{proof}
Take arbitrary $d\in \bN$, then $\sigma^\infty_{nL+r}\circ p^\infty_{nL+r}|_{\widetilde{\mathcal H}^{(r)}_{\infty}(n,k)_{\le d}}=\mathrm{Id}_{\widetilde{\mathcal H}^{(r)}_{\infty}(n,k)_{\le d}}$ for $L\gg 0$ by Proposition \ref{prop:p_N stabilizes}. For arbitrary $v\in \widetilde{\mathcal H}^{(r)}_{\infty}(n,k)_{\le d}$, we have
\begin{align*}
    \underset{L\to \infty}{\lim}\lVert \alpha v-\alpha_L\sigma^\infty_{nL+r}\circ p^\infty_{nL+r}(v)\rVert=\underset{L\to \infty}{\lim}\lVert \alpha v-\alpha_L v\rVert=0.
\end{align*}
\end{proof}

\begin{lemma}\label{lem:linearity and multiplication}
Suppose that $\{\prescript{L}{}{\mathcal A}\}_{L\in \bN}$ and $\{\prescript{L}{}{\mathcal B}\}_{L\in \bN}$ are two collections of  uniformly bounded degree linear operators, and assume that $\underset{L\to \infty}{\lim}\prescript{L}{}{\mathcal A}=\prescript{\infty}{}{\mathcal A}$ and $\underset{L\to \infty}{\lim}\prescript{L}{}{\mathcal B}=\prescript{\infty}{}{\mathcal B}$, where $\prescript{\infty}{}{\mathcal A}$ and $\prescript{\infty}{}{\mathcal B}$ are bounded degree linear operators. Then $\forall\alpha,\beta\in \bC$,
\begin{align}
    \underset{L\to \infty}{\lim}\left(\alpha\cdot\prescript{L}{}{\mathcal A}+\beta\cdot \prescript{L}{}{\mathcal B}\right)=\alpha\cdot\prescript{\infty}{}{\mathcal A}+\beta\cdot\prescript{\infty}{}{\mathcal B},\quad \underset{L\to \infty}{\lim}\left(\prescript{L}{}{\mathcal A}\cdot \prescript{L}{}{\mathcal B}\right)=\prescript{\infty}{}{\mathcal A}\cdot\prescript{\infty}{}{\mathcal B},
\end{align}
\end{lemma}

\begin{proof}
The first statement (linearity of limit) is straightforward, and we focus on the second statement. Let us fix $D\gg 0$ so that it bounds the degrees of $\{\prescript{L}{}{\mathcal A}\}_{L\in \bN}$, $\{\prescript{L}{}{\mathcal B}\}_{L\in \bN}$, $\prescript{\infty}{}{\mathcal A}$, and $\prescript{\infty}{}{\mathcal B}$. Let $d\in \bN$ and $\epsilon\in \bR_{>0}$ be arbitrary. Let us take $R\in \bR_{>0}$ such that 
\begin{align*}
    (\star)\quad \prescript{\infty}{}{\mathcal B}(\mathbb B_d)\subseteq R\cdot \mathbb B_{d+D},\text{ and }\prescript{\infty}{}{\mathcal A}(\mathbb B_{d+D})\subseteq R\cdot \mathbb B_{d+2D}
\end{align*}
simultaneously hold. Then we take $M\in \bN$ such that $\forall L\ge M$ the following simultaneously hold:
\begin{align*}
   \text{(1)}\quad & p^{\infty}_{nL+r}|_{\widetilde{\mathcal H}^{(r)}_{\infty}(n,k)_{\le d+D}}: \widetilde{\mathcal H}^{(r)}_{\infty}(n,k)_{\le d+D}\to \widetilde{\mathcal H}_{nL+r}(n,k)_{\le d+D}\; \text{ is isomorphism, and}\\
    \text{(2)}\quad &\left(\prescript{\infty}{}{\mathcal B}-\sigma^\infty_{nL+r}\circ \prescript{L}{}{\mathcal B}\circ p^\infty_{nL+r}\right)(\mathbb B_d)\subseteq \frac{\epsilon}{2R}\cdot\mathbb B_{d+D}, \text{ and } \\
    \text{(3)}\quad &\left(\prescript{\infty}{}{\mathcal A}-\sigma^\infty_{nL+r}\circ \prescript{L}{}{\mathcal A}\circ p^\infty_{nL+r}\right)(\mathbb B_{d+D})\subseteq \frac{\epsilon R}{2R^2+\epsilon}\cdot\mathbb B_{d+2D}.
\end{align*}
$M$ exists due to Proposition \ref{prop:p_N stabilizes}, and our assumption that $\underset{L\to \infty}{\lim}\prescript{L}{}{\mathcal A}=\prescript{\infty}{}{\mathcal A}$ and $\underset{L\to \infty}{\lim}\prescript{L}{}{\mathcal B}=\prescript{\infty}{}{\mathcal B}$, and Lemma \ref{lem:equivalent def of convergence}. We then have
\begin{align*}
&\left(\prescript{\infty}{}{\mathcal A}\cdot\prescript{\infty}{}{\mathcal B}-\sigma^\infty_{nL+r}\circ \prescript{L}{}{\mathcal A}\cdot \prescript{L}{}{\mathcal B}\circ p^\infty_{nL+r}\right)(\mathbb B_d)\\
\text{\small by (1)}\; &=\left(\prescript{\infty}{}{\mathcal A}\cdot\prescript{\infty}{}{\mathcal B}-\sigma^\infty_{nL+r}\circ \prescript{L}{}{\mathcal A}\circ p^\infty_{nL+r}\circ \sigma^\infty_{nL+r} \prescript{L}{}{\mathcal B}\circ p^\infty_{nL+r}\right)(\mathbb B_d)\\
&=\left(\prescript{\infty}{}{\mathcal A}\cdot(\prescript{\infty}{}{\mathcal B}-\sigma^\infty_{nL+r} \prescript{L}{}{\mathcal B}\circ p^\infty_{nL+r})+(\prescript{\infty}{}{\mathcal A}-\sigma^\infty_{nL+r}\circ \prescript{L}{}{\mathcal A}\circ p^\infty_{nL+r})\cdot \sigma^\infty_{nL+r}\prescript{L}{}{\mathcal B}\circ p^\infty_{nL+r}\right)(\mathbb B_d)\\
&\subseteq \prescript{\infty}{}{\mathcal A}\circ(\prescript{\infty}{}{\mathcal B}-\sigma^\infty_{nL+r} \prescript{L}{}{\mathcal B}\circ p^\infty_{nL+r})(\mathbb B_d)\\
&\quad +(\prescript{\infty}{}{\mathcal A}-\sigma^\infty_{nL+r}\circ \prescript{L}{}{\mathcal A}\circ p^\infty_{nL+r})\circ \sigma^\infty_{nL+r}\prescript{L}{}{\mathcal B}\circ p^\infty_{nL+r}(\mathbb B_d)\\
\text{\small by $(\star)$+(2)}\;&\subseteq \frac{\epsilon}{2R}\prescript{\infty}{}{\mathcal A}(\mathbb B_{d+D})+(R+\frac{\epsilon}{2R})(\prescript{\infty}{}{\mathcal A}-\sigma^\infty_{nL+r}\circ \prescript{L}{}{\mathcal A}\circ p^\infty_{nL+r})(\mathbb B_{d+D})\\
\text{\small by $(\star)$+(3)}\; &\subseteq \frac{\epsilon}{2}\mathbb B_{d+2D}+\frac{\epsilon}{2}\mathbb B_{d+2D}\\
&\subseteq \epsilon\cdot\mathbb B_{d+2D}.
\end{align*}
This implies that $\underset{L\to \infty}{\lim}\left(\prescript{L}{}{\mathcal A}\cdot \prescript{L}{}{\mathcal B}\right)=\prescript{\infty}{}{\mathcal A}\cdot\prescript{\infty}{}{\mathcal B}$.
\end{proof}

Our next result states that the algebraic inverse limit (for operators with bounded degree) is a special case of conformal limit.

\begin{proposition}\label{prop:inv lim is conf lim}
Let $\{\prescript{L}{}{\mathcal O}\}_{L\in \bN}$ be a collection of linear operators of uniformly bounded degree, such that the following holds for all $L\in \bN:$
\begin{align*}
    \prescript{L}{}{\mathcal O}\circ p_{nL+r}=p_{nL+r}\circ \prescript{L+1}{}{\mathcal O}.
\end{align*}
Then $\underset{L\to \infty}{\lim}\prescript{L}{}{\mathcal O}$ exists and it equals to $\underset{\substack{\longleftarrow\\L}}{\lim} \prescript{L}{}{\mathcal O}$, the inverse limit of $\{\prescript{L}{}{\mathcal O}\}_{L\in \bN}$.
\end{proposition}

\begin{proof}
Let us denote $\mathcal O:=\underset{\substack{\longleftarrow\\L}}{\lim} \prescript{L}{}{\mathcal O}$. By the definition of inverse limit, the equation $\prescript{L}{}{\mathcal O}\circ p^\infty_{nL+r}=p^\infty_{nL+r}\circ {\mathcal O}$ holds for all $L\in\bN$. Suppose that $C$ uniformly bounds the degrees of $\{\prescript{L}{}{\mathcal O}\}_{L\in \bN}$. Take arbitrary $d\in \bN$, then $\sigma^\infty_{nL+r}\circ p^\infty_{nL+r}|_{\widetilde{\mathcal H}^{(r)}_{\infty}(n,k)_{\le d+C}}=\mathrm{Id}_{\widetilde{\mathcal H}^{(r)}_{\infty}(n,k)_{\le d+C}}$ for $L\gg 0$ by Proposition \ref{prop:p_N stabilizes}. For arbitrary $v\in \widetilde{\mathcal H}^{(r)}_{\infty}(n,k)_{\le d}$, we have
\begin{align*}
    &\quad \underset{L\to \infty}{\lim}\lVert \mathcal O(v)-\sigma^\infty_{nL+r}\circ\prescript{L}{}{\mathcal O}\circ p^\infty_{nL+r}(v)\rVert\\
    &=\underset{L\to \infty}{\lim}\lVert \sigma^\infty_{nL+r}\circ p^\infty_{nL+r}\circ \mathcal O(v)-\sigma^\infty_{nL+r}\circ\prescript{L}{}{\mathcal O}\circ p^\infty_{nL+r}(v)\rVert\\
    &=\underset{L\to \infty}{\lim}\lVert \sigma^\infty_{nL+r}\circ (p^\infty_{nL+r}\circ \mathcal O-\prescript{L}{}{\mathcal O}\circ p^\infty_{nL+r})(v)\rVert\\
    &=0.
\end{align*}
This finishes the proof.
\end{proof}

\begin{corollary}\label{cor:limit of gl_n[z]}
Fix $r\in \{0,\cdots,n-1\}$. For every $m\in \bZ_{\ge 0}$, denote by $\prescript{L}{}{\mathsf T}_{0,m}(E^a_b)$ the operator acting on $\widetilde{\mathcal H}_{nL+r}(n,k)$ which is defined in \eqref{DDCA action}. Then we have
\begin{align}
    \underset{L\to \infty}{\lim}\: \left(\prescript{L}{}{\mathsf T}_{0,m}(E^a_b)-\delta_{m=0}\delta^a_bkL\cdot\mathrm{Id}\right)=J^a_{b,-m}. 
\end{align}
\end{corollary}

\begin{proof}
Since $J^a_{b,-m}$ is by definition the inverse limit of the collection of maps $\{\prescript{L}{}{\mathsf T}_{0,m}(E^a_b)-\delta_{m=0}\delta^a_bkL\cdot\mathrm{Id}\}_{L\in \bN}$ (see \eqref{eq:twisted gl_n[z] action}), the result follows from Proposition \ref{prop:inv lim is conf lim}.
\end{proof}

\begin{corollary}\label{cor:limit of Yangian}
Fix $r\in \{0,\cdots,n-1\}$. For every $m\in \bZ_{\ge 0}$, let $\prescript{\infty}{}{T}^a_{b;m}$ be the Yangian generators acting on $\widetilde{\mathcal H}^{(r)}_{\infty}(n,k)$ which is defined in Definition \ref{def:Yangian notation_infinite}, then we have
\begin{align}\label{eq:limit of Yangian}
    \underset{L\to \infty}{\lim}\: \prescript{L}{}{\widetilde{T}}^a_{b;m}=\prescript{\infty}{}{T}^a_{b;m}. 
\end{align}
\end{corollary}

\begin{proof}
Since $\prescript{\infty}{}{\widetilde{T}}^a_{b;m}$ is the inverse limit of $\prescript{L}{}{\widetilde{T}}^a_{b;m}$ (see \eqref{eq:Yangian} and Definition \ref{def:Yangian notation_infinite}), the result follows from Proposition \ref{prop:inv lim is conf lim}.
\end{proof}

To obtain more conformal limits, we shall use the following criterion.

\begin{theorem}\label{thm:recursion}
Let $\{\prescript{L}{}{\mathcal O}\}_{L\in \bN}$ be a collection of linear operators with negative degree, i.e. $\prescript{L}{}{\mathcal O}(\widetilde{\mathcal H}_{nL+r}(n,k)_{d})\subseteq \widetilde{\mathcal H}_{nL+r}(n,k)_{<d}$ holds for all $L,d\in \bN$. Suppose that $\mathcal O$ is a linear operator on $\widetilde{\mathcal H}^{(r)}_{\infty}(n,k)$ with negative degree. Then $\underset{L\to \infty}{\lim}\prescript{L}{}{\mathcal O}={\mathcal O}$ if and only if the followings hold $\forall m\in \bZ_{\ge 2}$ and $\forall 1\le b<a\le n$ and $\forall 1\le c,d\le n:$
\begin{equation}\label{eq:recursion}
\begin{split}
\underset{L\to \infty}{\lim}[\prescript{L}{}{\mathcal O},\prescript{L}{}{\mathsf T}_{0,0}(E^a_b)]&=[{\mathcal O}, J^a_{b,-i}],\quad\text{and}\\
\underset{L\to \infty}{\lim}[\prescript{L}{}{\mathcal O},\prescript{L}{}{\mathsf T}_{0,1}(E^c_d)]&=[{\mathcal O}, J^c_{d,-1}],\quad\text{and}\\
\underset{L\to \infty}{\lim}[\prescript{L}{}{\mathcal O},\prescript{L}{}{\mathsf t}_{0,m}]&=[{\mathcal O}, \alpha_{-m}/k].
\end{split}
\end{equation}
\end{theorem}

\begin{proof}
The ``only if'' part follows from Lemma \ref{lem:linearity and multiplication} and Corollary \ref{cor:limit of gl_n[z]}. By \ref{prop:cyclic at limit}, $\widetilde{\mathcal H}^{(r)}_{\infty}(n,k)$ is a highest weight module thus it is generated from a highest weight vector $\omega\in \widetilde{\mathcal H}^{(r)}_{\infty}(n,k)_0$ by the actions of $\{J^a_{b,0}\:|\:1\le b<a\le n\}\cup \{J^c_{d,-m}\:|\:1\le c,d\le n,\; m\in \bZ_{\ge 1}\}$. $\bar J_{<-1}$ can be obtained from taking iterative commutators between $\bar J_{-1}$, therefore $\widetilde{\mathcal H}^{(r)}_{\infty}(n,k)$ is generated from $\omega$ by the actions of operators in the set
\begin{align*}
    \mathcal G=\{J^a_{b,0}\:|\:1\le b<a\le n\}\cup\{J^c_{d,-1}\:|\:1\le c,d\le n\}\cup\{\alpha_{-m}\:|\:m\in \bZ_{\ge 2}\}
\end{align*}
We note that every operator $\mathcal A\in \mathcal G$ is an inverse limit, i.e. $\mathcal A=\underset{\substack{\longleftarrow\\L}}{\lim} \prescript{L}{}{\mathcal A}$. The condition \eqref{eq:recursion} is equivalent to
\begin{align}\label{eq:recursion 2}
    \underset{L\to \infty}{\lim}[\prescript{L}{}{\mathcal O},\prescript{L}{}{\mathcal A}]&=[{\mathcal O}, \mathcal A], \quad \forall \mathcal A\in \mathcal G.
\end{align}
Suppose that \eqref{eq:recursion 2} holds, then we shall prove that $\forall v\in \widetilde{\mathcal H}^{(r)}_{\infty}(n,k)$, $\underset{L\to \infty}{\lim}\lVert {\mathcal O}(v)-\sigma^\infty_{nL+r}\circ \prescript{L}{}{\mathcal O}\circ p^\infty_{nL+r}(v)\rVert=0$. Since $v$ is a linear combination of operators of form $\mathcal A_1\cdots\mathcal A_\ell (\omega)$ and $\mathcal A_i\in \mathcal G$, we can prove the statement for $\mathcal A_1\cdots\mathcal A_\ell (\omega)$ and the statement for $v$ follows from triangle inequality. We proceed by induction on $\ell$. When $\ell=0$, the statement is automatic because $\mathcal O(\omega)=0$ and $\prescript{L}{}{\mathcal O}\circ p^\infty_{nL+r}(\omega)$ by our assumption that $\mathcal O$ and $\prescript{L}{}{\mathcal O}$ have negative degree. Suppose that the statement has been proven for $v'=\mathcal A_2\cdots\mathcal A_\ell (\omega)$, then we have
\begin{align}\label{eq:recursion 3}
\begin{split}
&\underset{L\to \infty}{\lim}\lVert {\mathcal O}(\mathcal A_1(v'))-\sigma^\infty_{nL+r}\circ \prescript{L}{}{\mathcal O}\circ p^\infty_{nL+r}(\mathcal A_1(v'))\rVert\\
    &=\underset{L\to \infty}{\lim}\lVert {\mathcal O}\circ\mathcal A_1(v')-\sigma^\infty_{nL+r}\circ \prescript{L}{}{\mathcal O}\circ \prescript{L}{}{\mathcal A_1}\circ p^\infty_{nL+r}(v')\rVert\\
   \text{\small by triangle inequality}\; &\le \underset{L\to \infty}{\lim}\lVert [{\mathcal O},\mathcal A_1](v')-\sigma^\infty_{nL+r}\circ [\prescript{L}{}{\mathcal O}, \prescript{L}{}{\mathcal A_1}]\circ p^\infty_{nL+r}(v')\rVert\\
   &\quad +\underset{L\to \infty}{\lim}\lVert \mathcal A_1\circ {\mathcal O}(v')-\sigma^\infty_{nL+r}\circ \prescript{L}{}{\mathcal A_1}\circ \prescript{L}{}{\mathcal O}\circ p^\infty_{nL+r}(v')\rVert\\
   \text{\small by \eqref{eq:recursion 2}}\; &=\underset{L\to \infty}{\lim}\lVert \mathcal A_1\circ {\mathcal O}(v')-\sigma^\infty_{nL+r}\circ \prescript{L}{}{\mathcal A_1}\circ \prescript{L}{}{\mathcal O}\circ p^\infty_{nL+r}(v')\rVert\\
   \text{(!)}\; &=\underset{L\to \infty}{\lim}\lVert \mathcal A_1\circ {\mathcal O}(v')-\mathcal A_1\circ\sigma^\infty_{nL+r}\circ \prescript{L}{}{\mathcal O}\circ p^\infty_{nL+r}(v')\rVert\\
    \text{\small by induction and continuity of $\mathcal A_1$}\;&=0.
\end{split}
\end{align}
Here in the step labelled by (!), we apply the Proposition \ref{prop:p_N stabilizes}. Namely for $L\gg 0$ and fixed $n,k,d$,  $$\sigma^{\infty}_{nL+r}\circ p^{\infty}_{nL+r}|_{\widetilde{\mathcal H}^{(r)}_{\infty}(n,k)_{\le d}}=\mathrm{Id},$$
which implies that $\mathcal A_1\circ\sigma^\infty_{nL+r}|_{\widetilde{\mathcal H}_{nL+r}(n,k)_{\le d}}=\sigma^\infty_{nL+r}\circ \prescript{L}{}{\mathcal A_1}|_{\widetilde{\mathcal H}_{nL+r}(n,k)_{\le d}}$ for $L\gg 0$. Then we take $d$ such that $v'\in \widetilde{\mathcal H}^{(r)}_{\infty}(n,k)_{\le d}$. This finishes the induction step and therefore concludes the proof.
\end{proof}

\subsubsection{Order of error terms}

\begin{definition}\label{def:speed of convergence}
Fix $r\in \{0,\cdots,n-1\}$ and $d\in \bZ$. Let $\{\prescript{L}{}{\mathcal O}\in \End(\widetilde{\mathcal H}_{nL+r}(n,k))\}_{L\in \bN}$ be a collection of linear operators of uniformly bounded degree. Let $h>0$, then we say that $\prescript{L}{}{\mathcal O}$ converges to $\prescript{\infty}{}{\mathcal O}$ with error term of order $O(L^{-h})$, notation:
\begin{align*}
    \prescript{L}{}{\mathcal O}\xrightarrow{O(L^{-h})}\prescript{\infty}{}{\mathcal O},
\end{align*}
if for all $v\in \widetilde{\mathcal H}_{nL+r}(n,k)$ there exists constant $C_v>0$ such that
\begin{align}
    \lVert \prescript{\infty}{}{\mathcal O}(v)-\sigma^\infty_{nL+r}\circ \prescript{L}{}{\mathcal O}\circ p^\infty_{nL+r}(v)\rVert \le C_v L^{-h}\text{ holds for all }L.
\end{align}
Here $\rVert\cdot\rVert$ is a norm on $\widetilde{\mathcal H}^{(r)}_{\infty}(n,k)$.
\end{definition}

Similar argument as in the Remark \ref{rmk:independence of norm} shows that the definition does not depend on the choice of the norm $\rVert\cdot\rVert$. In the below we list some properties which are counterparts of the corresponding statements in the above.

\begin{lemma}[{cf. Lemma \ref{lem:scalar limit}}]\label{lem:scalar limit_error term}
Suppose that $\{\alpha_L\in \bC\}_{L\in \bN}$ is a sequence of complex numbers, and assume there exists $\alpha\in \bC$ and $C>0$ and $h>0$ such that $|\alpha_L-\alpha|\le C L^{-h}$ holds for all $L$, then
\begin{align}
    {\alpha_L\cdot \mathrm{Id}}_{\widetilde{\mathcal H}_{nL+r}(n,k)}\xrightarrow{O(L^{-h})}\alpha\cdot\mathrm{Id}_{\widetilde{\mathcal H}^{(r)}_{\infty}(n,k)}.
\end{align}
\end{lemma}

\begin{proof}
Take arbitrary $d\in \bN$, then there exists $M>0$ such that $\sigma^\infty_{nL+r}\circ p^\infty_{nL+r}|_{\widetilde{\mathcal H}^{(r)}_{\infty}(n,k)_{\le d}}=\mathrm{Id}_{\widetilde{\mathcal H}^{(r)}_{\infty}(n,k)_{\le d}}$ for all $L\ge M$ by Proposition \ref{prop:p_N stabilizes}. For arbitrary $v\in \widetilde{\mathcal H}^{(r)}_{\infty}(n,k)_{\le d}$ and $L\ge M$, we have
\begin{align*}
    \lVert \alpha v-\alpha_L\sigma^\infty_{nL+r}\circ p^\infty_{nL+r}(v)\rVert=|\alpha_L-\alpha|\lVert v\rVert\le C \lVert v\rVert L^{-h}.
\end{align*}
So we can take $C_v=\max(\max_{L\le M}(L^h\lVert \prescript{\infty}{}{\mathcal O}(v)-\sigma^\infty_{nL+r}\circ \prescript{L}{}{\mathcal O}\circ p^\infty_{nL+r}(v)\rVert), C\lVert v\rVert)$.
\end{proof}

\begin{lemma}[{cf. Lemma \ref{lem:linearity and multiplication}}]\label{lem:linearity and multiplication_error term}
Suppose that $\{\prescript{L}{}{\mathcal A}\}_{L\in \bN}$ and $\{\prescript{L}{}{\mathcal B}\}_{L\in \bN}$ are two collections of  uniformly bounded degree linear operators, and assume that $\prescript{L}{}{\mathcal A}\xrightarrow{O(L^{-h})}\prescript{\infty}{}{\mathcal A}$ and $\prescript{L}{}{\mathcal B}\xrightarrow{O(L^{-h'})}\prescript{\infty}{}{\mathcal B}$, where $\prescript{\infty}{}{\mathcal A}$ and $\prescript{\infty}{}{\mathcal B}$ are bounded degree linear operators. Then $\forall\alpha,\beta\in \bC$,
\begin{align}
    \alpha\cdot\prescript{L}{}{\mathcal A}+\beta\cdot \prescript{L}{}{\mathcal B}\xrightarrow{O(L^{-h''})}\alpha\cdot\prescript{\infty}{}{\mathcal A}+\beta\cdot\prescript{\infty}{}{\mathcal B},\quad \prescript{L}{}{\mathcal A}\cdot \prescript{L}{}{\mathcal B}\xrightarrow{O(L^{-h''})}\prescript{\infty}{}{\mathcal A}\cdot\prescript{\infty}{}{\mathcal B},
\end{align}
where $h''=\min(h,h')$.
\end{lemma}

\begin{proof}
In the proof of Lemma \ref{lem:linearity and multiplication}, we replace $\epsilon$ by $CL^{-h''}$ for some $C>0$, the rest remains the same.
\end{proof}

\begin{lemma}[{cf. Proposition \ref{prop:inv lim is conf lim}}]\label{lem:inv lim is conf lim_error term}
Let $\{\prescript{L}{}{\mathcal O}\}_{L\in \bN}$ be a collection of linear operators of uniformly bounded degree, such that the following holds for all $L\in \bN:$
\begin{align*}
    \prescript{L}{}{\mathcal O}\circ p_{nL+r}=p_{nL+r}\circ \prescript{L+1}{}{\mathcal O}.
\end{align*}
Then for all $h>0$, we have $\prescript{L}{}{\mathcal O}\xrightarrow{O(L^{-h})}\underset{\substack{\longleftarrow\\L}}{\lim} \prescript{L}{}{\mathcal O}$.
\end{lemma}

\begin{proof}
Let us denote $\mathcal O:=\underset{\substack{\longleftarrow\\L}}{\lim} \prescript{L}{}{\mathcal O}$. By the definition of inverse limit, the equation $\prescript{L}{}{\mathcal O}\circ p^\infty_{nL+r}=p^\infty_{nL+r}\circ {\mathcal O}$ holds for all $L\in\bN$. Suppose that $C$ uniformly bounds the degrees of $\{\prescript{L}{}{\mathcal O}\}_{L\in \bN}$. Take arbitrary $d\in \bN$, then $\sigma^\infty_{nL+r}\circ p^\infty_{nL+r}|_{\widetilde{\mathcal H}^{(r)}_{\infty}(n,k)_{\le d+C}}=\mathrm{Id}_{\widetilde{\mathcal H}^{(r)}_{\infty}(n,k)_{\le d+C}}$ for $L\gg 0$ by Proposition \ref{prop:p_N stabilizes}. Then $\mathcal O|_{\widetilde{\mathcal H}^{(r)}_{\infty}(n,k)_{\le d+C}}=\sigma^\infty_{nL+r}\circ\prescript{L}{}{\mathcal O}\circ p^\infty_{nL+r}|_{\widetilde{\mathcal H}^{(r)}_{\infty}(n,k)_{\le d+C}}$ for $L\gg 0$, whence the lemma follows.
\end{proof}

\begin{lemma}[{cf. Theorem \ref{thm:recursion}}]\label{lem:recursion_error term}
Let $\{\prescript{L}{}{\mathcal O}\}_{L\in \bN}$ be a collection of linear operators with negative degree, i.e. $\prescript{L}{}{\mathcal O}(\widetilde{\mathcal H}_{nL+r}(n,k)_{d})\subseteq \widetilde{\mathcal H}_{nL+r}(n,k)_{<d}$ holds for all $L,d\in \bN$. Suppose there exists $h>0$ and a linear operator $\mathcal O$ on $\widetilde{\mathcal H}^{(r)}_{\infty}(n,k)$ with negative degree such that the followings hold $\forall m\in \bZ_{\ge 2}$ and $\forall 1\le b<a\le n$ and $\forall 1\le c,d\le n:$
\begin{equation}\label{eq:recursion_error term}
\begin{split}
[\prescript{L}{}{\mathcal O},\prescript{L}{}{\mathsf T}_{0,0}(E^a_b)]&\xrightarrow{O(L^{-h})}[{\mathcal O}, J^a_{b,-i}],\quad\text{and}\\
[\prescript{L}{}{\mathcal O},\prescript{L}{}{\mathsf T}_{0,1}(E^c_d)]&\xrightarrow{O(L^{-h})}[{\mathcal O}, J^c_{d,-1}],\quad\text{and}\\
[\prescript{L}{}{\mathcal O},\prescript{L}{}{\mathsf t}_{0,m}]&\xrightarrow{O(L^{-h})}[{\mathcal O}, \alpha_{-m}/k].
\end{split}
\end{equation}
Then $\prescript{L}{}{\mathcal O}\xrightarrow{O(L^{-h})}{\mathcal O}$.
\end{lemma}

\begin{proof}
We modify the proof of Theorem \ref{thm:recursion} as follows. We notice that in the limit \eqref{eq:recursion 2} the error terms are of order $O(L^{-h})$, whence in the line marked with ``by \eqref{eq:recursion 2}" in \eqref{eq:recursion 3}, the error terms are of order $O(L^{-h})$. In the last line of \eqref{eq:recursion 3} we use the induction and continuity of $\mathcal A_1$ to bound the error term by the order $O(L^{-h})$. The rest of \eqref{eq:recursion 3} uses the fact that $\sigma^{\infty}_{nL+r}\circ p^{\infty}_{nL+r}|_{\widetilde{\mathcal H}^{(r)}_{\infty}(n,k)_{\le d}}=\mathrm{Id}$ for $L\gg 0$ and fixed $n,k,d$, so the error term vanishes for $L\gg 0$. In total, the error term of the convergence $\prescript{L}{}{\mathcal O}{\longrightarrow}{\mathcal O}$ is bounded by order $O(L^{-h})$.
\end{proof}

\subsection{Emergent \texorpdfstring{$\widehat{\gl}(n)$}{affine gl(n)} annihilation operators from conformal limit}

\begin{theorem}\label{thm:conformal limit of T[m,0]}
Fix $r\in \{0,\cdots,n-1\}$. For every $m\in \bZ_{>0}$, denote by $\prescript{L}{}{\mathsf T}_{m,0}(E^a_b)$ the operator acting on $\widetilde{\mathcal H}_{nL+r}(n,k)$ which is defined in \eqref{DDCA action}. Then we have
\begin{align}\label{eq:conformal limit of T[m,0]}
    \underset{L\to \infty}{\lim}\: \frac{1}{(k+n)^{m}L^{m}} \prescript{L}{}{\mathsf T}_{m,0}(E^a_b)=J^a_{b,m}-\frac{\delta^a_b}{k+n}\alpha_m, 
\end{align}
where $J^a_{b,m}$ and $\alpha_m$ are defined in \eqref{eq:trace+tracelss decomposition}. Moreover the error terms of the convergence \eqref{eq:conformal limit of T[m,0]} are of order $O(L^{-1})$.
\end{theorem}

\begin{theorem}\label{thm:L[1]}
Fix $r\in \{0,\cdots,n-1\}$. Denote by $\mathcal T(z)=\sum_{m\in \bZ}\mathcal L_m z^{-m-2}$ the Sugawara's stress-operator of affine vertex algebra associated to $\widehat{\mathfrak{sl}}(n)_k\oplus \widehat{\gl}(1)_{kn}$. Then we have
\begin{align}\label{eq:L[1]}
    \underset{L\to \infty}{\lim}\: \left[\frac{1}{(k+n)L} \prescript{L}{}{\mathsf t}_{2,1}-\prescript{L}{}{\mathsf t}_{1,0}\right]=-\mathcal L_{1}+C\cdot\alpha_{1},
\end{align}
where $C=-\frac{k(n+2r)}{n(n+k)}$. Moreover the error terms of the convergence \eqref{eq:L[1]} are of order $O(L^{-1})$.
\end{theorem}

\begin{remark}
Individual limits $\underset{L\to \infty}{\lim}\: \frac{1}{(k+n)L} \prescript{L}{}{\mathsf t}_{2,1} $ or $\underset{L\to \infty}{\lim}\:\prescript{L}{}{\mathsf t}_{1,0}$ do not exist.
\end{remark}

Our strategy to prove Theorem \ref{thm:conformal limit of T[m,0]} and Theorem \ref{thm:L[1]} is to first prove the $m=1$ case in \eqref{eq:conformal limit of T[m,0]}, then prove Theorem \ref{thm:L[1]}, next we take the adjoint action of stress-operator \eqref{eq:L[1]} on the $m=1$ case in \eqref{eq:conformal limit of T[m,0]} to get $m>1$ cases.

We begin with three technical lemma which are crucial in the proof of Theorem \ref{thm:L[1]}.

\begin{lemma}\label{lem:m>1}
For all $m\in \bZ_{\ge 2}$, we have
\begin{align}\label{eq:conformal limit of T[1,m]}
    \underset{L\to \infty}{\lim}\: \left[\frac{1}{(k+n)L} \prescript{L}{}{\mathsf T}_{1,m}(E^a_b)-\prescript{L}{}{\mathsf T}_{0,m-1}(E^a_b)\right]=0,
\end{align}
with error terms of order $O(L^{-1})$.
\end{lemma}

\begin{proof}
Taking the iterated adjoint action of $\frac{1}{k}\alpha_{-1}=\underset{L\to \infty}{\lim}\:\prescript{L}{}{\mathsf t}_{0,1}$ on two sides of \eqref{eq:limit of Yangian}, we get
\begin{align}\label{eq:rewrite limit of Yangian}
\underset{L\to \infty}{\lim}\:\frac{1}{m!}\mathrm{ad}^{m-1}_{\prescript{L}{}{\mathsf t}_{0,1}}\left(\prescript{L}{}{\widetilde{T}}^a_{b;m}\right)=\frac{1}{m!}\mathrm{ad}^{m-1}_{\alpha_{-1}/k}\left(\prescript{\infty}{}{\widetilde{T}}^a_{b;m}\right),
\end{align}
with error terms of order $O(L^{-h})$ for arbitrary $h>0$ by Lemma \ref{lem:linearity and multiplication_error term} and Lemma \ref{lem:inv lim is conf lim_error term}. Let us expand the left-hand-side (before taking $L\to\infty$):
\begin{align*}
    \frac{1}{m!}\mathrm{ad}^{m-1}_{\prescript{L}{}{\mathsf t}_{0,1}}\left(\prescript{L}{}{\widetilde{T}}^a_{b;m}\right)=\prescript{L}{}{\left[A^a X\cdot\mathrm{Sym}(Y X^{m-1})B_b\right]}-\left((k+n)L+\frac{kL}{m}\right)\cdot\prescript{L}{}{\mathsf T}_{0,m-1}(E^a_b).
\end{align*}
Let us compute the difference between $A^a X\cdot\mathrm{Sym}(Y X^{m-1})B_b$ and $\mathsf T_{1,m}(E^a_b)$:
\begin{align*}
&A^a X\cdot\mathrm{Sym}(YX^{m-1})B_b-A^a \mathrm{Sym}(YX^{m})B_b\\
&=\sum_{i=0}^{m-1}\frac{m-i}{m(m+1)}A^a X^i[X, Y]X^{m-1-i}B_b\\
\text{\small by \eqref{eq:[X,Y] inside}}\quad &= \sum_{i=0}^{m-1}\frac{m-i}{m(m+1)}\left[\mathsf T_{0,i}(E^a_c)\mathsf T_{0,m-1-i}(E^c_b)-(k+n)\mathsf T_{0,m-1}(E^a_b)\right]\\
&=-\frac{k+n}{2}\mathsf T_{0,m-1}(E^a_b)+\sum_{i=0}^{m-1}\frac{m-i}{m(m+1)}\mathsf T_{0,i}(E^a_c)\mathsf T_{0,m-1-i}(E^c_b).
\end{align*}
Therefore \eqref{eq:rewrite limit of Yangian} is equivalent to
\begin{multline}\label{eq:rewrite limit of Yangian 2}
\underset{L\to \infty}{\lim}\:\left[\prescript{L}{}{\mathsf T}_{1,m}(E^a_b)-\left((k+n)L+\frac{k+n}{2}\right)\cdot\prescript{L}{}{\mathsf T}_{0,m-1}(E^a_b)\right]\\
+\sum_{i=0}^{m-1}\frac{m-i}{m(m+1)}{J}^a_{c,-i}{J}^c_{b,i+1-m}=\frac{1}{m!}\mathrm{ad}^{m-1}_{\alpha_{-1}/k}\left(\prescript{\infty}{}{\widetilde{T}}^a_{b;m}\right).
\end{multline}
Dividing two sides by $(k+n)L$ and applying Lemma \ref{lem:scalar limit}, Lemma \ref{lem:linearity and multiplication}, and Proposition \ref{prop:inv lim is conf lim}, we get
\begin{align*}
    \underset{L\to \infty}{\lim}\:\left[\frac{1}{(k+n)L}\prescript{L}{}{\mathsf T}_{1,m}(E^a_b)-\prescript{L}{}{\mathsf T}_{0,m-1}(E^a_b)\right]=0. 
\end{align*}
The error terms of the above convergence are of order $O(L^{-1})$ by Lemma \ref{lem:linearity and multiplication_error term}.
\end{proof}

\begin{lemma}\label{lem:m=1 traceless part}
For an operator with adjoint $\gl_n$ indices $\mathcal O^a_b$, denote by $\left(\mathcal O^a_b\right)_{0}=\mathcal O^a_b-\frac{\delta^a_b}{n}\mathcal O^c_c$ the traceless part of $\mathcal O^a_b$. Then we have
\begin{align}\label{eq:conformal limit of T[1,1]}
    \underset{L\to \infty}{\lim}\:\left[\frac{1}{(k+n)L} \prescript{L}{}{\mathsf T}_{1,1}(E^a_b)_0-\prescript{L}{}{\mathsf T}_{0,0}(E^a_b)_0\right]=0,
\end{align}
with error terms of order $O(L^{-1})$.
\end{lemma}

\begin{proof}
Expanding \eqref{eq:Yangian} and extracting the traceless part of $u^{-2}$ term in \eqref{eq:Yangian}, we get
\begin{align}\label{eq:traceless part of T[1]}
    \left(\prescript{L}{}{\widetilde{T}}^a_{b;1}\right)_0=\prescript{L}{}{\left[A^a X Y B_b\right]}_0-(2k+n)L\cdot\prescript{L}{}{\mathsf T}_{0,0}(E^a_b)_0.
\end{align}
Using \eqref{eq:[X,Y] inside}, we get
\begin{align*}
A^a XY B_b-A^a \mathrm{Sym}(Y X)B_b=-\frac{k+n}{2}\mathsf T_{0,0}(E^a_b)+\frac{1}{2}\mathsf T_{0,0}(E^a_c)\mathsf T_{0,0}(E^c_b).
\end{align*}
Plug the above into \eqref{eq:traceless part of T[1]} and we get
\begin{equation}
\begin{split}
\left(\prescript{L}{}{\widetilde{T}}^a_{b;1}\right)_0&=\prescript{L}{}{\mathsf T}_{1,1}(E^a_b)_0-\left((2k+n)L+\frac{k+n}{2}\right)\cdot\prescript{L}{}{\mathsf T}_{0,0}(E^a_b)_0+\frac{1}{2}\left(\prescript{L}{}{\mathsf T}_{0,0}(E^a_c)\prescript{L}{}{\mathsf T}_{0,0}(E^c_b)\right)_0\\
&=\prescript{L}{}{\mathsf T}_{1,1}(E^a_b)_0-\left((k+n)L+\frac{k+n}{2}\right)\cdot(\prescript{L}{}{J}^a_{b,0})_0+\frac{1}{2}\left(\prescript{L}{}{J}^a_{c,0}\prescript{L}{}{J}^c_{b,0}\right)_0
\end{split}
\end{equation}
Dividing two sides by $(k+n)L$ and applying Lemma \ref{lem:scalar limit}, Lemma \ref{lem:linearity and multiplication}, and Proposition \ref{prop:inv lim is conf lim}, we get
\begin{align*}
0=\underset{L\to \infty}{\lim}\:\frac{1}{(k+n)L} \left(\prescript{L}{}{\widetilde{T}}^a_{b;1}\right)_0&=\underset{L\to \infty}{\lim}\: \left[\frac{1}{(k+n)L} \prescript{L}{}{\mathsf T}_{1,1}(E^a_b)_0-\left(1+\frac{1}{2L}\right)(\prescript{L}{}{J}^a_{b,0})_0\right]\\
&\quad + \underset{L\to \infty}{\lim}\: \frac{1}{2(k+n)L} \left(\prescript{L}{}{J}^a_{c,0}\prescript{L}{}{J}^c_{b,0}\right)_0\\
&=\underset{L\to \infty}{\lim}\:\left[\frac{1}{(k+n)L} \prescript{L}{}{\mathsf T}_{1,1}(E^a_b)_0-(\prescript{L}{}{J}^a_{b,0})_0\right]\\
&=\underset{L\to \infty}{\lim}\:\left[\frac{1}{(k+n)L} \prescript{L}{}{\mathsf T}_{1,1}(E^a_b)_0-\prescript{L}{}{\mathsf T}_{0,0}(E^a_b)_0\right].
\end{align*}
The error terms of the above convergence are of order $O(L^{-1})$ by Lemma \ref{lem:linearity and multiplication_error term}.
\end{proof}

\begin{lemma}\label{lem:m=1 trace part}
\begin{align}
\underset{L\to \infty}{\lim}\: \left[\frac{2}{(k+n)L} \prescript{L}{}{\mathsf t}_{1,1}-\prescript{L}{}{\mathsf t}_{0,0}\right]=\frac{rn-k(n+r)}{n+k},
\end{align}
with error terms of order $O(L^{-1})$.
\end{lemma}

\begin{proof}
We notice that 
\begin{align*}
    \prescript{L}{}{\mathsf t}_{1,1}|_{\widetilde{\mathcal H}_{nL+r}(n,k)_{d}}=d+\frac{knL(L-1)}{2}+knr+\frac{(nL+r)^2}{2},\quad  \prescript{L}{}{\mathsf t}_{0,0}|_{\widetilde{\mathcal H}_{nL+r}(n,k)}=nL+r.
\end{align*}
Then the result follows from direct computation.
\end{proof}

\begin{proof}[Proof of Theorem \ref{thm:L[1]}]
Since both $\frac{1}{(k+n)L} \prescript{L}{}{\mathsf t}_{2,1}-\prescript{L}{}{\mathsf t}_{1,0}$ and $-\mathcal L_{1}+C\cdot\alpha_{1}$ have negative degree, we can apply Theorem \ref{thm:recursion} to show the convergence and apply Lemma \ref{lem:recursion_error term} to bound the order of error terms. Namely, it is enough to show that
\begin{equation}\label{eq:recursion to check}
\begin{split}
&\left[\prescript{L}{}{\mathsf t}_{2,1}-\prescript{L}{}{\mathsf t}_{1,0},\;\prescript{L}{}{\mathsf T}_{0,0}(E^a_b)\right]\xrightarrow{O(L^{-1})}[-\mathcal L_{1}+C\cdot\alpha_{1}, \; J^a_{b,0}],\\
&\left[\prescript{L}{}{\mathsf t}_{2,1}-\prescript{L}{}{\mathsf t}_{1,0},\;\prescript{L}{}{\mathsf T}_{0,1}(E^a_b)_0\right]\xrightarrow{O(L^{-1})}[-\mathcal L_{1}+C\cdot\alpha_{1}, \;(J^a_{b,-1})_0],\\
& \left[\prescript{L}{}{\mathsf t}_{2,1}-\prescript{L}{}{\mathsf t}_{1,0},\;\prescript{L}{}{\mathsf t}_{0,1}\right]\xrightarrow{O(L^{-1})}[-\mathcal L_{1}+C\cdot\alpha_{1},\; \alpha_{-1}/k],\\
& \left[\prescript{L}{}{\mathsf t}_{2,1}-\prescript{L}{}{\mathsf t}_{1,0},\;\prescript{L}{}{\mathsf t}_{0,m}\right]\xrightarrow{O(L^{-1})}[-\mathcal L_{1}+C\cdot\alpha_{1},\; \alpha_{-m}/k] \quad (m\ge 2).
\end{split}
\end{equation}
The 1st line of \eqref{eq:recursion to check} holds because $\left[\prescript{L}{}{\mathsf t}_{2,1}-\prescript{L}{}{\mathsf t}_{1,0},\;\prescript{L}{}{\mathsf T}_{0,0}(E^a_b)\right]=0$ and $[-\mathcal L_{1}+C\cdot\alpha_{1}, \; J^a_{b,0}]=0$. The 2nd line of \eqref{eq:recursion to check} follows from Lemma \ref{lem:m=1 traceless part}. The 3rd line of \eqref{eq:recursion to check} follows from Lemma \ref{lem:m=1 trace part}. The 4th line of \eqref{eq:recursion to check} follows from Lemma \ref{lem:m>1}. This concludes the proof. 
\end{proof}

\begin{proof}[Proof of Theorem \ref{thm:conformal limit of T[m,0]}]
We begin with the $m=1, E^a_b=\mathrm{Id}$ case in \eqref{eq:conformal limit of T[m,0]}, which is equivalent to proving $\underset{L\to \infty}{\lim}\: \frac{1}{(k+n)^{m}L^{m}} \prescript{L}{}{\mathsf t}_{1,0}=\frac{1}{k+n}\alpha_1$. Since both $\frac{1}{(k+n)L }\prescript{L}{}{\mathsf t}_{1,0}$ and $\frac{1}{k+n}\alpha_1$ have negative degree, we can apply Theorem \ref{thm:recursion} to show the convergence and apply Lemma \ref{lem:recursion_error term} to bound the order of error terms. Namely, it suffices to show that 
\begin{equation}\label{eq:recursion to check 2}
\begin{split}
&\left[\frac{1}{(k+n)L }\prescript{L}{}{\mathsf t}_{1,0},\;\prescript{L}{}{\mathsf T}_{0,0}(E^c_d)\right]\xrightarrow{O(L^{-1})}[\frac{1}{k+n}\alpha_1, \; J^c_{d,0}] \quad (1\le d <c \le n),\\
&\left[\frac{1}{(k+n)L }\prescript{L}{}{\mathsf t}_{1,0},\;\prescript{L}{}{\mathsf T}_{0,1}(E^e_f)\right]\xrightarrow{O(L^{-1})}[\frac{1}{k+n}\alpha_1, \; J^e_{f,-1}] \quad (1\le e,f\le n),\\
&\left[\frac{1}{(k+n)L }\prescript{L}{}{\mathsf t}_{1,0},\;\prescript{L}{}{\mathsf t}_{0,\ell}\right]\xrightarrow{O(L^{-1})}[\frac{1}{k+n}\alpha_1,\; \alpha_{-\ell}/k] \quad (\ell\ge 2).
\end{split}
\end{equation}
The 1st line of \eqref{eq:recursion to check 2} holds because $[\prescript{L}{}{\mathsf t}_{1,0},\prescript{L}{}{\mathsf T}_{0,0}(E^c_d)]=0$ and $[\alpha_1, J^c_{d,0}]=0$. The left-hand-side of the 2nd line of \eqref{eq:recursion to check 2} equals to $\frac{1}{(k+n)L }\prescript{L}{}{\mathsf T}_{0,0}(E^e_f)$, which converges to $\frac{k}{k+n}\delta^e_f=[\frac{1}{k+n}\alpha_1, \; J^e_{f,-1}]$ with error terms of order $O(L^{-1})$. The left-hand-side of the 3rd line of \eqref{eq:recursion to check 2} equals to $\frac{\ell}{(k+n)L }\prescript{L}{}{\mathsf t}_{0,\ell-1}$, which converges to $0$ with error terms of order $O(L^{-1})$. The right-hand-side of the 3rd line of \eqref{eq:recursion to check 2} is constantly zero. 

\smallskip Next we prove the case of $m=1$ and $E^a_b=E^n_1$. By Theorem \ref{thm:recursion}, it s enough to show that 
\begin{equation}\label{eq:recursion to check 3}
\begin{split}
&\left[\frac{1}{(k+n)L }\prescript{L}{}{\mathsf T}_{1,0}(E^n_1),\;\prescript{L}{}{\mathsf T}_{0,0}(E^c_d)\right]\xrightarrow{O(L^{-1})}[J^n_{1,1}, \; J^c_{d,0}] \quad (1\le d <c \le n),\\
&\left[\frac{1}{(k+n)L }\prescript{L}{}{\mathsf T}_{1,0}(E^n_1),\;\prescript{L}{}{\mathsf T}_{0,1}(E^e_f)\right]\xrightarrow{O(L^{-1})}[J^n_{1,1}, \; J^e_{f,-1}] \quad (1\le e,f\le n),\\
&\left[\frac{1}{(k+n)L }\prescript{L}{}{\mathsf T}_{1,0}(E^n_1),\;\prescript{L}{}{\mathsf t}_{0,\ell}\right]\xrightarrow{O(L^{-1})}[J^n_{1,1},\; \alpha_{-\ell}/k] \quad (\ell\ge 2).
\end{split}
\end{equation}
The 1st line of \eqref{eq:recursion to check 3} holds because $[\prescript{L}{}{\mathsf T}_{1,0}(E^n_1),\prescript{L}{}{\mathsf T}_{0,0}(E^c_d)]=0$ and $[J^n_{1,1}, J^c_{d,0}]=0$ when $c>d$. To prove the 2nd line of \eqref{eq:recursion to check 3}, we notice that
\begin{align*}
[{\mathsf T}_{1,0}(E^n_1),\;{\mathsf T}_{0,1}(E^e_f)]=&{\mathsf T}_{1,1}([E^n_1,E^e_f])+\frac{n+k}{2}{\mathsf T}_{0,0}(\{E^n_1,E^e_f\})+{\mathsf T}_{0,0}(E^n_f){\mathsf T}_{0,0}(E^e_1)\\
&-\frac{1}{2}\delta^n_f {\mathsf T}_{0,0}(E^e_g){\mathsf T}_{0,0}(E^g_1)-\frac{1}{2}\delta^e_1 {\mathsf T}_{0,0}(E^n_g){\mathsf T}_{0,0}(E^g_f),
\end{align*}
therefore the left-hand-side converges to $\delta^e_1 J^n_{f,0}-\delta^n_f J^e_{1,0}+k\delta^n_1\delta^e_f$ with error terms of order $O(L^{-1})$ by Lemma \ref{lem:m=1 traceless part}. $\delta^e_1 J^n_{f,0}-\delta^n_f J^e_{1,0}+k\delta^n_1\delta^e_f$ equals to the right-hand-side of the 2nd line of \eqref{eq:recursion to check 3}. The left-hand-side of the 3rd line of \eqref{eq:recursion to check 3} equals to $\frac{\ell}{(k+n)L }\prescript{L}{}{\mathsf T}_{0,\ell-1}(E^n_1)$ which converges to $0$ with error terms of order $O(L^{-1})$, and the right-hand-side of the 3rd line of \eqref{eq:recursion to check 3} is constantly zero.

\smallskip Next we prove the case of $m=1$ and general traceless $E^a_b$. We notice that the adjoint action of $\prescript{L}{}{\mathsf T}_{0,0}(\mathfrak{sl}_n)$ on $\prescript{L}{}{\mathsf T}_{0,0}(E^n_1)$ spans the whole vector space $\prescript{L}{}{\mathsf T}_{0,0}(\mathfrak{sl}_n)$. Thus the result follows from the $E^a_b=E^n_1$ case and Lemma \ref{lem:linearity and multiplication} and Corollary \ref{cor:limit of gl_n[z]}. We use Lemma \ref{lem:linearity and multiplication_error term} to bound the order of error terms by $O(L^{-1})$. This finishes the proof of the case of $m=1$ and arbitrary $E^a_b$.

\smallskip We proceed to the $m>1$ cases using Theorem \ref{thm:L[1]} and Lemma \ref{lem:recursion_error term}. We notice that
\begin{multline*}
\frac{1}{(k+n)^{m}L^{m}} \prescript{L}{}{\mathsf T}_{m,0}(E^a_b)\\
=\frac{1}{(m-1)!}\underbrace{\left[\cdots\left[\frac{1}{(k+n)L} \prescript{L}{}{\mathsf T}_{1,0}(E^a_b),\; \frac{1}{(k+n)L} \prescript{L}{}{\mathsf t}_{2,1}-\prescript{L}{}{\mathsf t}_{1,0}\right],\cdots,\; \frac{1}{(k+n)L} \prescript{L}{}{\mathsf t}_{2,1}-\prescript{L}{}{\mathsf t}_{1,0}\right]}_{m-1\text{ times}}
\end{multline*}
Then according to the $m=1$ case that we just proved, and Theorem \ref{thm:L[1]}, and Lemma \ref{lem:linearity and multiplication_error term}, we conclude that
\begin{align*}
\frac{1}{(k+n)^{m}L^{m}} \prescript{L}{}{\mathsf T}_{m,0}(E^a_b)&\xrightarrow{O(L^{-1})}\frac{1}{(m-1)!}\underbrace{\left[\cdots\left[J^a_{b,1}-\frac{\delta^a_b}{k+n}\alpha_1,\; -\mathcal L_{1}+C\cdot\alpha_{1}\right],\cdots, -\mathcal L_{1}+C\cdot\alpha_{1}\right]}_{m-1\text{ times}}\\
&=J^a_{b,m}-\frac{\delta^a_b}{k+n}\alpha_m.
\end{align*}
This concludes the proof of Theorem \ref{thm:conformal limit of T[m,0]}.
\end{proof}

\subsection{Comparison with Dorey-Tong-Turner}\label{subsec:Comparison with DTT}

A different scaling limit was used in the original work of Dorey, Tong, and Turner \cite{dorey2016matrix}, where they consider the following operators
\begin{align}\label{DTT scaling}
    \prescript{L}{}{\mathcal J}^a_{b,m}=\left(\frac{n}{(n+k)N}\right)^{\frac{|m|}{2}}\begin{cases}
        \prescript{L}{}{\mathsf T}_{m,0}(E^a_b)_0 &,\text{ if }m\ge 0,\\
        \prescript{L}{}{\mathsf T}_{0,-m}(E^a_b)_0 &, \text{ if }m< 0.
    \end{cases}
\end{align}
To relate their scaling to our result (Theorem \ref{thm:conformal limit of T[m,0]}), we make the following observation. 

\smallskip Consider the linear automorphism $\Phi_N: \widetilde{\mathcal H}_N(n,k)\cong \widetilde{\mathcal H}_N(n,k)$ such that
\begin{align}
    \Phi_N\bigg\rvert_{\widetilde{\mathcal H}_N(n,k)_d}=\left(\frac{n}{(n+k)N}\right)^{d/2}\cdot\mathrm{Id}_{\widetilde{\mathcal H}_N(n,k)_d}.
\end{align}
Define 
\begin{align}
    \overline{p}_N:=\Phi_N^{-1}\circ p_N\circ \Phi_{N+n}: \widetilde{\mathcal H}_{N+n}(n,k)\to \widetilde{\mathcal H}_N(n,k).
\end{align}
Then $\left\{\widetilde{\mathcal H}_{nL+r}(n,k), \overline{p}_{nL+r}\right\}_{L\in \mathbb N}$ forms an inverse system, which is isomorphic to the inverse system \\
$\left\{\widetilde{\mathcal H}_{nL+r}(n,k), {p}_{nL+r}\right\}_{L\in \mathbb N}$ via the compatible family of isomorphisms $\{\Phi_{nL+r}\}_{L\in\bN}$. We use $\{\overline{p}_{nL+r}\}_{L\in \bN}$ to define the inverse limit
\begin{align}\label{eq:conformal limit_twisted}
    \overline{\mathcal H}^{(r)}_{\infty}(n,k)=\bigoplus_{d\ge 0}\overline{\mathcal H}^{(r)}_{\infty}(n,k)_d,\;\text{where }\;\overline{\mathcal H}^{(r)}_{\infty}(n,k)_d:=\lim_{\substack{\longleftarrow\\L}}\: \widetilde{\mathcal H}_{nL+r}(n,k)_d.
\end{align}
By construction, we have isomorphism:
\begin{align}
    \Phi_{\infty}:\overline{\mathcal H}^{(r)}_{\infty}(n,k)\cong \widetilde{\mathcal H}^{(r)}_{\infty}(n,k), \;\text{where $\Phi_{\infty}$ is the limit of }\{\Phi_{nL+r}\}_{L\in\bN}.
\end{align}
The formalism of conformal limit of operators in Section \ref{subsec:conformal limit of operators} works for $\overline{\mathcal H}^{(r)}_{\infty}(n,k)$. Apparently we have
\begin{align}
    \lim_{L\to \infty} \prescript{L}{}{\mathcal O}=\prescript{\infty}{}{\mathcal O} \text{ in }\widetilde{\mathcal H}^{(r)}_{\infty}(n,k)\:\Longleftrightarrow \: \lim_{L\to \infty} \Phi_{nL+r}^{-1}\prescript{L}{}{\mathcal O}\:\Phi_{nL+r}=\Phi_{\infty}^{-1}\prescript{\infty}{}{\mathcal O}\:\Phi_{\infty} \text{ in }\overline{\mathcal H}^{(r)}_{\infty}(n,k).
\end{align}
Moreover the error terms on two sides are of the same order, i.e.
\begin{align}
    \prescript{L}{}{\mathcal O}\xrightarrow{O(L^{-h})}\prescript{\infty}{}{\mathcal O} \text{ in }\widetilde{\mathcal H}^{(r)}_{\infty}(n,k)\:\Longleftrightarrow \: \Phi_{nL+r}^{-1}\prescript{L}{}{\mathcal O}\:\Phi_{nL+r}\xrightarrow{O(L^{-h})}\Phi_{\infty}^{-1}\prescript{\infty}{}{\mathcal O}\:\Phi_{\infty} \text{ in }\overline{\mathcal H}^{(r)}_{\infty}(n,k).
\end{align}
It is straightforward to compute that
\begin{align}
    \Phi_N^{-1}\prescript{L}{}{\mathsf T}_{m,m'}(E^a_b)\: \Phi_N= \left(\frac{n}{(n+k)N}\right)^{\frac{m'-m}{2}}\prescript{L}{}{\mathsf T}_{m,m'}(E^a_b).
\end{align}
Thus we have the following corollary to the Theorem \ref{thm:conformal limit of T[m,0]}.

\begin{corollary}\label{cor:conformal limit_twisted}
In $\overline{\mathcal H}^{(r)}_{\infty}(n,k)$, we have
\begin{align}\label{eq:renormalized J}
\lim_{L\to \infty}  \prescript{L}{}{\mathcal J}^a_{b,m}=\Phi_{\infty}^{-1}\bar{J}^a_{b,m}\Phi_{\infty},\quad 
\lim_{L\to \infty}  \prescript{L}{}{\beta}_{m} =\Phi_{\infty}^{-1}\alpha_{m}\Phi_{\infty},\quad (m\in \bZ),
\end{align}
where $\prescript{L}{}{\mathcal J}^a_{b,m}$ is defined in \eqref{DTT scaling} and $\prescript{L}{}{\beta}_{m}$ is defined as
\begin{align}\label{DTT scaling_central}
    \prescript{L}{}{\beta}_{m}=\left(\frac{n}{(n+k)N}\right)^{\frac{|m|}{2}}\begin{cases}
        (1+\frac{n}{k})\cdot\prescript{L}{}{\mathsf T}_{m,0}(E^a_a) &,\text{ if }m> 0,\\
        \prescript{L}{}{\mathsf T}_{0,0}(E^a_a)-nkL\cdot\mathrm{Id} &,\text{ if }m=0,\\
        \prescript{L}{}{\mathsf T}_{0,-m}(E^a_a) &, \text{ if }m< 0.
    \end{cases}
\end{align}
Moreover the error terms of the convergence \eqref{eq:renormalized J} are of order $O(L^{-1})$. In particular the limit $\prescript{\infty}{}{\mathcal J}^a_{b,m}=\Phi_{\infty}^{-1}\bar{J}^a_{b,m}\Phi_{\infty}$ and $\prescript{\infty}{}{\beta}_{m}=\Phi_{\infty}^{-1}\alpha_{m}\Phi_{\infty}$ satisfy $\widehat{\mathfrak{sl}}(n)_k\oplus \widehat{\mathfrak{gl}}(1)_{kn}$ commutation relation:
\begin{align}\label{affine gl(n)_k relation}
\begin{split}
&[\prescript{\infty}{}{\mathcal J}^a_{b,m},\prescript{\infty}{}{\mathcal J}^c_{d,l}]=\delta^c_b \prescript{\infty}{}{\mathcal J}^a_{d,m+l}-\delta^a_d \prescript{\infty}{}{\mathcal J}^c_{b,m+l}+km\delta_{m+l,0}\left(\delta^a_d\delta^b_c-\frac{1}{n}\delta^a_b\delta^c_d\right),\\
&[\prescript{\infty}{}{\beta}_{m},\prescript{\infty}{}{\beta}_{l}]=knm\delta_{m+l,0},\\
&[\prescript{\infty}{}{\mathcal J}^a_{b,m},\prescript{\infty}{}{\beta}_{l}]=0.
\end{split}
\end{align}
\end{corollary}

The conjecture by Dorey, Tong, and Turner in \cite{dorey2016matrix} is slightly different from our Corollary \ref{cor:conformal limit_twisted}. They did not define the conformal limit of Hilbert space, instead they stated the expected convergence property as a bound of the error terms of \eqref{affine gl(n)_k relation} for finite $L$. 

\begin{definition}\label{def:order of operator}
For an $h>0$, we say a sequence of operators $\{\prescript{L}{}{\mathcal A}\in \End(\widetilde{H}_{nL+r}(n,k))\}_{L\in\bN}$ is of order $O(L^{-h})$ if for every $d\in \mathbb N$ there exists a constant $C_d>0$ such that he operator norm of $\prescript{L}{}{\mathcal A}|_{\widetilde{\mathcal H}_{nL+r}(n,k)_{\le d}}$ is bounded by $C_dL^{-h}$, i.e.
\begin{align}\label{A is of order L^-h}
    \rVert \prescript{L}{}{\mathcal A}(v)\rVert_L \le C_d L^{-h} \rVert v\rVert_L\text{ holds for all }v\in \widetilde{\mathcal H}_{nL+r}(n,k)_{\le d}.
\end{align}
Here $\rVert\cdot\rVert_L$ is the norm on the Hilbert space $\widetilde{\mathcal H}_{nL+r}(n,k)$ induced by the Hermitian inner product (see Section \ref{subsec:Hermitian inner product}).
\end{definition}

\begin{theorem}[{Conjectured in \cite[(2.6)]{dorey2016matrix}, see also \eqref{eq:asymptotic hat sl(n)_intro}}]\label{thm:DTT scaling limit}
The operators
\begin{align*}
    [\prescript{L}{}{\mathcal J}^a_{b,m},\prescript{L}{}{\mathcal J}^c_{d,l}]-\delta^c_b \prescript{L}{}{\mathcal J}^a_{d,m+l}+\delta^a_d \prescript{L}{}{\mathcal J}^c_{b,m+l}-km\delta_{m+l,0}\left(\delta^a_d\delta^b_c-\frac{1}{n}\delta^a_b\delta^c_d\right).
\end{align*}
has order $O(L^{-1})$.
\end{theorem}

\begin{proof}
Corollary \ref{cor:conformal limit_twisted} implies that 
\begin{align*}
    [\prescript{L}{}{\mathcal J}^a_{b,m},\prescript{L}{}{\mathcal J}^c_{d,l}]-\delta^c_b \prescript{L}{}{\mathcal J}^a_{d,m+l}+\delta^a_d \prescript{L}{}{\mathcal J}^c_{b,m+l}-km\delta_{m+l,0}\left(\delta^a_d\delta^b_c-\frac{1}{n}\delta^a_b\delta^c_d\right)\xrightarrow{O(L^{-1})}0,
\end{align*}
then the Theorem \ref{thm:DTT scaling limit} follows from Proposition \ref{prop:equivalence of convergence} below.
\end{proof}

\begin{remark}\label{rmk:DTT scaling limit}
The same argument shows that the operators $[\prescript{L}{}{\beta}_{m},\prescript{L}{}{\beta}_{l}]-knm\delta_{m+l,0}$ and $[\prescript{L}{}{\mathcal J}^a_{b,m},\prescript{L}{}{\beta}_{l}]$ have order $O(L^{-1})$. This implies conjecture \eqref{eq:asymptotic hat gl(1)_intro} as
\begin{align*}
    \mathcal B_m=\begin{cases}
        \frac{1}{n+k}\beta_m &,\text{ if }m>0,\\
        \frac{1}{k}\beta_m &,\text{ if }m<0,\\
    \end{cases}\quad
    \text{ and $\mathcal B_0$ is central.}
\end{align*}
\end{remark}

\begin{proposition}\label{prop:equivalence of convergence}
For a sequence of operators $\{\prescript{L}{}{\mathcal A}\in \End(\widetilde{H}_{nL+r}(n,k))\}_{L\in\bN}$ with uniformly bounded degree, we have
\begin{align}
\prescript{L}{}{\mathcal A}\xrightarrow{O(L^{-h})} 0\in \End(\overline{\mathcal{H}}^{(r)}_{\infty})\;\Longleftrightarrow\; \prescript{L}{}{\mathcal A}\text{ has order }O(L^{-h}).
\end{align}
\end{proposition}

To prove Proposition \ref{prop:equivalence of convergence}, we need to compare the norm $\rVert\cdot\rVert_L$ on $\widetilde{\mathcal H}_{nL+r}(n,k)$ with a norm on $\overline{\mathcal H}^{(r)}_{\infty}(n,k)$. There is a natural choice of norm on $\overline{\mathcal H}^{(r)}_{\infty}(n,k)$. Namely, let $(\cdot|\cdot)$ be the unique $\widehat{\gl}(n)$-invariant Hermitian form on the irreducible integrable module $L_{k\varpi_r}(\widehat{\mathfrak{sl}}(n)_k)\otimes \mathrm{Fock}_{kr}(\widehat{\gl}(1)_{kn})$ which is normalized by requiring
\begin{align}
    (\omega|\omega)=1, \text{ where $\omega$ is the highest weight vector}.
\end{align}
The $\widehat{\gl}(n)$-invariance means the following equation
\begin{align}
    (v|J^a_{b,m}v')=(J^b_{a,-m}v|v')
\end{align}
holds $\forall v,v'\in L_{k\varpi_r}(\widehat{\mathfrak{sl}}(n)_k)\otimes \mathrm{Fock}_{kr}(\widehat{\gl}(1)_{kn})$ and $\forall 1\le a,b\le n$ and $\forall m\in \bZ$. Let $\rVert\cdot\rVert$ be the norm on $\overline{\mathcal H}^{(r)}_{\infty}(n,k)$ induced by $(\cdot|\cdot)$ via the isomorphism $\Phi_{\infty}:\overline{\mathcal H}^{(r)}_{\infty}(n,k)\cong L_{k\varpi_r}(\widehat{\mathfrak{sl}}(n)_k)\otimes \mathrm{Fock}_{kr}(\widehat{\gl}(1)_{kn})$.

\begin{lemma}\label{lem:bound of norm}
For arbitrary $d\in \bN$, there exists $B_d>1$ such that
\begin{align}\label{eq:bound of norm}
   B_d^{-1}\rVert v\rVert\le \frac{\rVert\overline{p}^{\infty}_{nL+r}(v)\rVert_L}{\rVert\overline{p}^{\infty}_{nL+r}(\omega)\rVert_L}\le B_d\rVert v\rVert \text{ holds for all $L$ and all }v\in \overline{\mathcal H}^{(r)}_{\infty}(n,k)_{\le d}.
\end{align}
\end{lemma}

\begin{proof}
We prove the following more precise statement that will lead to the lemma. For $i=1,2$, let $v_i=v_i'\otimes v_i''$ where 
$$v_i'\in L_{k\varpi_r}(\widehat{\mathfrak{sl}}(n)_k),\quad v_i''= \prod_{j=1}^{\ell_i}(\alpha_{-j})^{m_{i,j}}|0\rangle\in \mathrm{Fock}_{kr}(\widehat{\gl}(1)_{kn}),$$
then we have
\begin{align}\label{eq:inner product arbitrary}
    \lim_{L\to \infty}\frac{(\overline{p}^{\infty}_{nL+r}(v_1)|\overline{p}^{\infty}_{nL+r}(v_2))_L}{\rVert\overline{p}^{\infty}_{nL+r}(\omega)\rVert^2_L}=\left(\frac{k}{k+n}\right)^{\sum_{j=1}^{\ell_1}m_{1,j}}(v_1|v_2).
\end{align}
Here $(\cdot|\cdot)_L$ is the Hermitian inner product on $\widetilde{\mathcal H}_{nL+r}(n,k)$. And we have identified $\overline{\mathcal H}^{(r)}_{\infty}(n,k) $ with $ L_{k\varpi_r}(\widehat{\mathfrak{sl}}(n)_k)\otimes \mathrm{Fock}_{kr}(\widehat{\gl}(1)_{kn})$ using isomorphism $\Phi_\infty$.

As a preliminary step, we claim that 
\begin{align}\label{eq:inner product ground states}
    \frac{(\overline{p}^{\infty}_{nL+r}(\omega)|\overline{p}^{\infty}_{nL+r}(v))_L}{\rVert\overline{p}^{\infty}_{nL+r}(\omega)\rVert^2_L}=(\omega|v)\quad\text{for all $v$ and all $L$}.
\end{align}
To prove \eqref{eq:inner product ground states}, we notice that $\overline{p}^{\infty}_{nL+r}|_{\overline{\mathcal H}^{(r)}_{\infty}(n,k)_{0}}:\overline{\mathcal H}^{(r)}_{\infty}(n,k)_{0}\to \widetilde{\mathcal H}_{nL+r}(n,k)_{0}$ is an isomorphism between irreducible $\mathfrak{sl}_n$ modules. Moreover, the inner products $(\cdot|\cdot)$ and $(\cdot|\cdot)_L$ are $\mathfrak{sl}_n$-invariant, whence by the Schur lemma there exists a constant $\lambda_L\in \bC^{\times}$ for every $L$ such that
\begin{align}
    (\overline{p}^{\infty}_{nL+r}(\omega)|\overline{p}^{\infty}_{nL+r}(v))_L=\lambda_L (\omega|u), \quad \forall u\in \overline{\mathcal H}^{(r)}_{\infty}(n,k)_{0}.
\end{align}
Write $v=v_0+v_{>0}$ where $v_0\in \overline{\mathcal H}^{(r)}_{\infty}(n,k)_{0}$ and $v_{>0}\in \overline{\mathcal H}^{(r)}_{\infty}(n,k)_{>0}$, then we have
\begin{align}
    \frac{(\overline{p}^{\infty}_{nL+r}(\omega)|\overline{p}^{\infty}_{nL+r}(v))_L}{\rVert\overline{p}^{\infty}_{nL+r}(\omega)\rVert^2_L}=\frac{(\overline{p}^{\infty}_{nL+r}(\omega)|\overline{p}^{\infty}_{nL+r}(v_0))_L}{\rVert\overline{p}^{\infty}_{nL+r}(\omega)\rVert^2_L}=\frac{\lambda_L(\omega|v_0)}{\lambda_L(\omega|\omega)}=(\omega|v_0)=(\omega|v).
\end{align}
This proves \eqref{eq:inner product ground states}.

Next, we claim that if $\{\prescript{L}{}{\mathcal O}\in \End(\widetilde{\mathcal H}_{nL+r}(n,k))\}_{L\in \bN}$ is a sequence of bounded degree operators such that $\underset{L\to \infty}{\lim}\prescript{L}{}{\mathcal O}=\mathcal O$, then
\begin{align}\label{eq:inner product ground states_operator}
    \lim_{L\to \infty}\frac{(\overline{p}^{\infty}_{nL+r}(\omega)|\prescript{L}{}{\mathcal O}\circ \overline{p}^{\infty}_{nL+r}(v))_L}{\rVert\overline{p}^{\infty}_{nL+r}(\omega)\rVert^2_L}=(\omega|\mathcal O(v))\quad\text{for all $v$}.
\end{align}
Fix $v$, and let $d$ be a natural number such that $\prescript{L}{}{\mathcal O}\circ \overline{p}^{\infty}_{nL+r}(v)\in \widetilde{\mathcal H}_{nL+r}(n,k)_{\le d}$ for all $L$. Then by Proposition \ref{prop:p_N stabilizes} there exists $M$ such that $\overline{p}^{\infty}_{nL+r}\circ \overline{\sigma}^{\infty}_{nL+r}|_{\widetilde{\mathcal H}_{nL+r}(n,k)_{\le d}}=\mathrm{Id}_{\widetilde{\mathcal H}_{nL+r}(n,k)_{\le d}}$ for all $L\ge M$. Then it follows that
\begin{align*}
\lim_{L\to \infty}\frac{(\overline{p}^{\infty}_{nL+r}(\omega)|\prescript{L}{}{\mathcal O}\circ \overline{p}^{\infty}_{nL+r}(v))_L}{\rVert\overline{p}^{\infty}_{nL+r}(\omega)\rVert^2_L}&=\lim_{L\to \infty}\frac{(\overline{p}^{\infty}_{nL+r}(\omega)|\overline{p}^{\infty}_{nL+r}\circ \overline{\sigma}^{\infty}_{nL+r}\circ\prescript{L}{}{\mathcal O}\circ \overline{p}^{\infty}_{nL+r}(v))_L}{\rVert\overline{p}^{\infty}_{nL+r}(\omega)\rVert^2_L}\\
\text{\small by \eqref{eq:inner product ground states}}\quad &=\lim_{L\to \infty}(\omega| \overline{\sigma}^{\infty}_{nL+r}\circ\prescript{L}{}{\mathcal O}\circ \overline{p}^{\infty}_{nL+r}(v))\\
\text{\small since $\underset{L\to \infty}{\lim}\prescript{L}{}{\mathcal O}=\mathcal O$}\quad &=(\omega|\mathcal O(v)).
\end{align*}
This proves \eqref{eq:inner product ground states_operator}.

Next we apply \eqref{eq:inner product ground states_operator} to deduce \eqref{eq:inner product arbitrary}. Let us write $v'_i=f_i(\mathcal G)\ket{\omega}$ where $f_i\in \bC\langle\mathcal G\rangle$ are non-commutative polynomials in the set of variables $\mathcal G=\{\bar J^a_{b,-m}\:|\:1\le b,a\le n,m\in \bN\}$. We denote $f_i^{\mathrm{op}}(\mathcal G^{\mathrm{op}})$ to be the polynomial in variables $\mathcal G^{\mathrm{op}}=\{\bar J^b_{a,m}\:|\:1\le b,a\le n,m\in \bN \}$ obtained from $f_i(\mathcal G)$ by reversing the order of all monomials inside $f_i$ then followed by replacements $\bar J^a_{b,-m}\mapsto \bar J^b_{a,m}$. We also define $f_i(\prescript{L}{}{\mathcal G})$ by substituting $\bar J^a_{b,-m}\mapsto \prescript{L}{}{\mathcal J}^a_{b,-m}$. Similarly, $f_i^{\mathrm{op}}(\prescript{L}{}{\mathcal G}^{\mathrm{op}})$ is defined by substituting $\bar J^b_{a,m}\mapsto \prescript{L}{}{\mathcal J}^b_{a,m}$. Then we can rewrite
\begin{align}
\begin{split}
(\overline{p}^{\infty}_{nL+r}(v_1)|\overline{p}^{\infty}_{nL+r}(v_2))_L&=\left(f_1(\prescript{L}{}{\mathcal G})\prod_{j=1}^{\ell_1}(\prescript{L}{}{\beta}_{-j})^{m_{1,j}}\overline{p}^{\infty}_{nL+r}(\omega)\;\bigg\rvert \;\overline{p}^{\infty}_{nL+r}(v_2)\right)_L\\
&=\left(\overline{p}^{\infty}_{nL+r}(\omega)\;\bigg\rvert \;f_1^{\mathrm{op}}(\prescript{L}{}{\mathcal G}^{\mathrm{op}})^*\prod_{j=1}^{\ell_1}\left(\frac{k}{k+n}\prescript{L}{}{\beta}_{j}\right)^{m_{1,j}}\overline{p}^{\infty}_{nL+r}(v_2)\right)_L.
\end{split}
\end{align}
Here the superscript $*$ on the function $f_1^{\mathrm{op}}(\prescript{L}{}{\mathcal G}^{\mathrm{op}})$ means complex conjugation of all coefficients. By Corollary \ref{cor:conformal limit_twisted} we have:
\begin{align}
\underset{L\to \infty}{\lim}f_1^{\mathrm{op}}(\prescript{L}{}{\mathcal G}^{\mathrm{op}})^*\prod_{j=1}^{\ell_1}\left(\frac{k}{k+n}\prescript{L}{}{\beta}_{j}\right)^{m_{1,j}}=\left(\frac{k}{k+n}\right)^{\sum_{j=1}^{\ell_1}m_{1,j}} f_1^{\mathrm{op}}({\mathcal G}^{\mathrm{op}})^*\prod_{j=1}^{\ell_1}(\alpha_{j})^{m_{1,j}},
\end{align}
we note that the conjugation by $\Phi_\infty$ has been absorbed into the operators $\bar J^a_{b,m}$ and $\alpha_{m}$ since we have made identification $\Phi_{\infty}:\overline{\mathcal H}^{(r)}_{\infty}(n,k)\cong L_{k\varpi_r}(\widehat{\mathfrak{sl}}(n)_k)\otimes \mathrm{Fock}_{kr}(\widehat{\gl}(1)_{kn})$. Then \eqref{eq:inner product ground states_operator} implies that
\begin{align}
\begin{split}
\lim_{L\to \infty}\frac{(\overline{p}^{\infty}_{nL+r}(v_1)|\overline{p}^{\infty}_{nL+r}(v_2))_L}{\rVert\overline{p}^{\infty}_{nL+r}(\omega)\rVert^2_L}&=\left(\frac{k}{k+n}\right)^{\sum_{j=1}^{\ell_1}m_{1,j}}\left(\omega\;\bigg\rvert \;f_1^{\mathrm{op}}({\mathcal G}^{\mathrm{op}})^*\prod_{j=1}^{\ell_1}({\alpha}_{j})^{m_{1,j}}v_2\right)\\
&=\left(\frac{k}{k+n}\right)^{\sum_{j=1}^{\ell_1}m_{1,j}}\left(f_1({\mathcal G})\prod_{j=1}^{\ell_1}({\alpha}_{-j})^{m_{1,j}}\omega\;\bigg\rvert \;v_2\right)\\
&=\left(\frac{k}{k+n}\right)^{\sum_{j=1}^{\ell_1}m_{1,j}}(v_1|v_2)\\
\end{split}
\end{align}
This concludes the proof of \eqref{eq:inner product arbitrary}.

Finally, we demonstrate how to derive the original lemma from \eqref{eq:inner product arbitrary}. Every $v\in \overline{\mathcal H}^{(r)}_{\infty}(n,k)_{\le d}$ can be uniquely written as 
\begin{align}
\begin{split}
&v=\sum_{\substack{\vec m=(m_1,m_2,\cdots)\in\bN^{\infty}\\ \sum_{j}j\cdot m_j\le d}}v_{\vec m}, \quad v_{\vec m}=v'_{\vec m}\otimes \ket{\vec m},\quad \text{ where }\\
&v'_{\vec m}\in L_{k\varpi_r}(\widehat{\mathfrak{sl}}(n)_k),\quad \ket{\vec m}= \prod_{j}(\alpha_{-j})^{m_{j}}|0\rangle\in \mathrm{Fock}_{kr}(\widehat{\gl}(1)_{kn}).
\end{split}
\end{align}
Then \eqref{eq:inner product arbitrary} implies that 
\begin{align}\label{eq:limit of norm square}
\lim_{L\to \infty}\frac{\rVert\overline{p}^{\infty}_{nL+r}(v)\rVert^2_L}{\rVert\overline{p}^{\infty}_{nL+r}(\omega)\rVert^2_L}=\sum_{\vec m}\left(\frac{k}{k+n}\right)^{\sum_{j}m_{j}}\rVert v_{\vec m}\rVert^2.
\end{align}
By the polarization identity, the quadratic forms $\phi_L:\overline{\mathcal H}^{(r)}_{\infty}(n,k)_{\le d}\times \overline{\mathcal H}^{(r)}_{\infty}(n,k)_{\le d}\to \bC$ defined by
\begin{align}
    \phi_L(u,v):=\frac{(\overline{p}^{\infty}_{nL+r}(u)\:|\:\overline{p}^{\infty}_{nL+r}(v))_L}{\rVert\overline{p}^{\infty}_{nL+r}(\omega)\rVert^2_L}
\end{align}
point-wise converge to a quadratic form $\phi:\overline{\mathcal H}^{(r)}_{\infty}(n,k)_{\le d}\times \overline{\mathcal H}^{(r)}_{\infty}(n,k)_{\le d}\to \bC$. Choosing a basis of $\overline{\mathcal H}^{(r)}_{\infty}(n,k)_{\le d}$, then quadratic forms $\phi_L$ and $\phi$ are given by their matrix forms, and point-wise convergence implies that every matrix component $(\phi_L)_{ij}$ converge to $\phi_{ij}$. Therefore $\phi_L$ uniformly converges to $\phi$ on any compact subset $K\subset \overline{\mathcal H}^{(r)}_{\infty}(n,k)_{\le d}\times \overline{\mathcal H}^{(r)}_{\infty}(n,k)_{\le d}$. Let us take $K$ to be the diagonal unit sphere $\mathbb S^{\Delta}_d=\{(u,u)\in \overline{\mathcal H}^{(r)}_{\infty}(n,k)_{\le d}\times \overline{\mathcal H}^{(r)}_{\infty}(n,k)_{\le d} \:|\: \rVert u\rVert=1 \}$, then the uniform convergence implies that there exists $M>0$ such that $\forall L\ge M$, the inequality $|\phi_L(u,u)-\phi(u,u)|<\frac{1}{2}\left(\frac{k}{k+n}\right)^{d}$ holds for all $(u,u)\in \mathbb S^{\Delta}_d$. On the other hand, \eqref{eq:limit of norm square} implies that
\begin{align}
    \left(\frac{k}{k+n}\right)^{d}\le \phi(u,u)\le 1, \quad \forall (u,u)\in \mathbb S^{\Delta}_d.
\end{align}
Whence the following inequality holds for all $L\ge M$:
\begin{align}\label{eq:bound of phi_L}
    \frac{1}{2}\left(\frac{k}{k+n}\right)^{d}\le\phi_L(u,u)\le 1+\frac{1}{2}\left(\frac{k}{k+n}\right)^{d}, \quad \forall (u,u)\in \mathbb S^{\Delta}_d.
\end{align}
Lemma \ref{lem:bound of norm} follows from \eqref{eq:bound of phi_L}. In fact we can take
\begin{align}
    B_d=\max\left[\sqrt{2}\left(1+\frac{n}{k}\right)^{\frac{d}{2}},\;\max_{\ell<M}\sup_{(u,u)\in \mathbb S^{\Delta}_d}\left(\phi_\ell(u,u)^{\frac{1}{2}}+\phi_\ell(u,u)^{-\frac{1}{2}}\right)\right]
\end{align}
to fulfill \eqref{eq:bound of norm}.
\end{proof}

\begin{proof}[Proof of Proposition \ref{prop:equivalence of convergence}]
Since $\prescript{L}{}{\mathcal A}$ have uniformly bounded degree, there exists $m\in \bZ$ such that $\prescript{L}{}{\mathcal A}(\widetilde{\mathcal H}_{nL+r}(n,k)_{\le d})\subseteq \widetilde{\mathcal H}_{nL+r}(n,k)_{\le d+m}$ for all $d\in \bN$. Let us fix a choice of such $m$ in the rest of the proof. 

``$\Longrightarrow$''. By our assumption, for every $d\in \mathbb N$ there exists a constant $D_d>0$ such that
\begin{align}\label{A converges to zero}
    \rVert \overline{\sigma}^{\infty}_{nL+r}\circ\prescript{L}{}{\mathcal A}\circ \overline{p}^{\infty}_{nL+r}(u)\rVert \le D_d L^{-h} \rVert u\rVert\text{ holds for all }u\in \overline{\mathcal H}^{(r)}_{\infty}(n,k)_{\le d}.
\end{align}
Here $\rVert\cdot\rVert$ is the norm on $\overline{\mathcal H}^{(r)}_{\infty}(n,k)$ induced by $(\cdot|\cdot)$ via the isomorphism $\Phi_{\infty}:\overline{\mathcal H}^{(r)}_{\infty}(n,k)\cong L_{k\varpi_r}(\widehat{\mathfrak{sl}}(n)_k)\otimes \mathrm{Fock}_{kr}(\widehat{\gl}(1)_{kn})$. By Proposition \ref{prop:p_N stabilizes}, there exists $M>0$ such that such that $$\overline{p}^{\infty}_{nL+r}\circ \overline{\sigma}^{\infty}_{nL+r}|_{\widetilde{\mathcal H}_{nL+r}(n,k)_{\le d+m}}=\mathrm{Id}_{\widetilde{\mathcal H}_{nL+r}(n,k)_{\le d+m}}\text{ for all $L\ge M$}. $$
For $L\ge M$ and $u\in \overline{\mathcal H}^{(r)}_{\infty}(n,k)_{\le d}$, we have
\begin{align*}
\begin{split}
\frac{\rVert \prescript{L}{}{\mathcal A}\circ \overline{p}^{\infty}_{nL+r}(u)\rVert_L}{\rVert\overline{p}^{\infty}_{nL+r}(\omega)\rVert_L}&=\frac{\rVert \overline{p}^{\infty}_{nL+r}\circ \overline{\sigma}^{\infty}_{nL+r}\circ\prescript{L}{}{\mathcal A}\circ \overline{p}^{\infty}_{nL+r}(u)\rVert_L}{\rVert\overline{p}^{\infty}_{nL+r}(\omega)\rVert_L}\\
\text{\small by Lemma \ref{lem:bound of norm}}\quad &\le B_{d+m}\rVert \overline{\sigma}^{\infty}_{nL+r}\circ\prescript{L}{}{\mathcal A}\circ \overline{p}^{\infty}_{nL+r}(u)\rVert \\
\text{\small by \eqref{A converges to zero}}\quad &\le B_{d+m}D_d L^{-h} \rVert u\rVert\\
\text{\small by Lemma \ref{lem:bound of norm}}\quad &\le B_{d}B_{d+m}D_d L^{-h} \frac{\rVert\overline{p}^{\infty}_{nL+r}(u)\rVert_L}{\rVert\overline{p}^{\infty}_{nL+r}(\omega)\rVert_L},
\end{split}
\end{align*}
which implies that $\rVert \prescript{L}{}{\mathcal A}\circ \overline{p}^{\infty}_{nL+r}(u)\rVert_L\le B_{d}B_{d+m}D_d L^{-h}\rVert\overline{p}^{\infty}_{nL+r}(u)\rVert_L$. Since 
\begin{align*}
\overline{p}^{\infty}_{nL+r}|_{\overline{\mathcal H}^{(r)}_{\infty}(n,k)_{\le d}}:\overline{\mathcal H}^{(r)}_{\infty}(n,k)_{\le d}\to \widetilde{\mathcal H}_{nL+r}(n,k)_{\le d}\text{ is an isomorphism by our choice of $L$},
\end{align*}
it follows that
\begin{align}\label{A is of order 1/N_L>M}
\rVert \prescript{L}{}{\mathcal A}(v)\rVert_L \le B_{d}B_{d+m}D_d L^{-h} \rVert v\rVert_L \text{ holds for all }v\in \widetilde{\mathcal H}_{nL+r}(n,k)_{\le d}.
\end{align}
\eqref{A is of order 1/N_L>M} implies that $\prescript{L}{}{\mathcal A}$ has order $O(L^{-h})$. In fact we can take
\begin{align*}
    C_d=\max\left[ B_{d}B_{d+m}D_d,\;\max_{\ell<M}(\ell^h\cdot \rVert\prescript{\ell}{}{\mathcal A}|_{\widetilde{\mathcal H}_{n\ell+r}(n,k)_{\le d}}\rVert_\ell)\right]\text{ to fulfill \eqref{A is of order L^-h}},
\end{align*}
where $\rVert\prescript{\ell}{}{\mathcal A}|_{\widetilde{\mathcal H}_{n\ell+r}(n,k)_{\le d}}\rVert_\ell$ is the operator norm of the restriction of $\prescript{\ell}{}{\mathcal A}$ to subspace $\widetilde{\mathcal H}_{n\ell+r}(n,k)_{\le d}$, that is 
\begin{align*}
    \rVert\prescript{\ell}{}{\mathcal A}|_{\widetilde{\mathcal H}_{n\ell+r}(n,k)_{\le d}}\rVert_\ell=\sup_{\substack{u\in \widetilde{\mathcal H}_{n\ell+r}(n,k)_{\le d}\\ \rVert u\rVert_\ell=1}}\rVert\prescript{\ell}{}{\mathcal A}(u)\rVert_\ell.
\end{align*}

``$\Longleftarrow$''. By our assumption, for every $d\in \mathbb N$ there exists a constant $C'_d>0$ such that
\begin{align}\label{eq:bound C'_d}
    \rVert \prescript{L}{}{\mathcal A}(v)\rVert_L \le C'_d L^{-h} \rVert v\rVert_L\text{ holds for all }v\in \widetilde{\mathcal H}_{nL+r}(n,k)_{\le d}.
\end{align}
By Proposition \ref{prop:p_N stabilizes}, there exists $M>0$ such that such that $$\overline{p}^{\infty}_{nL+r}\circ \overline{\sigma}^{\infty}_{nL+r}|_{\widetilde{\mathcal H}_{nL+r}(n,k)_{\le d+m}}=\mathrm{Id}_{\widetilde{\mathcal H}_{nL+r}(n,k)_{\le d+m}}\text{ for all $L\ge M$}. $$
For $L\ge M$ and $u\in \overline{\mathcal H}^{(r)}_{\infty}(n,k)_{\le d}$, by Lemma \ref{lem:bound of norm} we have
\begin{align*}
\rVert \overline{\sigma}^{\infty}_{nL+r}\circ\prescript{L}{}{\mathcal A}\circ \overline{p}^{\infty}_{nL+r}(u)\rVert &\le B_{d+m}\frac{\rVert\overline{p}^{\infty}_{nL+r}\circ\overline{\sigma}^{\infty}_{nL+r}\circ\prescript{L}{}{\mathcal A}\circ \overline{p}^{\infty}_{nL+r}(u)\rVert_L}{\rVert\overline{p}^{\infty}_{nL+r}(\omega)\rVert_L}\\
&=B_{d+m}\frac{\rVert\prescript{L}{}{\mathcal A}\circ \overline{p}^{\infty}_{nL+r}(u)\rVert_L}{\rVert\overline{p}^{\infty}_{nL+r}(\omega)\rVert_L}\\
\text{\small by \eqref{eq:bound C'_d}}\quad &\le B_{d+m} C_d'L^{-h}\frac{\rVert\overline{p}^{\infty}_{nL+r}(u)\rVert_L}{\rVert\overline{p}^{\infty}_{nL+r}(\omega)\rVert_L}\\
\text{\small Lemma \ref{lem:bound of norm}}\quad &\le B_d B_{d+m} C_d'L^{-h}\rVert u\rVert.
\end{align*}
Let us take 
\begin{align*}
    D'_d=\max\left[ B_d B_{d+m} C_d',\; \max_{\ell<M}\sup_{\substack{v\in \overline{\mathcal H}^{(r)}_{\infty}(n,k)_{\le d}\\ \rVert v\rVert=1}}\rVert \overline{\sigma}^{\infty}_{n\ell+r}\circ\prescript{\ell}{}{\mathcal A}\circ \overline{p}^{\infty}_{n\ell+r}(v)\rVert\right],
\end{align*}
then $\rVert \overline{\sigma}^{\infty}_{nL+r}\circ\prescript{L}{}{\mathcal A}\circ \overline{p}^{\infty}_{nL+r}(u)\rVert \le D'_d L^{-h} \rVert u\rVert$ holds for all $u\in \overline{\mathcal H}^{(r)}_{\infty}(n,k)_{\le d}$ and all $L\in \bN$. Thus $\prescript{L}{}{\mathcal A}\xrightarrow{O(L^{-h})} 0$ by definition.
\end{proof}

\section{Applications}\label{sec:Application}

\subsection{Gelfand-Tsetlin bases of \texorpdfstring{$\widetilde{\mathcal H}_{N}(n,k)$}{H(n,k)}}\label{subsec:GZ}

\begin{definition}\label{def:GT pattern}
A GT pattern of height $n$ and length $N$ is an $(n+1)$-tuple of length-$N$ integer arrays $\Lambda=(\Lambda_1,\cdots,\Lambda_{n+1})$, $\Lambda_i=(\Lambda_{i,1}\ge \cdots\ge \Lambda_{i,N})\in \bZ^N$ such that
\begin{align}\label{GT condition}
    \Lambda_1\overset{1}{\searrow}\Lambda_2\overset{1}{\searrow}\cdots \overset{1}{\searrow}\Lambda_{n}\overset{1}{\searrow}\Lambda_{n+1}=\Lambda_1^\shortdownarrow.
\end{align}
See Definition \ref{def:decrease and increase} for the meaning of $\Lambda_1^\shortdownarrow$, and see Definition \ref{def:branching weights} for the meaning of $\overset{1}{\searrow}$. We denote the set of GT patterns of height $n$ and length $N$ by $\mathrm{GT}(N,n)$. In plain words, the condition \eqref{GT condition} is equivalent to the condition that $\Lambda_{n+1,i}=\Lambda_{1,i}-k$ and $\Lambda_{i,j}\ge \Lambda_{i+1,j}\ge \Lambda_{i,j+1}$ for all possible $i$ and $j$.

We say that a GT pattern $\Lambda=(\Lambda_1,\cdots,\Lambda_{n+1})$ is non-negative if and only if $\Lambda_{n+1}\in \bN^N$. We denote the set of non-negative GT patterns of height $n$ and length $N$ by $\mathrm{GT}(N,n)_+$.
\end{definition}

\begin{remark}\label{rmk:GT pattern}
For a  non-negative GT pattern $\Lambda\in \mathrm{GT}(N,n)_+$, we can associate $\Lambda$ to a conventional Gelfand-Tsetlin pattern \cite{gelfand1950finite} by the following arrangement:
\begin{equation}\label{standard GT}
\Lambda=(\Lambda_1,\cdots,\Lambda_{n+1})\mapsto \left(\lambda_{i,j}\in \bN\right)_{1\le i\le N+n}^{1\le j\le N+n+1-i},
\end{equation}
where
\begin{align}
\lambda_{i,j}=\begin{cases}
    \Lambda_{i,j} &,\text{ if }i\le n\text{ and }j\le N, \\
    0 &, \text{ if }i\le n\text{ and }j> N, \\
    \Lambda_{1,i+j-(n+1)}-k &, \text{ if }i > n.
\end{cases}
\end{align}
One can check that $(\lambda_{i,j})$ satisfy the defining condition for a conventional Gelfand-Tsetlin pattern: $\lambda_{i,j}\ge \lambda_{i+1,j}\ge \lambda_{i,j+1}$ for all possible $i$ and $j$. Conversely, if $\left(\lambda_{i,j}\in \bN\right)_{1\le i\le N+n}^{1\le j\le N+n+1-i}$ is a Gelfand-Tsetlin pattern such that $\lambda_{1,j}=0$ for all $j>N$, and $\lambda_{i,j}=\lambda_{1,i+j-(n+1)}-k$ for all $i>n$ and all possible $j$, we can associate $(\Lambda_1,\cdots,\Lambda_{n+1})\in \mathrm{GT}(N,n)_+$ by letting $\Lambda_i=\lambda_i$ for all $i\le n+1$.
\end{remark}

\begin{theorem}[Spectral Decomposition of $\widetilde{\mathcal H}_N(n,k)$]\label{thm:spectral decomposition_finite N}
The action of the quantum minors \\ $\{{A}_m(u)\}_{1\le m\le n}$ (see equation \eqref{GZ generators}) on $\widetilde{\mathcal H}_N(n,k)$ has simple spectrum, such that the eigenvectors are labelled by non-negative GT patterns of height $n$ and length $N:$
\begin{align}\label{eq:spectral decomposition_finite N}
    \widetilde{\mathcal H}_N(n,k)=\bigoplus_{\Lambda\in \mathrm{GT}(N,n)_+}V_{\Lambda},\quad \dim V_{\Lambda}=1,
\end{align}
and the eigenvalue of $A_m(u)$ on $V_{\Lambda}$, $\Lambda=(\Lambda_1,\cdots,\Lambda_{n+1})$ is
\begin{align}\label{eq:eigenvalue_finite N}
    \prod_{i=1}^N\frac{u-\Lambda_{m+1,i}-N+i}{u-\Lambda_{1,i}-N+i}.
\end{align}
\end{theorem}

\begin{proof}
Recall the decomposition \eqref{Schur-Weyl}:
\begin{align*}
    \bC[X,A]=\bigoplus_{\lambda=(\lambda_1\ge\cdots\ge\lambda_N)\in \bN^N}V^{N+n}_{\lambda}\otimes (V^{N}_{\lambda})^*,
\end{align*}
It is well-known that $V^{N+n}_{\lambda}$ has the Gelfand-Tsetlin bases \cite{gelfand1950finite}:
\begin{align*}
    V^{N+n}_{\lambda}=\bigoplus_{\Lambda\in \mathcal T_\lambda}V(\Lambda),\quad \dim V(\Lambda)=1,
\end{align*}
where $\mathcal T_\lambda$ is the set of arrays $\left(\lambda_{i,j}\in \bN\right)_{1\le i\le N+n}^{1\le j\le N+n+1-i}$ such that 
\begin{align*}
    \lambda_{1,j}=\begin{cases}
        \lambda_j &, \text{ if }j\le N,\\
        0 &, \text{ if }j>N,
    \end{cases}\quad\text{and $\lambda_{i,j}\ge \lambda_{i+1,j}\ge \lambda_{i,j+1}$ for all possible $i$ and $j$}.
\end{align*}
$V(\Lambda)$ is uniquely characterized by the condition that $V(\Lambda)$ belongs to the subspace $\bigoplus_{\nu}M^{\lambda}_{\nu,\lambda_{\ell+1}}\otimes V^{\ell}_{\nu}\otimes V^{N+n-\ell}_{\lambda_{\ell+1}}$ in the branching decomposition
\begin{align*}
    V^{N+n}_\lambda=\bigoplus_{\nu,\mu}M^{\lambda}_{\nu,\mu}\otimes V^{\ell}_{\nu}\otimes V^{N+n-\ell}_{\mu},\text{ with respect to subgroup }\begin{pNiceArray}{c:c}
  \GL_{\ell} & 0 \\
\hdottedline
 0 & \GL_{n+N-\ell}
\end{pNiceArray},
\end{align*}
for all $\ell\in \{1,\cdots,N+n-1\}$. Moreover, for a fixed $M\in \{1,\cdots,n+N-1\}$, $V(\Lambda)$ is a highest weight vector with weight $\mu$ with respect to $\GL_{n+N-M}$ in the decomposition
\begin{align*}
\begin{pNiceArray}{c:c}
  \GL_{M} & 0 \\
\hdottedline
 0 & \GL_{n+N-M}
\end{pNiceArray},
\end{align*}
if and only if $\forall i>M$ and $\forall$ possible $j$, $\lambda_{i,j}=\mu_{j+i-(M+1)}$. Since
\begin{align*}
    \widetilde{\mathcal H}_N(\lambda)=\left(V^{N+n}_{\lambda}\otimes (V^{N}_{\lambda})^*\right)^{\GL_{N}^\mathrm{diag},-k}\cong \Hom_{\GL_N}\left(V^{N}_{\lambda^\shortdownarrow},V^{N+n}_{\lambda}\right),
\end{align*}
and every $\GL_N$-invariant homomorphism is uniquely determined by a $\GL_N$ highest weight vector with weight $\lambda^\shortdownarrow$ in $V^{N+n}_{\lambda}$, we see that
$\widetilde{\mathcal H}_N(\lambda)$ has a basis labelled by elements in $\mathcal T_\lambda$ such that $\forall i>n$ and $\forall$ possible $j$, $\lambda_{i,j}=\lambda_{j+i-(n+1)}-k$. According to Remark \ref{rmk:GT pattern}, the subset of such elements in $\mathcal T_\lambda$ is one-to-one correspond to the subset of elements $(\Lambda_1,\cdots,\Lambda_{n+1})\in \mathrm{GT}(N,n)_+$ such that $\Lambda_1=\lambda$. This proves there exists a decomposition \eqref{eq:spectral decomposition_finite N}, and it remains to show that $V_\Lambda$ is an eigenvector for $A_m(u)$ with eigenvalue given by \eqref{eq:eigenvalue_finite N}.

In fact, for every element $(\Lambda_1,\cdots,\Lambda_{n+1})\in \mathrm{GT}(N,n)_+$, we have $\Lambda_1\overset{m}{\searrow}\Lambda_{m+1}\overset{n-m}{\searrow}\Lambda_1^\shortdownarrow$ by the Remark \ref{rmk:additive}. Then according to the construction of $V_\Lambda$ the Definition \ref{def:branching}, we have $V_\Lambda\subseteq \widetilde{\mathcal H}_{N}\left(\substack{\Lambda_1\\ \Lambda_{m+1}}\right)$. According to Remark \ref{rmk:Capelli det action on irrep},
\begin{align*}
    &{C}_m(u)\text{ acts on }\widetilde{\mathcal H}_{N}(\substack{\Lambda_1\\ \Lambda_{m+1}})\text{ as the scalar }\prod_{j=1}^{n+N-m}(u-\Lambda_{m+1,j}+j+m-n-N),\\
    &{C}_0(u)\text{ acts on }\widetilde{\mathcal H}_{N}(\Lambda_1)\text{ as the scalar }\prod_{j=1}^{n+N-m}(u-\Lambda_{1,j}+j+m-n-N),\\
    &\text{where we set $\Lambda_{i,j}=0$ if $j>N$}.
\end{align*}
Using Corollary \ref{cor:A_m in terms of C_m}, we see that $A_m(u)$ acts on $V_\Lambda$ as the scalar \eqref{eq:eigenvalue_finite N}, this concludes the proof of Theorem \ref{thm:spectral decomposition_finite N}.
\end{proof}

\begin{remark}\label{rmk:simple module and hwv}
Recall that for an admissible $\lambda$, $\widetilde{\mathcal H}_N(\lambda)$ is a simple $Y(\gl_n)$ submodule of $\widetilde{\mathcal H}_N(n,k)$ (Remark \ref{rmk:simple Y(gl_n) submodules}), and according to the proof of Theorem \ref{thm:spectral decomposition_finite N}, its eigenvector decomposition with respect to the $\{A_m(u)\}_{1\le m\le n}$ is 
\begin{align}\label{eq:spectral decomposition_finite N irrep}
    \widetilde{\mathcal H}_N(\lambda)=\bigoplus_{\substack{\Lambda\in \mathrm{GT}(N,n)_+\\ \Lambda_1=\lambda}}V_{\Lambda}.
\end{align}
Recall that the highest weight vector is defined to be the vector that is annihilated by all $T^a_b(u)$ such that $a<b$. Using the isomorphism $\widetilde{\mathcal H}_N(\lambda)\cong V_{\lambda/\lambda^\shortdownarrow}(k-N)^*$, and according to \cite[Corollary 2.5]{nazarov1998representations}, the highest weight vector in $\widetilde{\mathcal H}_N(\lambda)$ is $V_{\Lambda^0}$, where\footnote{Since $\widetilde{\mathcal H}_N(\lambda)$ is dual to $V_{\lambda/\lambda^\shortdownarrow}(k-N)$, highest weigh vector in $\widetilde{\mathcal H}_N(\lambda)$ is dual to the lowest weight vector in $V_{\lambda/\lambda^\shortdownarrow}(k-N)$, and vice versa. The Gelfand-Tsetlin pattern $\Lambda_0$ in \cite[(2.5)]{nazarov1998representations}, translated to our setting, becomes $$(\Lambda_0)_{i,j}=\begin{cases}
\lambda_j &, \text{ if }j+i\le n+1,\\
\min(\lambda_j,\lambda_{j+i-1-n}-k) &, \text{ if }j+i> n+1.
\end{cases}$$ $V_{\Lambda_0}$ is the lowest weight vector in $\widetilde{\mathcal H}_N(\lambda)$.} 
\begin{align}\label{eq:hwv}
\Lambda^0_{i,j}=
\begin{cases}
\lambda_j-k &, \text{ if }j+i> N+1,\\
\max(\lambda_j-k,\lambda_{j+i-1}) &, \text{ if }j+i\le N+1.
\end{cases}
\end{align}
\end{remark}

\begin{remark}
If we consider the Yangian action given by $\widetilde{T}(u)$ defined in \eqref{eq:Yangian}, then the $V_\Lambda$ in the decomposition \eqref{eq:spectral decomposition_finite N} is also an eigenvector for the quantum minors $\{\widetilde{A}_m(u)\}_{1\le m\le n}$, where
\begin{align}\label{GZ generators_twisted}
    \widetilde{A}_m(u):=\mathrm{qdet}\:\widetilde{T}^a_b(u)_{1\le a,b\le m}=\sum_{\sigma\in S_m}\mathrm{sgn}(\sigma)\prod_{1\le i\le m}^{\rightarrow}\widetilde{T}^{\sigma(i)}_i(u-i+1).
\end{align}
This is because $\widetilde{T}(u)$ and $T(u)$ only differ by a shift of spectral parameter $u\mapsto u+(k+n)L$ followed by multiplying a function $f(u)$. The eigenvalue of $\widetilde{A}_m(u)$ on $V_\Lambda$ is 
\begin{align}\label{eq:eigenvalue_finite N_twisted}
\begin{split}
&\prod_{i=1}^N\frac{u-r-\overline{\Lambda}_{m+1,i}+i}{u-r-\overline{\Lambda}_{1,i}+i}\prod_{j=0}^{m-1}\frac{\left(1+\frac{u-j}{n+k}\right)_L}{\left(1+\frac{u-j+k}{n+k}\right)_L},\\
&\text{where }\overline{\Lambda}_{i,j}=\Lambda_{i,j}-kL\text{ for all possible }i,j.
\end{split}
\end{align}
\end{remark}

\subsection{Semi-infinite Gelfand-Tsetlin bases of \texorpdfstring{$\widetilde{\mathcal H}^{(r)}_{\infty}(n,k)$}{H(n,k)}}\label{subsec:semi-infinite GZ}

By the Theorem \ref{thm:Yangian comptible with p_N and sigma_N}, the transition map $p_N$ is a Yangian module map with respect to the Yangian action given by $\widetilde{T}(u)$. Therefore $p_N$ maps common eigenvectors of $\{\widetilde{A}_m(u)\}_{1\le m\le n}$ to common eigenvectors of $\{\widetilde{A}_m(u)\}_{1\le m\le n}$. It can be read out from the eigenvalue \eqref{eq:eigenvalue_finite N_twisted} that
\begin{align}\label{transition map of GT basis}
    p_N(V_{\Lambda})=\begin{cases}
        V_{\Lambda^{\natural \shortdownarrow}} &, \text{ if $\Lambda_1$ is cuttable},\\
        0 &, \text{ otherwise},
    \end{cases}
\end{align}
where $\Lambda^{\natural \shortdownarrow}=(\Lambda_1^{\natural \shortdownarrow},\cdots,\Lambda_{n+1}^{\natural \shortdownarrow})$. This motivates our next definition.

\begin{definition}\label{def:semi-inf GT pattern}
A semi-infinite GT pattern of height $n$ and type $r$ is an $(n+1)$-tuple of integer arrays $\Lambda=(\Lambda_1,\cdots,\Lambda_{n+1})$, $\Lambda_i=(\Lambda_{i,1}\ge \Lambda_{i,2}\ge \cdots )\in \bZ^{\infty}$ such that
\begin{align}\label{semi-inf GT condition 1}
    \Lambda_1\overset{1}{\searrow}\Lambda_2\overset{1}{\searrow}\cdots \overset{1}{\searrow}\Lambda_{n}\overset{1}{\searrow}\Lambda_{n+1}=\Lambda_1^\shortdownarrow,
\end{align}
and that $\Lambda$ has the following asymptotic behaviour:
\begin{align}\label{semi-inf GT condition 2}
    \Lambda_{i,j}=-k\lfloor\frac{j+i-2-r}{n}\rfloor \text{ for }j\gg 0.
\end{align}
We denote the set of semi-infinite GT patterns of height $n$ and type $r$ by $\mathrm{GT}^{(r)}_\infty(n)$. And we define $\Lambda^{\mathrm{vac}}$ to be the semi-infinite GT pattern which saturates the condition \eqref{semi-inf GT condition 2}, i.e. 
\begin{align}
     \Lambda^{\mathrm{vac}}_{i,j}=-k\lfloor\frac{j+i-2-r}{n}\rfloor \text{ for all }i,j.
\end{align}
\end{definition}

The semi-infinite GT patterns uniformizes non-negative GT patterns in the following sense. For arbitrary $L\in \bN$, there exists a surjective map $\mathrm{GT}^{(r)}_\infty(n)\twoheadrightarrow \mathrm{GT}(nL+r,n)$ given by:
\begin{align}
f_L: \mathrm{GT}^{(r)}_\infty(n)\ni \Lambda \mapsto (\Lambda_{i,j}+kL)_{1\le i\le n+1}^{1\le j\le nL+r}\in \mathrm{GT}(nL+r,n).
\end{align}
Moreover, for a fixed element $\Lambda\in  \mathrm{GT}^{(r)}_\infty(n)$, $f_L(\Lambda)$ belongs to $\mathrm{GT}(nL+r,n)_+$ for sufficiently large $L$, due to the asymptotic behaviour \eqref{semi-inf GT condition 2}. 

\begin{definition}
In $\widetilde{\mathcal H}_{nL+r}(n,k)$, for an element $\Lambda\in  \mathrm{GT}^{(r)}_\infty(n)$, we set 
\begin{align}
V_{\Lambda}:=\begin{cases}
    V_{f_L(\Lambda)} &, \text{ if }f_L(\Lambda)\in \mathrm{GT}(nL+r,n)_+,\\
    0 &, \text{ otherwise}.
\end{cases}
\end{align}
\end{definition}

Using the above notation, \eqref{transition map of GT basis} can be compactly written as $p_{nL+r}(V_{\Lambda})=V_{\Lambda}$ for a $\Lambda\in \mathrm{GT}^{(r)}_\infty(n)$. Taking inverse limit $L\to \infty$, we get a decomposition of $\widetilde{\mathcal H}^{(r)}_\infty(n,k)$ into one-dimensional subspaces labelled by $\mathrm{GT}^{(r)}_\infty(n)$.

\begin{theorem}[Spectral Decomposition of $\widetilde{\mathcal H}^{(r)}_{\infty}(n,k)$]\label{thm:spectral decomposition}
The action of the quantum minors \newline $\{\widetilde{A}_m(u)\}_{1\le m\le n}$ on $\widetilde{\mathcal H}^{(r)}_{\infty}(n,k)$ has simple spectrum, such that the eigenvectors are labelled by semi-infinite GT patterns of height $n$ and type $r:$
\begin{align}\label{eq:spectral decomposition}
    \widetilde{\mathcal H}^{(r)}_\infty(n,k)=\bigoplus_{\Lambda\in \mathrm{GT}^{(r)}_{\infty}(n)}V_{\Lambda},\quad \dim V_{\Lambda}=1,
\end{align}
and the eigenvalue of $\widetilde{A}_m(u)$ on $V_{\Lambda}$ is
\begin{align}\label{eq:eigenvalue}
    \prod_{i=1}^r\left(\frac{u-r-\Lambda_{m+1,i}+i}{u-r-\Lambda_{1,i}+i}\right)\times \prod_{i=r+1}^\infty\left(\frac{u-r-\Lambda_{m+1,i}+i}{u-r-\Lambda^{\mathrm{vac}}_{m+1,i}+i}\times \frac{u-r-\Lambda^{\mathrm{vac}}_{1,i}+i}{u-r-\Lambda_{1,i}+i}\right).
\end{align}
\end{theorem}

\begin{proof}
We have constructed the decomposition \eqref{eq:spectral decomposition} in the above discussions, let us prove \eqref{eq:eigenvalue}. We observe that
\begin{align}
    \prod_{i=r+1}^{nL+r}\frac{u-r-\Lambda^{\mathrm{vac}}_{1,i}+i}{u-r-\Lambda^{\mathrm{vac}}_{m+1,i}+i}=\prod_{j=0}^{m-1}\frac{\left(1+\frac{u-j}{n+k}\right)_L}{\left(1+\frac{u-j+k}{n+k}\right)_L}.
\end{align}
Therefore, for a sufficiently large $L$, after substituting $\Lambda$ in \eqref{eq:eigenvalue_finite N_twisted} with the truncation $f_L(\Lambda)$, the eigenvalue of $\widetilde{A}_m(u)$ on $V_{\Lambda}$ is
\begin{align}\label{eq:eigenvalue large L}
    \prod_{i=1}^r\left(\frac{u-r-\Lambda_{m+1,i}+i}{u-r-\Lambda_{1,i}+i}\right)\times \prod_{i=r+1}^{nL+r}\left(\frac{u-r-\Lambda_{m+1,i}+i}{u-r-\Lambda^{\mathrm{vac}}_{m+1,i}+i}\times \frac{u-r-\Lambda^{\mathrm{vac}}_{1,i}+i}{u-r-\Lambda_{1,i}+i}\right).
\end{align}
The eigenvalue \eqref{eq:eigenvalue large L} stabilizes for $L\gg 0$ because $\Lambda_{i,j}=\Lambda^{\mathrm{vac}}_{i,j}$ for $j\gg 0$, therefore we can take $L\to \infty$ and get \eqref{eq:eigenvalue}.
\end{proof}

\begin{definition}\label{def:SYD}
For two sequences of non-increasing integers $\beta=(\beta_1\ge \beta_2\ge\cdots)$ and $\gamma=(\gamma_1\ge \gamma_2\ge\cdots)$ such that $\beta_i\ge \gamma_i$ for all $i$, we define the skew Young diagram associated to $\beta/\gamma$ to be
\begin{align}
    \mathrm{SYD}(\beta/\gamma):=\left\{(i,j)\in\bZ^2\:|\: i\ge 1, \beta_i\ge j>\gamma_i\right\}.
\end{align}
We represent $(i,j)\in\mathrm{SYD}(\beta/\gamma)$ pictorially as a unit box in $\bR^2$ with the center $(i,j)$, the coordinates $i$ and $j$ on $\bR^2$ increasing from top to bottom and from left to right respectively. For example:
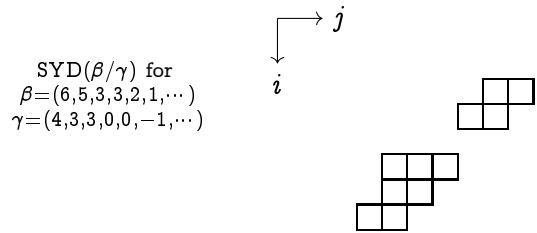
\begin{equation*}
\substack{\mathrm{SYD}(\beta/\gamma)\text{ for }\\ \beta=(6,5,3,3,2,1,\cdots)\\ \gamma=(4,3,3,0,0,-1,\cdots)}\qquad
\begin{tikzpicture}[scale=0.6]
    \draw[->] (0,0) -- (1,0) node[right]{$j$};
    \draw[->] (0,0) -- (0,-1) node[below]{$i$};
\end{tikzpicture}
\ytableausetup{smalltableaux}
\ydiagram{5+2,4+2,0,1+3,1+2,2}
\end{equation*}
\end{definition}

As we have seen in the Remark \ref{rmk:simple Y(gl_n) submodules}, for an admissible $\lambda$, $\widetilde{\mathcal H}_N(\lambda)$ is a simple $Y(\gl_n)$ submodule of $\widetilde{\mathcal H}_N(n,k)$. According to Theorem \ref{thm:Yangian comptible with p_N and sigma_N} the map $p_N:\widetilde{\mathcal H}_{N+n}(\lambda)\cong \widetilde{\mathcal H}_N(\lambda^{\natural\shortdownarrow})$ is a Yangian module isomorphism for a cuttable $\lambda$. Taking inverse limit, we arrive at the following decomposition result.

\begin{theorem}\label{thm:Yangian irrep decomposition}
$\widetilde{\mathcal H}^{(r)}_{\infty}(n,k)$ decomposes into simple $Y(\gl_n)$-modules:
\begin{align}\label{irrep decomposition}
    \widetilde{\mathcal H}^{(r)}_{\infty}(n,k)=\bigoplus_{\substack{\lambda=(\lambda_1\ge \lambda_2\ge \cdots)\in \bZ^\infty\\  \lambda_{j}=-k\lfloor\frac{j-1-r}{n}\rfloor \text{ for }j\gg 0.}} \widetilde{\mathcal H}(\lambda).
\end{align}
The Drinfeld polynomials of $\widetilde{\mathcal H}(\lambda)$ are
\begin{align}\label{Drinfeld poly}
    P_m(u)=\prod_{\substack{\text{$(i,j)$ is a top box of}\\ \text{a height $m$ column}\\ \text{in }\mathrm{SYD}(\lambda/\lambda^\shortdownarrow)}}(u-r+i-j),\quad m=1,\cdots, n-1.
\end{align}
The eigenvector decomposition of $\widetilde{\mathcal H}(\lambda)$ with respect to quantum minors $\{\widetilde A_m(u)\}_{1\le m\le n}$ is 
\begin{align}\label{eq:spectral decomposition_irrep}
    \widetilde{\mathcal H}(\lambda)=\bigoplus_{\substack{\Lambda\in \mathrm{GT}^{(r)}_{\infty}(n)\\ \Lambda_1=\lambda}}V_{\Lambda}.
\end{align}
\end{theorem}

\begin{proof}
The decomposition \eqref{irrep decomposition} follows from taking inverse limit, namely $\widetilde{\mathcal H}(\lambda)=\underset{\substack{\longleftarrow \\ L}}{\lim}\:\widetilde{\mathcal H}_{nL+r}(\lambda)$. The eigenvector decomposition \eqref{eq:spectral decomposition_irrep} follows from Theorem \ref{thm:spectral decomposition}. The Drinfeld polynomials \eqref{Drinfeld poly} essentially follow from \cite[Corollary 2.13]{nazarov1998representations}. One should be careful that $\widetilde{\mathcal H}_N(\lambda)$ is obtained from $V_{\lambda/\lambda^\shortdownarrow}$ by first taking the dual, then shifting spectral parameter, finally multiplying a function $f(u)$. In principal one can chase along the Yangian automorphisms and deduce \eqref{Drinfeld poly} from \cite[Corollary 2.13]{nazarov1998representations}. Here we provide a straightforward proof. 

By definition, we have
\begin{align}\label{def:Drinfeld poly}
    \frac{\widetilde{A}_{m-1}(u-1)\widetilde{A}_{m+1}(u)}{\widetilde{A}_{m}(u-1)\widetilde{A}_{m}(u)}\bigg\rvert_{\text{highest weight vector}}=\frac{P_m(u-1)}{P_m(u)}.
\end{align}
By Remark \ref{rmk:simple module and hwv}, the highest weight vector is $V_{\Lambda^0}$ where $\Lambda^0$ is given by \eqref{eq:hwv}. According to Theorem \ref{thm:spectral decomposition}, we have
\begin{align}\label{eq:Drinfeld poly_1}
\begin{split}
\frac{\widetilde{A}_{m-1}(u-1)\widetilde{A}_{m+1}(u)}{\widetilde{A}_{m}(u-1)\widetilde{A}_{m}(u)}\bigg\rvert_{V_{\Lambda^0}}&=\prod_{i=1}^\infty\left(\frac{u-1-r-\Lambda^0_{m,i}+i}{u-1-r-\Lambda^0_{m+1,i}+i}\times \frac{u-r-\Lambda^0_{m+2,i}+i}{u-r-\Lambda^0_{m+1,i}+i}\right)\\
\text{\small for $L\gg 0$}\quad &=\prod_{i=1}^{nL+r}\left(\frac{u-1-r-\Lambda^0_{m,i}+i}{u-1-r-\Lambda^0_{m+1,i}+i}\times \frac{u-r-\Lambda^0_{m+2,i}+i}{u-r-\Lambda^0_{m+1,i}+i}\right)\\
&=\frac{Q_m(u-1)_L}{Q_m(u)_L}\times \frac{Q_{m+1}(u+1)_L}{Q_{m+1}(u)_L},
\end{split}
\end{align}
where 
\begin{align}
    Q_m(u)_L=\prod_{\substack{(i,j)\in \mathrm{SYD}(\Lambda^0_m/\Lambda^0_{m+1})\\ i\le nL+r}}(u-r+i-j).
\end{align}
Using the formula \eqref{eq:hwv} for $\Lambda^0$, we observe that there is an equality between box configurations (see \cite[Lemma 2.6]{nazarov1998representations} for the dual statement):
\begin{align}\label{eq:top box of height m column}
    \mathrm{SYD}(\Lambda^0_m/\Lambda^0_{m+1})= 
    \left\{(i,j)\:|\: (i,j) \text{ is a top box of a height $m$ column in }\mathrm{SYD}(\Lambda^0_1/\Lambda^0_{m+1}) \right\}.
\end{align}
Whence it follows that:
\begin{align}\label{eq:Q_m(u)}
    Q_m(u)_L=\prod_{\substack{\text{$(i,j)$ is a top box of}\\ \text{a height $m$ column}\\ \text{in }\mathrm{SYD}(\Lambda^0_1/\Lambda^0_{m+1})\\ \text{and $i\le nL+r$}}}(u-r+i-j).
\end{align}
Using \eqref{eq:Q_m(u)} and \eqref{eq:top box of height m column}, one arrives at
\begin{align}\label{eq:P_m(u)}
    \frac{Q_m(u)_L}{Q_{m+1}(u+1)_L}=\prod_{\substack{\text{$(i,j)$ is a top box of}\\ \text{a height $m$ column}\\ \text{in }\mathrm{SYD}(\Lambda_1/\Lambda_{n+1})}}(u-r+i-j).
\end{align}
Plug \eqref{eq:P_m(u)} into \eqref{eq:Drinfeld poly_1}, and compare with \eqref{def:Drinfeld poly}, we conclude that $P_m(u)$ is given by \eqref{Drinfeld poly}.
\end{proof}

\begin{remark}
According to Corollary \ref{cor:affine gl(n) action on conformal limit}, $\widetilde{\mathcal H}^{(r)}_{\infty}(n,k)$ is isomorphic to $L_{k\varpi_r}(\widehat{\mathfrak{sl}}(n)_k)\otimes \mathrm{Fock}_{kr}(\widehat{\gl}(1)_{kn})$ as a $\widehat{\gl}(n)_k$ module. Theorem \ref{thm:Yangian irrep decomposition} describes the decomposition of $L_{k\varpi_r}(\widehat{\mathfrak{sl}}(n)_k)\otimes \mathrm{Fock}_{kr}(\widehat{\gl}(1)_{kn})$ with respect to the Yangian action given by $\prescript{\infty}{}{\widetilde{T}}(u)$. When $r=0$, $\widetilde{\mathcal H}^{(0)}_{\infty}(n,k)$ is the vacuum module for $\widehat{\gl}(n)_k$, and in this case our result recovers Uglov's result in \cite[Theorem 1.2]{uglov1998yangian}.
\end{remark}

\subsection{Solving a higher-spin generalization of Calogero-Sutherland model}

Recall that in Section \ref{subsec:Calogero rep} we defined a higher-spin analog of Calogero-Sutherland Hamiltonian:
\begin{align}
    H_{\mathrm{CS}}=\sum_{i=1}^N\Delta^{-1}(x_i\partial_{i})^2\Delta-2\sum_{i<j}^N\frac{x_ix_j(\Omega_{ij}+k)}{(x_i-x_j)^2}-(N-1)\sum_{i=1}^Nx_i\partial_i-\frac{N(N-1)(2N-1)}{6}.
\end{align}
$H_{\mathrm{CS}}$ is the Calogero representation of $\mathrm{Tr}((XY)^2)$ (Lemma \ref{lem:Calogero-Sutherland}). Using the technique that we developed in previous sections, we can completely solve this Hamiltonian.

\begin{theorem}\label{thm:spectrum of H_CS}
The eigenspace decomposition of $\mathcal H_N(n,k)$ with respect to the action of $H_{\mathrm{CS}}=\mathrm{Tr}((XY)^2)$ is given by
\begin{align}
    \mathcal H_N(n,k)=\bigoplus_{\lambda} \widetilde{\mathcal H}_N(\lambda),
\end{align}
where the sum is taken for all $\lambda=(\lambda_1\ge\cdots\ge\lambda_N)\in \bN^N$ such that $\lambda$ is admissible (Definition \ref{def:admissible and cuttable}), and $\widetilde{\mathcal H}_N(\lambda)$ is defined in Definition \ref{def:H(la)}. The eigenvalue of $H_{\mathrm{CS}}$ on $\widetilde{\mathcal H}_N(\lambda)$ is
\begin{align}
    \sum_{i=1}^N (\lambda_i-k)(\lambda_i-k+N+1-2i).
\end{align}
\end{theorem}

\begin{proof}
Since $\widetilde{\mathcal H}_N(\lambda)=\left(V^{N+n}_{\lambda}\otimes (V^{N}_{\lambda})^*\right)^{\GL_N^{\mathrm{diag}},-k}$, $\widetilde{\mathcal H}_N(\lambda)$ is a subspace of the summand $M^{\lambda}_{\mu,\lambda^\shortdownarrow}\otimes V^{n}_{\mu}\otimes V^{N}_{\lambda^\shortdownarrow}\otimes (V^{N}_{\lambda})^*$ in the decomposition
\begin{align}
    V^{N+n}_{\lambda}\otimes (V^{N}_{\lambda})^*=\bigoplus_{\mu=(\mu_1\ge\cdots\ge \mu_n)\in \bN^n}M^{\lambda}_{\mu,\lambda^\shortdownarrow}\otimes V^{n}_{\mu}\otimes V^{N}_{\lambda^\shortdownarrow}\otimes (V^{N}_{\lambda})^*.
\end{align}
We note that $\gl_N$ acts on $V^{N}_{\lambda^\shortdownarrow}$ in the above decomposition via the generator $XY$. Therefore the quadratic Casimir of $\gl_N$, which is $\mathrm{Tr}((XY)^2)$, acts on $\widetilde{\mathcal H}_N(\lambda)$ as the scalar\footnote{$2\rho_N=$ sum of positive roots of $\gl_N=(N-1,N-3,\cdots,3-N,1-N)$.}
\begin{align}
    \langle \lambda^\shortdownarrow,\lambda^\shortdownarrow+2\rho_N\rangle=\sum_{i=1}^N (\lambda_i-k)(\lambda_i-k+N+1-2i).
\end{align}
\end{proof}

\begin{remark}\label{rmk:ground states of H_CS}
The same argument as the proof of Theorem \ref{thm:spectrum of H_CS} shows that the energy grading operator $\mathrm{Tr}(XY)$ acts on $\widetilde{\mathcal H}_N(\lambda)$ as the scalar $d(\lambda)=\sum_{i=1}^{N}(\lambda_i-k)$, therefore $\widetilde{\mathcal H}_N(\lambda)$ is homogeneous of degree $d(\lambda)$ with respect to the original energy grading \eqref{eq:energy grading}. In particular, the lowest energy subspace with respect to the Hamiltonian $H_{\mathrm{CS}}$ is exactly the space of ground states $\mathcal H_N(n,k)_0$, that is,
\begin{align}
    \mathcal H_N(n,k)_0=\widetilde{\mathcal H}_N(\lambda_{\min}), \quad \lambda_{\min,j}=k\lfloor\frac{N+n-j}{n}\rfloor.
\end{align}
Using the shifted energy grading (see Definition \ref{def:p_N and sigma_N}), $\widetilde{\mathcal H}_N(\lambda)$ is homogeneous of degree
\begin{align}
    \widetilde{d}(\lambda)=\sum_{j=1}^N\left(\lambda_j-\lambda_{\min,j}\right)
\end{align}
Taking $N\to \infty$, the summand $\widetilde{\mathcal H}(\lambda)$ in \eqref{irrep decomposition} is homogeneous of degree
\begin{align}
    {d}_{\infty}(\lambda)=\sum_{j=1}^\infty\left(\lambda_j+k\lfloor\frac{j-1-r}{n}\rfloor\right)
\end{align}
with respect to the shifted energy grading.
\end{remark}

\section*{Acknowledgments}
The authors would like to thank Jean-Emile Bourgine, Ryosuke Kodera, Andrey Losev, and Yutaka Matsuo for stimulating discussions. Part of YZ's work was done during his visit at Shanghai Institute for Mathematics
and Interdisciplinary Sciences (SIMIS), he would like to thank the hospitality of SIMIS. Kavli IPMU is supported by World Premier International Research Center Initiative (WPI), MEXT, Japan. Work of SL is supported by the National Key R \& D Program of China (NO. 2020YFA0713000).
\appendix

\section{Hall-Littlewood Polynomial}\label{sec:hall-littlewood}
In this appendix we review the basics of the Hall-Littlewood polynomial and its transformed version, following \cite[Section 3]{haiman2003combinatorics}\footnote{The variable $q$ here is denoted by $t$ there.}.
\begin{definition}
For a Young diagram $\lambda=(\lambda_1\ge\cdots\ge \lambda_n)\in \bN^n$, we denote its associated partition by $(1^{\alpha_1},2^{\alpha_2},\cdots)$, then the Hall-Littlewood polynomial $P_{\lambda}(\mathbf a;q)$ in the variables $\mathbf a_1,\cdots,\mathbf a_n$ and $q$ is defined by the formula
\begin{align}
    P_{\lambda}(\mathbf a;q):=\frac{1}{\prod_{i\ge 0}[\alpha_i]_q!}\sum_{w\in S_n}w\left(\mathbf a^{\lambda}\prod_{i<j}\frac{1-q\mathbf a_j/\mathbf a_i}{1-\mathbf a_j/\mathbf a_i}\right).
\end{align}
Here $\alpha_0=n-\sum_{i\ge 1}\alpha_i$, and $\mathbf a^{\lambda}=\mathbf a_1^{\lambda_1}\cdots \mathbf a_n^{\lambda_n}$, and we use the following $q$-number notation
\begin{align*}
    [n]_q=\frac{1-q^n}{1-q},\;  [n]_q!=[n]_q[n-1]_q\cdots [1]_q,\; \left[\begin{matrix} n \\ k \end{matrix}\right]_q=\frac{[n]_q!}{[k]_q![n-k]_q!}.
\end{align*}
\end{definition}
The Hall-Littlewood polynomial $P_{\lambda}(\mathbf a;q)$ is an interpolation between Schur symmetric functions $s_{\lambda}(\mathbf a)$ and monomial symmetric functions $m_{\lambda}(\mathbf a)$, in fact we have 
\begin{align}\label{P(0) and P(1)}
    P_{\lambda}(\mathbf a;0)=s_{\lambda}(\mathbf a),\; P_{\lambda}(\mathbf a;1)=m_{\lambda}(\mathbf a).
\end{align}

\begin{definition}
For a pair of partitions $(\lambda,\mu)$, we define the Kostka-Foulkes functions $K_{\lambda\mu}(q)\in \mathbb Q(q)$ to be coefficient of the expansion
\begin{align}
    s_{\lambda}(\mathbf a)=\sum_{\mu}K_{\lambda \mu}(q)P_{\mu}(\mathbf a;q).
\end{align}
\end{definition}
It is true (but not obvious) that 
\begin{align*}
    K_{\lambda\mu}(q)\in \bZ_{\ge 0}[q],
\end{align*}
i.e. Kostka-Foulkes functions are polynomials in $q$ with non-negative integer coefficients, see \cite[Theorem 3.4.15]{haiman2003combinatorics}.

In \cite{jing1991vertex}, a transformed version of Hall-Littlewood polynomials is defined using the Jing operator, which we recall its definition here.
\begin{definition}
For an integer $m$, the $m$-th Jing operator $S^q_m$ is a linear map from $\bZ[\mathbf a_1,\cdots,\mathbf a_n,q]$ to itself, defined by
\begin{align}\label{Jing Operator}
    \bZ[\mathbf a_1,\cdots,\mathbf a_n,q]\ni f(\mathbf a;q)\mapsto (S^q_mf)(\mathbf a;q):=\sum_{i=1}^nf(\mathbf a_1,\cdots,q\mathbf a_i,\cdots,\mathbf a_n; q)\frac{\mathbf a_i^m}{\prod_{j\neq i}(1-\mathbf a_j/\mathbf a_i)}.
\end{align}
\end{definition}

\begin{remark}
The original definition of Jing operator in \cite{jing1991vertex} (see \cite[Definition 3.4.5]{haiman2003combinatorics}) is in terms of generating series $S^q(u)=\sum_{m\in \bZ} S^q_m u^m$. In other words, $S^q(u)$ is defined beforehand, and $S^q_m$ in \eqref{Jing Operator} is read out from the mode expansion of $S(u)$, see \cite[Lemma A.1]{moosavian2021towards}. The operator $S(u)$ is constructed such that it maps symmetric functions to symmetric functions, i.e. $S^q_m$ maps $\bZ[\mathbf a_1,\cdots,\mathbf a_n]^{S_n}[q]$ to itself for all $m\in \bZ$.
\end{remark}

\begin{definition}
Let $\mu=(\mu_1\ge\cdots \ge \mu_l)\in \bN^l$ be a Young diagram, define the transformed Hall-Littlewood polynomial $H_{\mu}(\mathbf a;q)\in \bZ[\mathbf a_1,\cdots,\mathbf a_n]^{S_n}[q]$ by
\begin{align}\label{definition: transformed HL}
    H_{\mu}(\mathbf a;q)=S^q_{\mu_1}S^q_{\mu_2}\cdots S^q_{\mu_l}(1).
\end{align}
\end{definition}

\begin{theorem}[{see \cite[Corollary 3.4.12]{haiman2003combinatorics}}]
$H_{\mu}(\mathbf a;q)$ is related to the Schur functions by
\begin{align}
    H_{\mu}(\mathbf a;q)=\sum_{\lambda}K_{\lambda \mu}(q)s_{\lambda}(\mathbf a).
\end{align}
\end{theorem}

\section{Affine Grassmannian}\label{sec:affine Gr}

In this appendix we review the basics of affine Grassmannian and the Beilinson-Drinfeld Grassmannian, following \cite{zhu2016introduction}. Then we give a geometric description of the Jing operator $S^q_m$ in terms of convolution product of coherent sheaves on affine Grassmannian for $\GL_n$ in Section \ref{subsec:Geometric description of Jing operator}.

\smallskip Let us denote $\mathscr K=\bC(\!(z)\!)$ and $\mathscr O=\bC[\![z]\!]$. For a reductive group $G$, we define the affine Grassmannian 
\begin{align}
    \Gr_G:=G(\mathscr K)/G(\mathscr O),
\end{align}
i.e. the coset of group of Laurent power series valued in $G$ modulo subgroup of formal power series valued in $G$. $G(\mathscr O)$ acts from the left on $\Gr_G$, and $G(\mathscr O)$-orbits are one-to-one correspond to dominant coweights of $G$ \cite[(2.1.2)]{zhu2016introduction}:
\begin{align}\label{Cartan decomposition}
    \Gr_G=\coprod_{\lambda\in \Lambda_G^+} \Gr^\lambda_G,\quad \text{where }\;\Gr^\lambda_G:=G(\mathscr O) z^\lambda G(\mathscr O)/G(\mathscr O).
\end{align}
Here $\Lambda_G$ is the coweight lattice of $G$ and $\Lambda_G^+$ is its dominant part. We note that $\mathbb G_m(\mathscr K)$ has a distinct point which corresponds to the completion morphism $\bC[z^\pm]\to \bC(\!(z)\!)$, and $z^\lambda$ is the image of the aforementioned distinct point in $G(\mathscr K)$ under the morphism $\lambda:\mathbb G_m\to G$.

$\Gr_G$ is endowed with an ind-scheme structure, i.e. $\Gr_G$ is the $\bC$-points of a contravariant functor (still denote it by $\Gr_G$) from the category of $\bC$-algebras to the category of sets, such that $\Gr_G$ is isomorphic to the functor $\underset{\substack{\longrightarrow\\ i\in I}}{\lim}X_i$ for some index set $I$ and schemes $X_i$ where $X_i\to X_j$ is a closed embedding for every arrow $i\to j$ in $I$. The functor $\Gr_G$ is defined by 
\begin{align}\label{Gr functor local}
    \bC\text{-Alg}\ni R\mapsto\Gr_G(R)=\left\{(\mathcal E,\phi)\:\bigg\rvert\: \substack{\mathcal E\text{ is a principal $G$-bundle on }\mathbb D_R\\ \phi:\mathcal E|_{\mathbb D^*_R}\cong \mathcal E^0|_{\mathbb D^*_R}\text{ is a trivialization}}\right\}.
\end{align}
Here $\mathbb D_R:=\Spec R[\![z]\!]$ and $\mathbb D^*_R:=\Spec R(\!(z)\!)$, and $\mathcal E^0$ is the trivial principal $G$-bundle on $\mathbb D_R$. 

Using the aforementioned ind-scheme structure, the $G(\mathscr O)$-orbit $\Gr^\lambda_G$ in \ref{Cartan decomposition} is a finite dimensional smooth subvariety and its Zariski closure $\bGr^\lambda_G$ (endowed with reduced scheme structure) is a projective variety \cite[Proposition 2.1.5]{zhu2016introduction}, which is called the (spherical) Schubert variety. Moreover, it is known that \cite[Proposition 2.1.4]{zhu2016introduction}:
\begin{align}
    \bGr^\lambda_G=\bigcup_{\mu\le \lambda}\Gr^\mu_G.
\end{align}
Here $\le$ is the Bruhat order, that is, $\mu\le \lambda$ if and only if $\lambda-\mu$ is a non-negative integral linear combinations of simple coroots.

We note that $\mathbb D_R$ can be regarded as a $\Spec R$-family of formal disks which is the germ of a smooth curve. Using a theorem of Beauville-Laszlo, the functor \eqref{Gr functor local} can be extend to a smooth curve from germ of a point. Namely, we have functorial isomorphism \cite[Theorem 1.4.2]{zhu2016introduction}
\begin{align}\label{Gr functor global}
\Gr_G(R)\cong\left\{(\mathcal E,\phi)\:\bigg\rvert\: \substack{\mathcal E\text{ is a principal $G$-bundle on }C_R\\ \phi:\mathcal E|_{C^*_R}\cong \mathcal E^0|_{C^*_R}\text{ is a trivialization}}\right\},
\end{align}
where $C$ is a smooth algebraic curve with a distinct pint $x\in C$ and $C^*:=C\setminus x$.

If we generalize \eqref{Gr functor global} to allow multiple points and furthermore allow them to vary and collide, we arrive at the following.
\begin{definition}
Let $C$ be a smooth curve and $I$ be a finite index set, then the Beilinson-Drinfeld Grassmannian $\Gr_{G,C^I}$ is defined to be the following contravariant functor
\begin{align}\label{BD Gr functor}
    \bC\text{-Alg}\ni R\mapsto\Gr_{G,C^I}(R)=\left\{(\mathcal E,D,\phi)\:\bigg\rvert\: \substack{D\in C^I(R),\;\mathcal E\text{ is a principal $G$-bundle on }C_R\\ \phi:\:\mathcal E|_{C_R\setminus D}\cong \mathcal E^0|_{C_R\setminus D}\text{ is a trivialization}}\right\}.
\end{align}
Here we identify $D\in C^I(R)$ with an $I$-colored divisor on $C_R$. We define the symmetrized Beilinson-Drinfeld Grassmannian $\Gr_{G,C^{(I)}}$ to be the functor \eqref{BD Gr functor} with $D\in C^I(R)$ replaced by $D\in \mathrm{Sym}^{|I|}C(R)$, that is, we identify $D$ with an divisor of degree $|I|$ on $C_R$.

We note that there are natural projections:
\begin{align}
    \Gr_{G,C^I}\to C^I,\text{ and }\Gr_{G,C^{(I)}}\to C^{(I)}=\mathrm{Sym}^{|I|}C,
\end{align}
which map $(\mathcal E,D,\phi)$ to $D$. We call them the structure maps for (symmetrized) Beilinson-Drinfeld Grassmannian.
\end{definition}

It is known that $\Gr_{G,C^I}$ and $\Gr_{G,C^{(I)}}$ are ind-schemes \cite[Theorem 3.1.3]{zhu2016introduction}. 

\smallskip $\Gr_{G,C^I}$ satisfies the following factorization property \cite[Theorem 3.2.1(iii)]{zhu2016introduction}. Given $f:J\twoheadrightarrow I$, let $J=\coprod_{i\in I}J_i$ denote the partition of $J$, and $C^{J/I}\subset C^J$ denote the open subset of those $\{x_j,j\in J\}$ such that $x_j\cap x_{j'}=\emptyset$ if $f(j)\neq f(j')$. Then there is a canonical isomorphism
\begin{align}\label{factorization}
    \Gr_{G,C^J}\times_{C^J} C^{J/I}\cong \left(\prod_{i\in I} \Gr_{G,C^{J_i}}\right)\times_{C^J} C^{J/I}.
\end{align}
In particular, for the identity map $J\to J$ we have $C^{J/J}=C^J\setminus \Delta$ where $\Delta$ is the divisor of those $\{x_j,j\in J\}$ such that $x_j=x_{j'}$ for some $j\neq j'$, and \eqref{factorization} reads:
\begin{align}\label{factorization_main diagonal}
    \Gr_{G,C^J}\times_{C^J} (C^J\setminus \Delta)\cong \left(\Gr_{G,C}\right)^J\times_{C^J} (C^J\setminus \Delta).
\end{align}
\eqref{factorization_main diagonal} also holds for symmtrized version:
\begin{align}\label{factorization_main diagonal_sym}
    \Gr_{G,C^{(J)}}\times_{C^{(J)}} (C^{(J)}\setminus \Delta)\cong \left(\Gr_{G,C}\right)^{(J)}\times_{C^{(J)}} (C^{(J)}\setminus \Delta),
\end{align}
where we still use $\Delta$ to denote its image in the symmetrization $C^{(J)}$. In the case of $C=\bA^1$, the coordinate $z$ on $\bA^1$ induces a coordinate for the formal disk $\mathbb D_x$ for every $x\in \bA^1$, so we have $\Gr_{G,\bA^1}\cong \Gr_G\times\bA^1$. Therefore we have
\begin{align}
\begin{split}
    \left(\Gr_{G,\bA^1}\right)^J\times_{\bA^J} (\bA^J\setminus \Delta)&\cong \left(\Gr_{G}\right)^J\times(\bA^J\setminus \Delta),\\
   \left(\Gr_{G,\bA^1}\right)^{(J)}\times_{\bA^{(J)}} (\bA^{(J)}\setminus \Delta)&\cong \left(\Gr_{G}\right)^J\times(\bA^{(J)}\setminus \Delta).
\end{split}
\end{align}

\begin{definition}\label{BD Schubert}
For a $I$-tuple of dominant coweights $\{\lambda_i,i\in I\}$ We define the Beilinson-Drinfeld Schubert variety $\bGr^{\{\lambda_i,i\in I\}}_{G,\bA^I}$ to be the Zariski closure of $\left(\prod_{i\in I}\bGr^{\lambda_i}_{G}\right)\times(\bA^I\setminus \Delta)$ in $\Gr_{G,\bA^I}$, endowed with reduced scheme structure. The symmetrized version $\bGr^{\{\lambda_i,i\in I\}}_{G,\bA^{(I)}}$ is defined to be the Zariski closure of $\left(\prod_{i\in I}\bGr^{\lambda_i}_{G}\right)\times(\bA^{(I)}\setminus \Delta)$ in $\Gr_{G,\bA^{(I)}}$, endowed with reduced scheme structure.
\end{definition}

\begin{remark}
For a general smooth curve $C$, one can define the Beilinson-Drinfeld Schubert varieties $\bGr^{\{\lambda_i,i\in I\}}_{G,\bA^I}$ and its symmetrized version $\bGr^{\{\lambda_i,i\in I\}}_{G,\bA^{(I)}}$ using a twist of the direct product $\left(\prod_{i\in I}\bGr^{\lambda_i}_{G}\right)\times(\bA^I\setminus \Delta)$ by a principal $\Aut(\mathbb D)^I$-bundle, see \cite[(3.1.10)]{zhu2016introduction}.
\end{remark}

\subsection{The affine Grassmannian for \texorpdfstring{$\GL_n$}{GL(n)}}\label{subsec:Gr fro GL(n)}

For $G=\GL_n$, a principal $\GL_n$-bundle is equivalent to a rank $n$ vector bundle, so we can rewrite \eqref{Gr functor local} and \eqref{BD Gr functor} as follows:
\begin{gather}
\Gr_{\GL_n}(R)=\left\{\mathcal E\:\bigg\rvert\: \substack{\mathcal E\subset j_* \mathcal O^{\oplus n}_{\mathbb D^*_R} \text{ is a rank $n$  } \\ \text{locally free sub-sheaf}}\right\},\; j:\mathbb D^*_R\hookrightarrow \mathbb D_R\text{ is the embedding},\\
\Gr_{\GL_n,C^I}(R)=\left\{(\mathcal E,D)\:\bigg\rvert\: \substack{D\in C^I(R),\; \mathcal E\subset j^D_*\mathcal O^{\oplus n}_{C_R\setminus D}\text{ is a }\\ \text{rank $n$ locally free sub-sheaf }}\right\},\; j^D:C_R\setminus D\hookrightarrow C_R\text{ is the embedding}.
\end{gather}
In particular, a $\bC$-point in $\Gr_{\GL_n}$ represents a lattice $\Lambda\subset \mathscr K^{\oplus n}$, that is, $\Lambda$ is required to be a rank $n$ free module over $\mathscr O$. 

Using Gauss elimination algorithm, one can show that there exists $g\in \GL_n(\mathscr O)$ such that $g\cdot\Lambda=\bigoplus_{i=1}^n z^{\lambda_i}\mathscr O$ where $\lambda=(\lambda_1\ge \cdots\ge \lambda_n)\in \bZ^n$ is a dominant coweight of $\GL_n$. $\lambda$ is uniquely determined by $\Lambda$, and this is exactly the decomposition \eqref{Cartan decomposition} in the case when $G=\GL_n$.

\begin{example}
$\Gr^{\omega_1}_{\GL_n}$ is defined to be the $\GL_n(\mathscr O)$-orbit through $z^{\omega_1}$ which represents the lattice $z\mathscr O\oplus\mathscr O^{\oplus(n-1)}$. Note that $(z\mathscr O)^{\oplus n}\subset z\mathscr O\oplus\mathscr O^{\oplus(n-1)}\subset \mathscr O^{\oplus n}$, and the $\GL_n(\mathscr O)$-action preserves these inclusions, that is, $(z\mathscr O)^{\oplus n}\subset \Lambda\subset \mathscr O^{\oplus n}$ for every $\Lambda\in \Gr^{\omega_1}_{\GL_n}$. Therefore $\Lambda$ is uniquely determined by $n-1$ dimensional linear subspace of $(\mathscr O/z\mathscr O)^{\oplus n}=\bC^n$, whence $\Gr^{\omega_1}_{\GL_n}\cong \bP^{n-1}$. In particular, we see that $\Gr^{\omega_1}_{\GL_n}$ is proper, thus $\Gr^{\omega_1}_{\GL_n}=\bGr^{\omega_1}_{\GL_n}$.
\end{example}

\begin{example}
The previous example shows that $\Gr^{\omega_1}_{\GL_n}$ is the moduli space of $\mathcal E\subset \mathscr O^{\oplus n}$ such that $\dim \mathscr O^{\oplus n}/\mathcal E=1$. This can be generalized to $\bGr^{N\omega_1}_{\GL_n}$: it is the moduli space of $\mathcal E\subset \mathscr O^{\oplus n}$ such that $\dim \mathscr O^{\oplus n}/\mathcal E=N$ \cite[Example 2.1.8]{zhu2016introduction}.
\end{example}

A connected component of $\Gr_{\GL_n}$ is indexed by an integer $N\in \bZ$, and is denoted by $\Gr^{(N)}_{\GL_n}$. It is known that a lattice $\Lambda\subset\mathscr K^{\oplus n}$ is in $\Gr^{(N)}_{\GL_n}$ if and only if \cite[4.1]{mirkovic2007quiver}
\begin{align*}
    \dim (\mathscr O^{\oplus n}/z^{M}\Lambda)=N+Mn,\quad\forall M\gg 0.
\end{align*}
In terms of $\GL_n(\mathscr O)$-orbit, we have
\begin{align}
    \Gr^{(N)}_{\GL_n}=\bigcup_{\substack{\lambda=(\lambda_1\ge \cdots\ge\lambda_n)\in \bZ^n \\ \sum_{i=1}^n \lambda_i=N}}\Gr^{\lambda}_{\GL_n}.
\end{align}

\subsection{Geometric description of Jing operator}\label{subsec:Geometric description of Jing operator}
In this section we give a geometric description of Jing operators $S^q_m$, and outine a proof of \cite[Corollary B.3]{moosavian2021towards}.

Consider the affine Grassmannian $\Gr_{\GL_n}=\GL_n(\mathscr K)/\GL_n(\mathscr O)$, and let $\omega_1=(1,0,\cdots,0)$ be the first fundamental coweight of $\GL_n$, then the $\GL_n(\mathscr O)$-orbit $\Gr^{\omega_1}$ is isomorphic to $\mathbb P^{n-1}$ and it is fixed by the $\mathbb C^{\times}_q$-rotation.

There is a convolution product on affine Grassmannian:
\begin{align}\label{eqn: convolution}
    m:\Gr_{\GL_n}\widetilde{\times} \Gr_{\GL_n}=\GL_n(\mathscr K)\overset{\GL_n(\mathscr O)}{\times}\GL_n(\mathscr K)/\GL_n(\mathscr O)\to \GL_n(\mathscr K)/\GL_n(\mathscr O).
\end{align}
Here the map sends $(g_1,g_2)$ to $g_1g_2$.

Consider $D^b_{\GL_n(\mathscr O)\rtimes \mathbb C^{\times}_q}(\Gr_{\GL_n})$, the $\GL_n(\mathscr O)\rtimes \mathbb C^{\times}_q$-equivariant bounded derive category of coherent sheaves on $\Gr_{\GL_n}$. Here coherent sheaves on the ind-scheme $\Gr_{\GL_n}$ are defined to have finite type support. The convolution map of $\Gr_{\GL_n}$ induces a functor 
\begin{gather*}
    \star : D^b_{\GL_n(\mathscr O)\rtimes \mathbb C^{\times}_q}(\Gr_{\GL_n})\times D^b_{\GL_n(\mathscr O)\rtimes \mathbb C^{\times}_q}(\Gr_{\GL_n})\to D^b_{\GL_n(\mathscr O)\rtimes \mathbb C^{\times}_q}(\Gr_{\GL_n}),\\
    \mathcal F\star \mathcal G=Rm_*(\mathcal F\widetilde{\boxtimes} \mathcal G).
\end{gather*}
Passing to the $K$-theory, we obtain an map 
\begin{align}
    \star: K_{\GL_n(\mathscr O)\rtimes \mathbb C^{\times}_q}(\Gr_{\GL_n})\otimes K_{\GL_n(\mathscr O)\rtimes \mathbb C^{\times}_q}(\Gr_{\GL_n})\longrightarrow K_{\GL_n(\mathscr O)\rtimes \mathbb C^{\times}_q}(\Gr_{\GL_n}).
\end{align}
It is known that the $\star$-product on $K_{\GL_n(\mathscr O)\rtimes \mathbb C^{\times}_q}(\Gr_{\GL_n})$ is associative, and the pushforward of the structure sheaf of origin $1\in \Gr_{\GL_n}$ is an identity element for $\star$.

The determinant line bundle $\mathcal O(1)$ on $\Gr_{\GL_n}$ \cite[1.5]{zhu2016introduction} is $\GL_n(\mathscr O)\rtimes \mathbb C^{\times}_q$-equivariant. Let us use $\mathcal O_{\Gr^{\omega_1}}(m)$ to denote $i_*i^*\mathcal O(1)^{\otimes m}$ where $i:\Gr^{\omega_1}\hookrightarrow \Gr_{\GL_n}$ is the natural embedding. Since $i$ is $\GL_n(\mathscr O)\rtimes \mathbb C^{\times}_q$-equivariant, $\mathcal O_{\Gr^{\omega_1}}(m)$ is also $\GL_n(\mathscr O)\rtimes \mathbb C^{\times}_q$-equivariant.

For an object $\mathcal F\in D^b_{\GL_n(\mathscr O)\rtimes \mathbb C^{\times}_q}(\Gr_{\GL_n})$, we denote its $\GL_n\times \mathbb C^{\times}_q$-equivariant Euler characteristic by $\chi_{q,\mathbf a}(\mathcal F)\in K_{\GL_n\times \mathbb C^{\times}_q}(\mathrm{pt})=\mathbb Z[\mathbf a_1^\pm,\cdots,\mathbf a_n^\pm]^{S_n}[q^{\pm}]$. The following result states that the operator $\mathcal O_{\Gr^{\omega_1}}(m)\star (-)$ is a geometrization of the Jing operator $S^q_m$.
\begin{proposition}[{see \cite[Corollary B.1]{moosavian2021towards}}]\label{prop:geometric Jing operator}
Let $\mathcal F\in D^b_{\GL_n(\mathscr O)\rtimes \mathbb C^{\times}_q}(\Gr_{\GL_n})$, then 
\begin{align}\label{geometric Jing operator}
    \chi_{q,\mathbf a}(\mathcal O_{\Gr^{\omega_1}}(m)\star \mathcal F)=S^q_m\left(\chi_{q,\mathbf a}(\mathcal F)\right).
\end{align}
\end{proposition}
The idea of proof of Proposition \ref{prop:geometric Jing operator} is to identify the right-hand-side of \eqref{Jing Operator} as the Atiyah-Bott localization formula for $\Gr^{\omega_1}$.

Combine \eqref{geometric Jing operator} with the definition of $H_{\mu}$ in terms of iterative action of $S_{\mu_i}$ \eqref{definition: transformed HL}, and we obtain the following.
\begin{corollary}[{see \cite[Corollary B.2]{moosavian2021towards}}]\label{corollary: character of convolution}
Let $\mu=(\mu_1\ge\cdots \ge\mu_l)\in \bN^l$ be a Young diagram, then
\begin{align}
    H_{\mu}(\mathbf a;q)=\chi_{q,\mathbf a}(\mathcal O_{\Gr^{\omega_1}}(\mu_1)\star\cdots \star\mathcal O_{\Gr^{\omega_1}}(\mu_l)).
\end{align}
\end{corollary}

Finally, the following result is deduced from Corollary \ref{corollary: character of convolution} by taking the spacial case $\mu_1=\cdots=\mu_N=k$.

\begin{corollary}[{see \cite[Corollary B.3]{moosavian2021towards}}]\label{corollary: character of O(k)}
Let $\overline{\Gr}^{N\omega_1}$ be the closure of the $\GL_n(\mathscr O)$-orbit through $z^{N\omega_1}$, then 
\begin{align}
    \chi_{q,\mathbf a}(\overline{\Gr}^{N\omega_1},\mathcal O(k))=H_{(k^N)}(\mathbf a;q).
\end{align}
Here $(k^N)$ is the partition consisting of $N$ copies of $k$.
\end{corollary}

The key idea of the proof of Corollary \ref{corollary: character of O(k)} is to use the isomorphism $m^*\mathcal O(1)\cong\mathcal O(1)\widetilde\boxtimes\cdots \widetilde\boxtimes \mathcal O(1)$ and rationality of the singularities of $\overline{\Gr}^{N\omega_1}$, i.e. $Rm_*\mathcal O\cong \mathcal O$. See \cite[Corollary B.3]{moosavian2021towards} for the details.

\section{Identities}\label{sec:Identities}
Let $f(x,y),g(x,y)$ be two polynomials of non-commutative variables $(x,y)$, then we have
\begin{multline}\label{eq:[X,Y] inside}
    A^a f(X,Y)[X,Y]g(X,Y)B_b\\
    =A^a f(X,Y)B_cA^cg(X,Y)B_b-(k+n)A^a f(X,Y)g(X,Y)B_b,
\end{multline}
\begin{multline}\label{eq:[X,Y] inside_trace}
    \mathrm{Tr}( f(X,Y)[X,Y]g(X,Y))\\
    =\mathrm{Tr}( f(X,Y)B_cA^cg(X,Y))-(k+n)\mathrm{Tr}(  f(X,Y)g(X,Y))-\mathrm{Tr}(f(X,Y))\mathrm{Tr}(g(X,Y)).
\end{multline}
In particular, we have
\begin{align}\label{eq:X^r[X,Y]X^s}
    \mathrm{Tr}( X^r[X,Y]X^s)
    =-\mathrm{Tr}(X^r)\mathrm{Tr}(X^s).
\end{align}

\begin{proof}
For a family of operator $\{\mathcal O^i\}_{1\le i\le N}$ transforming in the $\gl_N$ vector representation, we have \eqref{eq:moment map vector rep op}:
\begin{align*}
    [X,Y]^i_j\mathcal O^j=(BA)^i_j\mathcal O^j-(n+k)\mathcal{O}^i.
\end{align*}
\eqref{eq:[X,Y] inside} is obtained by applying the above equation to $\mathcal O=g(X,Y)B_b$.

For a family of operator $\{\mathcal A^i_j\}_{1\le i,j\le N}$ transforming in the $\gl_N$ adjoint representation, we have:
\begin{align*}
    [\mu^*_{\bC}(E^i_j),\mathcal A^j_l]= N\mathcal A^i_l-\delta^i_l\mathcal A^j_j.
\end{align*}
Using the identity $\mu^*_{\bC}(E^i_j)=(XY)^i_j-(Y X)^i_j-(BA)^i_j+(N+n)\delta^i_j$, we obtain
\begin{align*}
    [X,Y]^i_j\mathcal A^j_l=(BA)^i_j\mathcal A^j_l-(n+k)\mathcal{A}^i_l-\delta^i_l\mathcal A^j_j.
\end{align*}
\eqref{eq:[X,Y] inside_trace} is obtained by applying the above equation to $\mathcal A=g(X,Y)$.
\end{proof}

\bibliographystyle{unsrt}
\bibliography{Bib}

\end{document}